\documentclass{easychair}
\usepackage[table]{xcolor}
\usepackage{amsthm, amssymb, amsmath}
\usepackage{cleveref}
\usepackage{tikz-cd}
\usepackage{array}
\usepackage{mdframed}
\usepackage{thm-restate}
\usepackage{pgf-pie}
\usepackage{enumitem}
\usepackage{underscore}
\usepackage{pgfornament}

\usepackage{afterpage} 
\usepackage{hhline}
\usepackage{multirow}
\usepackage{longtable}
\usepackage{adjustbox}
\usepackage[tableposition=below]{caption}

\usepackage[
    backend=biber
]{biblatex}
\theoremstyle{definition}
\newtheorem{definition}{Definition}[section]
\newtheorem{lemma}{Lemma}
\newtheorem{theorem}{Theorem}

\newtheorem{example}{Example}
\usepackage[T1]{fontenc}
\usepackage[ttdefault=true]{AnonymousPro}

\newcommand{\A}{\ensuremath{\mathcal{A}}}

\newcommand{\B}{\ensuremath{\mathcal{B}}}

\newcommand{\C}{\ensuremath{\mathcal{C}}}

\newcommand{\D}{\ensuremath{\mathcal{D}}}

\newcommand{\E}{\ensuremath{\mathcal{E}}}

\newcommand{\T}{\ensuremath{\mathcal{T}}}

\renewcommand{\S}{\ensuremath{\mathcal{S}}}

\newcommand{\s}{\ensuremath{\sigma}}

\renewcommand{\t}{\ensuremath{\tau}}

\renewcommand{\P}{\ensuremath{\mathcal{P}}}

\newcommand{\N}{\ensuremath{\mathbb{N}_{\omega}}}

\newcommand{\qf}[1]{\ensuremath{QF(#1)}}

\newcommand{\eq}[1]{\ensuremath{Eq(#1)}}

\newcommand{\ax}[1]{\ensuremath{Ax(#1)}}

\newcommand{\adds}[1]{\ensuremath{(#1)^{2}}}

\newcommand{\addf}[1]{\ensuremath{(#1)_{s}}}

\newcommand{\subf}[1]{\ensuremath{\texttt{0}(#1)}}

\newcommand{\subff}[1]{\ensuremath{#1_{\texttt{0}}}}

\newcommand{\addfs}[1]{\ensuremath{s(#1)}}

\newcommand{\subs}[1]{\ensuremath{#1_{1}}}

\newcommand{\minmod}{\ensuremath{\textbf{minmod}}}

\newcommand{\minimal}{\ensuremath{\textbf{minimal}}}

\newcommand{\wrt}{w.r.t.}

\newcommand{\Forall}[1]{\ensuremath{\forall\,#1.\:}}

\newcommand{\Exists}[1]{\ensuremath{\exists\,#1.\:}}

\newcommand{\Teven}{\ensuremath{\T_{even}^{\infty}}}

\newcommand{\Tninfty}{\ensuremath{\T_{n,\infty}}}

\newcommand{\bb}{\ensuremath{\varsigma}}

\newcommand{\Tbb}{\ensuremath{\T_{\bb}}}

\newcommand{\Tone}{\ensuremath{\T_{\mathrm{I}}}}

\newcommand{\overarrow}[1]{\ensuremath{\overrightarrow{#1}}}

\newcommand{\vars}{\ensuremath{vars}}

\newcommand{\Tgeqn}{\ensuremath{\T_{\geq n}}}

\newcommand{\Tinfty}{\ensuremath{\T_{\infty}}}

\newcommand{\Tleqone}{\ensuremath{\T_{\leq 1}}}

\newcommand{\Tleqn}{\ensuremath{\T_{\leq n}}}

\newcommand{\Tmn}{\ensuremath{\T_{\langle m,n\rangle}}}

\newcommand{\Tbbs}{\ensuremath{\T^{s}_{\bb}}}

\newcommand{\Tbbeq}{\ensuremath{\T^{=}_{\bb}}}

\newcommand{\Tbbtwo}{\ensuremath{\T^{2}_{\bb}}}

\newcommand{\Tbbvee}{\ensuremath{\T^{\vee}_{\bb}}}

\newcommand{\Tbbveeeq}{\ensuremath{\T^{=}_{\bb\vee}}}

\newcommand{\Tbbinfty}{\ensuremath{\T^{\infty}_{\bb}}}

\newcommand{\Tonebb}{\ensuremath{\T^{\bb}_{1}}}

\newcommand{\Tbboneneq}{\ensuremath{\T^{\neq}_{\bb,1}}}

\newcommand{\Tbbinftythree}{\ensuremath{\T^{\infty,3}_{\bb}}}

\newcommand{\Tbbneqinfty}{\ensuremath{\T^{\infty}_{\bb\neq}}}

\newcommand{\Tmnbb}{\ensuremath{\T^{\bb}_{m,n}}}

\newcommand{\Tnbb}{\ensuremath{\T^{\bb}_{n}}}

\newcommand{\Tbbneq}{\ensuremath{\T^{\bb}_{n}}}

\newcommand{\Ttwothree}{\ensuremath{\T_{2,3}}}

\newcommand{\addnc}[1]{\ensuremath{(#1)_{\vee}}}

\newcommand{\Tupinfty}{\ensuremath{\T^{\infty}}}

\newcommand{\TM}{\ensuremath{\T_{f}}}

\newcommand{\TsM}{\ensuremath{\T^{s}_{f}}}

\newcommand{\Tmninfty}{\ensuremath{\T^{\infty}_{m,n}}}

\newcommand{\Toneinfty}{\ensuremath{\T^{\infty}_{1}}}

\newcommand{\Ttwoinfty}{\ensuremath{\T^{\infty}_{2}}}

\newcommand{\Toneinftyneq}{\ensuremath{\T_{1,\infty}^{\neq}}}

\newcommand{\Ttwoinftyneq}{\ensuremath{\T_{2,\infty}^{\neq}}}

\newcommand{\Toneodd}{\ensuremath{\T_{1}^{odd}}}

\newcommand{\Toddneq}{\ensuremath{\T_{odd}^{\neq}}}

\newcommand{\Ttwothreethree}{\ensuremath{\T_{2,3}^{3}}}

\newcommand{\Tnequpinfty}{\ensuremath{\T^{\infty}_{\neq}}}

\newcommand{\Ttwo}{\ensuremath{\T_{\mathrm{II}}}}

\newcommand{\Tthree}{\ensuremath{\T_{\mathrm{III}}}}

\newcommand{\wit}{\ensuremath{wit}}

\newcommand{\Tfour}{\ensuremath{\T_{\mathrm{IV}}}}

\newcommand{\Tfive}{\ensuremath{\T_{\mathrm{V}}}}

\newcommand{\Tsix}{\ensuremath{\T_{\mathrm{VI}}}}

\newcommand{\Tthreeone}{\ensuremath{\T^{3}_{\mathrm{I}}}}

\newcommand{\Tseven}{\ensuremath{\T_{\mathrm{VII}}}}

\newcommand{\Teight}{\ensuremath{\T_{\mathrm{VIII}}}}

\newcommand{\Tnine}{\ensuremath{\T_{\mathrm{IX}}}}

\newcommand{\Tten}{\ensuremath{\T_{\mathrm{X}}}}

\newcommand{\Televen}{\ensuremath{\T_{\mathrm{XI}}}}

\newcommand{\Tthreetwo}{\ensuremath{\T^{3}_{\mathrm{II}}}}

\newcommand{\Ttwelve}{\ensuremath{\T_{\mathrm{XII}}}}

\newcommand{\Tthirteen}{\ensuremath{\T_{\mathrm{XIII}}}}

\newcommand{\Tfourteen}{\ensuremath{\T_{\mathrm{XIV}}}}

\newcommand{\Tfifteen}{\ensuremath{\T_{\mathrm{XV}}}}

\newcommand{\Tsixteen}{\ensuremath{\T_{\mathrm{XVI}}}}

\newcommand{\Tseventeen}{\ensuremath{\T_{\mathrm{XVII}}}}

\newcommand{\Teighteen}{\ensuremath{\T_{\mathrm{XVIII}}}}

\newcommand{\smooth}{\textbf{SM}}

\newcommand{\stainf}{\textbf{SI}}

\newcommand{\convex}{\textbf{CV}}

\newcommand{\finwit}{\textbf{FW}}

\newcommand{\strfinwit}{\textbf{SW}}

\newcommand{\stafin}{\textbf{SF}}

\newcommand{\finmodpro}{\textbf{FM}}

\newcommand{\cmmf}{\textbf{CF}}

\newcolumntype{P}[1]{>{\centering\arraybackslash}p{#1}}

\newcommand{\F}{\ensuremath{\mathcal{F}}}

\renewcommand{\S}{\ensuremath{\mathcal{S}}}

\renewcommand{\P}{\ensuremath{\mathcal{P}}}

\newcommand{\NEQ}{\ensuremath{\neq}}

\newcommand{\NNEQ}[1]{\ensuremath{\neq(#1_{1},\ldots,#1_{n})}}

\newcommand{\NNNEQ}[2]{\ensuremath{\neq(#1_{1},\ldots,#1_{#2})}}

\newcommand{\fps}[1]{\ensuremath{\raisebox{.15\baselineskip}{\Large\ensuremath{\wp}}_{fin}(#1)}}

\newcommand{\nocontentsline}[3]{}
\newcommand{\toclesslab}[3]{\bgroup\let\addcontentsline=\nocontentsline#1{#2\label{#3}}\egroup}

\let\tp\texorpdfstring

\newcommand{\mytitle}{Combining Combination Properties: Minimal Models}

\title{\mytitle}
\author{Guilherme V. Toledo and Yoni Zohar}

\institute{
Bar-Ilan University,
Ramat Gan, Israel\\
}

\authorrunning{Toledo and Zohar}
\titlerunning{\mytitle}

\begin{document}

\maketitle

\begin{abstract}
This is a part of an ongoing research project, 
with the aim of finding the connections between 
properties related to theory combination in Satisfiability Modulo Theories.
In previous work, 7 properties were analyzed:
convexity, 
stable infiniteness, 
smoothness, 
finite witnessability, 
strong finite witnessability,
the finite model property, and
stable finiteness.
The first two properties are related to Nelson-Oppen combination,
the third and fourth to polite combination,
the fifth to strong politeness,
and the last two to shininess.
However, the remaining key property of shiny theories, namely,
the ability to compute the cardinalities of minimal models,
was not yet analyzed.
In this paper we study this property and its connection to the others.\footnote{This work was funded by NSF-BSF grant 2020704, ISF grant 619/21, and
the Colman-Soref fellowship.
}
\end{abstract}

\toclesslab\section{Introduction}{sec:introduction}

\newcommand{\question}[1]{\noindent {{\bf Q:} {#1}}\\}
\newcommand{\answer}[1]{\noindent {{\bf A:} {#1}}\vspace{.7em}}
\newcommand{\firstquestion}[1]{\noindent {{\bf Quisani:} {#1}}\\}
\newcommand{\firstanswer}[1]{\noindent {{\bf Author:} {#1}}\vspace{.7em}}

Quisani\footnote{Quisani, formerly known as Quizani, was born in \cite{DBLP:series/wsscs/Gurevich93a}.} and Author meet in the LPAR waiting room. 
Their papers 
are being reviewed. \\

\firstquestion{Again with the combination of properties?}
\firstanswer{Most certainly. We have some new exciting results about this. } 

\question{Now, there is a pun in your titles right?}
\answer{Yes. We deal with combination of properties
of theories. But the properties themselves are studied
in the field of theory combination in Satisfiability Modulo Theories (SMT)~\cite{BSST21}.}

\newcommand{\zzz}{{\bf z3}}
\newcommand{\bzla}{{\bf bitwuzla}}
\newcommand{\yices}{{\bf Yices}}
\newcommand{\cvcfive}{{\bf cvc5}}
\newcommand{\msat}{{\bf MathSAT}}
\newcommand{\python}{{\bf Python}}

\question{Sure, SMT-solvers, like
\cvcfive,
\zzz,
\bzla,
\yices, and
\msat~\cite{cvc5,DBLP:conf/tacas/MouraB08,DBLP:conf/cav/NiemetzP23,DBLP:conf/cav/Dutertre14,DBLP:conf/tacas/CimattiGSS13}.}
\answer{Yes, among others.} 

\question{So what is this theory combination all about?}
\answer{Well, SMT solvers implement decision procedures for various theories.
For example, many of them implement the bit-blasting~\cite{DBLP:series/txtcs/KroeningS16} algorithm for
bit-vectors, 
the simplex algorithm for 
arithmetic~\cite{DBLP:conf/cav/DutertreM06},
or the weak-equivalence algorithm for array formulas~\cite{DBLP:conf/frocos/ChristH15}.
But for many applications, reasoning in just one theory is not enough, and their combination is required.
}

\question{Like, if I have a formula about arrays of integers, for example?}
\answer{Exactly. And in such cases, one would hope to use
the existing decision procedures for arrays and integers,
rather than developing a completely new algorithm for this specific combination.}

\question{I see. I remember reading about this. This is about the Nelson-Oppen method right~\cite{NelsonOppen}?}
\answer{Nelson and Oppen's method is great when the theories
are stably-infinite, roughly meaning that satisfiable formulas, have infinite models. There is
also a variant of the method for convex theories (roughly meaning
that when a disjunction of equalities is implied, one of them is implied).
But some theories do not have these properties, and then this method does not apply.}

\question{That would cause problems when
considering arrays of bit-vectors, instead
of integers.}
\answer{Spot on, as usual. 
Several other combination techniques were introduced since the Nelson-Oppen approach.
Each method replaces the stable-infiniteness requirement by another one.
For example, the polite combination method~\cite{ranise:inria-00000570} requires a stronger property called {\em politeness},
but only from one of the theories. 
Another similar notion is that of
  shininess~\cite{TinZar05}.
}

\question{But these are properties for different combination methods. Why would you combine them?}
\answer{Besides theoretical interest,
this line of research can make it easier to prove that a given theory can be
combined using one of the above methods.
For example, Barrett et al. proved in~\cite{DBLP:conf/frocos/BarrettDS02} that under some conditions,
convexity implies stable-infiniteness.
}

\question{Sounds reasonable. Are there other connections between these  properties?}
\answer{Well, when politeness was introduced~\cite{ranise:inria-00000570}, the authors proved
that it is equivalent to shininess. 
Later, in~\cite{JB10-LPAR}, Jovanovic and Barrett
found a problem in the definition of politeness from~\cite{ranise:inria-00000570}
and corrected it.
Their fixed definition was later proven to be equivalent to shiny theories
for the one-sorted case in~\cite{CasalRasga}, and then for the many-sorted case~\cite{CasalRasga2},
under certain conditions.
\cite{CasalRasga} also named the property with the fixed definition from~\cite{JB10-LPAR} {\em strong politeness}.
In fact, it is possible that some of the authors of these
papers sat right here in this room --
\cite{CasalRasga} and \cite{JB10-LPAR} were published in LPAR 2010 and 2013, respectively.
}

\question{Cool! I wonder which one of them sat on my chair. Well, it's not mine per se, I mean the chair I am sitting on. So with all these results, weren't all the questions answered?}
\answer{Certainly not. To begin with, \cite{JB10-LPAR} did not 
prove that politeness and strong politeness are different.
They only fixed a problem in a proof from~\cite{ranise:inria-00000570},
by changing the definition.
Only later~\cite{SZRRBT-21} it was proven that these notions are indeed
different.
But there is much more into this. 
For example, (strong) politeness is actually a conjunction of two 
primitive properties, namely smoothness and (strong) finite witnessability.
Similarly, shininess is a conjunction of smoothness, stable finiteness (or its weaker form, the finite model property),
and the computability of the function that takes a quantifier-free formula
and returns the minimal cardinalities of its models.
For any combination of these properties (and their negations), it is 
interesting whether there is a theory exhibiting this combination
(for example, 
there are no theories that are smooth but not stably infinite,
but there are stably infinite theories that are not smooth, etc.).}

\question{Wait, so you are considering $2$, $4$, ... what, $256$ possible combinations?! }
\answer{Its actually worse, as some combinations
are possible only for empty signatures,
or for one-sorted and not many-sorted signatures. 
So the total number would be $2$, $4$, ... yeah, $1,024$ combinations. 
Now wait, no, no, no, stop jumping up and down.
It's not as crazy as it sounds. 
In fact, \cite{CADE,FroCoS} handle almost
all these properties. 
They obtained very special theories by use of non-computable functions like the busy beaver function, as unrelated as that may sound. 
The only property left is
the computability of the aforementioned function (the one with the minimal cardinalities and whatnot).
The current paper considers all combinations, including
this missing property. 
Now tell me, how would you like this to proceed?}

\question{Your next section (\Cref{preliminaries})
should review necessary definitions and notions. 
In \Cref{CMMF}, please provide more details on this minimal model function.
Can you clarify what is {\em minimal} about minimal models?
Then you can state your main theorem in \Cref{sec:maintheorem}, identifying all possible and impossible
combinations. 
Give it a nice name. ``\Cref{mainresult}" maybe?
Please split the proof 
to \Cref{relationships} and \Cref{sec:examples}.
In the first, provide all the impossible combinations, and in the second,
all the possible ones.
I wouldn't be surprised if the Busy Beaver function shows up there again, or even a completely new non-computable function.
If you have room, you could end with some concluding remarks 
and future directions in the last section, \Cref{conclusion}.
}
\answer{You got it.}\footnote{Due to lack of space, proofs are omitted and are given in an appendix.}

\toclesslab\section{Preliminaries}{preliminaries}

\toclesslab\subsection{First-Order Many-Sorted Logic}{logic}
We fix a many-sorted signature 
$\Sigma=(\S_{\Sigma},\F_{\Sigma},\P_{\Sigma})$ where: $\S_{\Sigma}$ is a countable set of \emph{sorts}; 
$\F_{\Sigma}$ a set of function symbols, each with an arity $\s_{1}\times\cdots\times\s_{n}\rightarrow\s$ for $\s_{1},\ldots,\s_{n},\s\in\S_{\Sigma}$; 
and 
$\P_{\Sigma}$ a set of predicate symbols, each with an arity $\s_{1}\times\cdots\times\s_{n}$. 
$\P_{\Sigma}$  includes for every $\s\in\S_{\Sigma}$ an equality $=_{\s}$ of arity $\s\times\s$ usually denoted by $=$ when $\s$ is clear. 
$\Sigma$ is \emph{empty} if it has no function or predicate symbols other than the equalities; $\Sigma$ is \emph{one-sorted} if $|\S_{\Sigma}|=1$. 
For each sort we assume a countably infinite set of variables, those sets being disjoint for distinct sorts. 
We then define terms, literals and formulas as usual.
The set of quantifier-free formulas on $\Sigma$ shall be written as $\qf{\Sigma}$.
The set of variables of sort $\s$ in a formula $\varphi$ is denoted by $\vars_{\s}(\varphi)$, and the set of all of its variables by $\vars(\varphi)$.
We will often use the signatures from \Cref{tab:signatures}:
for every $n$, $\Sigma_{n}$ is the empty signature with $n$ sorts,
and $\Sigma^{n}_{s}$ is the $n$-sorted signature
with a single function symbol $s$ whose arity is $\s_{1}\rightarrow\s_{1}$.
We often write $\Sigma_{s}$ instead of $\Sigma_{s}^{1}$.

\begin{table}[t]
\centering
\renewcommand{\arraystretch}{1}
    \begin{tabular}{c|c|c|c}
    Signature & Sorts & Function Symbols & Predicate Symbols\\\hline
         $\Sigma_{n}$ & $\{\s_{1},\ldots,\s_{n}\}$ & $\emptyset$ & $\emptyset$\\
         $\Sigma_{s}^{n}$ & $\{\s_{1},\ldots,\s_{n}\}$ & $\{s:\s_{1}\rightarrow\s_{1}\}$& $\emptyset$ \\
    \end{tabular}
    \renewcommand{\arraystretch}{1}
    \caption{Useful signatures. We often write $\Sigma_{s}$ instead of $\Sigma_{s}^{1}$.}
    \label{tab:signatures}
\end{table}

$\Sigma$-interpretations $\A$
are defined as usual (see, e.g., \cite{Monzano93}).  
$\s^{\A}$ denotes the domain of sort $\s$; 
for a function symbol $f$ and a predicate symbol $P$, $f^{\A}$ and $P^{\A}$ denote the related function and predicate in $\A$; 
for a term $\alpha$,  $\alpha^{\A}$ is its value in $\A$, and if $\Gamma$ is a set of terms, $\Gamma^{\A}=\{\alpha^{\A} : \alpha\in\Gamma\}$. 
If $\A$ satisfies  $\varphi$ we write $\A\vDash\varphi$, and say that $\varphi$ is \emph{satisfiable}.
The formulas in \Cref{card-formulas} will be important in what is to come:
an interpretation $\A$ satisfies 
$\NNEQ{x}$, for $x_{i}$ of sort $\s$ (or even its existential closure $\psi_{\geq n}^{\s}$)
  iff $|\s^{\A}|\geq n$; $\A\vDash\psi_{\leq n}^{\s}$ iff $|\s^{\A}|\leq n$; and $\A\vDash\psi_{=n}^{\s}$ iff $|\s^{\A}|=n$. 
If the signature at hand is one-sorted, we drop $\s$ (e.g., when writing $\psi_{\geq n}$).

\begin{figure}[t]
\begin{mdframed}
\renewcommand{\arraystretch}{1}
\[\NNEQ{x}=\bigwedge_{i=1}^{n-1}\bigwedge_{j=i+1}^{n}\neg(x_{i}=x_{j})\quad\quad\quad 
\psi_{\geq n}^{\s}=\Exists{{x_1\ldots x_n}}\NNEQ{x}\]
\[
\psi_{\leq n}^{\s}=\Exists{x_1,\ldots,x_n}\Forall{y}\bigvee_{i=1}^{n}y=x_{i}\quad\quad\quad \psi_{=n}^{\s}=\psi_{\geq n}^{\s}\wedge\psi_{\leq n}^{\s}\]
\end{mdframed}
\caption{Cardinality formulas. All of $x_{i}$, and $y$ are of sort $\s$.}
\label{card-formulas}
\end{figure}

 A \emph{theory} $\T$ is a class of interpretations (called $\T$-\emph{interpretations}, or the {\em models} of $\T$ when disregarding variables) 
 comprised of all the interpretations that 
 satisfy a set $\ax{\T}$ called the \emph{axiomatization} of $\T$; 
A formula is ($\T$-)\emph{satisfiable} if it is satisfied by a 
($\T$-)interpretation. 
Two formulas are ($\T$-)\emph{equivalent} if they are satisfied by the 
same ($\T$-)interpretations.
A formula $\varphi$ is $\T$-\emph{valid}, denoted $\vDash_{\T}\varphi$, if $\A\vDash\varphi$ for all $\T$-interpretations $\A$.
The following are many-sorted generalizations of the L{\"o}wenheim-Skolem and compactness theorems (see \cite{Monzano93,TinZar}).

\smallskip
\begin{theorem}\label{LowenheimSkolem}
    Let $\Sigma$ be a first-order, many-sorted signature; if a set of $\Sigma$-formulas $\Gamma$ is satisfiable, then there exists an interpretation $\A$ that satisfies $\Gamma$ where $|\s^{\A}|\leq\aleph_{0}$ for all $\s\in\S_{\Sigma}$.
\end{theorem}

\begin{theorem}\label{compactness}
    Let $\Sigma$ be a first-order, many-sorted signature; then a set of $\Sigma$-formulas $\Gamma$ is satisfiable if, and only if, each finite subset $\Gamma_{0}\subseteq\Gamma$ is satisfiable.
\end{theorem}

\medskip
\toclesslab\subsection{Theory Combination Properties}{modeltheory}

In what follows, 
$\Sigma$ denotes an arbitrary signature, 
$\T$ a $\Sigma$-theory
and 
$S\subseteq\S_{\Sigma}$.

\smallskip
\noindent
{\bf{Stable infiniteness  and smoothness }}
$\T$ is \emph{stably infinite}
\cite{NelsonOppen}
w.r.t. $S$ if, for every quantifier-free 
$\T$-satisfiable
formula $\phi$, there is a $\T$-interpretation $\A$ satisfying $\phi$ with $|\s^{\A}|\geq\aleph_{0}$ for each $\s\in S$. 
$\T$ is \emph{smooth}~\cite{TinZar05,RanRinZar} w.r.t. $S$
if, for every quantifier-free formula $\phi$, $\T$-interpretation $\A$ satisfying $\phi$, and function $\kappa$ from $S$ to the class of cardinals with $\kappa(\s)\geq|\s^{\A}|$ for each $\s\in S$, there is a $\T$-interpretation $\B$ satisfying $\phi$ with $|\s^{\B}|=\kappa(\s)$ for each $\s\in S$.

\smallskip
\noindent
{\bf Finite witnessability and strong finite witnessability }
Given a finite set of variables $V=\bigcup_{\s\in S}V_{\s}$, for $V_{\s}$ the subset of $V$ of variables of sort $\s$, and equivalence relations $E_{\s}$ on $V_{\s}$ whose union (also an equivalence relation) 
we denote by $E$, we define the formula 
$\delta_{V}^{E} := \bigwedge_{\s\in S}\big[\bigwedge_{xE_{\s}y}(x=y)\wedge\bigwedge_{x\not{E_{\s}}y}\neg(x=y)\big]$.
We then call $\delta_{V}^{E}$ an {\em arrangement} of $V$ and
often denote it by $\delta_{V}$ when $E$ is clear from the context.

$\T$ is \emph{finitely witnessable} 
\cite{RanRinZar} 
w.r.t. $S$ if there is a computable function $\wit$ (called a witness) from $\qf{\Sigma}$ into itself that satisfies: 
$(i)$~for any quantifier-free formula $\phi$, $\phi$ and $\Exists{\overarrow{x}}\wit(\phi)$ for $\overarrow{x}=\vars(\wit(\phi))\setminus\vars(\phi)$, are $\T$-equivalent; and 
$(ii)$~if $\wit(\phi)$ is $\T$-satisfiable, then there exists a $\T$-interpretation $\A$ satisfying it with $\s^{\A}=\vars_{\s}(\wit(\phi))^{\A}$ for each $\s\in S$. 

\emph{Strong finite witnessability}~\cite{JB10-LPAR} 
w.r.t. $S$ 
is defined similarly, replacing $(ii)$ by:
$(ii')$~given a finite set of variables $V$ and arrangement $\delta_{V}$ on $V$, if $\wit(\phi)\wedge\delta_{V}$ is $\T$-satisfiable, then there is a $\T$-interpretation $\A$ satisfying it  with $\s^{\A}=\vars_{\s}(\wit(\phi)\wedge\delta_{V})^{\A}$ for each $\s\in S$.
$\T$ is (strongly) polite w.r.t. $S$ if it is both smooth and (strongly) finitely witnessable w.r.t. $S$.

\smallskip
\noindent
{\bf Convexity }
 $\T$ is \emph{convex} 
 \cite{NelsonOppen}
 w.r.t. $S$ if 
$\vDash_{\T}\phi\rightarrow\bigvee_{i=1}^{n}x_{i}=y_{i}$,
for $\phi$ a conjunction of literals and $x_{i}$ and $y_{i}$ variables of sorts in $S$, implies that $\vDash_{\T}\phi\rightarrow(x_{i}=y_{i})$ for some $1\leq i\leq n$.

\smallskip
\noindent
{\bf Finite model property  and stable finiteness }
 $\T$ has the \emph{finite model property}
 \cite{CasalRasga}
 w.r.t. $S$ if, for every quantifier-free formula $\phi$ that is $\T$-satisfiable, there exists a $\T$-interpretation $\A$ that satisfies $\phi$ with $|\s^{\A}|<\aleph_{0}$ for all $\s\in S$. 
 %
 A theory is \emph{stably finite} 
 \cite{TinZar05}
 w.r.t. $S$ if, for every quantifier-free formula $\phi$ and $\T$-interpretation $\A$ that satisfies $\phi$, there exists a $\T$-interpretation $\B$ that satisfies $\phi$ with $|\s^{\B}|<\aleph_{0}$ and $|\s^{\B}|\leq |\s^{\A}|$ for every $\s\in S$.

\smallskip
\noindent
{\bf  Minimal model function }
Suppose $S$ is finite, and consider the set $\N := \mathbb{N}\cup\{\aleph_{0}\}$. 
A \emph{minimal model function}~\cite{TinZar05,RanRinZar,CasalRasga2} of  $\T$ w.r.t. $S$ is a function\footnote{Here, as is usual in set-theoretic notation: $\N^{S}$ is the set of functions $n:S\rightarrow\N$, themselves also denoted by $(n_{\s})_{\s\in S}$ where $n(\s)=n_{\s}$; and $\fps{X}$ is the set of finite subsets of $X$.}
\[\minmod_{\T,S}:\qf{\Sigma}\rightarrow\fps{\N^{S}}\]
such that, if $\phi$ is a quantifier-free, $\T$-satisfiable formula, then $(n_{\s})_{\s\in S}\in\minmod_{\T,S}(\phi)$ if, and only if, the following holds: first, there exists a $\T$-interpretation $\A$ that satisfies $\phi$ with $|\s^{\A}|=n_{\s}$ for each $\s\in S$; and second, if $\B$ is a $\T$-interpretation that satisfies $\phi$ with $(|\s^{\B}|)_{\s\in S}\neq (n_{\s})_{\s\in S}$, then there exists $\s\in S$ such that $n_{\s}<|\s^{\B}|$.\footnote{We can use $\N$ instead of the class of all cardinals thanks to \Cref{LowenheimSkolem}.}

When $S$ equals the set, for example, $\{\s_{1},\ldots,\s_{n}\}$, we will denote an element $(n_{\s})_{\s\in S}$ simply by $(n_{\s_{1}},\ldots,n_{\s_{n}})$, identifying $\N^{S}$ with $\N^{n}$.
When
$\Sigma$ is one-sorted, then $\minmod_{\T,S}(\phi)$ has precisely one element, and so we can identify the output of $\minmod_{\T,S}$, if not empty, with an element of $\N$.

\newcommand{\tuf}{\T_{{\mathsf UF}}}

\bigskip
\toclesslab\section{On Minimal Models}{CMMF}

In this section we analyze some of the characteristics of minimal model functions.
We start by noticing that, if $\phi$ is $\T$-satisfiable, there is a unique possibility for the set $\minmod_{\T,S}(\phi)$, given the bi-implication in its definition. It therefore follows that two minimal model functions always agree on $\T$-satisfiable inputs; the output can vary, however, on $\T$-unsatisfiable formulas.
We continue by showing that for satisfiable formulas, $\minmod_{\T,S}(\phi)$ is indeed a finite subset of $\N^{S}$, using
the following variant of
Dickson's Lemma~\cite{Dickson}.\footnote{We thank Benjamin Przybocki for pointing out this lemma to us.}\footnote{Dickson's original result simply exchanges $\N$ for $\mathbb{N}$ in \Cref{technical lemma}.}


\smallskip

\begin{restatable}{lemma}{dickson}\label{technical lemma}\label{Dickson}
    Let $n$ be a natural number, and consider any subset $A$ of $\N^{n}$ equipped with the order such that $(p_{1},\ldots,p_{n})\leq (q_{1},\ldots,q_{n})$ iff $p_{i}\leq q_{i}$ for all $1\leq i\leq n$: then $A$ possesses at most a finite number of minimal elements under this order.
\end{restatable}

\begin{restatable}{proposition}{uniqueandfinite}\label{unique and finite}
     For every $\Sigma$-theory $\T$, $S\subseteq\S_{\Sigma}$, and quantifier-free $\phi$, the subset $X$ of $\N^{S}$ is finite, where $(n_{\s})_{\s\in S}\in X$ iff: there is a $\T$-interpretation $\A$ that satisfies $\phi$ with $(|\s^{\A}|)_{\s\in S}=(n_{\s})_{\s\in S}$; if $\B$ is a $\T$-interpretation that satisfies $\phi$ with $(|\s^{\B}|)_{\s\in S}\neq (n_{\s})_{\s\in S}$, then there is a $\s\in S$ such that $n_{\s}<|\s^{\B}|$.
\end{restatable}

The following result thus neatly explains the choice of nomenclature for a minimal model function vis-{\`a}-vis its definition: it is so well known that its proof became rather elusive.

\smallskip
\begin{restatable}{proposition}{alternativedefinitionminmod}\label{alternative definition}
Take a quantifier-free formula $\phi$ and consider the set 
\[\textbf{Card}_{\T,S}(\phi)=\{(|\s^{\A}|)_{\s\in S} : \text{$\A$ is a $\T$-interpretation that satisfies $\phi$}\},\]
together with the partial order such that $(|\s^{\A}|)_{\s\in S}\leq(|\s^{\B}|)_{\s\in S}$ iff $|\s^{\A}|\leq |\s^{\B}|$ for all $\s\in S$; then, if $\phi$ is $\T$-satisfiable, $\minmod_{\T,S}(\phi)$ equals the set of $\leq$-minimal elements of $\textbf{Card}_{\T,S}(\phi)$.
\end{restatable}

If one looks at how a minimal model function is used in \cite{TinZar05,CasalRasga,CasalRasga2}, two differences become noticeable: first, that its codomain is taken to be $\fps{\mathbb{N}^{S}}$
rather than $\fps{\N^{S}}$; and second, that its domain is taken to be the subset of $\T$-satisfiable elements of $\qf{\Sigma}$, rather than $\qf{\Sigma}$.
In the first case, the difference boils down to assuming stable finiteness,
as we show in the next proposition.
We, however, do not make this assumption, as we are interested
in the characteristics of the separate properties,
and the connections between them.

\begin{restatable}{proposition}{CMMFplusFMP}\label{CMMF+FMP}
Let $\T$ be a theory, and $S$ a set of its sorts:
\begin{enumerate}
    \item $n_{\s}$ is in $\mathbb{N}$ for all quantifier-free, $\T$-satisfiable formulas $\phi$, some $(n_{\t})_{\t\in S}\in\minmod_{\T,S}(\phi)$, and all $\s\in S$, iff $\T$ has the finite model property with respect to $S$.
    \item  $n_{\s}$ is in $\mathbb{N}$ for all quantifier-free, $\T$-satisfiable formulas $\phi$, $(n_{\t})_{\t\in S}\in\minmod_{\T,S}(\phi)$, and $\s\in S$, iff $\T$ is stably finite with respect to $S$.
\end{enumerate}
\end{restatable}

Notice that while in the first item of \Cref{CMMF+FMP} we quantify existentially over the elements of $\minmod_{\T,S}(\phi)$, on the second this is done universally.

\newcommand{\MM}{\ensuremath{mm}}

As for the domain of $\minmod_{\T,S}$, notice that typically, the minimal model function is only considered for decidable theories (often implicitly), something that we do not wish to assume \textit{a priori} here. 
But, if decidability is assumed, both notions are one and the same.

\bigskip
\toclesslab\section{The Main Theorem}{sec:maintheorem}

\question{So, I am assuming that your main theorem provides a complete characterization of the possible and impossible combinations of properties?}
\answer{Not exactly.
Already in \cite{CADE}, 
one combination was problematic: 
no theories that are stably infinite and strongly finitely witnessable but not smooth were found, nor proven
not to exist. }

\question{Aha. These are the famous \emph{unicorn theories}, right?}
\answer{Yes. They were named that way because such theories were never seen, and were conjectured not to exist.
We are currently working on resolving this conjecture, using other techniques.
}

\question{And in the current paper, are you able to determine all the remaining combinations?}
\answer{Almost, except for two, that we also conjecture to be impossible (and feel that a proof would also lie beyond the scope and techniques of the present analysis).
}

\question{So I guess now you will be back to the paper and present the new conjectures, as well as the main result?}
\answer{Yes, if you do not mind.}\vspace{-.5em}\\--- 

We first make the following conjecture regarding two combinations:

\newtheorem{conjecture}{Conjecture}
\Crefname{conjecture}{Conjecture}{Conjecture}

\smallskip
\begin{definition}~
A {\em unicorn $2.0$} theory is
strongly finitely witnessable with no computable minimal model function.
%
A {\em unicorn $3.0$}
theory is
polite (smooth, and finitely witnessable) and shiny (smooth, stably finite, and has a computable minimal model function), but is not strongly polite (smooth, and strongly finitely witnessable).

\end{definition}

\begin{conjecture}
\label{conj:nouni23}
There are no
Unicorn $2.0$ and Unicorn $3.0$ theories.
\end{conjecture}

Notice that, according to \cite{CasalRasga,CasalRasga2}, unicorn $2.0$ theories that are smooth cannot have a decidable quantifier-free satisfiability problem; the same, however, is not necessarily true for the non-smooth case, or for unicorn $3.0$ theories. 

\renewcommand{\arraystretch}{1}
\begin{table}[t]
\centering
\adjustbox{scale=0.8,center}{\begin{tabular}{l|P{6.5cm}|P{6.5cm}}
& One-sorted & General\\
\hline
\multirow{6}{*}{\raisebox{0.2\height }{\rotatebox[origin=c]{90}{Empty}}} & 
\vspace{.01em}
$\color{red} \overline{\stainf}+\overline{\finwit}$ & 
\vspace{.01em}
$\color{red} \stainf+\overline{\convex}$\\
& $\color{red} \overline{\stainf}+\overline{\strfinwit}+\convex$ & $\color{blue} \smooth+\overline{\finwit}+\finmodpro$\\
&  & $\color{blue} 
\smooth+\overline{\strfinwit}+\stafin$\\
&$\overline{\stainf}+\overline{\cmmf}$ (\Cref{-SI+ES=>CMMF})& $\color{blue}\overline{\stainf}+\convex+\finmodpro+\overline{\stafin}+\Sigma_{2}$ \\
& $\overline{\finmodpro}+\overline{\cmmf}$ (\Cref{-FMP+ES+OS=>CMMF}) & $\smooth+\overline{\cmmf}$ (\Cref{SM+ES=>CMMF}) \\
&& $\overline{\stainf}+\convex+\overline{\finmodpro}+\overline{\cmmf}+\Sigma_{2}$ (\Cref{CV+-SI+-FMP+-CMMF})\\\hline
\multirow{6}{*}{\raisebox{0.2\height }{\rotatebox[origin=c]{90}{General}}} & 
\vspace{.01em}
$\color{red} \smooth+\finwit+\overline{\strfinwit}$ & 
\vspace{.01em}
$\color{red} \smooth+\overline{\stainf}$\\
&& $\color{red}\finwit+\overline{\strfinwit}$ \\
& $\color{red}\stainf+\overline{\smooth}+\strfinwit$ & $\color{blue}\overline{\finmodpro}+\stafin$\\
&& $\color{blue}\finwit+\overline{\finmodpro}$ \\
& $\color{blue} \finmodpro+\overline{\stafin}$ & $\color{blue}\strfinwit+\overline{\stafin}$\\
&& $\overline{\finwit}+\finmodpro+\cmmf$ (\Cref{CMMF+FM=>FW})\\
\end{tabular}}
\caption{
Impossible {\bf\color{red}red} combinations
were proven in~\cite{CADE} and
{\bf\color{blue}blue} 
in \cite{FroCoS}.
{\bf Black} combinations are proven in the current paper. 
}\label{table of impossibilities}
\end{table}
\renewcommand{\arraystretch}{1.0}

We may now state the following theorem, which refers to
\Cref{table of impossibilities}. 
We use abbreviated names of the properties: $\stainf$ for stably infinite, $\smooth$ for smooth, $\finwit$ for finitely witnessable, $\strfinwit$ for strongly finitely witnessable, $\convex$ for convex, $\finmodpro$ for finite model property, $\stafin$ for stably finite, and $\cmmf$ for computability of a minimal model function.
A line $\overline{X}$ over a property $X$ indicates its negation.
The table lists all the impossible combination of properties
found in~\cite{CADE} ({\bf\color{red}red}), in~\cite{FroCoS} ({\bf\color{blue}blue}), and the current paper ({\bf black}).
 A list of properties (or their negations) separated by plus signs indicates
 that the combination is impossible. 
 It is partitioned according to the complexity of the signatures:
 in the first column appear combinations
 that are impossible in one-sorted signatures, while
 the results of the second column hold generally, regardless
 of the number of sorts.
 Similarly, the combinations of the first row are only 
 proved for empty signatures, while in the second row they are general.
For example, the first line in the top left square
means that there are no theories over an empty one-sorted signature that are neither stably infinite
nor finitely witnessable. This was proved in \cite{CADE}.
Two entries in the table include $\Sigma_{2}$,
which means that the corresponding combination
is only impossible in $\Sigma_{2}$-theories.

With Conjecture \ref{conj:nouni23} and \Cref{table of impossibilities} in place, we can now state the following theorem, which summarizes the resutls from
\cite{CADE,FroCoS} and the present paper.

\begin{restatable}{theorem}{mainresult}\label{mainresult}
    Disregarding unicorn, unicorn $2.0$ and unicorn $3.0$ theories, a combination of properties 
    is impossible if, and only if, it occurs in \Cref{table of impossibilities}.
\end{restatable}

The ``if" part of \Cref{mainresult} is proved
by considering each combination that appears in the table separately.
The combinations that do not involve
the computability of a minimal model function were proven to be impossible
in~\cite{CADE,FroCoS}.
In \Cref{relationships}, we
consider the combinations that include this property.
In particular, 
\Cref{CMMF+FM=>FW,SM+ES=>CMMF,-SI+ES=>CMMF,CV+-SI+-FMP+-CMMF,-FMP+ES+OS=>CMMF}
include the precise formulations of these results.
The ``only if" part is proved by providing
examples for all combinations not mentioned in \Cref{table of impossibilities}.
Examples 
without conmputability of a minimal model function
were given in~\cite{CADE,FroCoS}.
In \Cref{sec:examples}, we determine this property
for each of the examples from \cite{CADE,FroCoS}, and also provide
new examples for the remaining combinations.

\bigskip
\toclesslab\section{Proof of \Cref{mainresult}: The Impossible Cases}{relationships}
In this section, we prove the impossibilty of some combinations of
the computability of a minimal model function with other properties that 
are related to theory combination.
The most general result is not restricted to any type of signature,
and is presented in \Cref{without}.
The other results, in \Cref{with}, hold only for empty signatures.
These results, along with results of previous work on the subject of combination of
properties, are summarized in \Cref{sec:sum}.

\medskip
\toclesslab\subsection{General Signatures}{without}

The following theorem holds for any signature,
and states that the computability of a minimal model function,
together with the finite model property,
 imply finite witnessability.


\smallskip
\begin{restatable}{theorem}{CMMFplusFMimpliesFW}\label{CMMF+FM=>FW}
    If $\T$ has a computable minimal model function and the finite model property with respect to a finite set $S$, then $\T$ is finitely witnessable with respect to $S$.
\end{restatable}

\begin{proof}[Proof sketch]
Take a quantifier-free formula $\phi$. 
As $\T$ has the finite model property, by \Cref{CMMF+FMP} there either is a minimal $\T$-interpretation $\A$ satisfying $\phi$ which has $(|\s^{\A}|)_{\s\in S}$ in $\mathbb{N}^{S}$, or $\phi$ is not $\T$-satisfiable. We can then produce $wit(\phi)$ by considering the conjunction of $\phi$ and a tautology involving $|\s^{\A}|$ many variables of sort $\s$ for a minimal $\T$-interpretation of $\phi$ in the first case, and on the second $wit(\phi)=\phi$: the resulting formula is obviously equivalent to $\phi$, can be found computably, and satisfies the witnessability property that a witness should. 
\end{proof}

\medskip
\toclesslab\subsection{Empty Signatures}{with}

We start with a superficially unexpected result: for empty signatures, smoothness implies the computability of a minimal model function. 

\smallskip
\begin{restatable}{theorem}{SMplusESimpliesCMMF}\label{SM+ES=>CMMF}
    If $\T$ is a $\Sigma_{n}$-theory smooth with respect to $\S_{\Sigma}$, then $\T$ has a computable minimal model function with respect to any $S\subseteq \S_{\Sigma}$.
\end{restatable}

Unexpected as smoothness is often considered a tool to increase models, while the computability of a minimal model function should decrease them. 
Only superficially as the proof is rather straightforward, although long,
with the following idea and crucial use of Dickson's lemma (\Cref{Dickson}):

\begin{proof}[Proof sketch]
On empty signatures, we can shrink interpretations until we reach minimal ones. 
Computability comes from there being only finitely many such minimal interpretations.
\end{proof}

In \Cref{SM+ES=>CMMF}, one actually needs smoothness w.r.t. all sorts to prove the computability of a minimal model function with respect to any set of sorts, as shown in the next example:

\begin{example}\label{counterexample 1}
   Take any increasing non-computable function $h:\mathbb{N}\setminus\{0\}\rightarrow\mathbb{N}\setminus\{0\}$ and consider the $\Sigma_{2}$-theory $\T$ with axiomatization
   $\{\psi^{\s_{2}}_{=n}\rightarrow \psi^{\s_{1}}_{\geq h(n)} : n\in\mathbb{N}\setminus\{0\}\}$.
   It is smooth with respect to $\{\s_{1}\}$, but not $\{\s_{1},\s_{2}\}$, and so \Cref{SM+ES=>CMMF} does not apply. 
   And indeed,
   due to the fact that $h$ is increasing,
   $h(n)=\minmod_{\T,\{\s_{1}\}}(\NNEQ{u})$,
   for $u_{i}$ of sort $\s_{2}$, and it is clear that, if $\T$ has a computable minimal model function with respect to $\{\s_{1}\}$, $h$ should be itself computable, leading to a contradiction.
\end{example}

The next theorem shows that a $\Sigma_{n}$-theory which is not stably infinite
w.r.t. none of its sorts (as singletons),
must have a computable minimal model function w.r.t. any subset of the sorts.

\begin{restatable}{theorem}{minusSIplusESimpliesCMMF}
\label{-SI+ES=>CMMF}
    If $\T$ is a $\Sigma_{n}$-theory that is not stably infinite with respect to 
    any $\s\in \S_{\Sigma_{n}}$, then $\T$ has a computable minimal model function with respect to any $S\subseteq \S_{\Sigma_{n}}$.
\end{restatable}

\begin{proof}[Proof sketch]
If $\T$ is not stably infinite w.r.t. any $\s\in S$, that means that there is only a finite set of possible finite cardinalities for $\s^{\A}$, for a $\T$-interpretation $\A$ that satisfies a formula $\phi$ and $\s\in S$; indeed, if there were infinitely many possible cardinalities, the pigeonhole principle would guarantee that there are infinite possible values $|\s_{0}^{\A}|$ for some $\s_{0}\in S$, and by \Cref{compactness} 
we would get a $\T$-interpretation $\A$ with $\s_{0}^{\A}$ infinite, 
contradicting the fact that $\T$ is not stably infinite w.r.t. $\s_{0}$,
as the fact that $\Sigma$ is empty implies that interpretations are determined by their cardinalities.  
A minimal model function w.r.t. $S$ can be calculated on $\phi$ by simply checking which of these finitely many interpretations satisfy $\phi$, what is of course computable. 
\end{proof}

The following two results are more restrictive than the previous two, not only demanding the signatures to be empty, but also with low numbers of sorts. 
%

\begin{restatable}{theorem}{minusFMPplusESplusOSimpliesCMMF}\label{-FMP+ES+OS=>CMMF}
    A $\Sigma_{1}$-theory without the finite model property \wrt ~its only sort has a computable minimal model function, also \wrt ~its only sort.
\end{restatable}

\begin{proof}[Proof sketch]
The proof 
is similar to that of \Cref{-SI+ES=>CMMF}. 
If $\T$ is a theory over the one-sorted, empty signature, without the finite model property, it has only finitely many finite interpretations up to isomorphism; indeed, were there infinitely many of them, their cardinalities would be unbounded, and then for any quantifier-free formula $\phi$ and $\T$-interpretation $\A$ that satisfies $\phi$, any finite $\T$-interpretation $\B$ with $|\s_{1}^{\B}|\geq \vars(\phi)^{\A}$ could be modified to satisfy $\phi$, giving us the finite model property (contradiction). A minimal model function of $\phi$ can then be simply calculated by checking which of these interpretations satisfy $\phi$, what is computable.
\end{proof}

\begin{restatable}{theorem}{lostofassumptions}\label{CV+-SI+-FMP+-CMMF}
    A convex $\Sigma_{2}$-theory admits at least one of the following properties: stable infiniteness, finite model property, or a computable minimal model function, all w.r.t. $\{\s_{1},\s_{2}\}$.
\end{restatable}

\begin{proof}[Proof sketch]
If $\T$ is not stably infinite, 
one can use \Cref{compactness} to show that
there exist $k_{1},k_{2}\in\mathbb{N}$ such that no $\T$-interpretation $\A$ has $(|\s_{1}^{\A}|,|\s_{2}^{\A}|)>(k_{1},k_{2})$. 
Then, We can take $k_{1}+1$ variables of sort $\s_{1}$ and $k_{2}+1$ 
variables of sort $\s_{2}$, and write a disjunction of those that is valid by the pigeonhole principle, meaning we must have either $|\s_{1}^{\A}|=1$ for all $\T$-interpretations, or $|\s_{2}^{\A}|=1$, so that $\T$ remains convex. W.l.o.g, we assume the latter. 
Assuming that a minimal model function of $\T$ is not computable, 
we get that there are infinitely many non-isomorphic $\T$-interpretations $\A$ with $|\s_{1}^{\A}|$ countable: otherwise we could enumerate the cardinalities of countable $\T$-interpretations as $\{(m_{1},1),\ldots,(m_{n},1),(\aleph_{0},1)\}$ and straightforwardly obtain a computable minimal model function. 
It follows that $\T$ has the finite model property, since given a quantifier-free formula $\phi$ and a $\T$-interpretation $\A$ where it is satisfied, we can find another $\T$-interpretation $\B$ with $\s_{1}^{\B}$ finite and greater than or equal to $|\vars_{\s}(\phi)^{\A}|$.
\end{proof}

\smallskip
\toclesslab\subsection{Summary}{sec:sum}

\Crefname{theorem}{Thm.}{Thms.}

\begin{figure}[t]
\centering
\begin{minipage}{0.19\textwidth}
\centering
\adjustbox{scale=0.7,center}{\begin{tikzpicture}[scale=0.7]
\def\firstcircle{(-0.75,0) coordinate (a) circle (1.2cm)}
\def\secondcircle{(0.75,0) coordinate (b)  circle (1.2cm)}
\def\firstellipse{(0,0) ellipse (0.8cm and 2.36cm)}
\draw \firstcircle;
\draw \secondcircle;
\draw \firstellipse;
\node[label={$\cmmf$}] (B) at (-1.4,-0.3) {};
\node[label={$\finmodpro$}] (B) at (1.4,-0.3) {};
\node[label={$\finwit$}] (B) at (0.0,1.0) {};
    \end{tikzpicture}}
    
    \Cref{CMMF+FM=>FW}
\end{minipage}
\begin{minipage}{0.19\textwidth}
\centering
\adjustbox{scale=0.7,center}{\begin{tikzpicture}[scale=0.7]
\def\firstcircle{(0,0) coordinate (a) circle (2.36cm)}
\def\secondcircle{(0,0) coordinate (b)  circle (1.18cm)}
\draw \firstcircle;
\draw \secondcircle;
\node[label={$\smooth$}] (B) at (0,-0.4) {};
\node[label={$\cmmf$}] (B) at (0,-2.1) {};
    \end{tikzpicture}}

    \vspace{0.75mm}
    
    \Cref{SM+ES=>CMMF} ($\Sigma_{n}$)
\end{minipage} 
\begin{minipage}{0.19\textwidth}
\centering
\adjustbox{scale=0.7,center}{\begin{tikzpicture}[scale=0.7]
\def\firstrectangle{(-2,-2) rectangle (2, -0.5)}
\def\secondrectangle{(-2,-0.5) rectangle (2, 0.5)}
\def\thirdrectangle{(-2,0.5) rectangle (2, 2)}
\draw \firstrectangle;
\draw \secondrectangle;
\draw \thirdrectangle;
\node[label={$\stainf\quad \&\quad \overline{\cmmf}$}] (B) at (0,0.7) {};
\node[label={$\stainf\quad \&\quad \cmmf$}] (B) at (0,-0.6) {};
\node[label={$\overline{\stainf}\quad \&\quad \cmmf$}] (B) at (0,-1.7) {};
    \end{tikzpicture}}

\vspace{3.5mm}
    
    \Cref{-SI+ES=>CMMF} ($\Sigma_{n}$) 
\end{minipage}
\begin{minipage}{0.19\textwidth}
\centering
\adjustbox{scale=0.7,center}{\begin{tikzpicture}[scale=0.7]
\def\firstrectangle{(-2,-2) rectangle (2, -0.5)}
\def\secondrectangle{(-2,-0.5) rectangle (2, 0.5)}
\def\thirdrectangle{(-2,0.5) rectangle (2, 2)}
\draw \firstrectangle;
\draw \secondrectangle;
\draw \thirdrectangle;
\node[label={$\finmodpro\quad \&\quad \overline{\cmmf}$}] (B) at (0,0.7) {};
\node[label={$\finmodpro\quad \&\quad \cmmf$}] (B) at (0,-0.6) {};
\node[label={$\overline{\finmodpro}\quad \&\quad \cmmf$}] (B) at (0,-1.7) {};
    \end{tikzpicture}}

    \vspace{3.5mm}
    
    \Cref{-FMP+ES+OS=>CMMF} ($\Sigma_{1}$)
\end{minipage}
\begin{minipage}{0.19\textwidth}
\centering
\adjustbox{scale=0.7,center}{\begin{tikzpicture}[scale=0.7]
\def\firstellipse{(0,0.0) ellipse (0.8cm and 2.3cm)}
\draw \firstellipse;
\begin{scope}[shift={(0,0)},rotate=120]
\draw \firstellipse;
\end{scope}
\begin{scope}[shift={(0,0)},rotate=240]
\draw \firstellipse;
\end{scope}
\draw (0,1.2) -- (1.04,-0.6) -- (-1.04,-0.6) -- cycle;
\node[label={$\cmmf$}] (B) at (-1.5,-1.3) {};
\node[label={$\finmodpro$}] (B) at (1.5,-1.3) {};
\node[label={$\stainf$}] (B) at (0.0,1.3) {};
\node[label={$\convex$}] (B) at (0.0,-0.2) {};
    \end{tikzpicture}}
    
    \Cref{CV+-SI+-FMP+-CMMF} ($\Sigma_{2}$)
\end{minipage}

\caption{Venn diagrams for \Cref{CMMF+FM=>FW,SM+ES=>CMMF,-SI+ES=>CMMF,CV+-SI+-FMP+-CMMF,-FMP+ES+OS=>CMMF}
}\label{fig: venn}
\end{figure}

\crefname{theorem}{theorem}{theorems}
\Crefname{theorem}{Theorem}{Theorems}

In \Cref{fig: venn} we represent the results of this section through Venn diagrams. 
For \Cref{CMMF+FM=>FW} the diagram is straightforward: 
the intersection of theories with a computable minimal model function and theories with the finite model property lies inside the domain of finitely witnessable theories. \Cref{SM+ES=>CMMF} is also quite clear, but we must restrict ourselves to empty signatures. In \Cref{CV+-SI+-FMP+-CMMF} the restriction is to the signature $\Sigma_{2}$ alone, and we must choose different shapes for our regions given bi-dimensional limitations. 
Meanwhile, for \Cref{-SI+ES=>CMMF}, 
we not only restrict ourselves to empty signatures, but as we are dealing with the negation of a property we represent all $\Sigma_{n}$-theories as the entire square. 
Notice that theories that are neither stably infinite nor have a computable
minimal model function are absent from the square.
For \Cref{-FMP+ES+OS=>CMMF} the square represents all $\Sigma_{1}$-theories.


\bigskip
\toclesslab\section{Proof of \Cref{mainresult}: The Possible Cases}{sec:examples}
\newcommand{\yz}[1]{{\color{red}YZ: {#1}}}

The  proof of the
``only if" part of \Cref{mainresult} amounts to providing examples of theories that admit the
combinations that are absent from \Cref{table of impossibilities}.
This part of the proof is quite long:
our analysis takes into account $10$ binary properties: stable infiniteness, smoothness, finite witnessability, strong finite witnessability, convexity, finite model property, stable finiteness and computability of a minimal model function, 
as well as emptiness and one-sortedness of the signature. That adds up to $1,024$ possible combinations, 
$884$ of which were proven to be impossible in \cite{CADE,FroCoS} and \Cref{relationships}, and
$20$ of which remain open on whether they are possible or not (corresponding to unicorn theories, as well as unicorn theories $2.0$ and $3.0$). 
Thus, $120$ possibilities have examples.
In $84$ cases, the existing theories from \cite{CADE,FroCoS} can be utilized.
For them, one only needs to determine whether a minimal model function is computable. 
The existing theories are, however, not enough. 
For the remaining $36$ combinations the current paper provides new theories that
exhibit them.

The structure of the remainder of this section is described 
in \Cref{tab:sumex}.
The examples from \cite{CADE,FroCoS}, and in particular,
the computability of their minimal model functions, 
are addressed in \Cref{sec:old}.
The new theories are described in \Cref{sec:new},
and are further sub-categorized into three classes of
theories, in \Cref{sec:newbb,sec:uncomputable,operators}.
\Cref{tab:sumex} also includes the number of theories in each class.
We provide several examples of each class of theories.
The remaining theories are defined in a similar manner.\footnote{Full details are provided in the appendix.}

\begin{table}[t]\centering
\begin{tabular}{|l|l|c|}\hline
Section & Theories & Quantity \\\hline
\Cref{sec:old} & Existing Theories & 84 \\
\Cref{sec:newbb} & New Busy Beaver Theories & 7 \\
\Cref{sec:uncomputable} & New Theories with a Non-computable Function & 13 \\
\Cref{operators} & New Derived Theories & 16\\\hline
\end{tabular}
\caption{Summary of examples.}    
\label{tab:sumex}
\end{table}

\medskip
\toclesslab\subsection{Existing Theories}{sec:old}
In this section, we describe how the computability of a minimal model function
can be determined for all the theories from 
We start with a theory that admits all the properties.

\smallskip
\begin{example}
\label{ex:tgeqn}
Consider the $\Sigma_{1}$-theory $\Tgeqn$ from \cite{CADE}, 
whose models have at least $n$ elements (axiomatized by $\psi_{\geq n}$). 
It was shown in \cite{CADE,FroCoS} that it admits all properties, except for 
the computability of a minimal model function, which was not studied there.
For $S=\{\s_{1}\}$ we can get a computable minimal model function by going over all arrangements
of variables that occur in the formula, and taking
the one with the least induced equivalence classes.
Of course, if this number is less then $n$, then the model
has to be enlarged to have $n$ elements: 
\[\minmod_{\Tgeqn,\{\s_{1}\}}(\phi)=\big\{\max\{n, \min\{|V/E| : \text{$\phi$ and $\delta_{V}^{E}$ are equivalent}, V=\vars(\phi)\}\}\big\}.\]
\end{example}

Now, let us move to an even simpler example, with a twist:

\begin{example}
    Consider the $\Sigma_{1}$-theory $\Tinfty$ from \cite{CADE}, whose models are infinite (axiomatized by $\{\psi_{\geq n} : n\in\mathbb{N}\setminus\{0\}\}$). 
    It was shown in \cite{CADE,FroCoS} that it is stably infinite, smooth, and convex,
    while not having any of the other properties (except for computability of the minimal
    model function, that was not considered there).
    It is easy to prove that, for $S=\{\s_{1}\}$,
$\minmod_{\Tinfty,S}(\phi)=\{\aleph_{0}\}$ is a minimal model function,
and also constant and therefore computable. The twist is that this theory wouldn't be considered to have a computable minimal model function by a definition that demands the theory to be stably finite, as discussed in \Cref{CMMF+FMP}. 
\end{example}

A class of examples without computable minimal model functions are those involving the Busy Beaver function 
$\bb:\mathbb{N}\rightarrow\mathbb{N}$, such that $\bb(n)$ is the maximum number of $1$'s a Turing machine with at most $n$ states can write to its tape when it halts (assuming the tape begin with only $0$'s). Its most important property for our purposes is that for any computable function $h:\mathbb{N}\rightarrow\mathbb{N}$, there exists $n_{0}\in\mathbb{N}$ such that $\bb(n)\geq h(n)$ for all $n\geq n_{0}$, as shown by Rad{\'o} (\cite{Rado}).
In \cite{FroCoS} this is used to separate 
the finite model property and stable finiteness from finite witnessability and strong finite witnessability, but it can also separate the former properties from the computability of a minimal model function. 

\begin{example}
\label{ex:bbold}
Consider the $\Sigma_{1}$-theory $\Tbb$ from \cite{FroCoS},
axiomatized by
$\{\psi_{\geq\bb(n+2)}\vee\bigvee_{i=2}^{n+2}\psi_{=\bb(i)} : n\in\mathbb{N}\}$.
A finite interpretation $\A$ is a $\Tbb$-interpretation if and only if $|\s_{1}^{\A}|=\bb(n)$ for some $n\in\mathbb{N}\setminus\{0\}$; and all infinite interpretations are $\Tbb$-interpretations.
From \cite{FroCoS},  we know that it is stably infinite, not smooth, convex, has the finite model property and is stably finite, without being neither finitely witnessable nor strongly finitely witnessable. 
It is then possible to show that $\Tbb$ does not have a computable minimal model function: indeed, suppose that it does. We can then define an auxiliary function $h:\mathbb{N}\rightarrow\mathbb{N}$ by making $h(0)=\bb(0)=0$, $h(1)=\bb(1)=1$, and, for $n\in\mathbb{N}\setminus\{0\}$,
$h(n+1)=m$ for $m\in \minmod_{\Tbb,\{\s_{1}\}}(\NNNEQ{x}{h(n)+1})$ (for $x_{i}$ of sort $\s_{1}$).
We then know that $h(n+1)$ equals the cardinality of a $\Tbb$-interpretation with at least $h(n)+1$ elements, and thus we prove by induction that $h(n)\geq \bb(n)$ for all $n\in\mathbb{N}$. This leads to a contradiction, as $h$ is computable from its definition,
but is not eventually bounded by $\bb$.
\end{example}

The process of proving computability (or non-computability) of
the other theories from \cite{CADE,FroCoS} is done in
a similar manner to the above examples.

\medskip
\toclesslab\subsection{New Theories}{sec:new}

In this section, we describe new theories, in order to show
the feasibility of the remaining combinations of properties.
We start with \Cref{sec:newbb}, where we show how the busy beaver function
can be used to obtain several new theories, beyond those
that were defined in \cite{FroCoS}.
We also show some special formulas 
that are useful for defining theories in
non-empty signatures.
Some combinations of properties were harder to handle,
and for them, we introduce a new non-computable function
over the natural numbers in \Cref{sec:uncomputable}.
Finally, in \Cref{operators}, we show how to extend the obtained theories
to more complex signatures.

\smallskip
\toclesslab\subsubsection{New Theories With the Busy Beaver Function}{sec:newbb}

We still use the Busy Beaver function for some of the theories in the current study, although the resulting examples are not as straightforward as $\Tbb$. 
The following example shows a theory that has none of the considered properties.

\smallskip
\begin{example}
Let
$\Tfive$ be the $\Sigma_{2}$-theory  with axiomatization
\[\{(\psi^{\s_{1}}_{=1}\wedge(\psi^{\s_{2}}_{\geq \bb(n+1)}\vee\bigvee_{i=1}^{n}\psi^{\s_{2}}_{=\bb(i)}))\vee(\psi^{\s_{1}}_{=2}\wedge(\psi^{\s_{2}}_{=2}\vee\psi^{\s_{2}}_{\geq n})) : n\in\mathbb{N}\setminus\{0\}\}.\]
By distributing the conjunctions over the disjunctions in the basic formula axiomatizing this theory, and then analyzing the conjuncts, we can glimpse at what the models of $\Tfive$ look like.
\begin{enumerate}[leftmargin=7.5em]
    \item[$\psi^{\s_{1}}_{=1}\wedge\psi^{\s_{2}}_{\geq\bb(n+1)}$] When considered over all $n\in\mathbb{N}\setminus\{0\}$, this gives us interpretations $\A$ with $|\s_{1}^{\A}|=1$ and $\s_{2}^{\A}$ infinite.
    \item[$\psi^{\s_{1}}_{=1}\wedge\psi^{\s_{2}}_{=\bb(i)}$] Corresponds to $\A$ with $|\s_{1}^{\A}|=1$ and $|\s^{2}_{\A}|=\bb(n)$ for some $n\in\mathbb{N}\setminus\{0\}$.

    \item[$\psi^{\s_{1}}_{=2}\wedge\psi^{\s_{2}}_{=2}$] The interpretations $\A$ for this formula have $|\s_{1}^{\A}|=|\s_{2}^{\A}|=2$.

    \item[$\psi^{\s_{1}}_{=2}\wedge\psi^{\s_{2}}_{\geq n}$] This gives us interpretations $\A$ where $|\s_{1}^{\A}|=2$ and $\s_{2}^{\A}$ is infinite.
\end{enumerate}

$\Tfive$ is then not stably infinite (and thus not smooth), as the cardinality of the domains of sort $\s_{1}$ of $\Tfive$-interpretations is bounded by $2$, what also helps prove $\Tfive$ is not convex: the disjunction of equalities $(x=y)\vee(y=z)\vee(x=z)$  
with $x,y,z$ of sort $\s_{1}$
is $\Tfive$-valid while none of the disjuncts is. 
$\Tfive$ is not finitely witnessable (and thus not strongly finitely witnessable), given the relationship between the cardinalities of its models and $\bb$, what is proven in a similar way to how that is done with $\Tbb$ in \Cref{ex:bbold}; $\Tfive$ does not have the finite model property (and is thus not stably finite), as $\NEQ(x_{1},x_{2})\wedge\NEQ(u_{1}, u_{2}, u_{3})$ (for $x_{i}$ of sort $\s_{1}$, and $u_{j}$ of sort $\s_{2}$) can only be satisfied by a $\Tfive$-interpretation $\A$ with $|\s_{1}^{\A}|=2$ and $\s_{2}^{\A}$ infinite; and, finally, $\Tfive$ does not have a computable minimal model function, by an argument similar to the one found in \Cref{ex:bbold}. To summarize, $\Tfive$ is a theory with none of the considered properties, over the empty signature.
\end{example}

\begin{figure}[t]
\begin{mdframed}
\[\psi_{\neq}=\Forall{x}\neg[s(x)=x]\]
\[\psi_{=}^{k}=\Forall{x}[s^{k}(x)=x]\] 
\[\psi_{\vee}^{k}=\Forall{x}[[s^{2k}(x)=s^{k}(x)]\vee[s^{2k}(x)=x]]\]
\end{mdframed}
\caption{Special formulas: $k\in\mathbb{N}\setminus\{0\}$, $x$ is of sort $\s_{1}$, and 
$s^{k}(x)$ is defined by $s^{1}(x)=s(x)$ and $s^{k+1}(x)=s(s^{k}(x))$.
We denote $\psi_{=}^{1}$ by $\psi_{=}$, and $\psi_{\vee}^{1}$ by $\psi_{\vee}$.}
\label{special-formulas}
\end{figure}

The following example
shows a theory that is smooth, but has none of the other considered properties.
It relies on $\Sigma_{s}$-formulas from \Cref{special-formulas},
that encode various shapes of cycles that functions can create
(a cycle is a scenario in which applying a function
some number of times results in the original input).
$\psi_{\neq}$ holds in interpretations where $s$ has
no cycles of size $1$.
For a positive $k$, $\psi_{=}^{k}$ holds when the interpretation of $s$
has cycles of size $k$,
and $\psi_{\vee}^{k}$ holds when the cycles that the interpretation
of $s$ creates have one of the two forms described in the disjunction.
In particular,
the formula $\psi_{\neq}\wedge\psi_{\vee}^{2}$ gets us the four scenarios represented in \Cref{possible scenarios two}. 
In a similar way, 
we may wish to combine  $\psi_{\neq}$ and $\psi_{=}^{2}$, implying that in an interpretation $\A$ where both hold $s^{\A}$ always induces a cycle  of size exactly $2$.
When $k=1$, we omit it from $\psi_{=}^{k}$ and $\psi_{\vee}^{k}$.

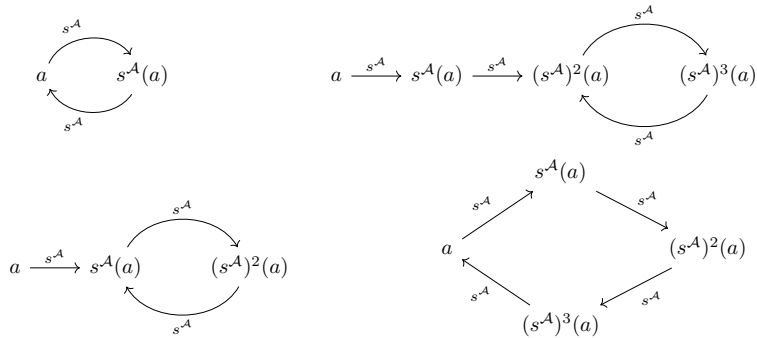
\begin{figure}[htp]
\centering
\begin{minipage}{0.3\textwidth}
\adjustbox{scale=0.8,center}{
\begin{tikzcd}
a \arrow[r, bend left=60, "s^{\A}"]
& s^{\A}(a)\arrow[l, bend left=60, "s^{\A}"]
\end{tikzcd}}
\end{minipage}
\begin{minipage}{0.5\textwidth}
\adjustbox{scale=0.8,center}{\begin{tikzcd}
a\arrow[r, "s^{\A}"]& s^{\A}(a) \arrow[r, "s^{\A}"]
& (s^{\A})^{2}(a) \arrow[r, bend left=60, "s^{\A}"]
& (s^{\A})^{3}(a)\arrow[l, bend left=60, "s^{\A}"]
\end{tikzcd}}
\end{minipage}
\centering
\begin{minipage}{0.4\textwidth}
\adjustbox{scale=0.8,center}{\begin{tikzcd}
a \arrow[r, "s^{\A}"]
& s^{\A}(a) \arrow[r, bend left=60, "s^{\A}"]
& (s^{\A})^{2}(a)\arrow[l, bend left=60, "s^{\A}"]
\end{tikzcd}}
\end{minipage}
\begin{minipage}{0.4\textwidth}
\adjustbox{scale=0.8,center}{\begin{tikzcd} 
& s^{\A}(a) \arrow[dr, "s^{\A}"] & \\
a\arrow[ur, "s^{\A}"]& 
& (s^{\A})^{2}(a) \arrow[dl, "s^{\A}"]\\
& (s^{\A})^{3}(a)\arrow[ul, "s^{\A}"] & \\
\end{tikzcd}}
\end{minipage}
\vspace{-8mm}
\caption{Possible scenarios when $\psi_{\neq}\wedge\psi_{\vee}^{2}$ holds}\label{possible scenarios two}
\end{figure}

\begin{example}

    Let $\Tnine$ be the $\Sigma_{s}$-interpretation with axiomatization
    \[\{(\psi_{\geq2\bb(n+1)}\wedge\psi_{\neq})\vee\bigvee_{i=2}^{n}(\psi_{=2\bb(i)}\wedge\psi_{\neq}\wedge\psi_{=}^{2})\vee\psi_{=1} : n\in\mathbb{N}\setminus\{0,1\}\}.\]
    It has, essentially, three types of models $\A$: a trivial one, with $|\s_{1}^{\A}|=1$, and where $s^{\A}$ is by force the identity function; those with $|\s_{1}^{\A}|$ finite and equal to some $2\bb(n)$, for $n\geq 2$, and where $\psi_{\neq}\wedge\psi_{=}^{2}$ holds, meaning that $s^{\A}$ always induces a cycle of size $2$ (notice that we multiply $\bb(n)$ by $2$ to accommodate these cycles); and those with $\s_{1}^{\A}$ infinite, and where $\psi_{\neq}$ holds.

    $\Tnine$ is not stably infinite (and thus not smooth), since $s(x)=x$ is only satisfied by the trivial model. It is not (strongly) finitely witnessable thanks to $\bb$. It is convex, because of the "determinism" of $s^{\A}$, which at most induces a cycle of size $2$. It is not stably finite (and thus does not have the finite model property), as $\neg(s^{4}(x)=x)$ is only satisfied by infinite $\Tnine$-interpretations. Finally, it does not have a computable minimal model function, thanks to $\bb$.

\end{example}

The Busy Beaver function is similarly used to define
5 more theories with various 
properties.

\smallskip
\toclesslab\subsubsection{Theories With a new Non-Computable Function}{sec:uncomputable}

The most interesting example out of the present paper is that of a $\Sigma_{1}$-theory that is stably infinite but not smooth, finitely witnessable but not strongly so, stably finite, and convex but without a computable minimal model function; we call this theory $\Tone$. 
We first prove the existence of a function $g:\mathbb{N}\setminus\{0\}\rightarrow\mathbb{N}\setminus\{0\}$ 
with some key properties.

\smallskip
\begin{lemma}
\label{lem:g-exists}
    There exists a function $g:\mathbb{N}\setminus\{0\}\rightarrow\mathbb{N}\setminus\{0\}$ that is
\textbf{1)}~increasing;
\textbf{2)}~unbounded;
\textbf{3)}~non-surjective;
\textbf{4)}~non-computable; and
\textbf{5)}~there exists an increasing computable function $\rho:\mathbb{N}\setminus\{0\}\rightarrow\mathbb{N}\setminus\{0\}$ such that
$g\circ\rho$ is computable.\footnote{Notice that this does not hold for $\bb$: it follows from \cite{Rado} that for any increasing function $\rho:\mathbb{N}\rightarrow\mathbb{N}$, $\bb\circ\rho$ grows at least as fast as $\bb$, and thus eventually faster than any computable function, thus being non-computable.}
\end{lemma}

\begin{proof}
Two of the theories in \cite{CADE} rely on the existence of a non-computable function $f:\mathbb{N}\setminus\{0\}\rightarrow\{0,1\}$ such that $f(1)=1$ and, for every $k\in\mathbb{N}\setminus\{0\}$, 
\[|\{1\leq i\leq 2^{k} : f(i)=0\}|=|\{1\leq i\leq 2^{k} : f(i)=1\}|\]
We can then define a new function $g:\mathbb{N}\setminus\{0\}\rightarrow\mathbb{N}\setminus\{0\}$ by making 
$g(n)=n+\sum_{i=1}^{n}f(i)=\sum_{i=1}^{n}(f(i)+1)$.
Then we have the following. \textbf{1)} $g$ is increasing, since $g(n+1)=g(n)+f(n+1)+1$, and $f(n+1)$ is either $0$ or $1$; \textbf{2)} $g$ is unbounded, as $f$ is not computable and therefore infinitely often equals $1$; \textbf{3)} it also follows that there are 
(infinitely many) values $n\in\mathbb{N}\setminus\{0\}$ such that $f(n+1)\neq 0$, and then $g(n+1)>g(n)+1$, meaning $g$ is not surjective; \textbf{4)} $g$ is non-computable, since if it were we could find an algorithm for calculating $f$:
$f(1)$ is fixed to be 1, and 
for $n\geq 1$, we notice that 
$f(n+1)=g(n+1)-g(n)-1$; and, finally, 
\textbf{5)}
for all $k\in\mathbb{N}\setminus\{0\}$, $g(2^{k})=3\times2^{k-1}$, since $|\{1\leq i\leq 2^{k} : f(i)=0\}|$ must equal $|\{1\leq i\leq 2^{k} : f(i)=1\}|$, meaning both equal $2^{k-1}$, from what follows that $g$ composed with $\rho(k)=2^{k}$ is computable.
\end{proof}

Now, we can define the theory $\Tone$.

\begin{example}
$\Tone$ is the $\Sigma_{1}$-theory with axiomatization
$\{\psi^{\s_{1}}_{\geq g(n)}\vee\bigvee_{i=1}^{n}\psi^{\s_{1}}_{=g(i)} : n\in\mathbb{N}\setminus\{0\}\}$.
It has all the $\Sigma_{1}$-interpretations whose domains have 
$g(k)$ elements for some $k$,
or infinitely many elements, and only those.

Since $g$ is increasing (property $\textbf{1}$), and the signature is empty,
different finite $\Tone$-in\-ter\-pre\-ta\-tions must have different cardinalities,
which makes reasoning about $\Tone$ easier.
$\Tone$ is stably infinite, as whenever
a quantifier-free formula is satisfied by one of its finite models,
it is also satisfied by (all) its infinite models.
According to a result from \cite{CADE}, it is then convex.
$\Tone$ is stably finite thanks to $g$ being unbounded (property {\bf 2}): that way, if a formula needs a certain number of elements in an interpretation in order to be satisfied, we can always find a finite $\Tone$-interpretation with at least that many elements.
Similarly, by a result from \cite{FroCoS}, it also has
the finite model property.
Further, since $g$ is not surjective (property {\bf 3}), $\Tone$ is not smooth,
as there are "holes" in the cardinalities of its models.
This, according to a result from \cite{CADE}, also guarantees that $\Tone$ is not strongly finitely witnessable.
The fact that $g$ is not computable (property {\bf 4}) implies that $\Tone$ does not have a computable minimal model function: if it had, this function could have been used to compute $g$ simply by 
hard-coding the value of $g(1)$,
and by making, for $S=\{\s_{1}\}$,
$g(n+1)\in \minmod_{\Tone,S}(\NNNEQ{x}{g(n)+1})$ (where the $x_{i}$ are of sort $\s_{1}$).

Finally, as for the finite witnessability of $\Tone$: given a quantifier-free formula $\phi$, let $n$ be the number of variables in it, and find the least $k$ such that $g(2^{k})=3\times 2^{k-1}\geq n$ (an equality that holds due to property $\textbf{5}$). For $x_{i}$ fresh variables, we then have that
$\wit(\phi)=\phi\wedge\bigwedge_{i=1}^{g(2^{k})}x_{i}=x_{i}$
is computable: it is furthermore clearly equivalent, and thus its existential closure is $\Tone$-equivalent, to $\phi$; and, lastly, if a $\Tone$-interpretation $\A$ satisfies $\phi$, it is enough to add $g(2^{k})-|\s_{1}^{\A}|$ elements to it in order to find a $\Tone$-interpretation where all elements are witnessed.
\end{example}  


Similarly to~\Cref{sec:newbb},
 we can use~\Cref{special-formulas} for theories related to $g$.
 In this way, we present
 a theory that is finitely witnessable, has the finite model property and is stably finite,
 without having any of the other properties.

\begin{example}\label{ex:gplusnc}
    Consider the $\Sigma_{s}$-theory $\Ttwelve$, with axiomatization
    \[\{\psi_{=1}\vee(\psi_{\neq}\wedge\psi^{2}_{\vee}\wedge(\psi_{\geq 2g(n+1)}\vee\bigvee_{i=1}^{n}\psi_{=2g(i)})) : n\in\mathbb{N}\setminus\{0\}\}.\]
    It has three types of model $\A$: a trivial one, with $|\s_{1}^{\A}|=1$; finite but non-trivial ones, with $|\s_{1}^{\A}|=2g(n)$ for some $n\in\mathbb{N}\setminus\{0\}$, and where $\psi_{\neq}\wedge\psi_{\vee}^{2}$ holds, and thus one of the scenarios in \Cref{possible scenarios two} holds;
    and infinite ones, where $\psi_{\neq}\wedge\psi_{\vee}^{2}$ still holds. It is not stably infinite (and thus not smooth), as $s(x)=x$ is only satisfied by the trivial model; it is finitely witnessable but not strongly so, thanks to the properties of $g$ related to computability, which also guarantee that $\Ttwelve$ does not have a computable minimal model function; it is not convex, due to $\psi_{\vee}^{2}$; and it has the finite model property, and is stably finite, given $g$ is unbounded. 
\end{example}

The function $g$ is used in a similar manner to define
11 more theories with various 
properties.

\smallskip
\toclesslab\subsubsection{Derived Theories}{operators}

We have defined in \cite{CADE}  {\em theory operators} that, given a theory in the signatures of \Cref{tab:signatures}, produce a different theory in a different signature, preserving  properties of the original theory. 

\smallskip
\begin{definition}[Theory Operators from \cite{CADE}]~
\label{def:operators}
\begin{enumerate}
\item 
If $\T$ is a $\Sigma_1$-(respectively $\Sigma_{s}$-)theory, then
$\adds{\T}$ is the $\Sigma_{2}$-($\Sigma_{s}^{2}$-)theory axiomatized by
$\ax{\T}$.\footnote{In fact, \cite{CADE} only defined this operator for $\Sigma_1$-theories, even though it was used there also for $\Sigma_s$-theories. Here we define it more generally, so that it also works in the presence of function symbols. In particular, this clarifies that its usage in \cite{CADE} for non-empty signatures was sound.}
\item If $\T$ is a $\Sigma_{n}$-theory, $\addf{\T}$ is the $\Sigma_{s}^{n}$-theory with axiomatization $\ax{\T}\cup\{\psi_{=}\}$.
\item If $\T$ is a $\Sigma_{n}$-theory, $\addnc{\T}$ is the $\Sigma_{s}^{n}$-theory with axiomatization $\ax{\T}\cup\{\psi_{\vee}\}$.
\end{enumerate}
\end{definition}

The first operator adds a sort to a theory, and the others
add a function symbol.
It was proven in \cite{CADE,FroCoS} that the first two operators 
preserve convexity, 
stable infiniteness, smoothness, finite witnessability, strong finite witnessability,
the finite model property and stable finiteness,
as well as the absence of these properties.
The third preserves all but convexity: $\addnc{\T}$ admits
all the properties of $\T$, except for the fact that it is guaranteed
to never be convex.
Here we prove that all three operators preserve the (non-)computability of a minimal model function.

\begin{restatable}{theorem}{addsresult}
\label{thm:addsresult}
Let $\T$ be a $\Sigma_{1}$ or $\Sigma_{s}$-theory. Then: $\T$ has a computable minimal model function with respect to $\{\s_{1}\}$ iff
$\adds{\T}$ has a computable minimal model function with respect to $\{\s_{1}, \s_{2}\}$.
\end{restatable} 

\begin{restatable}{theorem}{addfresult}
\label{thm:addfresult}
Let $\T$ be a $\Sigma_{n}$-theory. Then: $\T$
has a computable minimal model function w.r.t. $\{\s_{1},\ldots,\s_{n}\}$ if, and only if, $\addf{\T}$ has a computable minimal model function w.r.t. $\{\s_{1},\ldots,\s_{n}\}$.
\end{restatable}

\begin{restatable}{theorem}{addncresult}
\label{thm:addncresult}
Let $\T$ be a $\Sigma_{n}$-theory. Then: $\T$
has a computable minimal model function w.r.t. $\{\s_{1},\ldots,\s_{n}\}$ if, and only if, $\addnc{\T}$ has a computable minimal model function w.r.t. $\{\s_{1},\ldots,\s_{n}\}$.
\end{restatable}

These theorems are used in order to
construct new theories in different signatures
from existing theories (that are either defined in
\cite{CADE,FroCoS} or the current paper).

\begin{example}
The $\Tgeqn$ from \Cref{ex:tgeqn} admits all the discussed
properties w.r.t. to its only sort. 
By \cite{CADE,FroCoS}, all properties but the computability
of a minimal model function also hold for $\adds{\Tgeqn}$.
By \Cref{thm:addsresult}, this also holds for 
the computability of a minimal model function.
\end{example}

The operators are similarly used to define 14 more
theories with various 
properties.

\bigskip
\toclesslab\section{Conclusion and Future Work}{conclusion}

\question{OK, I see what you did here. 
In \cite{CADE,FroCoS} every combination
of the properties from a list of properties
related to politeness, shininess and the Nelson-Oppen method, was
considered. Now you have added the last property to this list, namely,
the computability of a minimal model function, thus 
completing all properties that relate to these combination methods.}
\answer{Yes. The impossibility results show that 
there are cases in which one must prove a certain combination
of properties, without the ability to reduce them to others.
The examples that we found constitute a thorough taxonomic analysis of the properties, and also provide an invaluable tool-set of techniques in order to explore them.
}

\question{So, are you finally done?}
\answer{The job is never done. 
Some open problems remain, namely the existence
of unicorns of the three flavors. 
Further, one can consider more properties,
such as
 decidability,
 finite axiomatizability, and more. 
 Moreover, we have only considered properties with respect to the whole set of sorts. 
 We could also consider subsets. 
 Finally, we have so many examples of theories that one wonders whether producing them could be somehow automated. 
 \Cref{thm:addsresult,thm:addfresult,thm:addncresult} definitely go in that direction, but perhaps even the 
 construction of the
 more primitive
 examples could be assisted by automation.
 }

\question{Then let me just ask: couldn't you think of a more serious name than {\em unicorn theories}?}
\answer{Oh yeah, 
because {\em smooth}, {\em polite} and {\em shiny} are very serious... We are just keeping up with the tradition!}

\printbibliography

\newpage

\newgeometry{
a4paper,
left=20mm,
right=20mm,
bottom=18mm,
top=15mm
}
\appendix

\renewcommand{\contentsname}{Contents of the Appendix}

\tableofcontents


\vspace{15mm}

\newcommand{\ornamentSep}{\noindent\hfil{{\pgfornament[width=0.618033988749895\columnwidth,color=black]{88}}}}

\ornamentSep

\vspace{15mm}

\crefname{theorem}{theorem}{theorems}
\Crefname{theorem}{Theo.}{Theo.}

\crefname{theorem}{theorem}{theorems}
\Crefname{theorem}{Theorem}{Theorems}


We briefly establish some simplifications, in terms of nomenclature, for the Appendix. In the signatures $\Sigma_{n}$ and $\Sigma_{s}^{n}$, we denote the sort $\s_{1}$ by $\s$ when there is no risk of confusion. An interpretation $\A$ is said to be trivial with respect to a sort $\s$ if $|\s^{\A}|=1$; we may omit the reference to the sort if the signature is one-sorted, or the sort is clear from context.
Also,
notice that we may conclude from  \Cref{alternative definition} that $\minmod_{\T,S}(\phi)=\minimal(\textbf{Card}_{\T,S}(\phi))$, where $\minimal$ is the operator that takes a subset of an ordered set and returns its minimal elements. 
This fact will be used throughout the appendix.
Finally, in some proofs we will refer to lemmas that are proved in \cite{arxivCADE}, which is the full version of~\cite{CADE}.

\section{\tp{Proof of \Cref{technical lemma}}{Proof of technical lemma}}
\dickson*

\begin{proof}
    We prove this by induction on $n$, being the cases $n=0$ and $n=1$ trivial; assume then the lemma is true for some $n$, and we shall work the case $n+1$. Take a minimal element $\overline{p}=(p_{1},\ldots,p_{n},p_{n+1})$ of $A$ (if there is none, we are done), and for any other minimal element $\overline{q}=(q_{1}, \ldots, q_{n}, q_{n+1})$ there must exist distinct $1\leq i,j\leq n+1$ such that $p_{i}>q_{i}$ and $p_{j}<q_{j}$ (keeping in mind that $p_{i}$, respectively $q_{j}$, may be infinite, in which case $q_{i}$, respectively $p_{j}$, may be any element of $\mathbb{N}$. $\overline{q}$ is then a minimal element of 
    \[A^{i}_{q_{i}}=\{(r_{1}, \ldots, r_{n}, r_{n+1})\in A : r_{i}=q_{i}\}\]
    with the order induced by $A$: if it weren't, it would also not be a minimal element of $A$. Define then 
    \[(A^{i}_{q_{i}})^{*}=\{(r_{1},\ldots, \widehat{r_{i}}, \ldots, r_{n+1})\in \N^{n} : (r_{1},\ldots, r_{i},\ldots, r_{n+1})\in A^{i}_{q_{i}}\},\]
    and it is obvious that $\overline{q}^{*}=(q_{1},\ldots,\widehat{q_{i}},\ldots,q_{n}, q_{n+1})$ is a minimal element of $(A^{i}_{q_{i}})^{*}$. This way, any minimal element of $A$ must be contained in 
    \[\bigcup_{i=1}^{n+1}\bigcup_{q_{i}=0}^{p_{i}-1}A^{i}_{q_{i}};\]
    since $(A^{i}_{q_{i}})^{*}$, as a subset of $\N^{n}$, has a finite number of minimal elements by induction hypothesis, we have that the number of minimal elements in $A$ is bounded from above (taking into consideration $\overline{p}$) by
    \[1+\sum_{i=1}^{n+1}\sum_{q_{i}=0}^{p_{i}-1}|(A^{i}_{q_{i}})^{*}|,\]
    being therefore finite as well.    
\end{proof}

\section{\tp{Proof of \Cref{unique and finite}}{Proof of unique and finite}}

In order to prove \Cref{unique and finite}, we begin with a lemma that proves half of the latter result \Cref{alternative definition}, and uses the accompanying notation.

\begin{lemma}\label{minmod is minimal}
    For every quantifier-free formula $\phi$, not necessarily $\T$-satisfiable, the subset $X$ of $\N^{S}$ is contained in the minimal elements of $\textbf{Card}_{\T,S}(\phi)$, where $(n_{\s})_{\s\in S}\in X$ iff: there is a $\T$-interpretation $\A$ that satisfies $\phi$ with $(|\s^{\A}|)_{\s\in S}=(n_{\s})_{\s\in S}$; if $\B$ is a $\T$-interpretation that satisfies $\phi$ with $(|\s^{\B}|)_{\s\in S}\neq (n_{\s})_{\s\in S}$, then there is a $\s\in S$ such that $n_{\s}<|\s^{\B}|$.
\end{lemma}

\begin{proof}
    This is obvious if $\phi$ is not $\T$-satisfiable, since in that case both $\textbf{Card}_{\T,S}(\phi)$, and the set $X$, are empty.

    So, assume that $\phi$ is $\T$-satisfiable, and let $(n_{\s})_{\s\in S}$ be in $X$: because there exists a $\T$-interpretation $\A$ that satisfies $\phi$ with $|\s^{\A}|=n_{\s}$, for all $\s\in S$, this means $(n_{\s})_{\s\in S}\in \textbf{Card}_{\T,S}(\phi)$; so assume that $(n_{\s})_{\s\in S}$ is not a minimal element. There must exist then a $\T$-interpretation $\B$ that satisfies $\phi$ with $(|\s^{\B}|)_{\s\in S}<(n_{\s})_{\s\in S}$, meaning that $|\s^{\B}|\leq n_{\s}$ for all $\s\in S$, and for at least one $\s_{0}\in S$ one has $|\s_{0}^{\B}|<n_{\s_{0}}$; but, in particular, $(|\s^{\B}|)_{\s\in S}\neq (n_{\s})_{\s\in S}$, what should imply from the fact that $(n_{\s})_{\s\in S}$ is in $X$ that there exists $\s\in S$ such that $n_{\s}<|\s^{\B}|$, leading to a contradiction. To summarize, $(n_{\s})_{\s\in S}$ is minimal.
\end{proof}

\uniqueandfinite*

\begin{proof}
    From \Cref{minmod is minimal}, we know that $X$ is contained in $\minimal(\textbf{Card}_{\T,S}(\phi))$. If $S$ has cardinality $m$, and 
    $\textbf{Card}_{\T,S}(\phi)$ is not empty, it can be seen in a natural way as a subset of $\N^{m}$, and by \Cref{technical lemma} it follows that $\minimal(\textbf{Card}_{\T,S}(\phi))$ is finite, what finishes the proof.
\end{proof}

\section{\tp{Proof of \Cref{alternative definition}}{Proof of alternative definition}}

\alternativedefinitionminmod*

\begin{proof}
One direction of this result is corresponds to \Cref{minmod is minimal}. For the other, suppose $(|\s^{\A}|)_{\s\in S}$ is a minimal element of $\textbf{Card}_{\T,S}(\phi)$: of course, $\A$ is a $\T$-interpretation that satisfies $\phi$, so $\phi$ is $\T$-satisfiable, and the first condition for $(|\s^{\A}|)_{\s\in S}$ to lie in $\minmod_{\T,S}(\phi)$ is satisfied. So, suppose that $\B$ is another $\T$-interpretation that satisfies $\phi$, and that $(|\s^{\B}|)_{\s\in S}\neq (|\s^{\A}|)_{\s\in S}$: since $(|\s^{\B}|)_{\s\in S}$ is in $\textbf{Card}_{\T,S}(\phi)$, and $(|\s^{\A}|)_{\s\in S}$ is a minimal element of the same set, we must have that: either $(|\s^{\B}|)_{\s\in S}\geq (|\s^{\A}|)_{\s\in S}$, and since both are different there must exist $\s\in S$ such that $|\s^{\B}|>|\s^{\A}|$; or $(|\s^{\B}|)_{\s\in S}$ and $(|\s^{\A}|)_{\s\in S}$ are not comparable, meaning there exists $\s, \s_{*}\in S$ such that $|\s_{*}^{\A}|>|\s_{*}^{\B}|$, and $|\s^{\B}|>|\s^{\A}|$, as we wished to show. So $(|\s^{\A}|)_{\s\in S}$ is in $\minmod_{\T,S}(\phi)$, and thus the equality this set with $\minimal(\textbf{Card}_{\T,S}(\phi))$ is proved.

\end{proof}

\section{\tp{Proof of \Cref{CMMF+FMP}}{Proof of CMMF+FMP}}

\CMMFplusFMP*

\begin{proof}
    \begin{enumerate}
        \item Suppose first that $\T$ has the finite model property with respect to $S$: then, for any $\T$-satisfiable, quantifier-free formula $\phi$, there is a $\T$-interpretation $\A$ that satisfies $\phi$ with $\s^{\A}$ finite for each $\s\in S$. Then, either $(|\s^{\A}|)_{\s\in S}$ itself lies in $\minmod_{\T,S}(\phi)$ (and of course $|\s^{\A}|\in\mathbb{N}$ for each $\s\in S$), or there is $(m_{\s})_{\s\in S}\in \minmod_{\T,S}(\phi)$, different from $(|\s^{\A}|)_{\s\in S}$, with $m_{\s}\leq |\s^{\A}|$ for each $\s\in S$, and $m_{\s}<|\s^{\A}|$ for at least one $\s\in S$ (and of course $m_{\s}$ is then in $\mathbb{N}$). 

        Indeed, suppose that the latter does not happen, and we will prove that this forces $(|\s^{\A}|)_{\s\in S}$ to be in $\minmod_{\T,S}(\phi)$. We already have that there is a $\T$-interpretation $\B$ that satisfies $\phi$ with $|\s^{\B}|=|\s^{\A}|$ for each $\s\in S$, namely $\A$ itself. Now, if there are $\T$-interpretations $\C$ that satisfy $\phi$ with $(|\s^{\C}|)_{\s\in S}<(|\s^{\A}|)_{\s\in S}$, we can take a $\C^{\prime}$ with minimal $(|\s^{\C^{\prime}}|)_{\s\in S}$, and it will lie in $\minmod_{\T,S}(\phi)$, against our hypothesis that no element in this set is strictly less than $(|\s^{\A}|)_{\s\in S}$.

        Now, for the reciprocal: take a $\T$-satisfiable, quantifier-free formula $\phi$, and we know there is an $(n_{\s})_{\s\in S}$ in $\minmod_{\T,S}(\phi)$ such that $n_{\s}\in\mathbb{N}$ for each $\s\in S$. But, from the definition of $\minmod_{\T,S}$, there is a $\T$-interpretation $\A$ that satisfies $\phi$ with $|\s^{\A}|=n_{\s}\in\mathbb{N}$ for each $\s\in S$, and of course this means each $\s^{\A}$ is finite for $\s\in S$.
        \item Suppose first that $\T$ is stably finite, and that there is (for a $\T$-satisfiable, quantifier-free $\phi$) an $(n_{\s})_{\s\in S}$ in $\minmod_{\T,S}(\phi)$ such that $n_{\s}=\aleph_{0}$ for some $\s\in S$. From the definition of $\minmod_{\T,S}$ there is then a $\T$-interpretation $\A$ that satisfies $\phi$ with $|\s^{\A}|=n_{\s}$ for all $\s\in S$; and, from the fact that $\T$ is stably finite with respect to $S$, we get that there is a $\T$-interpretation $\B$ that satisfies $\phi$ with $\s^{\B}$ finite, and $|\s^{\B}|\leq|\s^{\A}|$, both for every $\s\in S$. $\B$ is then a $\T$-interpretation that satisfies $\phi$ with $(|\s^{\B}|)_{\s\in S}\neq(n_{\s})_{\s\in S}$ (since one $n_{\s}$ equals $\aleph_{0}$, while all $|\s^{\B}|$ are in $\mathbb{N}$), and $|\s^{\B}|\leq n_{\s}$ for all $\s\in S$, contradicting the fact that $(n_{\s})_{\s\in S}$ is supposed to be in $\minmod_{\T,S}(\phi)$.

        Reciprocally, take a $\T$-satisfiable, quantifier-free formula $\phi$, and a $\T$-interpretation $\A$ that satisfies $\phi$: we state that, either there is an $(n_{\s})_{\s\in S}$ in $\minmod_{\T,S}(\phi)$ such that $(n_{\s})_{\s\in S}<(|\s^{\A}|)_{\s\in S}$, or $(|\s^{\A}|)_{\s\in S}$ itself is in $\minmod_{\T,S}(\phi)$. Indeed, suppose the former does not happen: we already know that there is a $\T$-interpretation $\B$ that satisfies $\phi$ with $(|\s^{\B}|)_{\s\in S}=(|\s^{\A}|)_{\s\in S}$, namely $\A$ itself. Then, suppose that there are $\T$-interpretations $\C$ that satisfy $\phi$ with $(|\s^{\C}|)_{\s\in S}<(|\s^{\A}|)_{\s\in S}$, and take a $\C^{\prime}$ among them with minimal $(|\s^{\C^{\prime}}|)_{\s\in S}$: it follows that $(|\s^{\C^{\prime}}|)_{\s\in S}$ is in $\minmod_{\T,S}(\phi)$, contradicting our assumption that no element in $\minmod_{\T,S}(\phi)$ is strictly less than $(|\s^{\A}|)_{\s\in S}$. 
        
        So, to summarize, there is $(n_{\s})_{\s\in S}\leq (|\s^{\A}|)_{\s\in S}$ in $\minmod_{\T,S}(\phi)$: from the hypothesis, we have that $n_{\s}\in\mathbb{N}$ for each $\s\in S$; at the same time, from the definition of $\minmod_{\T,S}$, there is a $\T$-interpretation $\D$ that satisfies $\phi$ with $|\s^{\D}|=n_{\s}$ for each $\s\in S$. We therefore reach the conclusion that $\D$ is a $\T$-interpretation that satisfies $\phi$ with $\s^{\D}$ finite, and $|\s^{\D}|\leq|\s^{\A}|$, both for every $\s\in S$, meaning $\T$ is stably finite with respect to $S$.
    \end{enumerate}
\end{proof}





\section{\tp{Proof of \Cref{CMMF+FM=>FW}}{Proof of CMMF+FM=>FW}}

\CMMFplusFMimpliesFW*

\begin{proof}
    Assume that $S=\{\s_{1},\ldots,\s_{n}\}$, and let $x_{j}^{i}$ be fresh variables of sort $\s_{i}$. We define a function $wit$ from $\qf(\Sigma)$ into itself by making
    \[wit(\phi)=\phi\wedge\bigwedge_{i=1}^{n}\bigwedge_{j=1}^{m_{i}}x_{j}^{i}=x_{j}^{i}\] 
    if there is a tuple $(m_{1},\ldots,m_{n})\in \minmod_{\T,S}(\phi)$ where all $m_{i}$ are finite (which is always possible to find if $\phi$ is $\T$-satisfiable since $\T$ has the finite model property with respect to $S$, as proven in \Cref{CMMF+FMP}), and otherwise $\wit(\phi)=\phi$. We state that this is a witness of $\T$, being first of all computable: indeed, $\minmod_{\T,S}(\phi)$ can be found algorithmically, and is always finite, allowing us to check whether there is a tuple with all finite coordinates in it. It is then obvious that $\phi$ and $\Exists{\overarrow{x}}wit(\phi)$ are $\T$-equivalent, for $\overarrow{x}=\vars(wit(\phi))\setminus\vars(\phi)$ equal to $\{x_{j}^{i} : 1\leq i\leq n, 1\leq j\leq m_{i}\}$ if $\phi$ is $\T$-satisfiable and $\emptyset$ otherwise, since $\phi$ and $wit(\phi)$ are themselves equivalent and thus $\T$-equivalent.

    Now, let $\A$ be a $\T$-interpretation that satisfies $\phi$ with $(|\s_{1}^{\A}|,\ldots,|\s_{n}^{\A}|)=(m_{1},\ldots,m_{n})$: we define a second $\T$-interpretation $\A^{\prime}$ by changing the value assigned by $\A$ only to the variables $x_{j}^{i}$, so that the map $x\in\{x_{1}^{i},\ldots,x_{m_{i}}^{i}\}\mapsto x^{\A^{\prime}}\in\s_{i}^{\A}$ becomes a bijection (this is possible as both $\{x_{1}^{i},\ldots,x_{m_{i}}^{i}\}$ and $\s_{i}^{\A}$ have cardinality $m_{i}$). Of course, $\A^{\prime}$ still satisfies $\phi$, and thus $wit(\phi)$, but since $x\in\{x_{1}^{i},\ldots,x_{m_{i}}^{i}\}\mapsto x^{\A^{\prime}}\in\s_{i}^{\A}$ is a bijection we now have that $\vars_{\s_{i}}(wit(\phi))^{\A^{\prime}}=\s_{i}^{\A^{\prime}}$ for each $1\leq i\leq n$, what finishes proving that $wit$ is indeed a witness.
\end{proof}

\section{\tp{Proof of \Cref{SM+ES=>CMMF}}{Proof of SM+ES=>CMMF}}

\begin{restatable}{lemma}{CMMFofsubset}\label{CMMF of subset}
    If $\T$ has a computable minimal model function with respect to $S$, then it has a computable minimal model function with respect to any subset $S^{\prime}\subseteq S$.
\end{restatable}


\begin{proof}
    For any $f:S\rightarrow\N$, element of $\N^{S}$, consider the function $S^{\prime}(f):S^{\prime}\rightarrow\N$ that is just the restriction of $f$ to $S^{\prime}$, that is $S^{\prime}(f)=f|_{S^{\prime}}$; for an $X\in\mathcal{P}(\N^{S})$, we define $S^{\prime}(X)=\{S^{\prime}(f) : f\in X\}$. The function $S^{\prime}:\N^{S}\rightarrow\N^{S^{\prime}}$ is then itself computable, given $S$ (and thus $S^{\prime}$) is finite.
    
    Suppose then that $\minmod_{\T,S}$ is computable: we state that
    \[\minmod_{\T,S^{\prime}}(\phi)=\minimal\big(S^{\prime}\circ\minmod_{\T,S}(\phi)\big)\]
    is a minimal model function with respect to $S^{\prime}$, being obviously computable as a composition of computable functions, and being its output always finite. Take a $\T$-satisfiable, quantifier-free formula $\phi$, and any element $(n_{\s})_{\s\in S^{\prime}}$ in $\minimal(S^{\prime}\circ\minmod_{\T,S}(\phi))$: we prove that it is also in $\minmod_{\T,S^{\prime}}(\phi)$.

    \begin{enumerate}
        \item Take a $(m_{\s})_{\s\in S}$ in $\minmod_{\T,S}(\phi)$ such that $m_{\s}=n_{\s}$ for all $\s\in S^{\prime}$; there is a $\T$-interpretation $\A$ that satisfies $\phi$ with $|\s^{\A}|=m_{\s}$ for all $\s\in S$, and of course this means that $|\s^{\A}|=n_{\s}$ for all $\s\in S^{\prime}$.

        \item Let $\B$ be a $\T$-interpretation that satisfies $\phi$ with $(|\s^{\B}|)_{\s\in S^{\prime}}\neq (n_{\s})_{\s\in S^{\prime}}$: assume, for a proof by contradiction, that there is no $\s\in S^{\prime}$ such that $|\s^{\B}|>n_{\s}$, meaning that $(|\s^{\B}|)_{\s\in S^{\prime}}<(n_{\s})_{\s\in S^{\prime}}$. Now, we state that there must exist a $(m_{\s}^{\prime})_{\s\in S}$ in $\minmod_{\T,S}(\phi)$ with $(m_{\s}^{\prime})_{\s\in S}\leq (|\s^{\B}|)_{\s\in S}$: if there isn't such an element different from $(|\s^{\B}|)_{\s\in S}$, then $(|\s^{\B}|)_{\s\in S}$ itself is in $\minmod_{\T,S}(\phi)$ and is (of course) less than or equal to itself. 
        
        Indeed, suppose that there is no element in $\minmod_{\T,S}(\phi)$ strictly less than $(|\s^{\B}|)_{\s\in S}$: $\B$ is already a $\T$-interpretation that satisfies $\phi$ by choice of $\B$. Now, if there are $\T$-interpretations $\C$ that satisfy $\phi$ with $(|\s^{\C}|)_{\s\in S}\neq(|\s^{\B}|)_{\s\in S}$, we can take a minimal one $\C^{\prime}$ among them (so that it lies in $\minmod_{\T,S}(\phi)$): since we are assuming that no element of $\minmod_{\T,S}(\phi)$ is strictly less than $(|\s^{\B}|)_{\s\in S}$, we cannot have $(|\s^{\C^{\prime}}|_{\s\in S})_{\s\in S}<(|\s^{\B}|)_{\s\in S}$. So: either $(|\s^{\C}|)_{\s\in S}\geq (|\s^{\C^{\prime}}|)_{\s\in S}>(|\s^{\B}|)_{\s\in S}$ (and thus there is a $\s\in S$ such that $|\s^{\B}|<|\s^{\C}|$); or $(|\s^{\C}|_{\s\in S})_{\s\in S}$ and $(|\s^{\C^{\prime}}|)_{\s\in S}$, and $(|\s^{\B}|)_{\s\in S}$ are not comparable (and again there is a $\s\in S$ such that $|\s^{\B}|<|\s^{\C}|$). To summarize, $(|\s^{\B}|)_{\s\in S}$ is in $\minmod_{\T,S}(\phi)$.

        So define $(n_{\s}^{\prime})_{\s\in S^{\prime}}$ to be $S^{\prime}((m_{\s}^{\prime})_{\s\in S})$, and thus an element of $S^{\prime}\circ\minmod_{\T,S}(\phi)$: we have that $(n_{\s}^{\prime})_{\s\in S^{\prime}}\leq (|\s^{\B}|)_{\s\in S}<(n_{\s})_{\s\in S^{\prime}}$, contradicting the fact that $(n_{\s})_{\s\in S^{\prime}}$ is a minimal element of $S^{\prime}\circ\minmod_{\T,S}(\phi)$. So there is a $\s\in S^{\prime}$ such that $|\s^{\B}|>n_{\s}$, and thus $(n_{\s})_{\s\in S^{\prime}}$ is in $\minmod_{\T,S^{\prime}}(\phi)$.

    \end{enumerate}

    Reciprocally, take an element $(n_{\s})_{\s\in S^{\prime}}$ in $\minmod_{\T,S^{\prime}}(\phi)$, and we prove that it is in $\minimal(S^{\prime}\circ\minmod_{\T,S}(\phi))$. Take a $\T$-interpretation $\A$ that satisfies $\phi$ such that $(|\s^{\A}|)_{\s\in S}$ is in
    \[\minimal\{(|\s^{\A}|)_{\s\in S} : \text{$\A$ is a $\T$-interpretation that satisfies $\phi$ with $(|\s^{\A}|)_{\s\in S^{\prime}}=(n_{\s})_{\s\in S^{\prime}}$}\},\]
    and with that we prove that $(|\s^{\A}|)_{\s\in S}$ is in $\minmod_{\T,S}(\phi)$, and thus $(n_{\s})_{\s\in S^{\prime}}$ is in $S^{\prime}\circ\minmod_{\T,S}(\phi)$.

    \begin{enumerate}
        \item Of course $\A$ is a $\T$-interpretation that satisfies $\phi$, so the first condition for $(|\s^{\A}|)_{\s\in S}$ to be in $\minmod_{\T,S}(\phi)$ is met.

        \item Let $\B$ be a $\T$-interpretation that satisfies $\phi$ such that $(|\s^{\A}|)_{\s\in S}\neq (|\s^{\B}|)_{\s\in S}$: if $(|\s^{\B}|)_{\s\in S^{\prime}}=(n_{\s})_{\s\in S^{\prime}}$, by choice of $\A$ we have that it is impossible to have $(|\s^{\A}|)_{\s\in S}>(|\s^{\B}|)_{\s\in S}$, and so there exists a $\s\in S\setminus S^{\prime}$ such that $|\s^{\B}|>|\s^{\A}|$; if, instead, $(|\s^{\B}|)_{\s\in S^{\prime}}\neq (n_{\s})_{\s\in S^{\prime}}$, from the fact that $(n_{\s})_{\s\in S^{\prime}}$ is in $\minmod_{\T,S^{\prime}}(\phi)$ there must exist $\s\in S^{\prime}\subseteq S$ such that $|\s^{\B}|>|\s^{\A}|$. Either way, we have proved $(|\s^{\A}|)_{\s\in S}$ meets the criteria to be in $\minmod_{\T,S}(\phi)$, and so $(n_{\s})_{\s\in S^{\prime}}$ is in $S^{\prime}\circ\minmod_{\T,S}(\phi)$.
    \end{enumerate}

    Finally, suppose $(|\s^{\B}|)_{\s\in S}$ is in $\minmod_{\T,S}(\phi)$, and that $(|\s^{\B}|)_{\s\in S^{\prime}}\neq (n_{\s})_{\s\in S^{\prime}}$. $(n_{\s})_{\s\in S^{\prime}}\in\minmod_{\T,S^{\prime}}(\phi)$, implies there exists a $\s\in S^{\prime}$ such that $|\s^{\B}|>n_{\s}$, implying therefore that one does not have $(|\s^{\B}|)_{\s\in S^{\prime}}<(n_{\s})_{\s\in S^{\prime}}$ and thus $(n_{\s})_{\s\in S^{\prime}}$ is indeed minimal in $S^{\prime}\circ\minmod_{\T,S}(\phi)$.
\end{proof}

\SMplusESimpliesCMMF*

\begin{proof}
We will prove the result for $S=\S_{\Sigma}$, and by \Cref{CMMF of subset} the result follows for any $S\subseteq\S_{\Sigma}$. Suppose $\S_{\Sigma}=\{\s_{1},\ldots,\s_{n}\}$, and consider the set
\[A=\{(|\s_{1}^{\A}|,\ldots,|\s_{n}^{\A}|): \text{$\A$ is a $\T$-interpretation}\}\cap\N^{n},\]
which is non-empty if $\T$ is not contradictory thanks to \Cref{LowenheimSkolem}: indeed, if $\T$ has an interpretation, the theorem guarantees it has an interpretation whose domains are at most countable; then, because of \Cref{technical lemma}, there must exist $\T$-interpretations $\A_{1}$ through $\A_{m}$ corresponding to the minimal elements $(p^{1}_{1},\ldots,p^{1}_{n})$ through $(p^{m}_{1},\ldots,p^{m}_{n})$ of $A$. Notice that, in what is to come, we will use the fact that the underlying structures of two $\Sigma$-interpretations $\A$ and $\B$ with $(|\s_{1}^{\A}|,\ldots,|\s_{n}^{\A}|)=(|\s_{1}^{\B}|,\ldots,|\s_{n}^{\B}|)$ are isomorphic, thanks to the fact that $\Sigma$ is empty.

Now, take a quantifier-free formula $\phi$, and consider the sets $V_{\s}$ of variables of sort $\s$ in $\phi$, which can be found alghoritmically  (let $V$ denote the union of all of these sets); next, let $\eq{\phi}$ be the set of all possible equivalence relations on $V$ that respect sorts (that is, if $x$ and $y$ are of sort different sorts, we cannot have $xEy$ for any $E\in \eq{\phi}$), a set that is clearly finite and can be found alghoritmically. We can then define the set
    \[B(\phi)=\{(|V_{\s_{1}}/E|,\ldots,|V_{\s_{n}}/E|):\text{$E\in\eq{\phi}$, and $\delta^{E}_{V}$ and $\phi$ are equivalent}\}\]
keeping in mind that: $B(\phi)$ is empty iff $\phi$ is not $\T$-satisfiable; $V_{\s}/E$ is simply the quotient of $V_{\s}$ by the restriction of $E$ to $V_{\s}$; and whether $\delta^{E}_{V}$ and $\phi$ are equivalent follows from the satisfiability problem of equality logic, which is decidable. Let $(q^{1}_{1},\ldots,q^{1}_{n})$ through $(q^{k}_{1},\ldots,q^{k}_{n})$ be the elements of $B(\phi)$ (dependent of $\phi$, unlike the $(p^{i}_{1},\ldots,p^{i}_{n})$), and define 
\[\minmod_{\T,\S_{\Sigma}}(\phi)=\minimal\{(\max\{p^{i}_{1}, q^{j}_{1}\},\ldots,\max\{p^{i}_{n}, q^{j}_{n}\}) : 1\leq i\leq m, 1\leq j\leq k\},\]
where $\max\{p^{i}_{l}, q^{j}_{l}\}$ is $\aleph_{0}$ if $p^{i}_{l}$ is $\aleph_{0}$: this is a computable function, and its output is always finite. Assume now that $\phi$ is $\T$-satisfiable, take then in $\minimal\{(\max\{p^{i}_{1}, q^{j}_{1}\},\ldots,\max\{p^{i}_{n}, q^{j}_{n}\})\}$ an element $(r_{1},\ldots,r_{n})$, say equal to $(\max\{p^{i}_{1}, q^{j}_{1}\},\ldots,\max\{p^{i}_{n}, q^{j}_{n}\})$.
\begin{enumerate}
    \item Pick a $\Sigma$-interpretation $\A$ with $(|\s_{1}^{\A},\ldots,|\s_{n}^{\A}|)=(r_{1},\ldots,r_{n})$: there is an equivalence $E\in\eq{\phi}$ with $(|V_{\s_{1}}/E|,\ldots,|V_{\s_{n}}/E|)\leq (q^{j}_{1},\ldots,q^{j}_{n})$, and so by changing the values assigned by $\A$ to variables we obtain a $\Sigma$-interpretation $\A^{\prime}$ with $(|\s_{1}^{\A^{\prime}}|,\ldots,|\s_{n}^{\A^{\prime}}|)=(|\s_{1}^{\A}|,\ldots,|\s_{n}^{\A}|)$, and that satisfies $\delta^{E}_{V}$ and thus $\phi$; since $(|\s_{1}^{\A^{\prime}},\ldots,|\s_{n}^{\A^{\prime}}|)\geq (p^{i}_{1},\ldots,p^{i}_{n})$ and $\T$ is smooth, we obtain that $\A^{\prime}$ is a $\T$-interpretation, which in addition satisfies $\phi$, as we needed to show.
    
    \item Now let $\B$ be a $\T$-interpretation that satisfies $\phi$ with $(|\s_{1}^{\B}|,\ldots,|\s_{n}^{\B}|)\neq(r_{1},\ldots,r_{n})$, and since $\B$ satisfies $\phi$ it also satisfies $\delta^{E}_{V}$ for some $E\in\eq{\phi}$, and thus $(|\s_{1}^{\B}|,\ldots,|\s_{n}^{\B}|)\geq (q^{j^{\prime}}_{1}, \ldots,q^{j^{\prime}}_{n})$ for some $1\leq j^{\prime}\leq k$. Furthermore, given that $\B$ is a $\T$-interpretation, $(|\s_{1}^{\B}|,\ldots,|\s_{n}^{\B}|)\geq (p^{i^{\prime}}_{1},\ldots,p^{i^{\prime}}_{n})$ for some $1\leq i^{\prime}\leq m$, meaning 
    \[(|\s_{1}^{\B}|,\ldots,|\s_{n}^{\B}|)\geq (\max\{p^{i^{\prime}}_{1}, q^{j^{\prime}}_{1}\},\ldots,\max\{p^{i^{\prime}}_{n},q^{j^{\prime}}_{n}\}).\]
    If $(\max\{p^{i^{\prime}}_{1}, q^{j^{\prime}}_{1}\},\ldots,\max\{p^{i^{\prime}}_{n},q^{j^{\prime}}_{n}\})=(r_{1},\ldots,r_{n})$, given that $(|\s_{1}^{\B}|,\ldots,|\s_{n}^{\B}|)$ must be different from the latter, we obtain there exists $\s_{l}\in\S_{\Sigma}$ such that $|\s_{l}^{\B}|>r_{l}$; otherwise, since $(r_{1},\ldots,r_{n})$ is a minimal element of the set $\{(\max\{p^{i}_{1}, q^{j}_{1}\},\ldots,\max\{p^{i}_{n}, q^{j}_{n}\})\}$, there must exist $\s_{l}\in\S_{\Sigma}$ such that $\max\{p^{i^{\prime}}_{l}, q^{j^{\prime}}_{l}\}>r_{l}$, and from the fact that $|\s_{l}^{\B}|\geq \max\{p^{i^{\prime}}_{l}, q^{j^{\prime}}_{l}\}$ we get $|\s_{l}^{\B}|>r_{l}$, as we needed to show.
    
\end{enumerate}
Now, for the reciprocal, still assuming that $\phi$ is $\T$-satisfiable, suppose that $(r_{1},\ldots,r_{n})$ is in the set $\minmod_{\T,\S_{\Sigma}}(\phi)$. We know, from the fact that $(r_{1},\ldots,r_{n})$ is in $\minmod_{\T,\S_{\Sigma}}(\phi)$, that there is a $\T$-interpretation $\A$ that satisfies $\phi$ such that $(r_{1},\ldots,r_{n})=(|\s_{1}^{\A}|,\ldots,|\s_{n}^{\A}|)$: because $\A$ is a $\T$-interpretation, there exists a $1\leq i\leq m$ such that $(|\s_{1}^{\A}|,\ldots,|\s_{n}^{\A}|)\geq (p^{i}_{1},\ldots,p^{i}_{n})$. And, because $\A$ satisfies $\phi$, it must satisfy some $\delta^{E}_{V}$, and thus $(|\s_{1}^{\A}|,\ldots,|\s_{n}^{\A}|)\geq (|V_{\s_{1}}/E|,\ldots,|V_{\s_{n}}/E|)=(q^{j}_{1},\ldots,q^{j}_{n})$, for some $1\leq j\leq k$. Thus $(r_{1},\ldots,r_{n})\geq (\max\{p^{i}_{1},q^{j}_{1}\},\ldots,\max\{p^{i}_{n},q^{j}_{n}\})$. Now, we prove that there is no $(\max\{p^{i^{\prime}}_{1},q^{j^{\prime}}_{1}\},\ldots,\max\{p^{i^{\prime}}_{n},q^{j^{\prime}}_{n}\})$ such that $(r_{1},\ldots,r_{n})>(\max\{p^{i^{\prime}}_{1},q^{j^{\prime}}_{1}\},\ldots,\max\{p^{i^{\prime}}_{n},q^{j^{\prime}}_{n}\})$, what will finish proving that not only $(r_{1},\ldots,r_{n})$ is in $\{(\max\{p^{i}_{1}, q^{j}_{1}\},\ldots,\max\{p^{i}_{n}, q^{j}_{n}\})\}$, but that is also a minimal element of that set.

Suppose that is not true: since there is a $\T$-interpretation $\B$ with $(|\s_{1}^{\B}|,\ldots,|\s_{n}^{\B}|)=(p^{i^{\prime}}_{1},\ldots,p^{i^{\prime}}_{n})$, and since $\T$ is smooth, there is a $\T$-interpretation $\A$ with $(|\s_{1}^{\A}|,\ldots,|\s_{n}^{\A}|)=(\max\{p^{i^{\prime}}_{1},q^{j^{\prime}}_{1}\},\ldots,\max\{p^{i^{\prime}}_{n},q^{j^{\prime}}_{n}\})$. Now, because $(|\s_{1}^{\A}|,\ldots,|\s_{n}^{\A}|)\geq (q^{j^{\prime}}_{1},\ldots,q^{j^{\prime}}_{n})$, which equals $(|V_{\s_{1}}/E|,\ldots,|V_{\s_{n}}/E|)$ for some $E\in\eq{\phi}$, we can change the values assigned by $\A$ to variables so to obtain a $\T$-interpretation $\A^{\prime}$ that satisfies $\delta^{E}_{V}$, and thus $\phi$. Of course 
\[(|\s_{1}^{\A^{\prime}}|,\ldots,|\s_{n}^{\A^{\prime}}|)=(\max\{p^{i^{\prime}}_{1},q^{j^{\prime}}_{1}\},\ldots,\max\{p^{i^{\prime}}_{n},q^{j^{\prime}}_{n}\})<(r_{1},\ldots,r_{n}),\]
and in particular both are different, implying from the fact that $(r_{1},\ldots,r_{n})$ lies in $\minmod_{\T,\S_{\Sigma}}(\phi)$ that there must be $\s_{l}\in\S_{\Sigma}$ such that $|\s_{l}^{\A^{\prime}}|<r_{l}$, what clearly constitutes a contradiction and finishes the proof.
\end{proof}

\section{\tp{Proof of \Cref{-SI+ES=>CMMF}}{Proof of -SI+ES=>CMMF}}

\minusSIplusESimpliesCMMF*

\begin{proof}
    We will prove the result for $S=\S_{\Sigma}$, and by \Cref{CMMF of subset} the full theorem follows; furthermore, we assume that $\S_{\Sigma}=\{\s_{1},\ldots,\s_{n}\}$. So, we begin by stating that there are only finitely many elements in the set 
    \[\textbf{Card}_{\T}=\{(|\s_{1}^{\A}|,\ldots,|\s_{n}^{\A}|) : \text{$\A$ is a $\T$-interpretation}\}\cap\N^{n}.\]
    If that were not true, consider the sets $\textbf{Card}_{\T}^{i}=\{|\s_{i}^{\A}|:\text{$\A$ is a $\T$-interpretation}\}\cap \N$, and since $\textbf{Card}_{\T}\subseteq \textbf{Card}_{\T}^{1}\times\cdots\times\textbf{Card}_{\T}^{n}$ we get that some $\textbf{Card}_{\T}^{i}$ must also be infinite. Take then the $\T$-interpretations $\A_{j}$, for $j\in\mathbb{N}$, corresponding to the different elements of $\textbf{Card}_{\T}^{i}$, and we have the set $|\s_{i}^{\A_{j}}|$ must be unbounded. Now, take the set $\Gamma=\{\psi_{\geq k}^{\s_{i}} : k\in\mathbb{N}\}$, and any finite subset of $\Gamma$ is satisfied by some $\A_{j}$ since their domains of sort $\s_{i}$ are unbounded, meaning by \Cref{compactness} that $\Gamma$ is $\T$-satisfiable, and thus $\T$ has an interpretation with an infinite domain of sort $\s_{i}$: but then $\T$ is stably infinite with respect to $\{\s_{i}\}$, since the signature $\Sigma$ is empty, against our hypothesis.

    So let $(m^{1}_{1}, \ldots, m^{1}_{n})$ through $(m^{m}_{1},\ldots,m^{m}_{n})$ be all elements of $\textbf{Card}_{\T}$, corresponding to the $\T$-interpretations $\A_{1}$ through $\A_{m}$; let as well $V$ be the set of all variables of a quantifier-free formula $\phi$, $V_{\s_{i}}$ be the variables in $\phi$ of sort $\s_{i}$, and $\eq{\phi}$ the set of all equivalence relations $E$ on $V$ such that, if $x$ and $y$ are of the same sort, $xEy$, and such that $\phi$ and $\delta_{V}^{E}$ are equivalent. We then state that
    \begin{multline*}
        \minmod_{\T,S}(\phi)=\minimal\{(m_{1},\ldots,m_{n})\in\textbf{Card}_{\T} :\\ (m_{1},\ldots,m_{n})\geq (|V_{\s_{1}}/E|,\ldots,|V_{\s_{n}}/E|)\text{ for some $E\in\eq{\phi}$}\},
    \end{multline*}
    is a computable minimal model function: indeed, $\textbf{Card}_{\T}$, as a finite set, is computable;
    and so is the set of tuples $(|V_{\s_{1}}/E|,\ldots,|V_{\s_{n}}/E|)$, given the number of equivalence relations $E$ in $\eq{\phi}$ is bounded by $2^{2^{|V|}}$, and finding for which ones $\phi$ and $\delta_{V}^{E}$ are equivalent boils down to the satisfiability problem of equality logic. Of course, the output of this function is always finite, given $\textbf{Card}_{\T}$ is finite, and it is actually easy to prove that it is empty whenever $\phi$ is not $\T$-satisfiable. For what follows, suppose $\phi$ is $\T$-satisfiable instead.

    \begin{enumerate}
        \item Suppose that $(m_{1},\ldots,m_{n})$ is in $\minmod_{\T,S}(\phi)$, and therefore equals $(|\s_{1}^{\A}|,\ldots,|\s_{n}^{\A}|)$ for some $\T$-interpretation $\A$ that satisfies $\phi$: let then $E$ be equivalence relation on $V$ induced by $\A$, and we have that $(|\s_{1}^{\A}|,\ldots,|\s_{n}^{\A}|)\geq (|V_{\s_{1}}/E|,\ldots,|V_{\s_{n}}/E|)$, meaning $(m_{1},\ldots,m_{n})$ is in the set 
        \begin{multline*}
            M(\phi)=\{(m_{1},\ldots,m_{n})\in\textbf{Card}_{\T} : (m_{1},\ldots,m_{n})\geq (|V_{\s_{1}}/E|,\ldots,|V_{\s_{n}}/E|)\text{ for some $E\in\eq{\phi}$}\}.
        \end{multline*}
        Suppose, however, that it is not a minimal element, meaning there exist a $(m_{1}^{\prime},\ldots,m_{n}^{\prime})$ in $\textbf{Card}_{\T}$ (and thus a $\T$-interpretation $\B$ with $(|\s_{1}^{\B}|,\ldots,|\s_{n}^{\B}|)=(m_{1}^{\prime},\ldots,m_{n}^{\prime})$) with $(m_{1}^{\prime},\ldots,m_{n}^{\prime})<(m_{1},\ldots,m_{n})$, and an $E\in\eq{\phi}$ such that $(m_{1}^{\prime},\ldots,m_{n}^{\prime})\geq (|V_{\s_{1}}/E|,\ldots,|V_{\s_{n}}/E|)$. Because of the last inequality, we get that by changing the values assigned to variables in $\B$ we can make a new interpretation $\B^{\prime}$ with $(|\s_{1}^{\B^{\prime}}|,\ldots,|\s_{n}^{\B^{\prime}}|)=(|\s_{1}^{\B}|,\ldots,|\s_{n}^{\B}|)$ (and thus making of $\B^{\prime}$ a $\T$-interpretation since $\Sigma$ is empty) that satisfies $\delta_{V}^{E}$, and thus $\phi$. But $(m_{1},\ldots,m_{n})$ is strictly smaller than $(m_{1}^{\prime},\ldots,m_{n}^{\prime})$, contradicting the fact that the first is in $\minmod_{\T,S}(\phi)$ and thus proving that this element is minimal in the relevant set.

        \item Reciprocally, suppose that $(m_{1},\ldots,m_{n})$ is a minimal element of the set $M(\phi)$: to start with, $(m_{1},\ldots,m_{n})\in\textbf{Card}_{\T}$, and so there is a $\T$-interpretation $\A_{j}$ such that $(m_{1},\ldots,m_{n})=(|\s_{1}^{\A_{j}}|,\ldots,|\s_{n}^{\A_{j}}|)$; furthermore, there is an equivalence $E$ in $\eq{\phi}$ such that $(m_{1},\ldots,m_{n})\geq (|V_{\s_{1}}/E|,\ldots,|V_{\s_{n}}/E|)$. Thus, by changing the values assigned by $\A_{j}$ to variables we may obtain a $\T$-interpretation $\A$ that satisfies $\delta_{V}^{E}$ (and thus $\phi$) with $(m_{1},\ldots,m_{n})=(|\s_{1}^{\A}|\ldots,|\s_{n}^{\A}|)$.

        Finally, take a $\T$-interpretation $\B$ that satisfies $\phi$ with $(|\s_{1}^{\B}|\ldots,|\s_{n}^{\B}|)\neq (m_{1},\ldots,m_{n})$: take then the equivalence $E$ induced by $\B$, and we know that $(|\s_{1}^{\B}|\ldots,|\s_{n}^{\B}|)\geq (|V_{\s_{1}}/E|,\ldots,|V_{\s_{n}}/E|)$, meaning, since $(|\s_{1}^{\B}|\ldots,|\s_{n}^{\B}|)$ must lie in $\textbf{Card}_{\T}$, that this tuple is in $M(\phi)$. Since $(m_{1},\ldots,m_{n})$ is a minimal element of this set, we must have $(m_{1},\ldots,m_{n})\leq(|\s_{1}^{\B}|\ldots,|\s_{n}^{\B}|)$, and thus there exists $1\leq i\leq n$ such that $m_{i}<|\s_{i}^{\B}|$, proving $(m_{1},\ldots,m_{n})\in\minmod_{\T,S}(\phi)$ and finishing the proof.
    \end{enumerate}
\end{proof}

\section{\tp{Proof of \Cref{-FMP+ES+OS=>CMMF}}{Proof of -FMP+ES+OS=>CMMF}}

\minusFMPplusESplusOSimpliesCMMF*

\begin{proof}
    Since $\T$ does not have the finite model property, there must exist a $\T$-interpretation $\A$ with infinite domain (we may assume of cardinality $\aleph_{0}$ by \Cref{LowenheimSkolem}), and a quantifier-free formula $\phi$ that is satisfied by $\A$ but not any finite $\T$-interpretation: we can then easily prove that no $\T$-interpretation $\B$ has cardinality $|\vars_{\s}(\phi)^{\A}|<|\s^{\B}|<\aleph_{0}$. Indeed, suppose that were not the case: we could then define a $\T$-interpretation $\B^{\prime}$ by changing only the value assigned by $\B$ to the variables in $\vars_{\s}(\phi)$, so that $x^{\B^{\prime}}=y^{\B^{\prime}}$ iff $x^{\A}=y^{\A}$. Since the atomic subformulas of $\phi$ are either equalities or disequalities, and $\A$ satisfies $\phi$, we get that $\B^{\prime}$ satisfies $\phi$, leading to a contradiction.

    So suppose that $\{|\s^{\A}|:\text{$\A$ is a $\T$-interpretation}\}\cap\N$ equals $\{m_{1},\ldots,m_{n}\}$, for $m_{1}<\cdots<m_{n}=\aleph_{0}$. Let then: $\phi$ be a quantifier-free $\Sigma_{2}$-formula: $V$ be the set of its variables of sort $\s$; and $\eq{\phi}$ be the set of equivalence relations $E$ on $V$ such that a $\phi$ and $\delta_{V}^{E}$ are equivalent (which is empty iff $\phi$ is not satisfiable). Finally we define $M(\phi)=\min\{|V/E| : E\in\eq{\phi}\}$ (that will equal $\aleph_{0}$ iff $\phi$ is not $\T$-satisfiable), and we then state that the following is a minimal model function: 
    \[\minmod_{\T,S}(\phi)=\{N(\phi)\},\quad\text{where}\quad N(\phi)=\min\{m_{i} : \text{$m_{i}\geq M(\phi)$, $1\leq i\leq n$}\},\]
    where $S=\{\s\}$, and if there is no $m_{i}$ greater than or equal to $M(\phi)$, as usual $N(\phi)=\aleph_{0}$. This is of course a computable function, as finding $\eq{\phi}$ can be done algorithmically, and the output is definitely finite (indeed, it is a singleton). Now assume $\phi$ is $\T$-satisfiable, and we begin by proving that there is a $\T$-interpretation $\A$ that satisfies $\phi$ with $|\s^{\A}|=N(\phi)$: let $E\in\eq{\phi}$ be the equivalence with $|V/E|=M(\phi)$, and take a $\Sigma_{1}$-interpretation $\A^{\prime}$ that satisfies $\delta_{V}^{E}$ (and thus $\phi$); take as well a set $A$ with cardinality $N(\phi)-M(\phi)$ disjoint from the domain of $\A^{\prime}$. We can then define the $\Sigma_{1}$-interpretation $\A$ by making $\s^{\A}=V^{\A}\cup A$ (that has cardinality $N(\phi)$, making of $\A$ a $\T$-interpretation), and $x^{\A}=x^{\A^{\prime}}$ for any $x\in V$ and arbitrarily for all other variables. $\A$ satisfies $\phi$ as it still satisfies $\delta_{V}^{E}$, as we wished to show.
    
    So, suppose there is a $\T$-interpretation $\B$ that satisfies $\phi$ with $|\s^{\B}|<N(\phi)$: the equivalence $F$ induced by $\B$ on $V$ is of course in $\eq{\phi}$, yet $|V/F|\leq|\s^{\B}|<N(\phi)$; this leads to a contradiction as $M(\phi)=\min\{|V/E| : E\in\eq{\phi}\}\leq |\s^{\B}|$ and 
    $N(\phi)=\min\{m_{1},\ldots,m_{n} : M(\phi)\leq m_{i}\}\geq M(\phi)$,
    since $|\s^{\B}|$ must be in $\{m_{1},\ldots,m_{n}\}$.
\end{proof}

\section{\tp{Proof of \Cref{CV+-SI+-FMP+-CMMF}}{Proof of CV+-SI+-FMP+-CMMF}}

\lostofassumptions*

\begin{proof}
    First, we state that because $\T$ is not stably infinite, there is a $(k_{1},k_{2})\in\mathbb{N}^{2}$ such that no $\T$-interpretation $\A$ has $(|\s^{\A}|,|\s_{2}^{\A}|)>(k_{1},k_{2})$. Indeed, suppose that that is not true, and we affirm that there exists a $\T$-interpretation $\A$ with $|\s^{\A}|,|\s_{2}^{\A}|\geq\aleph_{0}$: were that not true, the set 
    \[\Gamma=\ax{\T}\cup\{\psi^{\s}_{\geq k_{1}}\wedge\psi^{\s_{2}}_{\geq k_{2}} : (k_{1},k_{2})\in \mathbb{N}^{2}\}\]
    would be contradictory; by the many-sorted compactness Theorem, one would then be able to find $(k_{1}^{1},k_{2}^{1})$ through $(k_{1}^{m},k_{2}^{m})$ in $\mathbb{N}^{2}$ such that
    \[\Gamma_{0}=\ax{\T}\cup\{\psi^{\s}_{\geq k_{1}^{j}}\wedge\psi^{\s_{2}}_{\geq k_{2}^{j}} : 1\leq j\leq m\}\]
    is contradictory. Yet, by hypothesis, there exists a $\T$-interpretation $\A$ such that 
    \[(|\s^{\A}|,|\s_{2}^{\A}|)>(\max\{k_{1}^{1},\ldots,k_{1}^{m}\},\max\{k_{2}^{1},\ldots,k_{2}^{m}\},\]
    which therefore satisfies $\Gamma_{0}$, leading to a contradiction. But, if there is a $\T$-interpretation $\A$ with $|\s^{\A},|\s_{2}^{\A}|\geq\aleph_{0}$, then $\T$ is necessarily stably infinite: indeed, take any quantifier-free formula $\phi$ and a $\T$-interpretation $\B$ that satisfies $\phi$; take as well sets $A_{1}$ and $A_{2}$ with cardinalities, respectively, $|\s^{\A}|$ through $|\s_{2}^{\A}|$, both disjoint from the domains of $\B$. We then define a $\T$-interpretation $\C$ by making: $\s^{\C}=\vars_{\s}(\phi)^{\B}\cup A_{1}$ and $\s_{2}^{\C}=\vars_{\s_{2}}(\phi)^{\B}\cup A_{2}$, which have cardinalities $|A_{1}|=|\s^{\A}|$ and $|A_{2}|=|\s_{2}^{\A}|$; $x^{\C}=x^{\B}$ for $x\in\vars_{\s}(\phi)$, and $u^{\C}=u^{\B}$ for $u\in\vars_{\s_{2}}(\phi)$; and $x^{\C}$ and $u^{\C}$ arbitrary values from, respectively, $\s^{\C}$ and $\s_{2}^{\C}$, for all other variables $x$ of sort $\s$ and $u$ of sort $\s_{2}$. Then it is true that $\C$ satisfies $\phi$, has $|\s^{\C}|,|\s_{2}^{\C}|\geq\aleph_{0}$, and $(|\s^{\C}|,|\s_{2}^{\C}|)=(|\s^{\A}|,|\s_{2}^{\A}|)$, making of $\C$ a $\T$-interpretation.

    Now, we state that there cannot exist $\T$-interpretations $\A$ and $\B$ such that $|\s^{\A}|>1$ and $|\s_{2}^{\B}|>1$: if that were the case, then $\T$ would not be convex. Indeed, suppose that there are such $\A$ and $\B$, and it is clear that 
    \[\vDash_{\T}\bigvee_{i=1}^{k_{1}}\bigvee_{j=i+1}^{k_{1}+1}(x_{i}=x_{j})\vee\bigvee_{p=1}^{k_{2}}\bigvee_{q=p+1}^{k_{2}+1}(u_{p}=u_{q}),\]
    for variables $x_{i}$ of sort $\s$, and $u_{p}$ of sort $\s_{2}$: we know that no $\T$-interpretation $\C$ can have $(|\s^{\C}|,|\s_{2}^{\C}|)>(k_{1},k_{2})$, so either $k_{1}>|\s^{\C}|$, and the pigeonhole principle guarantees the satisfaction of $\bigvee_{i=1}^{k_{1}}\bigvee_{j=i+1}^{k_{1}+1}(x_{i}=x_{j})$; or $k_{2}>|\s_{2}^{\C}|$, and we get that $\bigvee_{p=1}^{k_{2}}\bigvee_{q=p+1}^{k_{2}+1}(u_{p}=u_{q})$ is satisfied. But then we can have neither $\vDash_{\T}x_{i}=x_{j}$ nor $\vDash_{\T}u_{p}=u_{q}$, for any $1\leq i<j\leq k_{1}+1$ or $1\leq p<q\leq k_{2}+1$: changing only the value assigned to $x_{j}$ by $\A$ so that $x_{j}^{\A^{\prime}}\neq x_{i}^{\A}$, we get a $\T$-interpretation $\A^{\prime}$ that falsifies $x_{i}=x_{j}$; and by changing only the value assigned to $u_{q}$, this time by $\B$, so that $u_{q}^{\B^{\prime}}\neq u_{p}^{\B}$, we get a $\T$-interpretation $\B^{\prime}$ that falsifies $u_{p}=u_{q}$.

    So, assume, without loss of generality, that $|\s_{2}^{\A}|=1$ for all $\T$-interpretation $\A$: we state that, if $\{|\s^{\A}| : \text{$\A$ is a $\T$-interpretation}\}\cap \N$ is finite, then $\T$ has a computable minimal model function. Indeed, let $m_{1}<m_{2}<\cdots<m_{n-1}<m_{n}$ be the elements of this set; let as well $\phi$ be a $\Sigma_{2}$-formula, $V$ the set of its variables of sort $\s$, and $\eq{\phi}$ the set of equivalences $E$ on $V$ such that a $\Sigma_{2}$-interpretation $\A$ satisfies $\phi$ iff it satisfies $\delta_{V}^{E}$ and has $|\s_{2}^{\A}|=1$ (notice $\eq{\phi}$ can be found algorithmically through an analysis of $\phi$ as a propositional formula). Define $M(\phi)=\min\{|V/E| : E\in\eq{\phi}\}$ (which will equal $\aleph_{0}$, as usual, iff $\eq{\phi}$ is empty), and we then state that 
    \[\minmod_{\T,S}(\phi)=\{(N(\phi),1)\},\quad\text{where}\quad N(\phi)=\min\{m_{i} : \text{$m_{i}\geq M(\phi)$, $1\leq i\leq n$}\}\]
    and $S=\{\s,\s_{2}\}$ (if $M(\phi)$ is greater than all $m_{i}$, what happens iff $\phi$ is not $\T$-satisfiable, $N(\phi)$ then equals $\aleph_{0}$), is a minimal model function: it is computable, as $\eq(\phi)$ can be found algorithmically, and always outputs a finite set, of cardinality one actually. 
    We start by proving that, if $\phi$ is $\T$-satisfiable, there is a $\T$-interpretation $\A$ that satisfies $\phi$ with $(|\s^{\A}|,|\s_{2}^{\A}|)=(N,1)$. So, let $\A^{\prime}$ be a $\Sigma_{2}$-interpretation that satisfies $\delta_{V}^{E}$ for $E\in\eq{\phi}$ such that $|V/E|=M$, and thus $\phi$: let as well $A$ be a set with $N-M$ elements disjoint from $\s^{\A^{\prime}}$ and $\s_{2}^{\A^{\prime}}$. We define the $\Sigma_{2}$-interpretation $\A$ by making: $\s_{2}^{\A}=\s_{2}^{\A^{\prime}}$ (so $|\s_{2}^{\A}|=1$); $\s^{\A}=V^{\A}\cup A$ (which has $|V^{\A}|+|A|=|V/E|+(N-M)=M+(N-M)=N$ elements, thus making of $\A$ a $\T$-interpretation as $N\in\{m_{1},\ldots,m_{n}\}$); $u^{\A}$ equal to the only element of $\s_{2}^{\A}$, for any variable $u$ of sort $\s_{2}$; $x^{\A}=x^{\A^{\prime}}$ for any variable $x\in V$, and arbitrarily for all other variables of sort $\s$. Then $\A$ satisfies $\delta_{V}^{E}$ and thus $\phi$, as we wished to show, meaning $(N,1)$ satisfies the first condition to be in $\minmod_{\T,S}(\phi)$.
    
    Now, suppose there is a $\T$-interpretation $\B$ that satisfies $\phi$ with $(|\s^{\B}|,|\s_{2}^{\B}|)=(p,q)$ and $(p,q)\neq (N,1)$: since $\B$ is a $\T$-interpretation, we already get that $q=1$. Suppose then that $p<N$, and let $F$ be the equivalence induced by $\B$: of course $F\in\eq{\phi}$, but we have that $|V/F|=V^{\B}\leq |\s^{\B}|<N$. This contradicts the fact that $M=\min\{|V/E| : E\in\eq{\phi}\}$ and $N=\min\{m_{1}, \ldots,m_{n} : M\leq m_{i}\}$, once we remember that $|\s^{\B}|\in\{m_{1},\ldots,m_{n}\}$. It is, furthermore, obvious that there cannot be any element $(s,t)$ in $\minmod_{\T,S}(\phi)$ different from $(N,1)$, as we forcibly have $t=1$ and thus either $s>N$ or $N>s$.

    As we assumed that $\T$ does not have a computable minimal model function, we must have then that $\{|\s^{\A}| : \text{$\A$ is a $\T$-interpretation}\}\cap \N$ is infinite: we can then prove that $\T$ has the finite model property, contradicting our hypothesis and proving the theorem. So, take a quantifier-free formula $\phi$ and a $\T$-interpretation $\A$ that satisfies $\phi$: because $\{|\s^{\A}| : \text{$\A$ is a $\T$-interpretation}\}\cap \mathbb{N}$ is infinite, there is a $\T$-interpretation $\B$ with $\aleph_{0}>|\s^{\B}|\geq |\vars_{\s}(\phi)^{\A}|$; let $A$ be a set with $|\s^{\B}|-|\vars_{\s}(\phi)^{\A}|$ elements, disjoint from $\s^{\A}$ and $\s_{2}^{\A}$. We define a $\Sigma_{2}$-interpretation $\C$ by making: $\s_{2}^{\C}=\s_{2}^{\A}$ (which then has cardinality $1$); $\s^{\C}=\vars_{\s}(\phi)^{\A}\cup A$ (which then has cardinality $|\s^{\B}|$, making of $\C$ a $\T$-interpretation); $u^{\C}$ equal to the only element of $\s_{2}^{\C}$ for all variables $u$ of sort $\s_{2}$; $x^{\C}=x^{\A}$ for all $x\in\vars_{\s}(\phi)$, and arbitrarily for all other variables of sort $\s$. We then have that $\C$ is a $\T$-interpretation that obviously satisfies $\phi$, with $|\s_{2}^{\C}|=1$ and $|\s^{\C}|<\aleph_{0}$, and thus we are done.
    
\end{proof}

\section{\tp{Proof of \Cref{mainresult}}{Proof of mainresult}}

\mainresult*

The ``if" part of \Cref{mainresult} was proven in \Cref{relationships}, as well as in \cite{CADE,FroCoS}.
The ``only if" part amounts to providing examples for theories that admit all the remaining
combinations of properties.
In \Cref{sec:newtheoriesapp,sec:oldtheoriesapp}, we provide
axiomatizations of examples.
Then, the remaining cases can be obtained by applying, on these examples,
\Cref{thm:addsresult,thm:addfresult,thm:addncresult}.
The name of the theories, as well as the properties that they admit,
are given in \Cref{tab-summary-app}.
Concretely, every line in the table corresponds to a combination of properties, defined by the first 8 columns, where
$T$ and $F$ correspond to the property and its negation,
respectively.
The table omits lines that correspond to combinations
that are impossible for all signature types.
Cells with red background correspond to impossible combinations,
white cells with black font correspond to examples,
and white cells with red font correspond to unknown cases.
The table distinguishes between empty and non-empty signatures,
as well as one-sorted (OS) and many-sorted (MS) signatures.

\newpage

\afterpage{
\renewcommand{\arraystretch}{1.0}
\begin{small}
\begin{longtable}{|P{0.5cm}|P{0.5cm}|P{0.5cm}|P{0.5cm}|P{0.5cm}|P{0.5cm}|P{0.5cm}|P{0.5cm}||P{1cm}P{1cm}P{1cm}P{1cm}|P{0.5cm}|}
\hline
\multicolumn{8}{|c|}{} & \multicolumn{2}{c|}{Empty} & \multicolumn{2}{c|}{Non-empty} & \\
\hline
$\stainf$ & $\smooth$ & $\finwit$ & $\strfinwit$ & $\convex$ & $\finmodpro$ & $\stafin$ & $\cmmf$ & \multicolumn{1}{c|}{OS} & \multicolumn{1}{c|}{MS} &\multicolumn{1}{c|}{OS} & \multicolumn{1}{c|}{MS}& $N^{\underline{o}}$\\\hhline{-------------}

\multirow{36}{*}{$T$ } & \multirow{20}{*}{$T$ } &
\multirow{12}{*}{$T$ } & \multirow{4}{*}{$T$ } &
\multirow{2}{*}{$T$ } & \multirow{2}{*}{$T$ } &
\multirow{2}{*}{$T$ } & \multirow{1}{*}{$T$ } &
\multicolumn{1}{c|}{$\Tgeqn$}&\multicolumn{1}{c|}{$\adds{\Tgeqn}$}&
\multicolumn{1}{c|}{$\addf{\Tgeqn}$}&\multicolumn{1}{c|}{$\addf{\adds{\Tgeqn}}$}&
1\\\hhline{~~~~~~~------}

 &  &
 &  &
 &  &
 & \multirow{1}{*}{$F$ } &
\multicolumn{2}{c|}{\Cref{SM+ES=>CMMF}\cellcolor{red!15}}&
\multicolumn{2}{c|}{\textcolor{red}{Unicorns 2.0}}&
2\\\hhline{~~~~---------}

 &  &
 &  &
\multirow{2}{*}{$F$ } & \multirow{2}{*}{$T$ } &
\multirow{2}{*}{$T$ } & \multirow{1}{*}{$T$ } &
\multicolumn{2}{c|}{\cellcolor{red!15}}&
\multicolumn{1}{c|}{$\addnc{\Tgeqn}$}&\multicolumn{1}{c|}{$\addnc{\adds{\Tgeqn}}$}&
3\\\hhline{~~~~~~~-*{1}{>{\arrayrulecolor{red!15}}|--}*{1}{>{\arrayrulecolor{black}}|-}--}

 &  &
 &  &
 &  &
 & \multirow{1}{*}{$F$ } &
\multicolumn{2}{c|}{\multirow{-2}{*}{\cite{CADE}\cellcolor{red!15}}}&
\multicolumn{2}{c|}{\textcolor{red}{Unicorns 2.0}}&
4\\\hhline{~~~----------}

 &  &
 & \multirow{8}{*}{$F$ } &
\multirow{4}{*}{$T$ } & \multirow{4}{*}{$T$ } &
\multirow{2}{*}{$T$ } & \multirow{1}{*}{$T$ } &
\multicolumn{2}{c|}{\cellcolor{red!15}}&
\multicolumn{2}{c|}{\textcolor{red}{Unicorns 3.0}}&
5\\\hhline{~~~~~~~-*{1}{>{\arrayrulecolor{red!15}}|--}*{1}{>{\arrayrulecolor{black}}|-}--}

 &  &
 &  &
 &  &
 & \multirow{1}{*}{$F$ } &
\multicolumn{2}{c|}{\multirow{-2}{*}{\cite{FroCoS}\cellcolor{red!15}}} &
\multicolumn{1}{c|}{$\TM$}&\multicolumn{1}{c|}{$\adds{\TM}$}&
6\\\hhline{~~~~~~-------}

 &  &
 &  &
 &  &
\multirow{2}{*}{$F$ } & \multirow{1}{*}{$T$ } &
\multicolumn{1}{c|}{\cellcolor{red!15}}&\multicolumn{1}{c|}{$\Ttwothree$}&
\multicolumn{1}{c|}{\cellcolor{red!15}}&\multicolumn{1}{c|}{$\addf{\Ttwothree}$}&
7\\\hhline{~~~~~~~-*{1}{>{\arrayrulecolor{red!15}}|-}*{1}{>{\arrayrulecolor{black}}|-}*{1}{>{\arrayrulecolor{red!15}}|-}*{1}{>{\arrayrulecolor{black}}|-}-}

 &  &
 &  &
 &  &
 & \multirow{1}{*}{$F$ } &
\multicolumn{1}{c|}{\multirow{-2}{*}{\cite{FroCoS}\cellcolor{red!15}}}&\multicolumn{1}{c|}{\Cref{SM+ES=>CMMF}\cellcolor{red!15}}&
\multicolumn{1}{c|}{\multirow{-2}{*}{\cite{FroCoS}\cellcolor{red!15}}}&\multicolumn{1}{c|}{$\Tseventeen$}&
8\\\hhline{~~~~---------}

 &  &
 &  &
\multirow{4}{*}{$F$ } & \multirow{4}{*}{$T$ } &
\multirow{2}{*}{$T$ } & \multirow{1}{*}{$T$ } &
\multicolumn{2}{c|}{\cellcolor{red!15}}&
\multicolumn{2}{c|}{\textcolor{red}{Unicorns 3.0}}&
9\\\hhline{~~~~~~~-*{1}{>{\arrayrulecolor{red!15}}|--}*{1}{>{\arrayrulecolor{black}}|-}--}

 &  &
 &  &
 &  &
 & \multirow{1}{*}{$F$ } &
\multicolumn{2}{c|}{\cellcolor{red!15}}&
\multicolumn{1}{c|}{$\TsM$}&\multicolumn{1}{c|}{$\adds{\TsM}$}&
10\\\hhline{~~~~~~--*{1}{>{\arrayrulecolor{red!15}}|--}*{1}{>{\arrayrulecolor{black}}|-}--}

 &  &
 &  &
 &  &
\multirow{2}{*}{$F$ } & \multirow{1}{*}{$T$ } &
\multicolumn{2}{c|}{\cellcolor{red!15}}&
\multicolumn{1}{c|}{\cellcolor{red!15}}&\multicolumn{1}{c|}{$\addnc{\Ttwothree}$}&
11\\\hhline{~~~~~~~-*{1}{>{\arrayrulecolor{red!15}}|---}*{1}{>{\arrayrulecolor{black}}|-}-}

 &  &
 &  &
 &  &
 & \multirow{1}{*}{$F$ } &
\multicolumn{2}{c|}{\multirow{-4}{*}{\cite{CADE}\cellcolor{red!15}}}&
\multicolumn{1}{c|}{\multirow{-2}{*}{\cite{FroCoS}\cellcolor{red!15}}}&\multicolumn{1}{c|}{$\Teighteen$}&
12\\\hhline{~~-----------}

 &  &
\multirow{8}{*}{$F$ } & \multirow{8}{*}{$F$ } &
\multirow{4}{*}{$T$ } & \multirow{2}{*}{$T$ } &
\multirow{1}{*}{$T$ } & \multirow{1}{*}{$F$ } &
\multicolumn{2}{c|}{\cellcolor{red!15}}&
\multicolumn{1}{c|}{$\Tbbs$}&\multicolumn{1}{c|}{$\Tbbeq$}&
13\\\hhline{~~~~~~--*{1}{>{\arrayrulecolor{red!15}}|--}*{1}{>{\arrayrulecolor{black}}|-}--}

 &  &
 &  &
 &  &
\multirow{1}{*}{$F$ } & \multirow{1}{*}{$F$ } &
\multicolumn{2}{c|}{\multirow{-2}{*}{\cite{CADE}\cellcolor{red!15}}}&
\multicolumn{1}{c|}{\cite{FroCoS}\cellcolor{red!15}}&\multicolumn{1}{c|}{$\Tbbtwo$}&
14\\\hhline{~~~~~--------}

 &  &
 &  &
 & \multirow{2}{*}{$F$ } & 
\multirow{2}{*}{$F$ } & \multirow{1}{*}{$T$ } &
\multicolumn{1}{c|}{$\Tinfty$}&\multicolumn{1}{c|}{$\adds{\Tinfty}$}&
\multicolumn{1}{c|}{$\addf{\Tinfty}$}&\multicolumn{1}{c|}{$\addf{\adds{\Tinfty}}$}&
15\\\hhline{~~~~~~~------}

 &  &
 &  &
 &  & 
 & \multirow{1}{*}{$F$ } &
\multicolumn{2}{c|}{\Cref{SM+ES=>CMMF}\cellcolor{red!15}}&
\multicolumn{1}{c|}{$\Tfifteen$}&\multicolumn{1}{c|}{$\adds{\Tfifteen}$}&
16\\\hhline{~~~~~--------}

 &  &
 &  &
\multirow{4}{*}{$F$ } & \multirow{2}{*}{$T$ } & 
\multirow{1}{*}{$T$ } & \multirow{1}{*}{$F$ } &
\multicolumn{2}{c|}{\cellcolor{red!15}}&
\multicolumn{1}{c|}{$\Tbbvee$}&\multicolumn{1}{c|}{$\adds{\Tbbvee}$}&
17\\\hhline{~~~~~~--*{1}{>{\arrayrulecolor{red!15}}|--}*{1}{>{\arrayrulecolor{black}}|-}--}

 &  &
 &  &
 &  & 
\multirow{1}{*}{$F$ } & \multirow{1}{*}{$F$ } &
\multicolumn{2}{c|}{\cellcolor{red!15}}&
\multicolumn{1}{c|}{\cite{FroCoS}\cellcolor{red!15}}&\multicolumn{1}{c|}{$\Tbbveeeq$}&
18\\\hhline{~~~~~---*{1}{>{\arrayrulecolor{red!15}}|--}*{1}{>{\arrayrulecolor{black}}|-}--}

 &  &
 &  &
 & \multirow{2}{*}{$F$ } & 
\multirow{2}{*}{$F$ } & \multirow{1}{*}{$T$ } &
\multicolumn{2}{c|}{\cellcolor{red!15}}&
\multicolumn{1}{c|}{$\addnc{\Tinfty}$}&\multicolumn{1}{c|}{$\addnc{\adds{\Tinfty}}$}&
19\\\hhline{~~~~~~~-*{1}{>{\arrayrulecolor{red!15}}|--}*{1}{>{\arrayrulecolor{black}}|-}--}

 &  &
 &  &
 &  & 
 & \multirow{1}{*}{$F$ } &
\multicolumn{2}{c|}{\multirow{-4}{*}{\cite{CADE}\cellcolor{red!15}}}&
\multicolumn{1}{c|}{$\Tsixteen$}&\multicolumn{1}{c|}{$\adds{\Tsixteen}$}&
20\\\hhline{~------------}

 & \multirow{16}{*}{$F$ } &
\multirow{8}{*}{$T$ } & \multirow{8}{*}{$F$ } &
\multirow{4}{*}{$T$ } & \multirow{4}{*}{$T$ } & 
\multirow{2}{*}{$T$ } & \multirow{1}{*}{$T$ } &
\multicolumn{1}{c|}{$\Teven$}&\multicolumn{1}{c|}{$\adds{\Teven}$}&
\multicolumn{1}{c|}{$\addf{\Teven}$}&\multicolumn{1}{c|}{$\addf{\adds{\Teven}}$}&
21\\\hhline{~~~~~~~------}

 &  &
 &  &
 &  & 
 & \multirow{1}{*}{$F$ } &
\multicolumn{1}{c|}{$\Tone$}&\multicolumn{1}{c|}{$\adds{\Tone}$}&
\multicolumn{1}{c|}{$\addf{\Tone}$}&\multicolumn{1}{c|}{$\addf{\adds{\Tone}}$}&
22\\\hhline{~~~~~~-------}

 &  &
 &  &
 &  & 
\multirow{2}{*}{$F$ } & \multirow{1}{*}{$T$ } &
\multicolumn{1}{c|}{\cellcolor{red!15}}&\multicolumn{1}{c|}{$\Tupinfty$}&
\multicolumn{1}{c|}{\cellcolor{red!15}}&\multicolumn{1}{c|}{$\adds{\Tupinfty}$}&
23\\\hhline{~~~~~~~-*{1}{>{\arrayrulecolor{red!15}}|-}*{1}{>{\arrayrulecolor{black}}|-}*{1}{>{\arrayrulecolor{red!15}}|-}*{1}{>{\arrayrulecolor{black}}|-}-}

 &  &
 &  &
 &  & 
 & \multirow{1}{*}{$F$ } &
\multicolumn{1}{c|}{\multirow{-2}{*}{\cite{FroCoS}\cellcolor{red!15}}}&\multicolumn{1}{c|}{$\Ttwo$}&
\multicolumn{1}{c|}{\multirow{-2}{*}{\cite{FroCoS}\cellcolor{red!15}}}&\multicolumn{1}{c|}{$\adds{\Ttwo}$}&
24\\\hhline{~~~~---------}

 &  &
 &  &
\multirow{4}{*}{$F$ } & \multirow{4}{*}{$T$ } & 
\multirow{2}{*}{$T$ } & \multirow{1}{*}{$T$ } &
\multicolumn{2}{c|}{\cellcolor{red!15}}&
\multicolumn{1}{c|}{$\addnc{\Teven}$}&\multicolumn{1}{c|}{$\addnc{\adds{\Teven}}$}&
25\\\hhline{~~~~~~~-*{1}{>{\arrayrulecolor{red!15}}|--}*{1}{>{\arrayrulecolor{black}}|-}--}

 &  &
 &  &
 &  & 
 & \multirow{1}{*}{$F$ } &
\multicolumn{2}{c|}{\cellcolor{red!15}}&
\multicolumn{1}{c|}{$\addnc{\Tone}$}&\multicolumn{1}{c|}{$\addnc{\adds{\Tone}}$}&
26\\\hhline{~~~~~~--*{1}{>{\arrayrulecolor{red!15}}|--}*{1}{>{\arrayrulecolor{black}}|-}--}

 &  &
 &  &
 &  & 
\multirow{2}{*}{$F$ } & \multirow{1}{*}{$T$ } &
\multicolumn{2}{c|}{\cellcolor{red!15}}&
\multicolumn{1}{c|}{\cellcolor{red!15}}&\multicolumn{1}{c|}{$\addnc{\Tupinfty}$}&
27\\\hhline{~~~~~~~-*{1}{>{\arrayrulecolor{red!15}}|---}*{1}{>{\arrayrulecolor{black}}|-}-}

 &  &
 &  &
 &  & 
 & \multirow{1}{*}{$F$ } &
\multicolumn{2}{c|}{\multirow{-4}{*}{\cite{CADE}\cellcolor{red!15}}}&
\multicolumn{1}{c|}{\multirow{-2}{*}{\cite{FroCoS}\cellcolor{red!15}}}&\multicolumn{1}{c|}{$\addnc{\Ttwo}$}&
28\\\hhline{~~-----------}

 &  &
\multirow{8}{*}{$F$ } & \multirow{8}{*}{$F$ } &
\multirow{4}{*}{$T$ } & \multirow{2}{*}{$T$ } &
\multirow{1}{*}{$T$ } & \multirow{1}{*}{$F$ } &
\multicolumn{1}{c|}{$\Tbb$}&\multicolumn{1}{c|}{$\adds{\Tbb}$}&
\multicolumn{1}{c|}{$\addf{\Tbb}$}&\multicolumn{1}{c|}{$\addf{\adds{\Tbb}}$}&
29\\\hhline{~~~~~~-------}

 &  &
 &  &
 &  &
\multirow{1}{*}{$F$ } & \multirow{1}{*}{$F$ } &
\multicolumn{1}{c|}{\cite{FroCoS}\cellcolor{red!15}}&\multicolumn{1}{c|}{$\Tbbinfty$}&
\multicolumn{1}{c|}{\cite{FroCoS}\cellcolor{red!15}}&\multicolumn{1}{c|}{$\addf{\Tbbinfty}$}&
30\\\hhline{~~~~~--------}

 &  &
 &  &
 & \multirow{2}{*}{$F$ } &
\multirow{2}{*}{$F$ } & \multirow{1}{*}{$T$ } &
\multicolumn{1}{c|}{$\Tninfty$}&\multicolumn{1}{c|}{$\adds{\Tinfty}$}&
\multicolumn{1}{c|}{$\addf{\Tninfty}$}&\multicolumn{1}{c|}{$\addf{\adds{\Tninfty}}$}&
31\\\hhline{~~~~~~~------}

 &  &
 &  &
 &  &
 & \multirow{1}{*}{$F$ } &
\multicolumn{1}{c|}{\Cref{-FMP+ES+OS=>CMMF}\cellcolor{red!15}}&\multicolumn{1}{c|}{$\Tthree$}&
\multicolumn{1}{c|}{$\Tseven$}&\multicolumn{1}{c|}{$\addf{\Tthree}$}&
32\\\hhline{~~~~---------}

 &  &
 &  &
\multirow{4}{*}{$F$ } & \multirow{2}{*}{$T$ } &
\multirow{1}{*}{$T$ } & \multirow{1}{*}{$F$ } &
\multicolumn{2}{c|}{\cellcolor{red!15}}&
\multicolumn{1}{c|}{$\addnc{\Tbb}$}&\multicolumn{1}{c|}{$\addnc{\adds{\Tbb}}$}&
33\\\hhline{~~~~~~--*{1}{>{\arrayrulecolor{red!15}}|--}*{1}{>{\arrayrulecolor{black}}|-}--}

 &  &
 &  &
 &  &
\multirow{1}{*}{$F$ } & \multirow{1}{*}{$F$ } &
\multicolumn{2}{c|}{\cellcolor{red!15}}&
\multicolumn{1}{c|}{\cite{FroCoS}\cellcolor{red!15}}&\multicolumn{1}{c|}{$\addnc{\Tbbinfty}$}&
34\\\hhline{~~~~~~--*{1}{>{\arrayrulecolor{red!15}}|--}*{1}{>{\arrayrulecolor{black}}|-}--}

 &  &
 &  &
 & \multirow{2}{*}{$F$ } &
\multirow{2}{*}{$F$ } & \multirow{1}{*}{$T$ } &
\multicolumn{2}{c|}{\cellcolor{red!15}}&
\multicolumn{1}{c|}{$\addnc{\Tninfty}$}&\multicolumn{1}{c|}{$\addnc{\adds{\Tninfty}}$}&
35\\\hhline{~~~~~~~-*{1}{>{\arrayrulecolor{red!15}}|--}*{1}{>{\arrayrulecolor{black}}|-}--}

 &  &
 &  &
 &  &
 & \multirow{1}{*}{$F$ } &
\multicolumn{2}{c|}{\multirow{-4}{*}{\cite{CADE}\cellcolor{red!15}}}&
\multicolumn{1}{c|}{$\Teight$}&\multicolumn{1}{c|}{$\addnc{\Tthree}$}&
36\\\hhline{-------------}

\multirow{20}{*}{$F$ } & \multirow{20}{*}{$F$ } &
\multirow{12}{*}{$T$ } & \multirow{4}{*}{$T$ } &
\multirow{2}{*}{$T$ } & \multirow{2}{*}{$T$ } &
\multirow{2}{*}{$T$ } & \multirow{1}{*}{$T$ } &
\multicolumn{1}{c|}{$\Tleqone$}&\multicolumn{1}{c|}{$\adds{\Tleqone}$}&
\multicolumn{1}{c|}{$\addf{\Tleqone}$}&\multicolumn{1}{c|}{$\addf{\adds{\Tleqone}}$}&
37\\\hhline{~~~~~~~------}

 &  &
 &  &
 &  &
 & \multirow{1}{*}{$F$ } &
\multicolumn{1}{c|}{\Cref{-SI+ES=>CMMF}\cellcolor{red!15}}&\multicolumn{3}{c|}{\textcolor{red}{Unicorns 2.0}}&
38\\\hhline{~~~~---------}

 &  &
 &  &
\multirow{2}{*}{$F$ } & \multirow{2}{*}{$T$ } &
\multirow{2}{*}{$T$ } & \multirow{1}{*}{$T$ } &
\multicolumn{1}{c|}{$\Tleqn$}&\multicolumn{1}{c|}{$\adds{\Tleqn}$}&
\multicolumn{1}{c|}{$\addf{\Tleqn}$}&\multicolumn{1}{c|}{$\addf{\adds{\Tleqn}}$}&
39\\\hhline{~~~~~~~------}

 &  &
 &  &
 &  &
 & \multirow{1}{*}{$F$ } &
\multicolumn{1}{c|}{\Cref{-SI+ES=>CMMF}\cellcolor{red!15}}&\multicolumn{3}{c|}{\textcolor{red}{Unicorns 2.0}}&
40\\\hhline{~~~----------}

 &  &
 & \multirow{8}{*}{$F$ } &
\multirow{4}{*}{$T$ } & \multirow{4}{*}{$T$ } &
\multirow{2}{*}{$T$ } & \multirow{1}{*}{$T$ } &
\multicolumn{1}{c|}{\cellcolor{red!15}}&\multicolumn{1}{c|}{$\Toneodd$}&
\multicolumn{1}{c|}{$\Toddneq$}&\multicolumn{1}{c|}{$\addf{\Toneodd}$}&
41\\\hhline{~~~~~~~-*{1}{>{\arrayrulecolor{red!15}}|-}*{1}{>{\arrayrulecolor{black}}|-}---}

 &  &
 &  &
 &  &
 & \multirow{1}{*}{$F$ } &
\multicolumn{1}{c|}{\multirow{-2}{*}{\cite{CADE}\cellcolor{red!15}}}&\multicolumn{1}{c|}{$\Televen$}&
\multicolumn{1}{c|}{$\Tfourteen$}&\multicolumn{1}{c|}{$\addf{\Televen}$}&
42\\\hhline{~~~~~~-------}

 &  &
 &  &
 &  &
\multirow{2}{*}{$F$ } & \multirow{1}{*}{$T$ } &
\multicolumn{1}{c|}{\cellcolor{red!15}}&\multicolumn{1}{c|}{$\Ttwothreethree$}&
\multicolumn{1}{c|}{\cellcolor{red!15}}&\multicolumn{1}{c|}{$\Tnequpinfty$}&
43\\\hhline{~~~~~~~-*{1}{>{\arrayrulecolor{red!15}}|-}*{1}{>{\arrayrulecolor{black}}|-}*{1}{>{\arrayrulecolor{red!15}}|-}*{1}{>{\arrayrulecolor{black}}|-}-}

 &  & 
 &  & 
 &  & 
 & \multirow{1}{*}{$F$ } &
\multicolumn{1}{c|}{\multirow{-2}{*}{\cite{FroCoS}\cellcolor{red!15}}}&\multicolumn{1}{c|}{$\Tthreetwo$}&
\multicolumn{1}{c|}{\multirow{-2}{*}{\cite{FroCoS}\cellcolor{red!15}}}&\multicolumn{1}{c|}{$\Tthirteen$}&
44\\\hhline{~~~~---------}

 &  &
 &  &
\multirow{4}{*}{$F$ } & \multirow{4}{*}{$T$ } &
\multirow{2}{*}{$T$ } & \multirow{1}{*}{$T$ } &
\multicolumn{1}{c|}{$\Tmn$}&\multicolumn{1}{c|}{$\adds{\Tmn}$}&
\multicolumn{1}{c|}{$\addf{\Tmn}$}&\multicolumn{1}{c|}{$\addf{\adds{\Tmn}}$}&
45\\\hhline{~~~~~~~------}

 &  &
 &  &
 &  &
 & \multirow{1}{*}{$F$ } &
\multicolumn{1}{c|}{\Cref{-SI+ES=>CMMF}\cellcolor{red!15}}&\multicolumn{1}{c|}{$\Tten$}&
\multicolumn{1}{c|}{$\Ttwelve$}&\multicolumn{1}{c|}{$\addf{\Tten}$}&
46\\\hhline{~~~~~~-------}

 &  &
 &  &
 &  &
\multirow{2}{*}{$F$ } & \multirow{1}{*}{$T$ } &
\multicolumn{1}{c|}{\cellcolor{red!15}}&\multicolumn{1}{c|}{$\Tmninfty$}&
\multicolumn{1}{c|}{\cellcolor{red!15}}&\multicolumn{1}{c|}{$\addf{\Tmninfty}$}&
47\\\hhline{~~~~~~~-*{1}{>{\arrayrulecolor{red!15}}|-}*{1}{>{\arrayrulecolor{black}}|-}*{1}{>{\arrayrulecolor{red!15}}|-}*{1}{>{\arrayrulecolor{black}}|-}-}

 &  &
 &  &
 &  &
 & \multirow{1}{*}{$F$ } &
\multicolumn{1}{c|}{\multirow{-2}{*}{\cite{FroCoS}\cellcolor{red!15}}}&\multicolumn{1}{c|}{$\Tsix$}&
\multicolumn{1}{c|}{\multirow{-2}{*}{\cite{FroCoS}\cellcolor{red!15}}}&\multicolumn{1}{c|}{$\addf{\Tsix}$}&
48\\\hhline{~~-----------}

 &  &
\multirow{8}{*}{$F$ } & \multirow{8}{*}{$F$ } &
\multirow{4}{*}{$T$ } & \multirow{2}{*}{$T$ } &
\multirow{1}{*}{$T$ } & \multirow{1}{*}{$F$ } &
\multicolumn{1}{c|}{\cellcolor{red!15}}&\multicolumn{1}{c|}{$\Tonebb$}&
\multicolumn{1}{c|}{$\Tbboneneq$}&\multicolumn{1}{c|}{$\addf{\Tonebb}$}&
49\\\hhline{~~~~~~--*{1}{>{\arrayrulecolor{red!15}}|-}*{1}{>{\arrayrulecolor{black}}|-}---}

 &  &
 &  &
 &  &
\multirow{1}{*}{$F$ } & \multirow{1}{*}{$F$ } &
\multicolumn{1}{c|}{\cellcolor{red!15}}&\multicolumn{1}{c|}{$\Tbbinftythree$}&
\multicolumn{1}{c|}{\cite{FroCoS}\cellcolor{red!15}}&\multicolumn{1}{c|}{$\Tbbneqinfty$}&
50\\\hhline{~~~~~---*{1}{>{\arrayrulecolor{red!15}}|-}*{1}{>{\arrayrulecolor{black}}|-}---}

 &  &
 &  &
 & \multirow{2}{*}{$F$ } &
\multirow{2}{*}{$F$ } & \multirow{1}{*}{$T$ } &
\multicolumn{1}{c|}{\cellcolor{red!15}}&\multicolumn{1}{c|}{$\Toneinfty$}&
\multicolumn{1}{c|}{$\Toneinftyneq$}&\multicolumn{1}{c|}{$\addf{\Toneinfty}$}&
51\\\hhline{~~~~~~~-*{1}{>{\arrayrulecolor{red!15}}|-}*{1}{>{\arrayrulecolor{black}}|-}---}

 &  &
 &  &
 &  &
 & \multirow{1}{*}{$F$ } &
\multicolumn{1}{c|}{\cellcolor{red!15}}&\multicolumn{1}{c|}{$\Tthreeone$}&
\multicolumn{1}{c|}{$\Tnine$}&\multicolumn{1}{c|}{$\adds{\Tnine}$}&
52\\\hhline{~~~~----*{1}{>{\arrayrulecolor{red!15}}|-}*{1}{>{\arrayrulecolor{black}}|-}---}

 &  &
 &  &
\multirow{4}{*}{$F$ } & \multirow{2}{*}{$T$ } &
\multirow{1}{*}{$T$ } & \multirow{1}{*}{$F$ } &
\multicolumn{1}{c|}{\cellcolor{red!15}}&\multicolumn{1}{c|}{$\Tnbb$}&
\multicolumn{1}{c|}{$\Tbbneq$}&\multicolumn{1}{c|}{$\addf{\Tnbb}$}&
53\\\hhline{~~~~~~--*{1}{>{\arrayrulecolor{red!15}}|-}*{1}{>{\arrayrulecolor{black}}|-}---}

 &  &
 &  &
 &  &
\multirow{1}{*}{$F$ } & \multirow{1}{*}{$F$ } &
\multicolumn{1}{c|}{\cellcolor{red!15}}&\multicolumn{1}{c|}{$\Tmnbb$}&
\multicolumn{1}{c|}{\cite{FroCoS}\cellcolor{red!15}}&\multicolumn{1}{c|}{$\addf{\Tmnbb}$}&
54\\\hhline{~~~~~---*{1}{>{\arrayrulecolor{red!15}}|-}*{1}{>{\arrayrulecolor{black}}|-}---}

 &  &
 &  &
 & \multirow{2}{*}{$F$ } &
\multirow{2}{*}{$F$ } & \multirow{1}{*}{$T$ } &
\multicolumn{1}{c|}{\cellcolor{red!15}}&\multicolumn{1}{c|}{$\Ttwoinfty$}&
\multicolumn{1}{c|}{$\Ttwoinftyneq$}&\multicolumn{1}{c|}{$\addf{\Ttwoinfty}$}&
55\\\hhline{~~~~~~~-*{1}{>{\arrayrulecolor{red!15}}|-}*{1}{>{\arrayrulecolor{black}}|-}---}

 &  &
 &  &
 &  &
 & \multirow{1}{*}{$F$ } &
\multicolumn{1}{c|}{\multirow{-8}{*}{\cite{FroCoS}\cellcolor{red!15}}}&\multicolumn{1}{c|}{$\Tfive$}&
\multicolumn{1}{c|}{$\Tfour$}&\multicolumn{1}{c|}{$\addf{\Tfive}$}&
56\\\hhline{-------------}
\caption[Caption for LOF]{Summarizing table.}
\label{tab-summary-app}
\end{longtable}
\end{small}
\renewcommand{\arraystretch}{1}}

\newpage

\section{New theories}
\label{sec:newtheoriesapp}

We start by reminding the reader that in \cite{CADE} we had proved the existence of a non-computable function $f:\mathbb{N}\setminus\{0\}\rightarrow\{0,1\}$ such that $f(1)=1$ and, for every $k\in\mathbb{N}\setminus\{0\}$, 
\[|\{1\leq i\leq 2^{k} : f(i)=0\}|=|\{1\leq i\leq 2^{k} : f(i)=1\}|\]
From that, we then define the function $g:\mathbb{N}\setminus\{0\}\rightarrow\mathbb{N}\setminus\{0\}$ by making 
$g(n)=n+\sum_{i=1}^{n}f(i)=\sum_{i=1}^{n}(f(i)+1)$, and we prove in \Cref{lem:g-exists} that $g$ is: increasing; unbounded; non-surjective; ~non-computable; and that there exists an increasing computable function $\rho:\mathbb{N}\setminus\{0\}\rightarrow\mathbb{N}\setminus\{0\}$ such that
$g\circ\rho$ is computable.

\subsection{\tp{$\Tone$}{Tone}}

We define $\Tone$ as the $\Sigma_{1}$-theory with axiomatization
\[\{\psi^{\s}_{\geq g(n)}\vee\bigvee_{i=1}^{n}\psi^{\s}_{=g(i)} : n\in\mathbb{N}\setminus\{0\}\}.\]

\begin{lemma}\label{SI Tone}
    The $\Sigma_{1}$-theory $\Tone$ is stably infinite with respect to its only sort.
\end{lemma}

\begin{proof}
    Obvious, since $\Tone$ has infinite interpretations and the signature $\Sigma_{1}$ is empty.
\end{proof}

\begin{lemma}
    The $\Sigma_{1}$-theory $\Tone$ is not smooth with respect to its only sort.
\end{lemma}

\begin{proof}
    If $f(3)=0$, then $g(3)=4$ but, since $f(4)=1$, we get $g(4)=6$, meaning $\Tone$ has no interpretation with $5$ elements, although it has one with $4$. If $f(3)=1$, $g(2)=3$ but $g(3)=5$ and the same argument applies.
\end{proof}

\begin{lemma}
    The $\Sigma_{1}$-theory $\Tone$ is finitely witnessable with respect to its only sort.
\end{lemma}

\begin{proof}
    Let $\phi$ be a quantifier-free formula and let $2^{m(\phi)}$ be the smallest power of two larger than the number of variables of $\phi$: this, of course, may be found algorithmically starting from $\phi$. If $x_{i}$ are fresh variables, we state that 
    \[wit(\phi)=\phi\wedge\bigwedge_{i=1}^{M(\phi)}x_{i}=x_{i}\]
    is a witness for $\Tone$, where $M(\phi)=3\times 2^{m(\phi)-1}$. For $\overarrow{x}=\vars(wit(\phi))\setminus\vars(\phi)=\{x_{1},\ldots,x_{M(\phi)}\}$, it is of course true that $\phi$ and $\Exists{\overarrow{x}}wit(\phi)$ are $\Tone$-equivalent since $\phi$ and $wit(\phi)$ are already equivalent. Now, suppose that $\phi$ is $\Tone$-satisfiable, and let $\A$ be a $\Tone$-interpretation that satisfies $\phi$: take $V=\vars(\phi)$, and define an interpretation $\B$ with $\s^{\B}=V^{\A}\cup X$, where $X$ is any set of cardinality $M(\phi)-|V^{\A}|$; $x^{\B}=x^{\A}$ for all $x\in V$; $x_{i}\in\{x_{1},\ldots,x_{M(\phi)}\}\mapsto x_{i}^{\B}\in \s^{\B}$ is a bijection, what is possible given that $\s^{\B}$ contains $M(\phi)$ elements (and, since $M(\phi)=3\times 2^{m(\phi)-1}=g(2^{m(\phi)-1})$ we have that $\B$ is indeed a $\Tone$-interpretation); and $x^{\B}$ may be defined arbitrarily for variables outside these two sets. Then it is true that $\B$ is a $\Tone$-interpretation that satisfies $\phi$, and thus $wit(\phi)$, with $\s^{\B}=\vars(wit(\phi))^{\B}$.
\end{proof}

\begin{lemma}
    The $\Sigma_{1}$-theory $\Tone$ is not strongly finitely witnessable with respect to its only sort.
\end{lemma}

\begin{proof}
    From Lemma $73$ of \cite{arxivCADE} and \Cref{SI Tone}, if $\Tone$ has a strong witness then it would be smooth up to $\aleph_{0}$.
\end{proof}

\begin{lemma}
    The $\Sigma_{1}$-theory $\Tone$ is convex with respect to its only sort.
\end{lemma}

\begin{proof}
    Follows from Theorem $5$ of \cite{CADE}.
\end{proof}

\begin{lemma}\label{cmmf:tone}
    The $\Sigma_{1}$-theory $\Tone$ does not have a computable minimal model function with respect to its only sort.
\end{lemma}

\begin{proof}
    For a proof by contradiction, assume $\Tone$ has a computable minimal model function: we then define a function $h:\mathbb{N}\setminus\{0\}\rightarrow\mathbb{N}\setminus\{0\}$ by making $h(1)=2$ and $h(n+1)\in\minmod_{\Tone,\{\s\}}(\NNNEQ{x}{h(n)+1})$ (notice $\NNNEQ{x}{h(n)+1}$ is $\Tone$-satisfiable, since the theory is stably infinite).
    $h$ is obviously recursively computable since $\minmod_{\Tone,\{\s\}}$ is computable, and once $h(n)$ is known the formula $\psi^{\s}_{\geq h(n)+1}$ can be found algorithmically. One notices, however, that $h(n)=g(n)$ for all $n\in\mathbb{N}\setminus\{0\}$, contradicting the fact that $g$ is not computable.
\end{proof}

\subsection{\tp{$\Ttwo$}{Ttwo}}

$\Ttwo$ is the $\Sigma_{2}$-theory with axiomatization
\[\{(\psi^{\s}_{=1}\wedge\psi^{\s_{2}}_{\geq n})\vee((\psi^{\s}_{\geq g(n+1)}\wedge\psi^{\s_{2}}_{\geq g(n+1)})\vee\bigvee_{i=1}^{n}(\psi^{\s}_{=g(i)}\wedge\psi^{\s_{2}}_{=g(i)}) : n\in\mathbb{N}\setminus\{0\}\}.\]
In words, $\Ttwo$ has essentially three classes of model: those whose domain of sort $\s$ has cardinality $1$, and whose domain of sort $\s_{2}$ is infinite; those whose domains of sort $\s$ and $\s_{2}$ have the same finite cardinality $g(n)$ for some $n\in\mathbb{N}\setminus\{0\}$; and those where both domains are infinite.

\begin{lemma}\label{Ttwo is SI}
    The $\Sigma_{2}$-theory $\Ttwo$ is stably infinite with respect to $\{\s,\s_{2}\}$.
\end{lemma}

\begin{proof}
    Let $\phi$ be a quantifier-free formula and $\A$ a $\Ttwo$-interpretation that satisfies $\phi$: we define the $\Sigma_{2}$-interpretation $\B$ with $\s^{\B}=\s^{\A}\cup X$ and $\s_{2}^{\B}=\s_{2}^{\A}\cup X$, where $X$ is an infinite set disjoint from $\s^{\A}$ and $\s_{2}^{\A}$, and $x^{\B}=x^{\A}$ for any variable $x$, of sort either $\s$ or $\s_{2}$. Then it is clear that $\B$ has both domains infinite, is a $\Ttwo$-interpretation, and satisfies $\phi$.
\end{proof}

\begin{lemma}\label{Ttwo is not SM}
    The $\Sigma_{2}$-theory $\Ttwo$ is not smooth with respect to $\{\s,\s_{2}\}$.
\end{lemma}

\begin{proof}
    Obvious, given that $\Ttwo$ has a model $\A$ with $|\s^{\A}|=|\s_{2}^{\A}|=2$, but none $\B$ with, say $|\s^{\B}|=2$ and $|\s_{2}^{\B}|=3$.
\end{proof}

\begin{lemma}\label{Ttwo is FW}
    The $\Sigma_{2}$-theory $\Ttwo$ is finitely witnessable with respect to $\{\s,\s_{2}\}$.
\end{lemma}

\begin{proof}
    Given a quantifier-free formula $\phi$, let $k$ be the minimum of positive integers such that $3\times2^{k-1}<M\leq 3\times 2^{k}$ (meaning $g(2^{k})<M\leq g(2^{k+1})$), where $M=\max\{|\vars_{\s}(\phi)|, |\vars_{\s_{2}}(\phi)|\}$: it is not difficult to see that finding $k$ can be done algorithmically, and so the function
    \[\wit(\phi)=\phi\wedge\bigwedge_{i=1}^{g(2^{k+1})}(x_{i}=x_{i})\wedge(u_{i}=u_{i})\]
    is computable, where $X=\{x_{1},\ldots,x_{g(2^{k+1})}\}$ is a set of fresh variables of sort $\s$, and $U=\{u_{1},\ldots,u_{g(2^{k+1})}\}$ is a set of fresh variables of sort $\s_{2}$. We state that $\wit$ is indeed a witness, being it obvious that $\Exists{\overarrow{x}}\wit(\phi)$ and $\phi$ are $\Ttwo$-equivalent, for $\overarrow{x}=\vars(\wit(\phi))\setminus\vars(\phi)$, since $\wit(\phi)$ and $\phi$ are themselves equivalent (as $\bigwedge_{i=1}^{g(2^{k+1})}(x_{i}=x_{i})\wedge(u_{i}=u_{i})$ is a tautology).

    Now, assume that $\A$ is a $\Ttwo$-interpretation that satisfies $\wit(\phi)$, or equivalently $\phi$: take sets $A$ and $B$, disjoint respectively from $\s^{\A}$ and $\s_{2}^{\A}$, with $g(2^{k+1})-|\vars_{\s}(\phi)^{\A}|$ and $g(2^{k+1})-|\vars_{\s_{2}}(\phi)^{\A}|$ elements. We define a new interpretation $\B$ by making: $\s^{\B}=\vars_{\s}(\phi)^{\A}\cup A$ (which has therefore $g(2^{k+1})$ elements); $\s_{2}^{\B}=\vars_{\s_{2}}(\phi)^{\A}\cup B$ (meaning $\s_{2}^{\B}$ has $g(2^{k+1})$ elements, and thus $\B$ is a $\Ttwo$-interpretation); $x^{\B}=x^{\A}$ for all variables $x$ in $\phi$; $x_{i}\in X\mapsto x_{i}^{\B}\in \s^{\B}$ and $u_{i}\in U\mapsto u_{i}^{\B}\in\s_{2}^{\B}$ bijective maps, and arbitrarily for all other variables. It is then easy to see that $\B$ is a $\Ttwo$-interpretation, that satisfies $\phi$ and $\wit(\phi)$, and with $\vars_{\s}(\phi)^{\B}=\s^{\B}$ and $\vars_{\s_{2}}(\phi)^{\B}=\s_{2}^{\B}$.
\end{proof}

\begin{lemma}
    The $\Sigma_{2}$-theory $\Ttwo$ is not strongly finitely witnessable with respect to $\{\s,\s_{2}\}$.
\end{lemma}

\begin{proof}
    By \cite{CADE}, since $\Ttwo$ is stably infinite by \Cref{Ttwo is SI}, if it were strongly finitely witnessable it would also have, for a quantifier-free formula $\phi$, a $\Ttwo$-interpretation $\A$ satisfying $\phi$, and cardinals $|\s^{\A}|\leq \kappa\leq\aleph_{0}$ and $|\s_{2}^{\A}|\leq \kappa_{2}\leq\aleph_{0}$, a $\Ttwo$-interpretation $\B$ that satisfies $\phi$ with $|\s^{\B}|=\kappa$ and $|\s_{2}^{\B}|=\kappa_{2}$, which it does not have according to \Cref{Ttwo is not SM}.
\end{proof}

\begin{lemma}\label{Ttwo is CV}
    The $\Sigma_{2}$-theory $\Ttwo$ is convex with respect to $\{\s,\s_{2}\}$.
\end{lemma}

\begin{proof}
    Follows from Theorem $5$ of \cite{CADE}.
\end{proof}

\begin{lemma}
    The $\Sigma_{2}$-theory $\Ttwo$ has the finite model property with respect to $\{\s,\s_{2}\}$.
\end{lemma}

\begin{proof}
    Follows from Theorem 2 of \cite{FroCoS} and \Cref{Ttwo is FW}.
\end{proof}

\begin{lemma}
    The $\Sigma_{2}$-theory $\Ttwo$ is not stably finite with respect to $\{\s,\s_{2}\}$.
\end{lemma}

\begin{proof}
    Obvious, since $\Ttwo$ has a model $\A$ with $|\s^{\A}|=1$ and $|\s_{2}^{\A}|=\aleph_{0}$, but no models $\B$ with both $|\s^{\B}|$ and $|\s_{2}^{\B}|$ finite, $|\s^{\B}|\leq |\s^{\A}|$ and $|\s_{2}^{\B}|\leq |\s_{2}^{\A}|$.
\end{proof}

\begin{lemma}\label{cmmf ttwo}
    The $\Sigma_{2}$-theory $\Ttwo$ does not have a computable minimal model function with respect to $\{\s,\s_{2}\}$.
\end{lemma}

\begin{proof}
     Assume by contradiction that $\Ttwo$ has a computable minimal model function, and we define a function $h:\mathbb{N}\setminus\{0\}\rightarrow\mathbb{N}\setminus\{0\}$ by making $h(1)=2$ and 
    \[h(n+1)=\min\{p : (p,q)\in \minmod_{\Ttwo,\{\s,\s_{2}\}}(\NNNEQ{x}{h(n)+1})\},\]
    for $x_{i}$ of sort $\s$
    ($\NNNEQ{x}{h(n)+1}$ is then $\Ttwo$-satisfiable).
    Such a function is computable given that a minimal model function of $\Ttwo$ is computable; furthermore, it is easy to see that $h(n)=g(n)$ for all $n\in\mathbb{N}\setminus\{0\}$, contradicting the fact that $g$ is not computable.
\end{proof}

\subsection{\tp{$\Tthree$}{Tthree}}

$\Tthree$ is the $\Sigma_{2}$-theory with axiomatization
\[\{\psi^{\s}_{\geq n}\wedge(\psi^{\s_{2}}_{\geq g(n+1)}\vee\bigvee_{i=1}^{n}\psi^{\s_{2}}_{=g(i)}) : n\in\mathbb{N}\setminus\{0\}\}.\]
Its models always have the domain of sort $\s$ infinite, while the domain of sort $\s_{2}$ may be infinite or equal to $g(n)$ for some $n\in\mathbb{N}\setminus\{0\}$.

\begin{lemma}
    The $\Sigma_{2}$-theory $\Tthree$ is stably infinite with respect to $\{\s,\s_{2}\}$.
\end{lemma}

\begin{proof}
    Take a quantifier-free formula $\phi$, and let $\A$ be a $\Tthree$-interpretation that satisfies $\phi$. We define a $\Sigma_{2}$-interpretation $\B$ by making $\s^{\B}=\s^{\A}\cup X$ and $\s_{2}^{\B}=\s_{2}^{\A}\cup X$ (for $X$ an infinite set disjoint from both $\s^{\A}$ and $\s_{2}^{\A}$), and $x^{\B}=x^{\A}$ for a variable $x$ of either sort. Then $\B$ has both domains infinite, is a $\Tthree$-interpretation, and satisfies $\phi$.
\end{proof}

\begin{lemma}
    The $\Sigma_{2}$-theory $\Tthree$ is not smooth with respect to $\{\s,\s_{2}\}$.
\end{lemma}

\begin{proof}
    Obvious, given that there is an $n\in\mathbb{N}\setminus\{0\}$ such that $g(n)<g(n)+1<g(n+1)$, and $\Tthree$ has a model $\A$ with $|\s^{\A}|=\aleph_{0}$ and $|\s_{2}^{\A}|=g(n)$, but no model $\B$ with $|\s^{\B}|=\aleph_{0}$ and $|\s_{2}^{\B}|=g(n)+1$.
\end{proof}

\begin{lemma}
    The $\Sigma_{2}$-theory $\Tthree$ is not finitely witnessable, and thus not strongly finitely witnessable, with respect to $\{\s,\s_{2}\}$.
\end{lemma}

\begin{proof}
    Obvious, given that if $\Tthree$ had a witness $\wit$, then for any $\Tthree$-satisfiable quantifier-free formula $\phi$ there would exist a $\Tthree$-interpretation $\A$ with $\s^{\A}=\vars_{\s}(\wit(\phi))^{\A}$ and $\s_{2}^{\A}=\vars_{\s_{2}}(\wit(\phi))^{\A}$, what is impossible given that any $\A$ must have $\s^{\A}$ infinite.
\end{proof}

\begin{lemma}\label{Tthree is convex}
    The $\Sigma_{2}$-theory $\Tthree$ is convex with respect to $\{\s,\s_{2}\}$.
\end{lemma}

\begin{proof}
    Follows from Theorem $5$ of \cite{CADE}.
\end{proof}

\begin{lemma}
    The $\Sigma_{2}$-theory $\Tthree$ does not have the finite model property, and thus is not stably finite, with respect to $\{\s,\s_{2}\}$.
\end{lemma}

\begin{proof}
    Obvious, given that $\Tthree$ has no models where both domains are finite.
\end{proof}

\begin{lemma}
    The $\Sigma_{2}$-theory $\Tthree$ does not have a computable minimal model function with respect to $\{\s,\s_{2}\}$.
\end{lemma}

\begin{proof}
    The proof is the same as the one of \Cref{cmmf ttwo}, but looking at the second sort instead of the first.
\end{proof}

\subsection{\tp{$\Tfour$}{Tfour}}

$\Tfour$ is the $\Sigma_{s}$-theory axiomatized by
\[\{(\psi_{\neq}\wedge\psi_{\vee}^{6}\wedge\psi_{\geq \bb(n+1)})\vee(\psi_{\neq}\wedge\psi_{\vee}^{2}\wedge\bigvee_{i=2}^{n}\psi_{=\bb(i)})\vee\psi_{=1} : n\in\mathbb{N}\setminus\{0,1\}\}.\]
The theory $\Tfour$ has essentially three classes of models: those $\A$ with infinite elements in $\s^{\A}$, no fixed points for $s^{\A}$, and where either $(s^{\A})^{12}(a)=(s^{\A})^{6}(a)$ or $(s^{\A})^{12}(a)=a$, for all $a\in\s^{\A}$; those with one element in $\s^{\A}$, and $s^{\A}$ the identity; and those with $\bb(n)$ elements in $\s^{\A}$, for $n\in\mathbb{N}\setminus\{0,1\}$, no fixed points for $s^{\A}$, and where either $(s^{\A})^{4}(a)=(s^{\A})^{2}(a)$ or $(s^{\A})^{4}(a)=a$, for all $a\in \s^{\A}$. Notice this is well defined as $\psi_{\vee}^{4}$ implies $\psi_{\vee}^{6}$: indeed, if $s^{4}(x)=s^{2}(x)$, for all $x$,
\[s^{12}(x)=s^{4}(s^{4}(s^{4}(x)))=s^{2}(s^{2}(s^{2}(x)))=s^{6}(x);\]
and if $s^{4}(x)=x$, 
\[s^{12}(x)=s^{4}(s^{4}(s^{4}(x)))=s^{4}(s^{4}(x))=s^{4}(x)=x.\]

\begin{lemma}
    The $\Sigma_{s}$-theory $\Tfour$ is not stably infinite, and thus not smooth, with respect to its only sort.
\end{lemma}

\begin{proof}
    Obvious, as no infinite $\Tfour$-interpretation can satisfy the formula $s(x)=x$, although that formula is satisfied by the $\Tfour$-interpretation with only one element in its domain.
\end{proof}

\begin{lemma}
    The $\Sigma_{s}$-theory $\Tfour$ is not finitely witnessable, and thus not strongly finitely witnessable, with respect to its only sort.
\end{lemma}

\begin{proof}
    Suppose that $\Tfour$ has a computable witness $\wit$, and define a function $h:\mathbb{N}\rightarrow\mathbb{N}$ by making $h(0)=0$, $h(1)=1$, and 
    \[h(n+1)=|\vars(\wit(\NNNEQ{x}{h(n)+1}))|,\]
    for $n\geq 2$. Now, it is clear that $h(0)\geq \bb(0)$ and $h(1)\geq \bb(1)$, so assume that $h(n)\geq \bb(n)$: because $\NNNEQ{x}{h(n)+1}$ is $\Tfour$-satisfiable, there is a $\Tfour$-interpretation $\A$ that satisfies $\wit(\NNNEQ{x}{h(n)+1})$, and thus $\Exists{\overarrow{x}}\wit(\NNNEQ{x}{h(n)+1})$ and $\NNNEQ{x}{h(n)+1}$, for $\overarrow{x}=\vars(\wit(\NNNEQ{x}{h(n)+1}))\setminus\vars(\NNNEQ{x}{h(n)+1})$, with $\s^{\A}=\vars(\wit(\NNNEQ{x}{h(n)+1}))^{\A}$. As $\A$ is a $\Tfour$-interpretation with at least $h(n)+1\geq \bb(n)+1$ elements, $|\s^{\A}|\geq \bb(n+1)$, and thus $h(n+1)=|\vars(\wit(\NNNEQ{x}{h(n)+1}))|\geq\bb(n+1)$, as we wanted to prove; but this contradicts the fact that $h$ is computable.
\end{proof}

\begin{lemma}
    The $\Sigma_{s}$-theory $\Tfour$ is not convex with respect to its only sort.
\end{lemma}

\begin{proof}
    We state that 
    \[(s^{6}(x)=y)\wedge(s^{6}(y)=z)\vDash_{\Tfour}(x=y)\vee(y=z):\]
    this is obviously true for the trivial $\Tfour$-interpretation; for finite non-trivial $\Tfour$-interpretations, we have either $s^{4}(x)=s^{2}(x)$, and thus 
    \[z=s^{12}(x)=s^{4}(s^{4}(s^{4}(x)))=s^{2}(s^{2}(s^{2}(x)))=s^{6}(x)=y,\]
    or $s^{4}(x)=x$, and then 
    \[z=s^{12}(x)=s^{4}(s^{4}(s^{4}(x)))=s^{4}(s^{4}(x))=s^{4}(x)=x;\]
    finally, for infinite $\Tfour$-interpretations, we already know that either $z=s^{12}(x)=s^{6}(x)=y$ or $y=s^{6}(x)=x$.

    But we state that neither
    \[(s^{6}(x)=y)\wedge(s^{6}(y)=z)\vDash_{\Tfour}x=y\quad\text{nor}\quad (s^{6}(x)=y)\wedge(s^{6}(y)=z)\vDash_{\Tfour}y=z.\]
    Indeed, let $\A$ be the $\Tfour$-interpretation with: $\s^{\A}=\{a, b, c, d\}$; $s^{\A}(a)=b$, $s^{\A}(b)=c$, $s^{\A}(c)=d$ and $s^{\A}(d)=a$ (for all $i\in\mathbb{N}$); $x^{\A}=z^{\A}=a$ and $y^{\A}=c$. Then $\A$ satisfies $(s^{6}(x)=y)\wedge(s^{6}(y)=z)$, but not $x=y$.

    On the other hand, let $\B$ be the $\Tfour$-interpretation with: $\s^{\B}=\s^{\A}$; $s^{\B}(a)=b$, $s^{\B}(b)=c$, $s^{\B}(c)=d$ and $s^{\B}(d)=c$; $x^{\B}=a$ and $y^{\B}=z^{\B}=c$. Then $\B$ satisfies, again, $(s^{6}(x)=y)\wedge(s^{6}(y)=z)$, but not $y=z$.
\end{proof}

\begin{lemma}
    The $\Sigma_{s}$-theory $\Tfour$ does not have the finite model property, and thus is not stably finite, with respect to its only sort.
\end{lemma}

\begin{proof}
    Take the $\Tfour$-interpretation $\A$ with $\s^{\A}=\{a_{i}, b_{i}, c_{i} : i\in\mathbb{N}\}$, $s^{\A}(a_{i})=b_{i}$, $s^{\A}(b_{i})=c_{i}$ and $s^{\A}(c_{i})=a_{i}$ for all $i\in \mathbb{N}$, and $x^{\A}=a$ (and arbitrarily for other variables): we have that $\A$ is indeed a $\Tfour$-interpretation as both $(s^{\A})^{12}(a)$ and $(s^{\A})^{6}(a)$ both equal $a$ for all $a\in\s^{\A}$, since $(s^{\A})^{3}(a)=a$. Then $\A$ satisfies $(s^{3}(x)=x)\wedge\neg(s(x)=x)$, but no finite $\Tfour$-interpretation $\B$ can satisfy that formula. This is obvious if $|\s^{\B}|=1$, so assume that $\B$ with $|\s^{\B}|>1$ satisfies $s^{3}(x)=x$, being obvious that it satisfies $\neg(s(x)=x)$: then we have that $\B$ satisfies $s^{4}(x)=s(s^{3}(x))=s(x)$, as $s$ is a function, as well as $s^{4}(x)=s^{2}(x)$ and $s^{2}(x)=s(s(x))\neq s(x)$, as $\B$ is a $\Tfour$-interpretation, leading to a contradiction.
\end{proof}

\begin{lemma}\label{cmmf:tfour}
    The $\Sigma_{s}$-theory $\Tfour$ does not have a computable minimal model function with respect to its only sort.
\end{lemma}

\begin{proof}
    Suppose that instead $\minmod_{\Tfour,S}$ is computable, where $S=\{\s\}$: we define a function $h:\mathbb{N}\rightarrow\mathbb{N}$ by making $h(0)=0$, $h(1)=1$, and 
    \[h(n+1)\in\minmod_{\Tfour,S}(\NNNEQ{x}{h(n)+1}),\]
    for all $n\geq 2$, which is obviously computable (being the set on the right computable and finite). It is easy to see, by induction, that $h(n)=\bb(n)$, getting us a contradiction.
\end{proof}

\subsection{\tp{$\Tfive$}{Tfive}}

$\Tfive$ is the $\Sigma_{2}$-theory with axiomatization
\[\{(\psi^{\s}_{=1}\wedge(\psi^{\s_{2}}_{\geq \bb(n+1)}\vee\bigvee_{i=1}^{n}\psi^{\s_{2}}_{=\bb(i)}))\vee(\psi^{\s}_{=2}\wedge(\psi^{\s_{2}}_{=2}\vee\psi^{\s_{2}}_{\geq n})) : n\in\mathbb{N}\setminus\{0\}\},\]
which has four classes of models $\A$: with $\s^{\A}$ trivial and $\s_{2}^{\A}$ infinite; with $\s^{\A}$ trivial and $|\s_{2}^{\A}|=\bb(n)$ for some $n\in\mathbb{N}\setminus\{0\}$; with $|\s^{\A}|=|\s_{2}^{\A}|=2$; and with $|\s^{\A}|=2$ and $\s_{2}^{\A}$ infinite.

\begin{lemma}
    The $\Sigma_{2}$-theory $\Tfive$ is not stably infinite, and thus not smooth, with respect to $\{\s,\s_{2}\}$.
\end{lemma}

\begin{proof}
    Obvious, as no model of $\Tfive$ has domain of sort $\s$ infinite.
\end{proof}

\begin{lemma}
    The $\Sigma_{2}$-theory $\Tfive$ is not finitely witnessable, and thus not strongly finitely witnessable, with respect to $\{\s,\s_{2}\}$.
\end{lemma}

\begin{proof}
    Suppose that $\Tfive$ has a witness: we define a function $h:\mathbb{N}\rightarrow\mathbb{N}$ by making $h(0)=0$, $h(1)=1$, $h(2)=4$, and 
    \[h(n+1)=|\vars_{\s_{2}}(\wit(\NNNEQ{u}{h(n)+1}))|,\]
    for $n\geq 3$, which is obviously computable. We state that $h(n)\geq\bb(n)$, leading to a contradiction, that being obvious for $n=0$, $n=1$ or $n=2$: suppose then that $h(n)\geq \bb(n)$, and since $\NNNEQ{u}{h(n)+1}$ is $\Tfive$-satisfiable, there is a $\Tfive$-interpretation that satisfies $\NNNEQ{u}{h(n)+1}$ (and thus has at least $\bb(n)+1$ elements, and thus at least $\bb(n+1)$) with $\vars_{\s_{2}}(\wit(\NNNEQ{u}{h(n)+1}))^{\A}=\s_{2}^{\A}$, meaning that
    \[h(n+1)=|\vars_{\s_{2}}(\wit(\NNNEQ{u}{h(n)+1}))|\geq |\vars_{\s_{2}}(\wit(\NNNEQ{u}{h(n)+1}))^{\A}|=|\s_{2}^{\A}|\geq \bb(n+1),\] as we wanted to show.
\end{proof}

\begin{lemma}
    The $\Sigma_{2}$-theory $\Tfive$ is not convex with respect to $\{\s,\s_{2}\}$.
\end{lemma}

\begin{proof}
    It is clear that $\vDash_{\Tfive}(x=y)\vee(x=z)\vee(y=z)$
    for variables $x$, $y$ and $z$ of sort $\s$, as no $\Tfive$-interpretation has more than two elements in the domain of sort $\s$. But $\Tfive$ cannot prove neither of the disjuncts: consider the $\Tfive$-interpretations $\A$ and $\B$ with $\s^{\A}=\s^{\B}=\{a,b\}$, $\s_{2}^{\A}=\s_{2}^{\B}=\{c,d\}$, $x^{\A}=z^{\A}=x^{\B}=y^{\B}=a$, and $y^{\A}=z^{\B}=b$; then $\A$ falsifies $x=y$ and $y=z$, while $\B$ falsifies $x=z$. 
\end{proof}

\begin{lemma}
    The $\Sigma_{2}$-theory $\Tfive$ does not have the finite model property, and thus is not stably finite, with respect to $\{\s,\s_{2}\}$.
\end{lemma}

\begin{proof}
    The formula $\NEQ(x_{1},x_{2})\wedge\NEQ(u_{1},u_{2},u_{3})$ (for $x_{i}$ of sort $\s$ and $u_{j}$ of sort $\s_{2}$) is satisfied by a $\Tfive$-interpretation $\A$ with $|\s^{\A}|=2$ and $\s_{2}^{\A}$ infinite, but cannot be satisfied by a $\Tfive$-interpretation $\B$ with both $\s^{\B}$ and $\s_{2}^{\B}$ finite, as $|\s^{\B}|=2$ and $|\s_{2}^{\B}|\geq 3$ imply $|\s_{2}^{\B}|\geq\aleph_{0}$.
\end{proof}

\begin{lemma}
    The $\Sigma_{2}$-theory $\Tfive$ does not have a computable minimal model function with respect to $\{\s,\s_{2}\}$.
\end{lemma}

\begin{proof}
    This is the same as the proof of \Cref{cmmf:tfour}, taking into consideration the second sort instead of the first.
\end{proof}

\subsection{\tp{$\Tsix$}{Tsix}}

$\Tsix$ is the $\Sigma_{2}$ theory with axiomatization, for $m>n\geq 2$, with axiomatization
\[\{(\psi^{\s}_{=n}\wedge\psi^{\s_{2}}_{\geq k})\vee(\psi^{\s}_{=m}\wedge\psi^{\s_{2}}_{\geq g(k+1)})\vee(\psi^{\s}_{=m}\wedge\bigvee_{i=1}^{k}\psi^{\s_{2}}_{=g(i)}) : k\in\mathbb{N}\setminus\{0\}\},\]
which has essentially three classes of models $\A$: those with $|\s^{\A}|=n$ and $\s_{2}^{\A}$ infinite; those with $|\s^{\A}|=m$ and $\s_{2}^{\A}$ infinite; and those with $|\s^{\A}|=m$ and $|\s_{2}^{\A}|=g(k)$, for some $k\in\mathbb{N}\setminus\{0\}$.

\begin{lemma}
    The $\Sigma_{2}$-theory $\Tsix$ is not stably infinite, and thus not smooth, with respect to $\{\s,\s_{2}\}$.
\end{lemma}

\begin{proof}
    Obvious, as no $\Tsix$-interpretation has infinite elements in its domain of sort $\s$.
\end{proof}

\begin{lemma}
    The $\Sigma_{2}$-theory $\Tsix$ is finitely witnessable with respect to $\{\s,\s_{2}\}$.
\end{lemma}

\begin{proof}
    Let $\phi$ be a quantifier-free formula, $k$ an integer such that $g(2^{k})=3\times 2^{k-1}<|\vars_{\s_{2}}(\phi)|\leq 3\times 2^{k}=g(2^{k+1})$, $x_{1}$ through $x_{m}$ fresh variables of sort $\s$, and $u_{1}$ through $u_{g(2^{k+1})}$ fresh variables of sort $\s_{2}$: we define the, clearly computable, function
    \[\wit(\phi)=\phi\wedge\bigwedge_{j=1}^{m}(x_{j}=x_{j})\wedge\bigwedge_{i=1}^{g(2^{k+1})}(u_{i}=u_{i}),\]
    which we state is a witness. Of course $\Exists{\overarrow{x}}\wit(\phi)$ and $\phi$ are $\Tsix$-equivalent, for $\overarrow{x}=\vars(\wit(\phi))\setminus\vars(\phi)$, since $\wit(\phi)$ and $\phi$ are themselves equivalent, as $\wit(\phi)$ is the conjunction of $\phi$ and a tautology.

    Now, let $\A$ be a $\Tsix$-interpretation that satisfies $\wit(\phi)$, and thus $\phi$: we let $A$ and $B$ be sets of cardinality, respectively, $m-|\vars_{\s}(\phi)^{\A}|$ and $g(2^{k+1})-|\vars_{\s_{2}}(\phi)^{\A}|$, disjoint from $\s^{\A}$ and $\s_{2}^{\A}$. Then, we construct a $\Tsix$-interpretation $\B$ by making $\s^{\B}=\vars_{\s}(\phi)^{\A}\cup A$ (that has $m$ elements); $\s_{2}^{\B}=\vars_{\s_{2}}(\phi)^{\A}\cup B$ (that has $g(2^{k+1})$ elements, thus making of $\B$ a $\Tsix$-interpretation); $x^{\B}=x^{\A}$ for all $x\in\vars_{\s}(\phi)$, $x\in \{x_{1},\ldots,x_{m}\}\mapsto x^{\B}\in\s^{\B}$ a bijection, and $x^{\B}$ arbitrary for any other variables of sort $\s$; $u^{\B}=u^{\A}$ for all $u\in\vars_{\s_{2}}(\phi)$, $u\in \{u_{1},\ldots,u_{g(2^{k+1})}\}\mapsto u^{\B}\in\s_{2}^{\B}$ a bijection, and $u^{\B}$ again arbitrary for other variables of sort $\s_{2}$. Of course $\B$ satisfies $\phi$, and thus $\wit(\phi)$, and has $\s^{\B}=\vars_{\s}(\wit(\phi))^{\B}$ and $\s_{2}^{\B}=\vars_{\s_{2}}(\wit(\phi))^{\B}$.
\end{proof}

\begin{lemma}
    The $\Sigma_{2}$-theory $\Tsix$ is not strongly finitely witnessable with respect to $\{\s,\s_{2}\}$.
\end{lemma}

\begin{proof}
    Suppose that $\Tsix$ has a strong witness $\wit$, let $\phi$ be a $\Tsix$-satisfiable formula (such as a tautology), and $\A$ a $\Tsix$-interpretation that satisfies $\phi$: since $\phi$ and $\Exists{\overarrow{x}}\wit(\phi)$ are $\Tsix$-equivalent, for $\overarrow{x}=\vars(\wit(\phi))\setminus\vars(\phi)$, there exists a $\Tsix$-interpretation $\A^{\prime}$, differing from $\A$ at most on the value assigned to $\overarrow{x}$, such that $\A^{\prime}$ satisfies $\wit(\phi)$. Let $V=\vars(\wit(\phi))$, and $\delta_{V}$ the arrangement induced by $\A^{\prime}$ on $V$, and it is clear that $\A^{\prime}$ satisfies $\wit(\phi)\wedge\delta_{V}$: there is, then, a $\Tsix$-interpretation $\B$ that satisfies $\wit(\phi)\wedge\delta_{V}$ with $\vars_{\s_{2}}(\wit(\phi)\wedge\delta_{V})^{\B}=\s_{2}^{\B}$ (and the same holds for $\s$, but that is not as important); let $g(q)$ be the cardinality of $|\s_{2}^{\B}|$, which must be in the image of $g$ as it is finite.

    As $f$ is not computable, there is a $p>q$ such that $f(p)=1$, and thus $g(p)=g(p-1)+f(p)+1=g(p-1)+2$: consider fresh variables $u_{1}$ through $u_{k}$, for $k=g(p-1)+1-g(q)$, define $U=V\cup\{u_{1},\ldots,u_{k}\}$, and let $\delta_{U}$ be the arrangement on $U$ that matches $\delta_{V}$ on $V$, and where each $u_{i}$ is on its own equivalence class. The formula $\wit(\phi)\wedge\delta_{U}$ is $\Tsix$-satisfiable. Indeed, consider the interpretation $\B^{\prime}$ with: $\s^{\B^{\prime}}=\s^{\B}$, $\s_{2}^{\B^{\prime}}=\s_{2}^{\B}\cup A$, for some infinite set $A$ disjoint from the domains of $\B$ (what makes of $\B^{\prime}$ a $\Tsix$-interpretation); $x^{\B^{\prime}}=x^{\B}$ for all variables $x$ of sort $\s$; $u\in \{u_{1},\ldots,u_{k}\}\mapsto u^{\B^{\prime}}\in A$ an injective map; and $u^{\B^{\prime}}=u^{\B}$ for all other variables $u$ of sort $\s_{2}$, we have that $\B^{\prime}$ satisfies $\wit(\phi)\wedge\delta_{U}$. But no $\Tsix$-interpretation $\C$ that satisfies $\wit(\phi)\wedge\delta_{U}$ can have $\s_{2}^{\C}=\vars_{\s_{2}}(\wit(\phi)\wedge\delta_{U})^{\C}$, leading to a contradiction: indeed, the set $\vars_{\s_{2}}(\wit(\phi)\wedge\delta_{U})^{\C}$ must have $q(p-1)+1$ elements, an integer which is not in the image of $g$.

\end{proof}

\begin{lemma}
    The $\Sigma_{2}$-theory $\Tsix$ is not convex with respect to $\{\s,\s_{2}\}$.
\end{lemma}

\begin{proof}
    From Theorem $4$ of \cite{CADE}, if $\Tsix$ were convex, as all of its models have at least $2$ elements in each of its domains, then we would have that it is stably infinite, which is not.
\end{proof}

\begin{lemma}
    The $\Sigma_{2}$-theory $\Tsix$ has the finite model property with respect to $\{\s,\s_{2}\}$.
\end{lemma}

\begin{proof}
    Let $\phi$ be a quantifier-free formula, and $\A$ a $\Tsix$-interpretation that satisfies $\phi$: take an integer $k$ such that 
    \[g(2^{k})=3\times 2^{k-1}\leq |\vars_{\s_{2}}(\phi)^{\A}|<3\times 2^{k}=g(2^{k+1}),\]
    take sets $A$ and $B$, with respectively $m-|\vars_{\s}(\phi)^{\A}|$ and $g(2^{k+1})-|\vars_{\s_{2}}(\phi)^{\A}|$ elements. We can then define the interpretation $\B$ by making: $\s^{\B}=\vars_{\s}(\phi)^{\A}\cup A$, which has then $m$ elements; $\s_{2}^{\B}=\vars_{\s_{2}}(\phi)^{\A}\cup B$, which has then $g(2^{k+1})$ elements, making of $\B$ a $\Tsix$-interpretation; $x^{\B}=x^{\A}$ for all $x\in\vars_{\s}(\phi)$, and arbitrarily for other variables $x$ of sort $\s$; and $u^{\B}=u^{\A}$ for all $u\in\vars_{\s_{2}}(\phi)$, and arbitrarily for other variables $u$ of sort $\s_{2}$. It is then clear that $\B$ is a $\Tsix$-interpretation that satisfies $\phi$ with finitely many elements in each of its domains.
\end{proof}

\begin{lemma}
    The $\Sigma_{2}$-theory $\Tsix$ is not stably finite with respect to $\{\s,\s_{2}\}$.
\end{lemma}

\begin{proof}
    Obvious, as $\Tsix$ has models $\A$ with $|\s^{\A}|=n$ and $|\s_{2}^{\A}|=\aleph_{0}$, but no models $\B$ with $|\s^{\B}|=n$ and $\s_{2}^{\B}$ finite, nor any models $\C$ with $|\s^{\C}|<n$.
\end{proof}

\begin{lemma}
    The $\Sigma_{2}$-theory $\Tsix$ does not have a computable minimal model function with respect to $\{\s,\s_{2}\}$.
\end{lemma}

\begin{proof}
    This proof mimics that of \Cref{cmmf:tone}, referencing the second sort instead of the first.
\end{proof}

\subsection{\tp{$\Tthreeone$}{Tthreeone}}

$\Tthreeone$ is the $\Sigma_{2}$-theory with axiomatization
\[\{\psi^{\s}_{\geq n}\wedge(\psi^{\s_{2}}_{\geq g(n+1)}\vee\bigvee_{i=1}^{n}\psi^{\s_{2}}_{g_{i}})\wedge\psi^{\s_{3}}_{=1} : n\in\mathbb{N}\setminus\{0\}\}.\]
The models $\A$ of $\Tthreeone$ have: $\s^{\A}$ infinite; $\s_{2}^{\A}$ either infinite or of cardinality $g(n)$ for some $n\in\mathbb{N}\setminus\{0\}$; and $\s_{3}^{\A}$ of cardinality $1$.

\begin{lemma}
    The $\Sigma_{3}$-theory $\Tthreeone$ is not stably infinite, and thus not smooth, with respect to $\{\s,\s_{2},\s_{3}\}$.
\end{lemma}

\begin{proof}
    Obvious, as no $\Tthreeone$-interpretation $\A$ has $\s_{3}^{\A}$ infinite.
\end{proof}

\begin{lemma}
    The $\Sigma_{3}$-theory $\Tthreeone$ is not finitely witnessable, and thus not strongly finitely witnessable, with respect to $\{\s,\s_{2},\s_{3}\}$.
\end{lemma}

\begin{proof}
    Obvious: if $\Tthree$ had a witness $\wit$, then for any $\Tthree$-satisfiable quantifier-free formula $\phi$ (such as a tautology) there would exist a $\Tthreeone$-interpretation $\A$ with $\s^{\A}=\vars_{\s}(\wit(\phi))^{\A}$, $\s_{2}^{\A}=\vars_{\s_{2}}(\wit(\phi))^{\A}$ and $\s_{3}^{\A}=\vars_{\s_{3}}(\wit(\phi))^{\A}$; this is not possible given that any $\Tthreeone$-interpretation $\A$ must have $\s^{\A}$ infinite.
\end{proof}

\begin{lemma}
    The $\Sigma_{3}$-theory $\Tthreeone$ is convex with respect to $\{\s,\s_{2},\s_{3}\}$.
\end{lemma}

\begin{proof}
    Suppose that we have
    \[\phi\vDash_{\Tthreeone}\bigvee_{i=1}^{m}(x_{i}=y_{i})\vee\bigvee_{j=1}^{n}(u_{j}=v_{j})\vee\bigvee_{k=1}^{p}(w_{k}=z_{k}),\]
    for a conjunction of literals $\phi$, variables $x_{i}$ and $y_{i}$ of sort $\s$, $u_{j}$ and $v_{j}$ of sort $\s_{2}$, and $w_{k}$ and $z_{k}$ of sort $\s_{3}$. If there are indeed variables of sort $\s_{3}$ in the above disjunction we are done, since then $\phi\vDash_{\Tthreeone}w_{k}=z_{k}$ for each $1\leq k\leq p$. So we may safely assume there are no such variables: we can then remove the literals involving variables of sort $\s_{3}$ to obtain the cube $\phi^{\prime}$, as either $\phi$ is a contradiction (when we are done), or $\phi$ and $\phi^{\prime}$ are $\Tthreeone$-equivalent. But it is easy to see that
    \[\phi^{\prime}\vDash_{\Tthreeone}\bigvee_{i=1}^{m}(x_{i}=y_{i})\vee\bigvee_{j=1}^{n}(u_{j}=v_{j})\quad\text{and}\quad \phi^{\prime}\vDash_{\Tthree}\bigvee_{i=1}^{m}(x_{i}=y_{i})\vee\bigvee_{j=1}^{n}(u_{j}=v_{j})\]
    are equivalent, given the axiomatization of $\Tthreeone$: since $\Tthree$ is convex by \Cref{Tthree is convex}, we have that either $\phi^{\prime}\vDash_{\Tthree}x_{i}=y_{i}$ for some $1\leq i\leq m$, or $\phi^{\prime}\vDash_{\Tthree}u_{j}=v_{j}$ for some $1\leq j\leq n$. That way, either $\phi\vDash_{\Tthreeone}x_{i}=y_{i}$ or $\phi\vDash_{\Tthreeone}u_{j}=v_{j}$, proving $\Tthreeone$ is convex.
\end{proof}

\begin{lemma}
    The $\Sigma_{3}$-theory $\Tthreeone$ does not have the finite model property, and thus is not stably finite, with respect to $\{\s,\s_{2},\s_{3}\}$.
\end{lemma}

\begin{proof}
    Obvious, as no $\Tthreeone$-interpretation $\A$ has $\s^{\A}$ finite.
\end{proof}

\begin{lemma}
    The $\Sigma_{3}$-theory $\Tthreeone$ does not have a computable minimal model function with respect to $\{\s,\s_{2},\s_{3}\}$.
\end{lemma}

\begin{proof}
    Suppose, instead, that a minimal model function of $\Tthreeone$ is computable: we define a function $h:\mathbb{N}\setminus\{0\}\rightarrow\mathbb{N}\setminus\{0\}$ by making $h(1)=2$ and 
    \[h(n+1)=\min\{q : (p,q,r)\in \minmod_{\Tthreeone,S}(\NNNEQ{u}{h(n)+1})\},\]
    where $S=\{\s,\s_{2},\s_{3}\}$, the variables $u_{i}$ are of sort $\s_{2}$, and $\NNNEQ{u}{m}$ is $\Tthreeone$-satisfiable for every value of $m$;    $h$ is computable given our hypothesis on a minimal model function, which is computable and outputs a finite set. Of course this is contradictory, since $h(n)=g(n)$ for all $n\in\mathbb{N}\setminus\{0\}$, while $g$ is not computable.
\end{proof}

\subsection{\tp{$\Tseven$}{Tseven}}

$\Tseven$ is the $\Sigma_{s}$-theory with axiomatization
\[\psi_{\geq \bb(n+1)}\vee\bigvee_{i=1}^{n}(\psi_{=\bb(i)}\wedge\psi_{=}) : n\in\mathbb{N}\setminus\{0\}\},\]
where $\psi_{=}=\Forall{x}(s(x)=x)$.
Its models $\A$ have: either an infinite number of elements, and $s^{\A}$ any function; or a finite number $\bb(n)$ of elements, for $n\in\mathbb{N}\setminus\{0\}$, and $s^{\A}$ the identity.

\begin{lemma}\label{Tseven is SI}
    The $\Sigma_{s}$-theory $\Tseven$ is stably infinite with respect to its only sort.
\end{lemma}

\begin{proof}
    Take a quantifier-free formula $\phi$, a $\Tseven$-interpretation $\A$ that satisfies $\phi$, and an infinite set $A$ disjoint from $\s^{\A}$, and we define a $\Tseven$-interpretation $\B$ by making: $\s^{\B}=\s^{\A}\cup A$, which is then infinite; $s^{\B}(a)=a$ for all $a\in\s^{\B}$, so $\B$ is indeed a $\Tseven$-interpretation; and $x^{\B}=x^{\A}$ for all variables $x$ of sort $\s$, meaning $\B$ satisfies $\phi$.
\end{proof}

\begin{lemma}\label{Tseven is -SM}
    The $\Sigma_{s}$-theory $\Tseven$ is not smooth with respect to its only sort.
\end{lemma}

\begin{proof}
    Obvious, as $\bb(1)=1$, $\bb(2)=4$ and $\bb(n+1)>\bb(n)$ for all $n\in\mathbb{N}$, meaning $\Tseven$ has models of cardinality $1$, but none with cardinality $2$.
\end{proof}

\begin{lemma}\label{Tseven is -FW}
    The $\Sigma_{s}$-theory $\Tseven$ is not finitely witnessable, and thus not strongly finitely witnessable, with respect to its only sort.
\end{lemma}

\begin{proof}
    We proceed by contradiction. Start by assuming $\Tseven$ has a witness $\wit$: we define $h:\mathbb{N}\rightarrow\mathbb{N}$ by making $h(0)=0$, $h(1)=1$, and 
    \[h(n+1)=|\vars_{\s}(\wit(\NNNEQ{x}{h(n)+1}))|,\]
    for $n\geq 2$; this is a recursively computable function as constructing $\NNNEQ{x}{h(n)+1}$ can be done algorithmically once $h(n)+1$ is known, and $\wit$ is computable. 
    
    Now we prove, by induction, that $h(n)\geq\bb(n)$, what leads to a contradiction: that is obvious for $n=0$ and $n=1$; suppose then that $h(n)\geq \bb(n)$. Since $\NNNEQ{x}{h(n)+1}$ is $\Tseven$-satisfiable (by, say, an infinite $\Tseven$-interpretation), there is a $\Tseven$-interpretation $\A$ that satisfies $\NNNEQ{x}{h(n)+1}$ with $\vars_{\s}(\wit(\NNNEQ{x}{h(n)+1}))^{\A}=\s^{\A}$. Since $\A$ satisfies $\NNNEQ{x}{h(n)+1}$, it has at least $\bb(n)+1$ elements, and therefore at least $\bb(n+1)$ since $\A$ (as a $\Tseven$-interpretation) must have, if finite, $\bb(m)$ elements, for some $m\in\mathbb{N}\setminus\{0\}$;  then we get that
    \[h(n+1)=|\vars_{\s}(\wit(\NNNEQ{x}{h(n)+1}))|\geq |\vars_{\s}(\wit(\NNNEQ{x}{h(n)+1}))^{\A}|=|\s^{\A}|\geq \bb(n+1),\] as we wanted to show.
\end{proof}

\begin{lemma}
    The $\Sigma_{s}$-theory $\Tseven$ is convex with respect to its only sort.
\end{lemma}

\begin{proof}
    This proof is similar to that of Lemma $66$ in \cite{arxivCADE}.
\end{proof}

\begin{lemma}\label{Tseven is -FMP}
    The $\Sigma_{s}$-theory $\Tseven$ does not have the finite model property, and is thus not stably finite, with respect to its only sort.
\end{lemma}

\begin{proof}
    Define the $\Tseven$-interpretation $\A$ with $\s^{\A}=\{a_{n} : n\in\mathbb{N}\}$ and $s^{\A}(a_{n})=a_{n+1}$: then, regardless of the value assigned to variables, $\A$ satisfies the quantifier-free formula $\phi$ given by $\neg(s(x)=x)$. Yet all finite $\Tseven$-interpretations satisfy $\Forall{x}(s(x)=x)$, so none of them can satisfy $\phi$.
\end{proof}

\begin{lemma}\label{Tseven is -CMMF}
    The $\Sigma_{s}$-theory $\Tseven$ does not have a computable minimal model function with respect to its only sort.
\end{lemma}

\begin{proof}
    This proof follows in the steps of \Cref{cmmf:tfour}.
\end{proof}

\subsection{\tp{$\Teight$}{Teight}}

$\Teight$ is the $\Sigma_{s}$-theory with axiomatization
\[(\psi_{\geq \bb(n+1)}\wedge\psi_{\vee})\vee\bigvee_{i=1}^{n}(\psi_{=\bb(i)}\wedge\psi_{=}) : n\in\mathbb{N}\setminus\{0\}\},\]
where $\psi_{\vee}=\Forall{x}((s^{2}(x)=x)\vee(s^{2}(x)=s(x))$ and $\psi_{=}=\Forall{x}(s(x)=x)$.
Its models $\A$ have: either an infinite number of elements, and $(s^{\A})^{2}(a)$ equal to either $s^{\A}(a)$ or $a$ itself; or a finite number $\bb(n)$ of elements, for $n\in\mathbb{N}\setminus\{0\}$, and $s^{\A}$ the identity. Notice that $\Teight$ is nothing but the result of applying the operator $(\cdot)_{\vee}$, understood in a generalized way, to the theory $\Tseven$: so it is somewhat clear that $\Teight$ shares all of the properties (or absence of them) with $\Tseven$, except for convexity.

\begin{lemma}
    The $\Sigma_{s}$-theory $\Teight$ is stably infinite with respect to its only sort.
\end{lemma}

\begin{proof}
    Similar to the proof of \Cref{Tseven is SI}.
\end{proof}

\begin{lemma}
    The $\Sigma_{s}$-theory $\Teight$ is not smooth with respect to its only sort.
\end{lemma}

\begin{proof}
Similar to the proof of \Cref{Tseven is -SM}.
\end{proof}

\begin{lemma}
    The $\Sigma_{s}$-theory $\Teight$ is not finitely witnessable, and thus not strongly finitely witnessable, with respect to its only sort.
\end{lemma}

\begin{proof}
Similar to the proof of \Cref{Tseven is -FW}.
\end{proof}

\begin{lemma}
    The $\Sigma_{s}$-theory $\Teight$ is not convex with respect to its only sort.
\end{lemma}

\begin{proof}
    It is clear that 
    \[(y=s(x))\wedge(z=s(y))\vDash_{\Teight}(x=z)\vee(y=z)\]
    since: if $\A$ is a finite $\Teight$-interpretation it satisfies $\Forall{x}(s(x)=x)$ and therefore both $x=z$ and $y=z$ are satisfied; or, if $\A$ is an infinite $\Teight$-interpretation, it satisfies $\psi_{\vee}$.

    However, take the $\Teight$-interpretations $\A$ and $\B$ with $\s^{\A}=\s^{\B}=\mathbb{N}$;
    \[s^{\A}(n)=\begin{cases}
        n+1 & \text{if $n$ is even;}\\
        n-1 & \text{if $n$ is odd;}\\
    \end{cases}\quad
    s^{\B}(n)=\begin{cases}
        n+1 & \text{if $n$ is even;}\\
        n & \text{if $n$ is odd;}\\
    \end{cases}\]
    and $x^{\A}=x^{\B}=z^{\A}=0$, and $y^{\A}=y^{\B}=z^{\B}=1$. Then $\A$ does not satisfy $y=z$, while $\B$ does not satisfy $x=z$, meaning that neither
    \[(y=s(x))\wedge(z=s(y))\vDash_{\Teight}(x=z)\quad\text{nor}\quad(y=s(x))\wedge(z=s(y))\vDash_{\Teight}(y=z).\]
\end{proof}

\begin{lemma}
    The $\Sigma_{s}$-theory $\Teight$ does not have the finite model property, and is thus not stably finite, with respect to its only sort.
\end{lemma}

\begin{proof}
Similar to the proof of \Cref{Tseven is -FMP}.
\end{proof}

\begin{lemma}
    The $\Sigma_{s}$-theory $\Tseven$ does not have a computable minimal model function with respect to its only sort.
\end{lemma}

\begin{proof}
Similar to the proof of \Cref{Tseven is -CMMF}.
\end{proof}

\subsection{\tp{$\Tnine$}{Tnine}}

$\Tnine$ is the $\Sigma_{s}$-theory with axiomatization
\[\{(\psi_{\geq2\bb(n+1)}\wedge\psi_{\neq})\vee\bigvee_{i=2}^{n}(\psi_{=2\bb(i)}\wedge\psi_{\neq}\wedge\psi_{=}^{2})\vee\psi_{=1} : n\in\mathbb{N}\setminus\{0,1\}\},\]
where $\psi_{\neq}=\Forall{x}\neg(s(x)=x)$ and $\psi_{=}^{2}=\Forall{x}(s^{2}(x)=x)$. The models $\A$ of $\Tnine$ have either: $|\s^{\A}|=1$, and $s^{\A}$ the identity; an infinite number of elements, and $s^{\A}(a)\neq a$ for all $a\in\s^{\A}$; or a finite number $2\bb(n)$ of elements, for $n\in\mathbb{N}\setminus\{0,1\}$, $s^{\A}(a)\neq a$ for all $a\in\s^{\A}$, but $(s^{\A})^{2}(a)=a$. Notice we use $2\bb(n)$ for the cardinalities of (most of) the finite $\Tnine$-interpretations because the facts that $s^{\A}(a)\neq a$ and that $(s^{\A})^{2}(a)=a$ imply that $\s^{\A}$ must have an even number of elements.

\begin{lemma}
    The $\Sigma_{s}$-theory $\Tnine$ is not stably infinite, and thus not smooth, with respect to its only sort.
\end{lemma}

\begin{proof}
    Obvious, as the $\Tnine$-interpretation $\A$ with $|\s^{\A}|=1$ satisfies $s(x)=x$, while no other $\Tnine$-interpretation satisfies that formula.
\end{proof}

\begin{lemma}
    The $\Sigma_{s}$-theory $\Tnine$ is not finitely witnessable, and thus not strongly finitely witnessable, with respect to its only sort.
\end{lemma}

\begin{proof}
    Assume instead that $\Tnine$ actually has a witness, and then define $h:\mathbb{N}\rightarrow\mathbb{N}$ by making $h(0)=0$, $h(1)=2$, and 
    \[h(n+1)=|\vars_{\s}(\wit(\NNNEQ{x}{h(n)+1}))|,\]
    for $n\geq 2$: this is a computable function as $\wit$ is itself computable. But it is easy to prove, by induction, that $h(n)\geq2\bb(n)$, what leads to a contradiction.
\end{proof}

\begin{lemma}
    The $\Sigma_{s}$-theory $\Tnine$ is convex with respect to its only sort.
\end{lemma}

\begin{proof}
This proof is similar to that of Lemma $66$ in \cite{arxivCADE}.
\end{proof}

\begin{lemma}
    The $\Sigma_{s}$-theory $\Tnine$ does not have the finite model property, and is thus not stably finite, with respect to its only sort.
\end{lemma}

\begin{proof}
    Obvious, as no finite $\Tnine$-interpretation can satisfy $\neg(s^{2}(x)=x)$, while the infinite ones can.
\end{proof}

\begin{lemma}
    The $\Sigma_{s}$-theory $\Tnine$ does not have a computable minimal model function with respect to its only sort.
\end{lemma}

\begin{proof}
    This proof is similar enough to that of \Cref{cmmf:tfour}, but arriving at $h(n)=2\bb(n)$ instead of $h(n)=\bb(n)$.
\end{proof}

\subsection{\tp{$\Tten$}{Tten}}

$\Tten$ is the $\Sigma_{2}$-theory with axiomatization
\[\{\psi^{\s}_{=2}\wedge(\psi^{\s_{2}}_{\geq g(n+1)}\vee\bigvee_{i=1}^{n}\psi^{\s_{2}}_{=g(i)}) : n\in\mathbb{N}\setminus\{0\}\}.\]
The models $\A$ of $\Tten$ have: $\s^{\A}$ of cardinality $2$; and either $\s_{2}^{\A}$ infinite, or finite and equal to $g(n)$ for some $n\in\mathbb{N}\setminus\{0\}$.

\begin{lemma}
    The $\Sigma_{2}$-theory $\Tten$ is not stably infinite, and thus not smooth, with respect to $\{\s,\s_{2}\}$.
\end{lemma}

\begin{proof}
    Obvious, as no $\Tten$-interpretation $\A$ has $\s^{\A}$ infinite.
\end{proof}

\begin{lemma}
    The $\Sigma_{2}$-theory $\Tten$ is finitely witnessable with respect to $\{\s,\s_{2}\}$.
\end{lemma}

\begin{proof}
    Take a quantifier-free formula $\phi$, let $V$ be its set of variables of sort $\s_{2}$, take $k$ as the smallest positive integer such that $g(2^{k-1})=3\times 2^{k-2}< |V|\leq 3\times 2^{k-1}=g(2^{k})$, and sets of fresh variables $\{x_{1},x_{2}\}$ of sort $\s$ and $\{u_{1},\ldots,u_{g(2^{k})}\}$ of sort $\s_{2}$: we define the witness of $\Tten$ as
    \[\wit(\phi)=\phi\wedge(x_{1}=x_{1})\wedge(x_{2}=x_{2})\wedge\bigwedge_{i=1}^{g(2^{k})}(u_{i}=u_{i}),\]
    which is clearly computable. Let $\overarrow{x}=\vars(\wit(\phi))\setminus\vars(\phi)$: it is clear that $\phi$ and $\Exists{\overarrow{x}}\wit(\phi)$ are $\Tten$-equivalent as $\phi$ and $\wit(\phi)$ are already themselves equivalent.

    Now suppose $\A$ is a $\Tten$-interpretation that satisfies $\wit(\phi)$, let $A$ be a set of cardinality $g(2^{k})-|V^{\A}|$, and we define an interpretation $\B$ by making: $\s^{\B}=\s^{\A}$ (which has then cardinality $2$); $\s_{2}^{\B}=V^{\A}\cup A$ (which has then cardinality $g(2^{k})$, making of $\B$ a $\Tten$-interpretation); $x^{\B}=x^{\A}$ for all variables $x$ of sort $\s$; $u^{\B}=u^{\A}$ for all $u\in V$, and arbitrarily for all other variables of sort $\s_{2}$. It is then clear that $\B$ satisfies $\phi$.

    We then change the value assigned by $\B$ to the variables in $\overarrow{x}$ to obtain a second $\Tten$-interpretation $\B^{\prime}$ which still satisfies $\phi$ and thus $\wit(\phi)$: by making $x_{i}\in \{x_{1},x_{2}\}\mapsto x_{i}^{\B^{\prime}}\in\s^{\B}$ and $u_{i}\in\{u_{1},\ldots,u_{g(2^{k})}\}\mapsto u_{i}^{\B^{\prime}}\in\s_{2}^{\B}$ bijections, we get that $\s^{\B^{\prime}}=\vars_{\s}(\wit(\phi))^{\B^{\prime}}$ and $\s_{2}^{\B^{\prime}}=\vars_{\s_{2}}(\wit(\phi))^{\B^{\prime}}$, and so $\wit$ is indeed a witness.
\end{proof}

\begin{lemma}
    The $\Sigma_{2}$-theory $\Tten$ is not strongly finitely witnessable with respect to $\{\s,\s_{2}\}$.
\end{lemma}

\begin{proof}
    We proceed by contradiction: suppose then that $\Tten$ has a strong witness, let $\phi$ be a $\Tten$-satisfiable formula (a tautology will do), and $\A$ a $\Tten$-interpretation that satisfies $\phi$. Since, for $\overarrow{x}=\vars(\wit(\phi))\setminus\vars(\phi)$, $\phi$ and $\Exists{\overarrow{x}}\wit(\phi)$ are $\Tten$-equivalent, there exists a $\Tten$-interpretation $\A^{\prime}$ (differing from $\A$ at most on the value assigned to the variables in $\overarrow{x}$) that satisfies $\wit(\phi)$. 
    
    Let $V$ be the set of variables of sort $\s_{2}$ in $\wit(\phi)$, and take the arrangement $\delta_{V}$ induced by $\A^{\prime}$ on $V$, so that $\A^{\prime}$ satisfies $\wit(\phi)\wedge\delta_{V}$: there must then exist a $\Tten$-interpretation $\B$ that satisfies $\wit(\phi)\wedge\delta_{V}$ with $\vars_{\s_{2}}(\wit(\phi)\wedge\delta_{V})^{\B}=\s_{2}^{\B}$; let $q$ be the positive integer such that $g(q)=|\s_{2}^{\B}|$. Since $f$ is not computable (and therefore not eventually constant), there is a $p>q$ such that $f(p)=1$, and thus $g(p)=g(p-1)+f(p)+1=g(p-1)+2$; we then take fresh variables $U^{\prime}=\{u_{1},\ldots,u_{k}\}$, for $k=g(p-1)+1-g(q)$, define $U=V\cup U^{\prime}$, and let $\delta_{U}$ be the arrangement on $U$ that matches $\delta_{V}$ on $V$, with $k$ more equivalence classes in total. 
    
    Notice that the formula $\wit(\phi)\wedge\delta_{U}$ is $\Tten$-satisfiable. Indeed, consider the interpretation $\B^{\prime}$ with: $\s^{\B^{\prime}}=\s^{\B}$; $\s_{2}^{\B^{\prime}}=\s_{2}^{\B}\cup A$ (for an infinite set $A$ disjoint from $\s^{\B}$ and $\s_{2}^{\B}$, making of $\B^{\prime}$ a $\Tten$-interpretation); $x^{\B^{\prime}}=x^{\B}$ for all variables $x$ of sort $\s$; $u\in U^{\prime}\mapsto u^{\B^{\prime}}\in A$ an injective map; and $u^{\B^{\prime}}=u^{\B}$ for all other variables $u$ of sort $\s_{2}$. This way, $\B^{\prime}$ satisfies $\wit(\phi)\wedge\delta_{U}$; yet no $\Tten$-interpretation $\C$ that satisfies $\wit(\phi)\wedge\delta_{U}$ can have $\s_{2}^{\C}=\vars_{\s_{2}}(\wit(\phi)\wedge\delta_{U})^{\C}$, as the set $\vars_{\s_{2}}(\wit(\phi)\wedge\delta_{U})^{\C}$ must have $g(p-1)+1\notin\{g(n) : n\in\mathbb{N}\setminus\{0\}\}$ elements, leading to the desired contradiction.
\end{proof}

\begin{lemma}
    The $\Sigma_{2}$-theory $\Tten$ is not convex with respect to $\{\s,\s_{2}\}$.
\end{lemma}

\begin{proof}
    By the pigeonhole principle, it is clear that $\vDash_{\Tten}(x=y)\vee(x=z)\vee(y=z)$
    for variables $x$, $y$ and $z$ of sort $\s$. Yet we have neither $\vDash_{\Tten} x=y$, $\vDash_{\Tten} x=z$ or $\vDash_{\Tten} y=z$. Indeed, take the $\Tten$-interpretations $\A$ and $\B$ with: $\s^{\A}=\s^{\B}=\{a,b\}$; $\s_{2}^{\A}=\s_{2}^{\B}=\{c,d\}$ (making of them indeed $\Tten$-interpretations, as $g(1)=2$); $x^{\A}=y^{\A}=x^{\B}=z^{\B}=a$, and $z^{\A}=y^{\B}=b$. Then $\A$ does not satisfy neither $x=z$ nor $y=z$, while $\B$ does not satisfy $x=y$.
\end{proof}

\begin{lemma}
    The $\Sigma_{2}$-theory $\Tten$ is stably finite, and thus has the finite model property, with respect to $\{\s,\s_{2}\}$.
\end{lemma}

\begin{proof}
    Let $\phi$ be a quantifier-free formula, $\A$ a $\Tten$-interpretation that satisfies $\phi$ (we may assume with $|\s_{2}^{\A}|\geq \aleph_{0}$, as otherwise we would have nothing to prove), $V$ the set of variables of sort $\s_{2}$ in $\phi$, $k$ the least positive integer such that $g(2^{k-1})=3\times 2^{k-2}<|V|\leq 3\times 2^{k-1}=g(2^{k})$, and $A$ a set of cardinality $g(2^{k})-|V^{\A}|$ disjoint from both domains of $\A$. We then define an interpretation $\B$ by making: $\s^{\B}=\s^{\A}$ (which has therefore cardinality $2$); $\s_{2}^{\B}=V^{\A}\cup A$ (which has therefore cardinality $g(2^{k})$, thus making of $\B$ a $\Tten$-interpretation); $x^{\B}=x^{\A}$ for all variables $x$ of sort $\s$; $u^{\B}=u^{\A}$ for all $u\in V$, and arbitrarily for all other variables $u$ of sort $\s_{2}$. Then we have that both $\s^{\B}$ and $\s_{2}^{\B}$ are finite, that $|\s^{\B}|\leq|\s^{\A}|$ and $|\s_{2}^{\B}|\leq|\s_{2}^{\A}|$, and that $\B$ satisfies $\phi$, finishing our proof.
\end{proof}

\begin{lemma}
    The $\Sigma_{2}$-theory $\Tten$ does not have a computable minimal model function with respect to $\{\s,\s_{2}\}$.
\end{lemma}

\begin{proof}
    The proof is by contradiction: if $\minmod_{\Tten,S}$ is computable, for $S=\{\s,\s_{2}\}$, we may define $h:\mathbb{N}\setminus\{0\}\rightarrow\mathbb{N}$ by $h(1)=2$, and for $n\geq 2$
    \[h(n+1)=m,\quad\text{where}\quad\minmod_{\Tten,S}(\NNNEQ{u}{h(n)+1})=\{(2,m)\},\]
    the variables $u_{i}$ are of sort $\s_{2}$, and $\NNNEQ{u}{m}$ is $\Tten$-satisfiable for every value of $m$: $h$ is computable as this minimal model function of $\Tten$ is assumed computable, and always outputs finite sets; at the same time it is clear that $h(n)=g(n)$, giving us the contradiction.
\end{proof}

\subsection{\tp{$\Tthreetwo$}{Tthreetwo}}

$\Tthreetwo$ is the $\Sigma_{3}$-theory with axiomatization
\[\{\psi^{\s_{3}}_{=1}\}\cup\{(\psi^{\s}_{=1}\wedge\psi^{\s_{2}}_{\geq n})\vee((\psi^{\s}_{\geq g(n+1)}\wedge\psi^{\s_{2}}_{\geq g(n+1)})\vee\bigvee_{i=1}^{n}(\psi^{\s}_{=g(i)}\wedge\psi^{\s_{2}}_{=g(i)})) : n\in\mathbb{N}\setminus\{0\}\}.\]
$\Tthreetwo$ has three classes of model $\A$, all satisfying $|\s_{3}^{\A}|=1$: those with $|\s^{\A}|=1$, and $|\s_{2}^{\A}|\geq\aleph_{0}$; those with $|\s^{\A}|=|\s_{2}^{\A}|=g(n)$, for some $n\in\mathbb{N}\setminus\{0\}$; and those with $|\s^{\A}|,|\s_{2}^{\A}|\geq\aleph_{0}$.

\begin{lemma}
    The $\Sigma_{3}$-theory $\Tthreetwo$ is not stably infinite, and thus not smooth, with respect to $\{\s,\s_{2},\s_{3}\}$.
\end{lemma}

\begin{proof}
    Obvious, as no $\Tthreetwo$-interpretation has the domain of sort $\s_{3}$ infinite.
\end{proof}

\begin{lemma}\label{Tthreetwo is FW}
    The $\Sigma_{3}$-theory $\Tthreetwo$ is finitely witnessable with respect to $\{\s,\s_{2},\s_{3}\}$.
\end{lemma}

\begin{proof}
    Take a quantifier-free formula $\phi$, let $V=\vars_{\s}(\phi)$ and $V_{2}=\vars_{\s_{2}}(\phi)$, the least positive integer $k$ such that 
    \[g(2^{k-1})=3\times 2^{k-2}<\max\{|V|,|V_{2}|\}\leq 3\times 2^{k-1}=g(2^{k}),\]
    and fresh variables $\{x_{1},\ldots,x_{g(2^{k})}\}$ of sort $\s$, $\{u_{1},\ldots,u_{g(2^{k})}\}$ of sort $\s_{2}$, and $\{z\}$ of sort $\s_{3}$. We then define the witness of $\phi$ as
    \[\wit(\phi)=\phi\wedge\bigwedge_{i=1}^{g(2^{k})}(x_{i}=x_{i})\wedge\wedge\bigwedge_{i=1}^{g(2^{k})}(u_{i}=u_{i})\wedge(z=z),\]
    which is of course computable. Since $\phi$ and $\wit(\phi)$ are equivalent, $\phi$ and $\Exists{\overarrow{x}}\wit(\phi)$ are $\Tthreetwo$-equivalent, for $\overarrow{x}=\vars(\wit(\phi))\setminus\vars(\phi)$.

    Furthermore, take a $\Tthreetwo$-interpretation $\A$ that satisfies $\wit(\phi)$, and thus $\phi$: we take sets $A$ and $B$, disjoint from the domains of $\A$, with $g(2^{k})-|V^{\A}|$ and $g(2^{k})-|V_{2}^{\A}|$ elements, respectively. We define an interpretation $\B$ by making: $\s^{\B}=V^{\B}\cup A$ (with cardinality $g(2^{k})$); $\s_{2}^{\B}=V_{2}^{\A}\cup B$ (with cardinality $g(2^{k})$); $\s_{3}^{\B}=\s_{3}^{\A}$ (with cardinality $1$, thus making of $\B$ a $\Tthreetwo$-interpretation); $x^{\B}=x^{\A}$ for all $x\in V$, $x_{i}\in\{x_{1},\ldots,x_{g(2^{k})}\}\mapsto x_{i}^{\B}\in\s^{\B}$ a bijection, and arbitrarily for all other variables $x$ of sort $\s$; $u^{\B}=u^{\A}$ for all $u\in V_{2}$, $u_{i}\in\{u_{1},\ldots,u_{g(2^{k})}\}\mapsto u_{i}^{\B}\in\s_{2}^{\B}$ a bijection, and arbitrarily for all other variables $u$ of sort $\s_{2}$; and $w^{\B}$ the only element of $\s_{3}^{\B}$, for all variables $w$ of sort $\s_{3}$. This way $\B$ satisfies $\phi$, and thus $\wit(\phi)$, and has $\s^{\B}=\vars_{\s}(\wit(\phi))^{\B}$, and $\s_{j}^{\B}=\vars_{\s_{j}}(\wit(\phi))^{\B}$ for $j\in\{2,3\}$.
\end{proof}

\begin{lemma}
    The $\Sigma_{3}$-theory $\Tthreetwo$ is not strongly finitely witnessable with respect to $\{\s,\s_{2},\s_{3}\}$.
\end{lemma}

\begin{proof}
    Suppose $\Tthreetwo$ actually has a strong witness $\wit$, take a $\Tthreetwo$-satisfiable quantifier-free formula $\phi$ (a tautology works), and a $\Tthreetwo$-interpretation $\A$ that satisfies $\phi$: since $\phi$ and $\Exists{\overarrow{x}}\wit(\phi)$ are $\Tthreetwo$-equivalent, for $\overarrow{x}=\vars(\wit(\phi))\setminus\vars(\phi)$, we have that there is a $\Tthreetwo$-interpretation $\A^{\prime}$ differing from $\A$ at most on the values assigned to the variables in $\overarrow{x}$ that satisfies $\wit(\phi)$. Let $V=\vars(\wit(\phi))$, and $\delta_{V}$ the arrangement induced by $\A^{\prime}$ on $V$, meaning we also have that $\A^{\prime}$ satisfies $\wit(\phi)\wedge\delta_{V}$. We thus know that there is a $\Tthreetwo$-interpretation $\B$ that satisfies $\wit(\phi)\wedge\delta_{V}$ with $\s^{\B}=\vars_{\s}(\wit(\phi)\wedge\delta_{V})^{\B}$, and $\s_{j}^{\B}=\vars_{\s_{j}}(\wit(\phi)\wedge\delta_{V})^{\B}$ for $j\in\{2,3\}$.

    Because all domains of $\B$ are finite, we know that $|\s^{\B}|=|\s_{2}^{\B}|=g(n)$, for some $n\in\mathbb{N}\setminus\{0\}$, and $|\s_{3}^{\B}|=1$. Take then a fresh variable $x$ of sort $\s$, consider the set $U=V\cup\{x\}$, and the arrangement $\delta_{U}$ that matches $\delta_{V}$ on $V$, and where $x$ is not in the same equivalence class of any variable on $V$. It is easy to see that $\wit(\phi)\wedge\delta_{U}$ is $\Tthreetwo$-satisfiable, and thus there must exist a $\Tthreetwo$-interpretation $\C$ that satisfies that formula with $\s^{\C}=\vars_{\s}(\wit(\phi)\wedge\delta_{U})^{\C}$, and $\s_{j}^{\C}=\vars_{\s_{j}}(\wit(\phi)\wedge\delta_{U})^{\C}$ for $j\in\{2,3\}$: but that implies $|\s^{\C}|=g(n)+1$, $|\s_{2}^{\C}|=g(n)$ and $|\s_{3}^{\C}|=1$, what constitutes a contradiction.
\end{proof}

\begin{lemma}
    The $\Sigma_{3}$-theory $\Tthreetwo$ is convex with respect to $\{\s,\s_{2},\s_{3}\}$.
\end{lemma}

\begin{proof}
    Suppose that 
    \[\phi\vDash_{\Tthreetwo}\bigvee_{j=1}^{3}\bigvee_{i=1}^{n_{j}}x_{i}^{j}=y_{i}^{j},\]
    where $\phi$ is a cube, and $x_{i}^{1}$ and $y_{i}^{1}$ are variables of sort $\s$, while $x_{i}^{j}$ and $y_{i}^{j}$ are variables of sort $\s_{j}$, for $j\in\{2,3\}$. Now, since for any $\Tthreetwo$-interpretation $\A$, $|\s_{3}^{\A}|=1$, we can safely assume that there are actually no variables of sort $\s_{3}$ in the disjunction, otherwise there is nothing to prove; likewise, we may remove any literals from $\phi$ including variables of sort $\s_{3}$, as they are either tautologies (meaning removing them returns an equivalent formula) or contradictions (when there is nothing to prove).

    So we obtain 
    \[\phi^{\prime}\vDash_{\Tthreetwo}\bigvee_{j=1}^{2}\bigvee_{i=1}^{n_{j}}x_{i}^{j}=y_{i}^{j},\] 
    for a $\Sigma_{2}$-formula $\phi^{\prime}$. It is easy to see, as the axiomatization of $\Tthreetwo$ just adds the formula $\psi^{\s_{3}}_{=1}$ to the axiomatization of $\Ttwo$, that this implies $\phi^{\prime}\vDash_{\Ttwo}\bigvee_{j=1}^{2}\bigvee_{i=1}^{n_{j}}x_{i}^{j}=y_{i}^{j}$. Since $\Ttwo$ is convex by \Cref{Ttwo is CV}, there must exist $1\leq j\leq 2$ and $1\leq i\leq n_{j}$ such that $\phi^{\prime}\vDash_{\Ttwo}x_{i}^{j}=y_{i}^{j}$, and we thus get that $\phi^{\prime}\vDash_{\Tthreetwo}x_{i}^{j}=y_{i}^{j}$.

\end{proof}

\begin{lemma}
    The $\Sigma_{3}$-theory $\Tthreetwo$ has the finite model property with respect to $\{\s,\s_{2},\s_{3}\}$.
\end{lemma}

\begin{proof}
    Follows from Theorem 2 of \cite{FroCoS} and \Cref{Tthreetwo is FW}.
\end{proof}

\begin{lemma}
    The $\Sigma_{3}$-theory $\Tthreetwo$ is not stably finite with respect to $\{\s,\s_{2},\s_{3}\}$.
\end{lemma}

\begin{proof}
        Obvious, since there is a $\Tthreetwo$-interpretation $\A$ with $|\s^{\A}|=|\s_{3}^{\A}|=1$ and $|\s_{2}^{\A}|=\aleph_{0}$, while there are no $\Tthreetwo$-interpretations with all domains finite, $|\s^{\B}|\leq |\s^{\A}|$, $|\s_{2}^{\B}|\leq |\s_{2}^{\A}|$ and $|\s_{3}^{\B}|\leq|\s_{3}^{\A}|$.
\end{proof}

\begin{lemma}
    The $\Sigma_{3}$-theory $\Tthreetwo$ does not have a computable minimal model function with respect to $\{\s,\s_{2},\s_{3}\}$.
\end{lemma}

\begin{proof}
    Suppose instead that $\Tthreetwo$ has a computable minimal model function: we then define a function $h:\mathbb{N}\setminus\{0\}\rightarrow\mathbb{N}$ by making $h(1)=2$ and, for $S=\{\s,\s_{2},\s_{3}\}$ and $n\geq 2$, 
    \[h(n+1)=m,\quad\text{where}\quad \minmod_{\Tthreetwo,S}(\NNNEQ{u}{h(n)+1})=\{(m,m,1)\}.\]
    We get $h$ is computable, on one hand, given that $\minmod_{\Tthreetwo,S}$ is assumed computable; one the other, it is easy to prove that $h(n)=g(n)$ for all $n\in\mathbb{N}\setminus\{0\}$, and so $h$ is not computable, meaning we reach a contradiction.
\end{proof}

\subsection{\tp{$\Televen$}{Televen}}

$\Televen$ is the $\Sigma_{2}$-theory with axiomatization 
\[\{\psi^{\s}_{=1}\wedge(\psi^{\s_{2}}_{\geq g(n+1)}\vee\bigvee_{i=1}^{n}\psi^{\s_{2}}_{=g(i)}) : n\in\mathbb{N}\setminus\{0\}\}.\]
The models $\A$ of $\Televen$ have: $\s^{\A}$ of cardinality $1$; and either $\s_{2}^{\A}$ infinite, or finite and equal to $g(n)$, for $n\in\mathbb{N}\setminus\{0\}$.

\begin{lemma}
    The $\Sigma_{2}$-theory $\Televen$ is not stably infinite, and thus not smooth, with respect to $\{\s,\s_{2}\}$.
\end{lemma}

\begin{proof}
    Obvious: no $\Televen$-interpretation $\A$ has $\s^{\A}$ infinite.
\end{proof}

\begin{lemma}
    The $\Sigma_{2}$-theory $\Televen$ is finitely witnessable with respect to $\{\s,\s_{2}\}$.
\end{lemma}

\begin{proof}
    Take: a quantifier-free formula $\phi$; $V$ its set of variables of sort $\s_{2}$; $k$ the least positive integer such that 
    \[g(2^{k-1})=3\times 2^{k-2}< |V|\leq 3\times 2^{k-1}=g(2^{k});\]
    and sets $\{x_{1}\}$ and $\{u_{1},\ldots,u_{g(2^{k})}\}$ of fresh variables of sorts, respectively, $\s$ and $\s_{2}$. We then state that
    \[\wit(\phi)=\phi\wedge(x_{1}=x_{1})\wedge\bigwedge_{i=1}^{g(2^{k})}(u_{i}=u_{i}),\]
    is a witness, being obviously computable. It is obvious that, since $\phi$ and $\wit(\phi)$ are equivalent, $\phi$ and $\Exists{\overarrow{x}}\wit(\phi)$ are $\Televen$-equivalent, where $\overarrow{x}=\vars(\wit(\phi))\setminus\vars(\phi)$. Now: suppose $\A$ is a $\Televen$-interpretation that satisfies $\wit(\phi)$; and take a set $A$ of cardinality $g(2^{k})-|V^{\A}|$ disjoint from the domains of $\A$. We define an interpretation $\B$ by making: $\s^{\B}=\s^{\A}$ (so $|\s^{\B}|=|\s^{\A}|=1$); $\s_{2}^{\B}=V^{\A}\cup A$ (so $|\s^{\B}|=|V^{\A}|+|A|=|V^{\A}|+(g(2^{k})-|V^{\A}|)=g(2^{k})$, what makes of $\B$ a $\Televen$-interpretation); $x^{\B}=x^{\A}$ for all variables $x$ of sort $\s$; $u^{\B}=u^{\A}$ for all $u\in V$; and $u^{\B}$ can be defined arbitrarily for all other variables $u$ of sort $\s_{2}$. With this definition $\B$ satisfies $\phi$.

    We construct yet another $\Televen$-interpretation $\B^{\prime}$ by changing only the value assigned by $\B$ to the variables in $\overarrow{x}$, meaning that $\B^{\prime}$ still satisfies $\phi$ and thus $\wit(\phi)$: by making $u_{i}\in\{u_{1},\ldots,u_{g(2^{k})}\}\mapsto u_{i}^{\B^{\prime}}\in\s_{2}^{\B}$ a bijection, we get that $\s^{\B^{\prime}}=\vars_{\s}(\wit(\phi))^{\B^{\prime}}$ and $\s_{2}^{\B^{\prime}}=\vars_{\s_{2}}(\wit(\phi))^{\B^{\prime}}$, proving $\wit$ is in fact a witness.
\end{proof}

\begin{lemma}
    The $\Sigma_{2}$-theory $\Televen$ is not strongly finitely witnessable with respect to $\{\s,\s_{2}\}$.
\end{lemma}

\begin{proof}
    Suppose that $\Televen$ has a strong witness: let $\phi$ be a $\Tten$-satisfiable quantifier-free formula, and $\A$ a $\Tten$-interpretation that satisfies $\phi$; since $\phi$ and $\Exists{\overarrow{x}}\wit(\phi)$, for $\overarrow{x}=\vars(\wit(\phi))\setminus\vars(\phi)$, are $\Televen$-equivalent, there is a $\Televen$-interpretation $\A^{\prime}$ that satisfies $\wit(\phi)$. Let $V=\vars_{\s_{2}}(\wit(\phi))$, take the arrangement $\delta_{V}$ induced by $\A^{\prime}$ on $V$, and since $\A^{\prime}$ satisfies $\wit(\phi)\wedge\delta_{V}$ there is a $\Televen$-interpretation $\B$ that satisfies $\wit(\phi)\wedge\delta_{V}$ with $\vars_{\s_{2}}(\wit(\phi)\wedge\delta_{V})^{\B}=\s_{2}^{\B}$: say $|\s_{2}^{\B}|=g(q)$. 
    
    Since $f$ is not computable there is a $p>q$ such that $g(p)=g(p-1)+2$; we then take fresh variables $U^{\prime}=\{u_{1},\ldots,u_{k}\}$, for $k=g(p-1)+1-g(q)$, define $U=V\cup U^{\prime}$, and let $\delta_{U}$ be the arrangement on $U$ that agrees with $\delta_{V}$ on $V$, and where each $u_{i}$ is in its own equivalence class. It is then easy to see that $\wit(\phi)\wedge\delta_{U}$ is $\Televen$-satisfiable, by say an interpretation with domain of sort $\s_{2}$ infinite: however no $\Televen$-interpretation $\C$ that satisfies $\wit(\phi)\wedge\delta_{U}$ can have $\s_{2}^{\C}=\vars_{\s_{2}}(\wit(\phi)\wedge\delta_{U})^{\C}$, as  $\vars_{\s_{2}}(\wit(\phi)\wedge\delta_{U})^{\C}$ must have $g(p-1)+1$ elements, leading to the desired contradiction.
\end{proof}

\begin{lemma}
    The $\Sigma_{2}$-theory $\Televen$ is convex with respect to $\{\s,\s_{2}\}$.
\end{lemma}

\begin{proof}
    This proof is similar to that of Lemma $66$ in \cite{arxivCADE}.
\end{proof}

\begin{lemma}
    The $\Sigma_{2}$-theory $\Televen$ is stably finite, and thus has the finite model property, with respect to $\{\s,\s_{2}\}$.
\end{lemma}

\begin{proof}
    Take: a quantifier-free formula $\phi$; a $\Televen$-interpretation $\A$ that satisfies $\phi$; $V=\vars_{\s_{2}}(\phi)$; the least positive integer $k$ such that 
    \[g(2^{k-1})=3\times 2^{k-2}<|V|\leq 3\times 2^{k-1}=g(2^{k});\]
    and a set $A$ of cardinality $g(2^{k})-|V^{\A}|$ disjoint from both $\s^{\A}$ and $\s_{2}^{\A}$. We can then define the interpretation $\B$ by making: $\s^{\B}=\s^{\A}$ (so $|\s^{\B}|=|\s^{\A}|=1$); $\s_{2}^{\B}=V^{\A}\cup A$ (so $|\s_{2}^{\B}|=g(2^{k})$, meaning $\B$ is a $\Televen$-interpretation); $x^{\B}=x^{\A}$ for all variables $x$ of sort $\s$; $u^{\B}=u^{\A}$ for all $u\in V$; and $u^{\B}$ an arbitrary value in $\s_{2}^{\B}$ for all other variables $u$ of sort $\s_{2}$. Then we have that $\B$ satisfies $\phi$, that $|\s^{\B}|,|\s_{2}^{\B}|<\aleph_{0}$, and that $|\s^{\B}|\leq|\s^{\A}|$ and $|\s_{2}^{\B}|\leq|\s_{2}^{\A}|$.
\end{proof}

\begin{lemma}
    The $\Sigma_{2}$-theory $\Televen$ does not have a computable minimal model function with respect to $\{\s,\s_{2}\}$.
\end{lemma}

\begin{proof}
    We proceed by contradiction, starting by assuming that $\minmod_{\Tten,S}$ is computable, for $S=\{\s,\s_{2}\}$: we henceforth define $h:\mathbb{N}\setminus\{0\}\rightarrow\mathbb{N}$ by making $h(1)=2$, and for $n\geq 2$
    \[h(n+1)=m,\quad\text{where}\quad\minmod_{\Tten,S}(\NNNEQ{u}{h(n)+1})=\{(1,m)\},\]
    which is computable. At the same time, we prove by induction that  $h(n)=g(n)$, giving us the contradiction that proves the result: indeed, this is obviously true for $n=0$ and $n=1$, so assume the results holds for a certain $n$. It is clear that a $\Televen$-interpretation $\A$ with $|\s_{2}^{\A}|=g(n+1)$ satisfies $\NNNEQ{u}{h(n)+1}=\NNNEQ{u}{g(n)+1}$ (after possibly changing the value assigned to the variables of that formula), as $g(n+1)\geq g(n)+1$; and no $\Televen$-interpretation $\B$ with $|\s_{2}^{\B}|<|\s_{2}^{\A}|$ can satisfy $\NNNEQ{u}{g(n)+1}$, as we have necessarily in this case $|\s_{2}^{\B}|\geq g(n)$, so the result holds.
\end{proof}

\subsection{\tp{$\Ttwelve$}{Ttwelve}}

$\Ttwelve$ is the $\Sigma_{s}$-theory with axiomatization
\[\{\psi_{=1}\vee(\psi_{\neq}\wedge\psi^{2}_{\vee}\wedge(\psi_{\geq 2g(n+1)}\vee\bigvee_{i=1}^{n}\psi_{=2g(i)})) : n\in\mathbb{N}\setminus\{0\}\}.\]
The models $\A$ of $\Ttwelve$ have either: $|\s^{\A}|=1$, and $s^{\A}$ the identity; $|\s^{\A}|\geq\aleph_{0}$, $s^{\A}(a)\neq a$ for all $a\in \s^{\A}$, but either $(s^{\A})^{4}(a)=a$ or $(s^{\A})^{4}(a)=(s^{\A})^{2}(a)$; or $|\s^{\A}|=2g(n)$, for some $n\in\mathbb{N}\setminus\{0\}$, and again $s^{\A}(a)\neq a$, for all $a\in \s^{\A}$, but either $(s^{\A})^{4}(a)=a$ or $(s^{\A})^{4}(a)=(s^{\A})^{2}(a)$. We must have, for finite non-trivial $\Ttwelve$-interpretations $\A$, $|\s^{\A}|=2g(n)$ instead of simply $|\s^{\A}|=g(n)$, in order to guarantee that both $\psi_{\neq}$ and $\psi^{2}_{\vee}$ can be simultaneously satisfied (by making, to give an example, for every pair $\{a,b\}$ of elements, $s^{\A}(a)=b$ and $s^{\A}(b)=a$).

\begin{lemma}
    The $\Sigma_{s}$-theory $\Ttwelve$ is not stably infinite, and thus not smooth, with respect to its only sort.
\end{lemma}

\begin{proof}
    Obvious, as the only $\Ttwelve$-interpretation $\A$ that satisfies $s(x)=x$ must have $|\s^{\A}|=1$.
\end{proof}

\begin{lemma}
    The $\Sigma_{s}$-theory $\Ttwelve$ is finitely witnessable with respect to its only sort.
\end{lemma}

\begin{proof}
    This proof is similar to that of Lemma 64 in \cite{arxivCADE}.
\end{proof}

\begin{lemma}
    The $\Sigma_{s}$-theory $\Ttwelve$ is not strongly finitely witnessable with respect to its only sort.
\end{lemma}

\begin{proof}
    Suppose $\Ttwelve$ has a strong witness $\wit$. We then take the $\Ttwelve$-satisfiable quantifier-free formula $\phi=\neg(x_{1}=x_{2})$ (satisfiable by all $\Ttwelve$-interpretations except the trivial one), and a $\Ttwelve$-interpretation $\A$ that satisfies $\phi$. From the definition of a strong witness we must have that $\A$ satisfies $\Exists{\overarrow{X}}\wit(\phi)$ (for $\overarrow{x}=\vars(\wit(\phi))\setminus\vars(\phi)$), as $\phi$ and $\Exists{\overarrow{x}}\wit(\phi)$ are $\Ttwelve$-equivalent, and thus there is a $\Ttwelve$-interpretation $\A^{\prime}$ (differing from $\A$ at most on the values assigned to the variables in $\overarrow{x}$) that satisfies $\wit(\phi)$. Let $V=\vars(\wit(\phi))$, and $\delta_{V}$ be the arrangement on $V$ induced by $\A^{\prime}$, and then we know that $\A^{\prime}$ satisfies $\wit(\phi)\wedge\delta_{V}$: again from the definition of a strong witness, there must exist a $\Ttwelve$-interpretation $\B$ that satisfies $\wit(\phi)\wedge\delta_{V}$ with $\s^{\B}=\vars(\wit(\phi)\wedge\delta_{V})^{\B}=V^{\B}$.

    As $\vars(\wit(\phi)\wedge\delta_{V})$ must be finite, $\s^{\B}$ must be finite; and since $\B$ satisfies $\wit(\phi)$, and thus $\Exists{\overarrow{x}}\wit(\phi)$ and $\phi=\neg(x_{1}=x_{2})$, $|\s^{\B}|$ is thus equal to $2g(p)$ for some $p\in\mathbb{N}\setminus\{0\}$: consider then a fresh variable $x$, define the set $U=V\cup\{x\}$, and take the arrangement $\delta_{U}$ obtained from $\delta_{V}$ by making $x$ different from all variables in $V$. We state that $\wit(\phi)\wedge\delta_{U}$ is $\Ttwelve$-satisfiable: indeed, consider the interpretation $\C$ with
    \[\s^{\C}=\begin{cases}\s^{\B}\cup \{a_{1},a_{2}\} & \text{if $g(p+1)=g(p)+1$,}\\
       \s^{\B}\cup\{a_{1},a_{2},a_{3},a_{4}\} & \text{if $g(p+1)=g(p)+2$,} 
    \end{cases}\]
    for elements $\{a_{1},a_{2},a_{3},a_{4}\}\notin \s^{\B}$, so that $|\s^{\C}|=2g(p+1)$; $s^{\C}(a)=s^{\B}(a)$ for all $a\in \s^{\B}$, $s^{\C}(a_{1})=a_{2}$, $s^{\C}(a_{2})=a_{1}$, and, if $a_{3},a_{4}\in\s^{\C}$, $s^{\C}(a_{3})=a_{4}$ and $s^{\C}(a_{4})=a_{3}$ (so that $\C$ satisfies both $\psi_{\neq}$ and $\psi^{2}_{\vee}$, making of $\C$ a $\Ttwelve$-interpretation); and $x^{\C}=a_{1}$, and $y^{\C}=y^{\B}$ for all other variables. This way $\C$ clearly satisfies $\wit(\phi)\wedge\delta_{U}$, meaning there should be a $\Ttwelve$-interpretation $\D$ that satisfies $\wit(\phi)\wedge\delta_{U}$ with $\s^{\D}=\vars(\wit(\phi)\wedge\delta_{U})^{\D}$: but this is impossible. Indeed, $|\vars(\wit(\phi)\wedge\delta_{U})^{\D}|=|U|=2g(p)+1$, which cannot equal either $1$ or $2g(q)$ for a $q\in\mathbb{N}\setminus\{0\}$.
\end{proof}

\begin{lemma}
    The $\Sigma_{s}$-theory $\Ttwelve$ is not convex with respect to its only sort.
\end{lemma}

\begin{proof}
    It is clear that 
    \[(y=s^{2}(x))\wedge(z=s^{2}(y))\vDash_{\Ttwelve}(x=z)\vee(y=z):\]
    this is built in the axioms through the formula $\psi^{2}_{\vee}$ for those interpretations with more than one element, and obvious for the trivial interpretation as there $s$ is interpreted as the identity. Let $\phi$ denote $(y=s^{2}(x))\wedge(z=s^{2}(y))$: we can have, however, neither $\phi\vDash_{\Ttwelve}x=z$ nor $\phi\vDash_{\Ttwelve}y=z$. Indeed, define the $\Ttwelve$-interpretations $\A$ and $\B$ with: $\s^{\A}=\s^{\B}=\{a_{0}, a_{1}, a_{2}, a_{3}\}$ (so $|\s^{\A}|=|\s^{\B}|=4=2g(2)$); 
    \[s^{\A}(a_{n})=\begin{cases}
        a_{n+1} & \text{if $n\in\{0,1,2\}$;}\\
        a_{1} & \text{if $n=3$;}\\
    \end{cases}
    \quad
    s^{\B}(a_{n})=\begin{cases}
        a_{n+1} & \text{if $n\in\{0,1,2\}$;}\\
        a_{2} & \text{if $n=3$;}\\
    \end{cases}\]
    and $x^{\A}=z^{\A}=x^{\B}=a_{0}$ and $y^{\A}=y^{\B}=z^{\B}=a_{2}$. This way, $\A$ falsifies $y=z$, while $\B$ falsifies $x=z$.
\end{proof}

\begin{lemma}
    The $\Sigma_{s}$-theory $\Ttwelve$ is stably finite, and thus has the finite model property, with respect to its only sort.
\end{lemma}

\begin{proof}
    Let $\phi$ be a quantifier-free formula, and $\A$ a $\Ttwelve$-interpretation that satisfies $\phi$: without loss of generality we may assume $|\s^{\A}|>1$, and thus that $\A$ satisfies $\psi_{\neq}$ and $\psi_{\vee}^{2}$. Notice that $s^{\mathbb{N}}(a)=\{a\}\cup\{(s^{\A})^{n}(a) : n\in\mathbb{N}\setminus\{0\}\}$, for $a\in\s^{\A}$, equals one of the cases found in \Cref{possible scenarios two} as a simple exhaustive search can show. Define then $V=\vars(\phi)$ and $A=\bigcup_{a\in V^{\A}}s^{\mathbb{N}}(a)$, which are necessarily finite sets (being true that $1<2|V|\leq |A|\leq 4|V|$): we then take an $n\in\mathbb{N}\setminus\{0\}$ such that $2g(n)>|A|+1$. We define an interpretation $\B$ by making:
\[\s^{\B}=\begin{cases}
    A\cup\{a_{1}, b_{1}, \ldots, a_{m},b_{m}\} & \text{if $2g(n)-|A|$ is even (and equals $2m$);}\\
    A\cup\{c_{1},c_{2},c_{3}, a_{1}, b_{1},\ldots,a_{m},b_{m}\} & \text{if $2g(n)-|A|$ is odd (and equals $2m+3$),}
\end{cases}\]
where, in the first case $m>1$, and in the second $m\geq 0$ (as $2g(n)>|A|+1$), and the $a_{i}$, $b_{j}$ and $c_{k}$ are elements not already in $\s^{\A}$ (meaning $|\s^{\B}|=2g(n)$); $s^{\B}(a)=s^{\A}(a)$ for all $a\in A$, $s^{\B}(a_{i})=b_{i}$, $s^{\B}(b_{i})=a_{i}$, $s^{\A}(c_{1})=c_{2}$, $s^{\B}(c_{2})=c_{3}$ and $s^{\B}(c_{3})=c_{2}$ (meaning $\B$ satisfies $\psi_{\neq}$ and $\psi^{2}_{\vee}$, and is therefore a $\Ttwelve$-interpretation); and $x^{\B}=x^{\A}$ for all variables $x\in V$, and arbitrary otherwise. This way $\B$ clearly satisfies $\phi$, and is of course finite.
\end{proof}

\begin{lemma}
    The $\Sigma_{s}$-theory $\Ttwelve$ does not have a computable minimal model function with respect to its only sort.
\end{lemma}

\begin{proof}
    Suppose $\Ttwelve$ has a computable minimal model function, and define $h:\mathbb{N}\setminus\{0\}\rightarrow\mathbb{N}$ by making $h(1)=4$ and 
    \[h(n+1)\in\minmod_{\Ttwelve,S}(\NNNEQ{x}{h(n)+1}),\]
    for $S=\{\s\}$ and $n\geq 2$ ($\NNNEQ{x}{h(n)+1}$ being $\Ttwelve$-satisfiable, given the theory has infinite models): this is, of course, a computable function, as $\minmod_{\Ttwelve,S}$ is computable. At the same time, it is easy to prove that $h(n)=2g(n)$ for all $n\in\mathbb{N}\setminus\{0\}$, meaning that $h$ is not computable, giving us the desired contradiction.
\end{proof}

\subsection{\tp{$\Tthirteen$}{Tthirteen}}

$\Tthirteen$ is the $\Sigma_{s}^{2}$-theory with axiomatization
\[\{(\psi^{\s}_{=1}\wedge\psi^{\s_{2}}_{=1})\vee(\psi_{\neq}\wedge(\psi^{\s}_{=2}\wedge\psi^{\s_{2}}_{\geq n})\vee(\psi^{\s}_{\geq g(n+1)}\wedge\psi^{\s_{2}}_{\geq g(n+1)})\vee\bigvee_{i=2}^{n}(\psi^{\s}_{=g(i)}\wedge\psi^{\s_{2}}_{=g(i)})) : n\in\mathbb{N}\setminus\{0,1\}\},\]
where $\psi_{\neq}=\Forall{x}\neg(s(x)=x)$. Essentially, $\Tthirteen$ can be seen as having four classes of models $\A$: the trivial one, with $|\s^{\A}|=|\s_{2}^{\A}|=1$, and $s^{\A}$ the identity; those with $|\s^{\A}|=2$ and $|\s_{2}^{\A}|\geq\aleph_{0}$, and $s^{\A}$ with no fixed points; those with $|\s^{\A}|=|\s_{2}^{\A}|=g(n)$, for some $n\in\mathbb{N}\setminus\{0,1\}$, and $s^{\A}$ with no fixed points; and finally, those with $|\s^{\A}|,|\s_{2}^{\A}|\geq\aleph_{0}$, and again $s^{\A}$ with no fixed points.

\begin{lemma}
    The $\Sigma_{s}^{2}$-theory $\Tthirteen$ is not stably infinite, and thus not smooth, with respect to $\{\s,\s_{2}\}$.
\end{lemma}

\begin{proof}
    Obvious, as $s(x)=x$ is satisfied by the trivial $\Tthirteen$-interpretation $\A$ with $|\s^{\A}|=|\s_{2}^{\A}|=1$, while all other $\Tthirteen$-interpretations satisfy $\psi_{\neq}$.
\end{proof}

\begin{lemma}
    The $\Sigma_{s}^{2}$-theory $\Tthirteen$ is finitely witnessable with respect to $\{\s,\s_{2}\}$.
\end{lemma}

\begin{proof}
    This proof is similar to that of Lemma $64$ in \cite{arxivCADE}.
\end{proof}

\begin{lemma}
    The $\Sigma_{s}^{2}$-theory $\Tthirteen$ is not strongly finitely witnessable with respect to $\{\s,\s_{2}\}$.
\end{lemma}

\begin{proof}
    Once the trivial interpretation is removed, it is clear that $\Tthirteen$ is stably infinite: according to Lemma 73 of \cite{arxivCADE}, it should then be smooth up to countable cardinals, which is not; for an example, no $\Tthirteen$-interpretation has $|\s^{\A}|=3$ and $|\s_{2}^{\A}|=4$.
\end{proof}

\begin{lemma}
    The $\Sigma_{s}^{2}$-theory $\Tthirteen$ is convex with respect to $\{\s,\s_{2}\}$.
\end{lemma}

\begin{proof}
    This proof is similar to that of Lemma $66$ in \cite{arxivCADE}.
\end{proof}

\begin{lemma}
    The $\Sigma_{s}^{2}$-theory $\Tthirteen$ has the finite model property with respect to $\{\s,\s_{2}\}$.
\end{lemma}

\begin{proof}
    Take a quantifier-free formula $\phi$, and a $\Tthirteen$-interpretation $\A$ that satisfies $\phi$: we may assume, without loss of generality, that $\A$ satisfies $\psi_{\neq}$. Take then the sets $V=\vars_{\s}(\phi)$ and $U=\vars_{\s_{2}}(\phi)$, and for each $x\in V$ define $M_{x}$ as the maximum of $j$ such that $s^{j}(x)$ occurs in $\phi$ (where we identify $x$ with $s^{0}(x)$): we then take the minimum $n\in\mathbb{N}\setminus\{0,1\}$ such that 
    \[g(n)\geq \max\{|U^{\A}|, |\Gamma^{\A}|\},\quad\text{where}\quad\Gamma=\{s^{j}(x) : x\in V, 0\leq j\leq M_{x}+1\}\]
    Take sets $A$ and $B$ with cardinalities, respectively, $g(n)-|\Gamma^{\A}|$ and $g(n)-|U^{\A}|$, disjoint from each other and the domains of $\A$, and we can then define an interpretation $\B$ by making: $\s^{\B}=\Gamma^{\A}\cup A$ and $\s_{2}^{\B}=U^{\A}\cup B$ (meaning both of these have $g(n)$ elements, and therefore are finite); $s^{\B}(a)=s^{\A}(a)$ for all those $a\in\Gamma^{\A}$ such that $s^{\A}(a)\in\Gamma^{\A}$, and $s^{\B}((s^{\A})^{M_{x}+1}(x^{\A}))=(s^{\A})^{M_{x}}(x^{\A})$; $s^{\B}(a)\in \Gamma^{\A}$ for all $a\in A$ (meaning $\B$ satisfies $\psi_{\neq}$, and is therefore a $\Tthirteen$-interpretation); $x^{\B}=x^{\A}$ for all $x\in V$, and arbitrary for all other variables of sort $\s$; $u^{\B}=u^{\A}$ for all $u\in U$, and arbitrary for all other variables of sort $\s_{2}$. As $\B$ satisfies $\phi$, we are done.
\end{proof}

\begin{lemma}
    The $\Sigma_{s}^{2}$-theory $\Tthirteen$ is not stably finite with respect to $\{\s,\s_{2}\}$.
\end{lemma}

\begin{proof}
    Obvious, as $\neg(s(x)=x)$ is satisfied by a $\Tthirteen$-interpretation $\A$ with $|\s^{\A}|=2$ and $|\s_{2}^{\A}|=\aleph_{0}$, but there are no $\Tthirteen$-interpretations $\B$ with both $\s^{\B}$ and $\s_{2}^{\B}$ finite, $|\s^{\B}|\leq|\s^{\A}|$ and $|\s_{2}^{\B}|\leq|\s_{2}^{\A}|$, and that satisfies $\neg(s(x)=x)$.
\end{proof}

\begin{lemma}
    The $\Sigma_{s}^{2}$-theory $\Tthirteen$ does not have a computable minimal model function with respect to $\{\s,\s_{2}\}$.
\end{lemma}

\begin{proof}
    Suppose, instead, that $\Tthirteen$ has a computable minimal model function: we define a function $h:\mathbb{N}\setminus\{0\}\rightarrow\mathbb{N}$ by making $h(1)=2$ and 
    \[h(n+1)=m,\quad \text{where}\quad \minmod_{\Tthirteen,S}(\NNNEQ{x}{h(n)+1})=\{(m,m)\},\]
    for $S=\{\s,\s_{2}\}$ and $n\geq 2$ (notice $\NNNEQ{x}{h(n)+1}$ is $\Tthirteen$-satisfiable); on one hand, this is a computable function as $\minmod_{\Tthirteen,S}$ is assumed computable and always has a finite output. At the other, one easily sees by induction that $h(n)=g(n)$ for all $n\in\mathbb{N}\setminus\{0\}$, meaning $h$ is not computable, and thus leading to a contradiction.
\end{proof}

\subsection{\tp{$\Tfourteen$}{Tfourteen}}

$\Tfourteen$ is the $\Sigma_{s}$-theory with axiomatization 
\[\{\psi_{=1}\vee(\psi_{\neq}\wedge(\psi_{\geq g(n+1)}\vee\bigvee_{i=1}^{n}\psi_{=g(i)})) : n\in\mathbb{N}\setminus\{0\}\},\]
where $\psi_{\neq}=\Forall{x}\neg(s(x)=x)$. $\Tfourteen$ can be seen as having three distinct groups of models $\A$: those with $|\s^{\A}|=1$, and $s^{\A}$ the identity; those with $|\s^{\A}|=g(n)$, for some $n\in\mathbb{N}\setminus\{0\}$, and $s^{\A}$ with no fixed points; and those with $|\s^{\A}|\geq \aleph_{0}$, and again $s^{\A}$ with no fixed points.

\begin{lemma}
    The $\Sigma_{s}$-theory $\Tfourteen$ is not stably infinite, and thus not smooth, with respect to its only sort.
\end{lemma}

\begin{proof}
    Obvious, as the trivial $\Tfourteen$-interpretation $\A$ with $|\s^{\A}|=1$ satisfies $s(x)=x$, while all other $\Tfourteen$-interpretations satisfy $\psi_{\neq}$.
\end{proof}

\begin{lemma}
    The $\Sigma_{s}$-theory $\Tfourteen$ is finitely witnessable with respect to its only sort.
\end{lemma}

\begin{proof}
    This proof is similar to that of Lemma $64$ in \cite{arxivCADE}.
\end{proof}

\begin{lemma}
    The $\Sigma_{s}$-theory $\Tfourteen$ is not strongly finitely witnessable with respect to its only sort.
\end{lemma}

\begin{proof}
    Suppose instead we have a strong witness $\wit$. Let $\phi$ be the quantifier-free formula $\neg(s(x)=x)$, which is $\Tfourteen$-satisfiable: take a $\Tfourteen$-interpretation $\A$ that satisfies $\phi$, and as $\phi$ and $\Exists{\overarrow{x}}\wit(\phi)$ are $\Tfourteen$-equivalent (where $\overarrow{x}=\vars(\wit(\phi))\setminus\vars(\phi)$), we have that there is an interpretation $\A^{\prime}$ (differing from $\A$ at most on the value assigned to those variables in $\overarrow{x}$) that satisfies $\wit(\phi)$. Let $V$ be the set of variables of $\wit(\phi)$, and consider the arrangement $\delta_{V}$ induced by $\A^{\prime}$ on $V$, so that $\A^{\prime}$ satisfies $\wit(\phi)\wedge\delta_{V}$: there must thus exist a $\Tfourteen$-interpretation $\B$ that satisfies $\wit(\phi)\wedge\delta_{V}$ with $\s^{\B}=\vars(\wit(\phi)\wedge\delta_{V})^{\B}$. Of course, as $\B$ satisfies $\wit(\phi)$, it also satisfies $\Exists{\overarrow{x}}\wit(\phi)$ and $\phi$, meaning that $|\s^{\B}|>1$, and thus there must be an $n\in\mathbb{N}\setminus\{0\}$ such that $|\s^{\B}|=g(n)$.

    Since $f$ is not computable, there is an $m>n$ such that $f(m)=1$, and thus $g(m)=g(m-1)+2$. Make $k=g(m)-g(n)-1$, and take a set $U^{\prime}$ of fresh variables $x_{1}$ through $x_{k}$, define $U=V\cup U^{\prime}$, and consider the arrangement $\delta_{U}$ which equals $\delta_{V}$ on $V$, and where each $x_{i}$ is in its own equivalence class. Now, $\wit(\phi)\wedge\delta_{U}$ is $\Tfourteen$-satisfiable. Indeed, consider the interpretation $\B^{\prime}$ with: $\s^{\B^{\prime}}=\s^{\B}\cup A$, where $A=\{a_{n} : n\in\mathbb{N}\}$ is disjoint from $\s^{\B}$ (meaning $|\s^{\B^{\prime}}|=\aleph_{0}$); $s^{\B^{\prime}}(a)=s^{\B}(a)$ for all $a\in \s^{\B}$, and $s^{\B^{\prime}}(a_{n})=a_{n+1}$ for all $a_{n}\in A$ (meaning $\B^{\prime}$ satisfies $\psi_{\neq}$, and is therefore a $\Tfourteen$-interpretation); and $x^{\B^{\prime}}=x^{\B}$ for all variables $x$ not in $U^{\prime}$, and $x_{i}^{\B^{\prime}}=a_{i}$ for each $x_{i}\in U^{\prime}$. Then $\B^{\prime}$ satisfies $\wit(\phi)\wedge\delta_{U}$, meaning there should exist a $\Tfourteen$-interpretation $\C$ that also satisfies that formula with $\s^{\C}=\vars(\wit(\phi)\wedge\delta_{U})^{\C}$. But that is impossible: indeed, that would force $|\s^{\C}|$ to be the number of equivalence classes on $\delta_{U}$, which is the number of equivalence classes on $\delta_{V}$ (that is, $g(n)$) plus $k=g(m)-g(n)-1$, adding up to a total of $g(m)-1$; but, by our choice of $m$, $g(m)-1\notin\{g(n) : n\in\mathbb{N}\setminus\{0\}\}$.
\end{proof}

\begin{lemma}
    The $\Sigma_{s}$-theory $\Tfourteen$ is convex with respect to its only sort.
\end{lemma}

\begin{proof}
    This proof is similar to that of Lemma $66$ in \cite{arxivCADE}.
\end{proof}

\begin{lemma}
    The $\Sigma_{s}$-theory $\Tfourteen$ is stably finite, and thus has the finite model property, with respect to its only sort.
\end{lemma}

\begin{proof}
    Let $\phi$ be a quantifier-free formula and $\A$ a $\Tfourteen$-interpretation that satisfies $\phi$: without loss of generality, we may assume that $\A$ satisfies $\psi_{\neq}$. Take the set $V=\vars_{\s}(\phi)$, and for each $x\in V$ define $M_{x}$ to be the maximum of $j$ such that $s^{j}(x)$ appears in $\phi$, where we identify $s^{0}(x)$ with $x$ for simplicity: we then define the set of terms $\Gamma=\{s^{j}(x) : x\in V, 0\leq j\leq M_{x}+1\}$, and take the minimum $n$ of $\mathbb{N}\setminus\{0\}$ such that $g(n)\geq |\Gamma^{\A}|$; finally, take a set $A$ of cardinality $g(n)-|\Gamma^{\A}|$ disjoint from $\s^{\A}$. 
    
    We can define an interpretation $\B$ by making: $\s^{\B}=\Gamma^{\A}\cup A$ (which has then $g(n)$ elements); $s^{\B}(a)=s^{\A}(a)$ for all $a\in \Gamma^{\A}$ such that $s^{\A}(a)\in \Gamma^{\A}$, and for those $a=(S^{\A})^{M_{x}+1}(x^{\A})$ in $\Gamma^{\A}$ such that $s^{\A}(a)$ is not in that set, $s^{\B}(a)=(s^{\A})^{M_{x}}(x^{\A})$; $s^{\B}(a)\in \Gamma^{\A}$ for all $a\in A$ (meaning that $\B$ satisfies $\psi_{\neq}$, and thus that it is a $\Tfourteen$-interpretation); and finally $x^{\B}=x^{\A}$ for all $x\in V$, and arbitrary otherwise. This way $\B$ is a finite $\Tfourteen$-interpretation that satisfies $\phi$.
\end{proof}

\begin{lemma}
    The $\Sigma_{s}$-theory $\Tfourteen$ does not have a computable minimal model function with respect to its only sort.
\end{lemma}

\begin{proof}
    We proceed by contradiction, so suppose that $\minmod_{\Tfourteen,S}$ is computable, where $S=\{\s\}$. We can thus define $h:\mathbb{N}\setminus\{0\}\rightarrow\mathbb{N}$ by: $h(1)=2$, and, for all $n\geq 2$,
    \[h(n+1)=m,\quad\text{where}\quad\minmod_{\Tfourteen,S}(\NNNEQ{x}{h(n)+1})=\{m\}.\]
    $h$ is computable as $\minmod_{\Tfourteen,S}$ is assumed computable; at the same time, a simple inductive argument shows that $h(n)=g(n)$ for all $n\in\mathbb{N}\setminus\{0\}$, implying $h$ is not computable.
\end{proof}

\subsection{\tp{$\Tfifteen$}{Tfifteen}}\label{Tfifteen}

$\Tfifteen$ is the $\Sigma_{s}$-theory with axiomatization
\[\{(\psi^{\neq}_{\geq 2n}\rightarrow\psi_{\geq \bb(n+1)})\wedge(\psi_{\geq n}\vee\psi^{2}_{=}) : n\in\mathbb{N}\},\]
where
\[\psi^{\neq}_{\geq n}=\Exists{x_{1},\ldots,x_{n}}\bigwedge_{i=1}^{n-1}\bigwedge_{j=i+1}^{n}\neg(x_{i}=x_{j})\wedge\bigwedge_{i=1}^{n}\neg(s(x_{i})=x_{i})\]
and $\psi^{2}_{=}=\Forall{x}(s^{2}(x)=x)$. The models $\A$ of $\Tfifteen$ can be divided into two classes: those with $\s^{\A}$ finite, where $(s^{\A})^{2}(a)=a$ for all $a\in \s^{\A}$ and, if there are at least $2n$ elements $a$ such that $s^{\A}(a)\neq a$, then $|\s^{\A}|\geq\bb(n+1)$ (notice that $\bb(n+1)\geq 2n$ for all $n\in\mathbb{N}$); and those with $\s^{\A}$ infinite. Now, notice that while $\psi^{\neq}_{\geq 0}\rightarrow\psi_{\geq\bb(1)}$ does not tell us anything, $\psi^{\neq}_{\geq 2}\rightarrow\psi_{\geq\bb(2)}$ guarantees that any $\Tfifteen$-interpretation $\A$ where $s^{\A}$ is not the identity will have at least $\bb(2)=4$ elements in total, and so on; notice too that the fact finite models of $\Tfifteen$ satisfy $\psi^{2}_{=}$ guarantees that elements $a$ such that $s^{\A}(a)\neq a$ come in pairs, justifying the $2n$ on the formula $\psi^{\neq}_{\geq 2n}\rightarrow\psi_{\geq\bb(n+1)}$.

\begin{lemma}\label{Tfifteen is SM}
    The $\Sigma_{s}$-theory $\Tfifteen$ is smooth, and thus stably infinite, with respect to its only sort.
\end{lemma}

\begin{proof}
    Take a quantifier-free formula $\phi$, a $\Tfifteen$-interpretation $\A$ that satisfies $\phi$, and a cardinal $\kappa>|\s^{\A}|$. Take a set $A$ disjoint from $\s^{\A}$ with cardinality $\kappa$ if $\kappa\geq\aleph_{0}$, and $\kappa-|\s^{\A}|$ if $\kappa<\aleph_{0}$, and we define an interpretation $\B$ by making: $\s^{\B}=\s^{\A}\cup A$ (so $|\s^{\B}|=\kappa$); $s^{\B}(a)=s^{\A}(a)$ for every $a\in \s^{\A}$, and $s^{\B}(a)=a$ for every $a\in A$ (so if $\A$ satisfies $\psi^{2}_{=}$ so does $\B$, and otherwise $\B$ is automatically infinite); and $x^{\B}=x^{\A}$ for all variables $x$ (meaning $\B$ satisfies $\phi$). Finally, since $\B$ has as many elements satisfying $s^{\B}(a)\neq a$ as $\A$, we have that $\B$ satisfies $\psi^{\neq}_{\geq 2n}\rightarrow\psi_{\geq\bb(n+1)}$ for each $n\in\mathbb{N}$, and is therefore a $\Tfifteen$-interpretation.
\end{proof}

\begin{lemma}\label{Tfifteen is not FW}
    The $\Sigma_{s}$-theory $\Tfifteen$ is not finitely witnessable, and thus not strongly finitely witnessable, with respect to its only sort.
\end{lemma}

\begin{proof}
    Suppose we have a witness $\wit$: we define a function $h:\mathbb{N}\rightarrow\mathbb{N}$ by making $h(0)=0$, $h(1)=1$, and for $n>1$
    \[h(n+1)=|\vars(\wit(\bigwedge_{i=1}^{2n+1}\bigwedge_{j=i+1}^{2n+2}\neg(x_{i}=x_{j})\wedge\bigwedge_{i=1}^{2n+2}\neg(s(x_{i})=x_{i})))|,\]
    where we may call the formula being witnessed by $\neq(x_{1},\ldots,x_{2(n+1)})$ for simplicity. Now, $h$ is computable, since constructing $\neq(x_{1},\ldots,x_{2(n+1)})$ once $h(n)$ is know can be done algorithmically, as well as finding its witness (as $\wit$ is assumed computable), the set of variables on the witness, and the cardinality of this set; at the same time, we prove that $h(n)\geq\bb(n)$ for all $n\in\mathbb{N}$, meaning $h$ is not computable and leading to a contradiction.

    The statement is obvious for $n=0$ and $n=1$. For $n\geq 2$, as $\neq(x_{1},\ldots,x_{2(n+1)})$ is $\Tfifteen$-satisfiable (by, say, an infinite model of the theory), there exists a $\Tfifteen$-interpretation $\A$ that satisfies $\wit(\neq(x_{1},\ldots,x_{2(n+1)}))$, and thus $\neq(x_{1},\ldots,x_{2(n+1)})$, with $\s^{\A}=\vars(\wit(\neq(x_{1},\ldots,x_{2(n+1)})))^{\A}$. But, as $\A$ satisfies $\neq(x_{1},\ldots,x_{2(n+1)})$, by hypothesis we have $|\s^{\A}|\geq\bb(n+1)$, and thus 
    \[h(n+1)=|\vars(\wit(\neq(x_{1},\ldots,x_{2(n+1)})))|\geq|\vars(\wit(\neq(x_{1},\ldots,x_{2(n+1)})))^{\A}|=|\s^{\A}|\geq\bb(n+1).\]
\end{proof}

\begin{lemma}
    The $\Sigma_{s}$-theory $\Tfifteen$ is convex with respect to its only sort.
\end{lemma}

\begin{proof}
    This proof is similar to that of Lemma $66$ in \cite{arxivCADE}.
\end{proof}

\begin{lemma}\label{Tfifteen is not FMP}
    The $\Sigma_{s}$-theory $\Tfifteen$ does not have the finite model property, and thus is not stably finite, with respect to its only sort.
\end{lemma}

\begin{proof}
    Consider the quantifier-free formula $\phi$ given by $\neg(s^{2}(x)=x)$, which is not satisfied by any finite $\Tfifteen$-interpretations. We define the $\Tfifteen$-interpretation $\A$ with $\s^{\A}=\{a_{n} : n\in\mathbb{N}\}$ and such that $s^{\A}(a_{n})=a_{n+1}$: making $x^{\A}=a_{0}$, it is of course true that $\A$ satisfies $\phi$, proving the result.
\end{proof}

\begin{lemma}\label{Tfifteen is not CMMF}
    The $\Sigma_{s}$-theory $\Tfifteen$ does not have a computable minimal model function with respect to its only sort.
\end{lemma}

\begin{proof}
    The proof is by contradiction, so let us assume that $\minmod_{\Tfifteen,S}$ is computable, where $S=\{\s\}$: we then define a function $h:\mathbb{N}\rightarrow\mathbb{N}$ by: $h(0)=0$, $h(1)=1$, and, for all $n\geq 2$, 
    \[h(n)\in\minmod_{\Tfifteen,S}\big(\neq(x_{1},\ldots,x_{2n})\wedge\bigwedge_{i=1}^{2n}\neg(s(x_{i})=x_{i})\big).\]
    It is obvious that $h(n)=\bb(n)$ for all $n\in\mathbb{N}$, leading to a contradiction: while $h$ only uses supposedly computable functions in its definition, it is still not computable.
\end{proof}

\subsection{\tp{$\Tsixteen$}{Tsixteen}}

$\Tsixteen$ is the $\Sigma_{s}$-theory with axiomatization
\[\{\psi_{\vee}\wedge(\psi^{\neq}_{\geq 2n}\rightarrow\psi_{\geq \bb(n+1)})\wedge(\psi_{\geq n}\vee\psi^{2}_{=}) : n\in\mathbb{N}\},\]
where the definition of $\psi^{\neq}_{\geq n}$ can be found in \Cref{Tfifteen},
and $\psi^{2}_{=}=\Forall{x}(s^{2}(x)=x)$. The models $\A$ of $\Tsixteen$ belong to two distinct classes: those with $\s^{\A}$ finite, where $(s^{\A})^{2}(a)=a$ for all $a\in \s^{\A}$ and, if there are at least $2n$ elements $a$ such that $s^{\A}(a)\neq a$, then $|\s^{\A}|\geq\bb(n+1)$; and those with $\s^{\A}$ infinite, where $(s^{\A})^{2}(a)$ equals either $a$ or $s^{\A}(a)$ for each $a\in\s^{\A}$. Notice that, in a way, $\Tsixteen$ can be seen as the result of applying the operator $\addnc{\cdot}$ to $\Tfifteen$: while this cannot be formally performed as $\Tfifteen$ is defined on a non-empty signature, the proofs of the various lemmas (except the one related to convexity) for $\Tsixteen$ are essentially copy-pasted from $\Tfifteen$.

\begin{lemma}
    The $\Sigma_{s}$-theory $\Tsixteen$ is smooth, and thus stably infinite, with respect to its only sort.
\end{lemma}

\begin{proof}
Similar to that of \Cref{Tfifteen is SM}.
\end{proof}

\begin{lemma}
    The $\Sigma_{s}$-theory $\Tsixteen$ is not finitely witnessable, and thus not strongly finitely witnessable, with respect to its only sort.
\end{lemma}

\begin{proof}
Similar to that of \Cref{Tfifteen is not FW}.
\end{proof}

\begin{lemma}
    The $\Sigma_{s}$-theory $\Tsixteen$ is convex with respect to its only sort.
\end{lemma}

\begin{proof}
    It is obvious that $\phi\vDash_{\Tsixteen}(x=z)\vee(y=z)$, where $\phi=(y=s(x))\wedge(z=s(y))$,
    as every model of $\Tsixteen$ satisfies $\psi_{\vee}$. Take then the $\Tsixteen$-interpretations $\A$ and $\B$ with: $\s^{\A}=\s^{\B}=\{a_{n}, b_{n}: n\in\mathbb{N}\}$; $s^{\A}(a_{n})=s^{\B}(a_{n})=b_{n}$, $s^{\A}(b_{n})=a_{n}$ and $s^{\B}(b_{n})=b_{n}$, for all $n\in\mathbb{N}$ (meaning $\A$ and $\B$ satisfy $\psi_{\vee}$ and, being both infinite, are models of $\Tsixteen$); $x^{\A}=z^{\A}=x^{\B}=a_{0}$, and $y^{\A}=y^{\B}=z^{\B}=b_{0}$. This way, $\A$ falsifies $y=z$, while $\B$ falsifies $x=z$, meaning that neither $\phi\vDash_{\Tsixteen}x=z$ nor $\phi\vDash_{\Tsixteen}y=z$.
\end{proof}

\begin{lemma}
    The $\Sigma_{s}$-theory $\Tsixteen$ does not have the finite model property, and thus is not stably finite, with respect to its only sort.
\end{lemma}

\begin{proof}
Similar to that of \Cref{Tfifteen is not FMP}.
\end{proof}

\begin{lemma}
    The $\Sigma_{s}$-theory $\Tsixteen$ does not have a computable minimal model function with respect to its only sort.
\end{lemma}

\begin{proof}
Similar to that of \Cref{Tfifteen is not CMMF}.
\end{proof}

\subsection{\tp{$\Tseventeen$}{Tseventeen}}

$\Tseventeen$ is the $\Sigma_{s}^{2}$-theory with axiomatization
\[\{(\psi^{\s}_{=1}\rightarrow\psi^{\s_{2}}_{\geq n})\wedge(\psi^{\neq}_{\geq n}\rightarrow\psi^{\s}_{\geq g(n)}) : n\in\mathbb{N}\setminus\{0\}\},\]
where the definition of $\psi^{\neq}_{\geq n}$ can be found in \Cref{Tfifteen},
and $x_{1}$ through $x_{n}$ are variables of sort $\s$. The models $\A$ of $\Tseventeen$ may be more or less divided in three classes: those $\A$ with $|\s^{\A}|=1$ (where necessarily $s^{\A}$ is the identity) and $|\s_{2}^{\A}|\geq\aleph_{0}$; those with $|\s^{\A}|\geq\aleph_{0}$, where $s^{\A}$ can be any function, and $\s_{2}^{\B}$ can have any cardinality; and those with $1<|\s^{\A}|<\aleph_{0}$, where $\s_{2}^{\A}$ can have any cardinality, and if there are at least $n$ elements $a\in\s^{\A}$ such that $s^{\A}(a)\neq a$, then $|\s^{\A}|\geq g(n)$.

\begin{lemma}\label{Tseventeen is SM}
    The $\Sigma_{s}^{2}$-theory $\Tseventeen$ is smooth, and thus stably infinite, with respect to $\{\s,\s_{2}\}$.
\end{lemma}

\begin{proof}
    Take a quantifier-free formula $\phi$, a $\Tseventeen$-interpretation $\A$ that satisfies $\phi$, and cardinals $\kappa\geq|\s^{\A}|$ and $\kappa_{2}\geq|\s_{2}^{\A}|$. In the next step, take sets $A$ and $B$, both disjoint from the domains of $\A$, with respectively: $\kappa$ elements if $\kappa\geq\aleph_{0}$, and $\kappa-|\s^{\A}|$ elements if $\kappa<\aleph_{0}$; and, analogously, $\kappa_{2}$ elements if $\kappa_{2}\geq\aleph_{0}$, and $\kappa_{2}-|\s_{2}^{\A}|$ elements if $\kappa_{2}<\aleph_{0}$. We then define an interpretation $\B$ by making: $\s^{\B}=\s^{\A}\cup A$, and $\s_{2}^{\B}=\s_{2}^{\A}\cup B$ (so, if $|\s^{\B}|=1$, $|\s_{2}^{\B}|\geq|\s_{2}^{\A}|\geq\aleph_{0}$); $s^{\B}(a)=s^{\A}(a)$ for all $a\in \s^{\A}$, and $s^{\B}(a)=a$ for all $a\in A$ (so $\B$ has at most as many elements satisfying $s^{\B}(a)\neq a$ as $\A$, and since the latter interpretation satisfies $\psi^{\neq}_{\geq n}\rightarrow\psi^{\s}_{\geq g(n)}$ for every $n\in\mathbb{N}\setminus\{0\}$, so does the former, meaning $\B$ is a $\Tseventeen$-interpretation); and $x^{\B}=x^{\A}$ for all variables $x$ of sort $\s$, and $u^{\B}=u^{\A}$ for all variables $u$ of sort $\s_{2}$, meaning that $\B$ satisfies $\phi$. Of course $|\s^{\B}|=\kappa$ and $|\s_{2}^{\B}|=\kappa_{2}$, so the proof is done.
\end{proof}

\begin{lemma}\label{Tseventeen is FW}
    The $\Sigma_{s}^{2}$-theory $\Tseventeen$ is finitely witnessable with respect to $\{\s,\s_{2}\}$.
\end{lemma}

\begin{proof}
    This proof is similar to that of Lemma $55$ in \cite{arxivCADE}.
\end{proof}

\begin{lemma}\label{Tseventeen is not SFW}
    The $\Sigma_{s}^{2}$-theory $\Tseventeen$ is not strongly finitely witnessable with respect to $\{\s,\s_{2}\}$.
\end{lemma}

\begin{proof}
    This proof is similar to that of Lemma $56$ in \cite{arxivCADE}.
\end{proof}

\begin{lemma}\label{Tseventeen is CV}
    The $\Sigma_{s}^{2}$-theory $\Tseventeen$ is convex with respect to $\{\s,\s_{2}\}$.
\end{lemma}

\begin{proof}
    This proof is similar to that of Lemma $66$ in \cite{arxivCADE}.
\end{proof}

\begin{lemma}\label{Tseventeen is FMP}
    The $\Sigma_{s}^{2}$-theory $\Tseventeen$ has the finite model property with respect to $\{\s,\s_{2}\}$.
\end{lemma}

\begin{proof}
    Let $\phi$ be a quantifier-free formula, and $\A$ a $\Tseventeen$-interpretation that satisfies $\phi$. Let $V$ be the set of variables in $\phi$ of sort $\s$ (and analogously for $V_{2}$ and $\s_{2}$), and for each $x\in V$ define $M_{x}$ as the maximum of $j$ such that $s^{j}$ occurs in $\phi$: consider next the set of terms 
    \[\Gamma=\{s^{j}(x) : x\in V, 0\leq j\leq M_{x}+1\},\]
    where $s^{0}(x)$ is to be understood simply as $x$. Let $m$ be the number of elements $a$ in $\Gamma^{\A}$ such that $s^{\A}(a)\in\Gamma^{\A}$ and $s^{\A}(a)\neq a$, $A$ be a set of cardinality $g(m+1)-|\Gamma^{\A}|$, and $b$ be an element neither in $\A$ nor in $A$. Finally, we define an interpretation $\B$ by making: $\s^{\B}=\Gamma^{\A}\cup A$ and $\s_{2}^{\B}=V_{2}^{\A}\cup\{b\}$ (meaning $|\s^{\B}|=g(m+1)$, for some $m\in\mathbb{N}$, and $\s_{2}^{\B}$ is not empty); $s^{\B}(a)=s^{\A}(a)$ for all $a\in\Gamma^{\A}$ such that $s^{\A}(a)\in\Gamma^{\A}$, and for all other elements $a$ of $\s^{\A}$, $s^{\A}(a)=a$ (meaning that $\A$ has $m$ elements $a$ satisfying $s^{\B}(a)\neq a$, and therefore $\A$ is a $\Tseventeen$-interpretation); $x^{\B}=x^{\A}$ for all $x\in V$, and arbitrarily otherwise; and $u^{\B}=u^{\A}$ for all $u\in V_{2}$, and arbitrarily otherwise. This way $\B$ has both domains finite, while clearly satisfying $\phi$.
\end{proof}

\begin{lemma}\label{Tseventeen is not SF}
    The $\Sigma_{s}^{2}$-theory $\Tseventeen$ is not stably finite with respect to $\{\s,\s_{2}\}$.
\end{lemma}

\begin{proof}
    Obvious, as $\Tseventeen$ possesses a model $\A$ with $|\s^{\A}|=1$ and $|\s_{2}^{\A}|=\aleph_{0}$, but no models $\B$ with $|\s^{\B}|=1$ and $\s_{2}^{\B}$ finite.
\end{proof}

\begin{lemma}\label{Tseventeen is not CMMF}
    The $\Sigma_{s}^{2}$-theory $\Tseventeen$ does not have a computable minimal model function with respect to $\{\s,\s_{2}\}$.
\end{lemma}

\begin{proof}
    Suppose that $\Tseventeen$ has a computable minimal model function, so that we can define a function $h:\mathbb{N}\setminus\{0\}\rightarrow\mathbb{N}$ by making, for $S=\{\s,\s_{2}\}$,
    \[h(n)=\min\{p : (p,q)\in\minmod_{\Tseventeen,S}(\neq(x_{1},\ldots,x_{n}))\}.\]
    $h$ is computable as $\minmod_{\Tseventeen,S}(\neq(x_{1},\ldots,x_{n}))$ can be found algorithmically and is finite (in fact, it is a singleton). At the same time, it is easy to see that $h(n)=g(n)$ for all $n\in\mathbb{N}\setminus\{0\}$, leading to a contradiction.
\end{proof}

\subsection{\tp{$\Teighteen$}{Teighteen}}

$\Tseventeen$ is the $\Sigma_{s}^{2}$-theory with axiomatization
\[\{\psi_{\vee}\wedge(\psi^{\s}_{=1}\rightarrow\psi^{\s_{2}}_{\geq n})\wedge(\psi^{\neq}_{\geq n}\rightarrow\psi^{\s}_{\geq g(n)}) : n\in\mathbb{N}\setminus\{0\}\},\]
where the definition of $\psi^{\neq}_{\geq n}$ can be found in \Cref{Tfifteen},
and $x_{1}$ through $x_{n}$ are variables of sort $\s$. We can see the models $\A$ of $\Teighteen$ as belonging to three main classes: those $\A$ with $|\s^{\A}|=1$ (and $s^{\A}$ the identity) and $|\s_{2}^{\A}|\geq\aleph_{0}$; those with $|\s^{\A}|\geq\aleph_{0}$, $(s^{\A})^{2}(a)$ equal to either $a$ or $s^{\A}(a)$ for all $a\in\s^{\A}$, and where $\s_{2}^{\B}$ can have any cardinality; and those with $1<|\s^{\A}|<\aleph_{0}$, $|\s_{2}^{\A}|$ any cardinality, where $(s^{\A})^{2}(a)$ equals either $a$ or $s^{\A}(a)$ for all $a\in\s^{\A}$, and if there are at least $n$ elements $a\in\s^{\A}$ such that $s^{\A}(a)\neq a$, then $|\s^{\A}|\geq g(n)$. Notice that, once more, $\Teighteen$ could be seen as the result of a generalized operator $\addnc{\cdot}$ being applied to the theory $\Tseventeen$: most proofs are then very similar, so we choose to omit them.

\begin{lemma}
    The $\Sigma_{s}^{2}$-theory $\Teighteen$ is smooth, and thus stably infinite, with respect to $\{\s,\s_{2}\}$.
\end{lemma}

\begin{proof}
   This proof is similar to that of \Cref{Tseventeen is SM}.
\end{proof}

\begin{lemma}
    The $\Sigma_{s}^{2}$-theory $\Teighteen$ is finitely witnessable with respect to $\{\s,\s_{2}\}$.
\end{lemma}

\begin{proof}
    This proof is similar to that of \Cref{Tseventeen is FW}.
\end{proof}

\begin{lemma}
    The $\Sigma_{s}^{2}$-theory $\Teighteen$ is not strongly finitely witnessable with respect to $\{\s,\s_{2}\}$.
\end{lemma}

\begin{proof}
    This proof is similar to that of \Cref{Tseventeen is not SFW}.
\end{proof}

\begin{lemma}
    The $\Sigma_{s}^{2}$-theory $\Teighteen$ is convex with respect to $\{\s,\s_{2}\}$.
\end{lemma}

\begin{proof}
    As $\Teighteen$ satisfies $\psi_{\vee}$, we get that, for $\phi$ the cube $(y=s(x))\wedge(z=s(y))$,
    \[
    \phi\vDash_{\Teighteen}(x=z)\vee(y=z).\]
    We then take the $\Teighteen$-interpretations $\A$ and $\B$ given by: $\s^{\A}=\s^{\B}=\{a_{0},a_{1},a_{2}\}$, and $\s_{2}^{\A}=\s_{2}^{\B}=\{b\}$ (so that $|\s^{\A}|=|\s^{\B}|=3=g(2)$); $s^{\A}(a_{0})=s^{\B}(a_{0})=a_{1}$, $s^{\A}(a_{1})=a_{0}$, $s^{\B}(a_{1})=a_{1}$ and $s^{\A}(a_{2})=s^{\B}(a_{2})=a_{2}$ (so $\A$ satisfies $\psi^{\neq}_{\geq 1}$ and $\psi^{\neq}_{\geq 2}$ but not $\psi^{\neq}_{\geq 3}$, and $\B$ satisfies $\psi^{\neq}_{\geq 1}$ but not $\psi^{\neq}_{\geq 2}$, meaning both are $\Teighteen$-interpretations); $x^{\A}=z^{\A}=x^{\B}=a_{0}$, and $y^{\A}=y^{\B}=z^{\B}=a_{1}$. So $\A$ does not satisfy $y=z$, and $\B$ does not satisfy $x=z$, meaning that neither $\phi\vDash_{\Teighteen}x=z$ nor $\phi\vDash_{\Teighteen}y=z$.
\end{proof}

\begin{lemma}
    The $\Sigma_{s}^{2}$-theory $\Teighteen$ has the finite model property with respect to $\{\s,\s_{2}\}$.
\end{lemma}

\begin{proof}
    This proof is similar to that of \Cref{Tseventeen is FMP}.
\end{proof}

\begin{lemma}
    The $\Sigma_{s}^{2}$-theory $\Teighteen$ is not stably finite with respect to $\{\s,\s_{2}\}$.
\end{lemma}

\begin{proof}
    This proof is similar to that of \Cref{Tseventeen is not SF}.
\end{proof}

\begin{lemma}
    The $\Sigma_{s}^{2}$-theory $\Teighteen$ does not have a computable minimal model function with respect to $\{\s,\s_{2}\}$.
\end{lemma}

\begin{proof}
    This proof is similar to that of \Cref{Tseventeen is not CMMF}.
\end{proof}

\section{Old Theories}
\label{sec:oldtheoriesapp}

In this section we consider the computability of a minimal model function
for theories that were defined in~\cite{CADE,FroCoS}.
Specifically, the axiomatizations of all the theories of this 
section can be found in Section 4 of \cite{FroCoS}.

\subsection{\tp{$\Teven$, $\Toneodd$}{Teven, Toneodd}}

\begin{lemma}\label{Teven}
    The $\Sigma_{1}$-theory $\Teven$ and the $\Sigma_{2}$-theory $\Toneodd$ have a computable minimal model function with respect to, respectively, $\{\s\}$ and $\{\s,\s_{2}\}$.
\end{lemma}

\begin{proof}
    Take: a quantifier-free $\Sigma_{1}$-formula $\phi$; its set of variables $V$; the set $\eq{\phi}$ of equivalences on $V$ such that $\phi$ and $\delta_{V}^{E}$ are equivalent, a computable set given that it can be reduced to the satisfiability problem in equality logic; $m=\min\{|V/E| : E\in\eq{\phi}\}$, again computable (and equal to $\aleph_{0}$ iff $\eq{\phi}$ is empty, what happens iff $\phi$ is not satisfiable); and, finally, we state that
    \[\minmod_{\Teven,\{\s\}}(\phi)=\begin{cases}
        \{m\} & \text{if $m$ is even;}\\
        \{m+1\} & \text{if $m$ is odd;}\\
        \emptyset & \text{if $m=\aleph_{0}$}
    \end{cases}\]
    is a computable minimal model function, being certainly at least computable and with a finite output. Now assume $\phi$ is satisfiable, and thus $m\in\mathbb{N}$, take an $E\in\eq{\phi}$ such that $m=|E/V|$ and define an interpretation $\A$ whose domain has $m$ elements if that number is even, and $m+1$ elements if it is odd (so $\A$ is a $\Teven$-interpretation given it has an even number of elements), and where $x^{\A}=y^{\A}$ iff $xEy$ for any variables $x$ and $y$ in $V$ (and arbitrarily otherwise). Of course $\A$ satisfies $\delta_{V}^{E}$, and thus $\phi$.

    Now, suppose that $\B$ is another $\Teven$-interpretation that satisfies $\phi$ with $|\s^{\B}|\neq |\s^{\A}|$. Then $\B$ must satisfy $\delta_{V}^{F}$ for some $F\in\eq{\phi}$, and by definition of $E$ we have $|\s^{\B}|\geq |V/F|\geq |V/E|=m$. If $|\s^{\A}|=m$ we are done; otherwise, $|\s^{\A}|=m+1$, which is even, but since $|\s^{\B}|$ is also even and $|\s^{\B}|\geq m$ we get $|\s^{\B}|>|\s^{\A}|$, what finishes the proof for $\Teven$. The proof for $\Toneodd$ is essentially the same, changing even numbers for odd ones, and setting the cardinalities for the other sort to $1$.
\end{proof}
    
\subsection{\tp{$\Tninfty$}{Tninfty}}

\begin{lemma}
    The $\Sigma_{1}$-theory $\Tninfty$ has a computable minimal model function with respect to its only sort.
\end{lemma}

\begin{proof}
    This proof is similar to that of \Cref{Teven}. Take again: a quantifier-free $\Sigma_{1}$-formula $\phi$; its variables $V$; the set $\eq{\phi}$ of equivalences $E$ on $V$ such that $\phi$ and $\delta_{V}^{E}$ are equivalent (since this can be reduced to the satisfiability problem of equality logic, we have an algorithmic way to find $\eq{\phi}$); $m=\min\{|V/E| : E\in\eq{\phi}\}$, for which there is also an algorithm; and then we state that
    \[\minmod_{\Tninfty,\{\s\}}(\phi)=\begin{cases}
        \{n\} & \text{if $m\leq n$;}\\
        \{\aleph_{0}\} & \text{if $m>n$}
    \end{cases}\]
    is a computable minimal model function, being clearly computable. Now assume $\phi$ is satisfiable, so that $m$ is actually in $\mathbb{N}$, and to prove that we indeed have a minimal model function, take an interpretation $\A$ whose domains possesses $n$ elements if $m\leq n$, and $\aleph_{0}$ otherwise (clearly making of it a $\Tninfty$-interpretation), and where $x^{\A}=y^{\A}$ iff $xEy$ (for an $E\in\eq{\phi}$ such that $|V/E|=m$) for variables $x,y\in V$, and arbitrarily otherwise. Since $\A$ satisfies $\delta_{V}^{E}$, it also satisfies $\phi$.

    Now, suppose $\B$ is another $\Tninfty$-interpretation that satisfies $\phi$ with $|\s^{\B}|\neq|\s^{\A}|$: let $F$ be the equivalence induced by $\B$ on $V$, and we have that $|\s^{\B}|\geq |V/F|\geq |V/E|=m$. If $|\s^{\A}|=n$ we have that, since $\B$ is a $\Tninfty$, that $|\s^{\B}|\geq n$ and thus $|\s^{\B}|>|\s^{\A}|$; otherwise $|\s^{\A}|=\aleph_{0}$. But in the latter case, since $|\s^{\B}|\geq m>n$ and $\B$ is a $\Tninfty$-interpretation, we get $|\s^{\B}|\geq\aleph_{0}$ and then $|\s^{\B}|>|\s^{\A}|$, finishing the proof.

\end{proof}

\subsection{\tp{$\Tbb$}{Tbb}}

\begin{lemma}
    The $\Sigma_{1}$-theory $\Tbb$ does not have a computable minimal model function with respect to its only sort.
\end{lemma}

\begin{proof}
    Suppose $\Tbb$ has a computable minimal model function. We construct a sequence $f:\mathbb{N}\rightarrow\mathbb{N}$ by making $f(0)=f(1)=1$ and, for $n\geq1$, 
    \[f(n+1)\in\minmod_{\Tbb,\{\s\}}(\NNNEQ{x}{f(n)+1}):\]
    of course $f$ is recursively computable since $\minmod_{\Tbb,\{\s\}}$ is itself computable, and building $\NNNEQ{x}{f(n)+1}$ can be done algorithmically once $f(n)$ is known. However, we state that $f(n)\geq \bb(n)$ for all $n\in\mathbb{N}$, that being obvious for $n=0$ and $n=1$; for $n\geq 1$, since $f(n)\geq \bb(n)$, $\NNNEQ{x}{f(n)+1}$ can only be satisfied by a $\Tbb$-interpretation with more than $\bb(n)+1$ elements, and therefore at least $\bb(n+1)$ elements, meaning that $f(n+1)\geq\bb(n+1)$. This is, of course, absurd, given that $\bb$ grows eventually faster than any computable function.
\end{proof}

\subsection{\tp{$\TM$, $\TsM$}{TM, TsM}}

\begin{lemma}\label{TM does not have CMMF}
    The $\Sigma_{s}$-theories $\TM$ and $\TsM$ do not have a computable minimal model function with respect to its only sort.
\end{lemma}

\begin{proof}
    We prove this for $\TM$, but for $\TsM$ the proof is identical. Suppose $\TM$ has a computable minimal model function. We construct functions $h:\mathbb{N}\setminus\{0\}\rightarrow\{0,1\}$ and $h_{0},h_{1}:\mathbb{N}\setminus\{0\}\rightarrow\mathbb{N}$ by making $h(1)=1$, $h_{1}(n)=\sum_{i=1}^{n}h(i)$, $h_{0}(n)=n-h_{1}(n)$ and, assuming $h(n)$ defined,
    \[h(n+1)=\begin{cases}
        0 & \text{if $n+1\in\minmod_{\TM,S}\big(\NNNEQ{x}{h_{0}(n)+1}\wedge\bigwedge_{i=1}^{h_{0}(n)+1}\neg(s(x_{i})=x_{i})\big)$}\\
        1 & \text{if $n+1\in\minmod_{\TM,S}\big(\NNNEQ{x}{h_{1}(n)+1}\wedge\bigwedge_{i=1}^{h_{1}(n)+1}s(x_{i})=x_{i}\big)$}
    \end{cases},\]
    where $S=\{\s\}$ (notice that, from the axiomatization of $\TM$, $n+1$ cannot belong to both sets in this definition). It is clear that $h$ is recursively computable, given that $\minmod_{\TM,S}$ is computable, and as we shall prove $h=f$, what leads to a contradiction. Of course $h(1)=f(1)$, so assume $f$ and $h$ coincide on all values up to and including $n$, while $h_{1}$ coincides with $f_{1}$, and $h_{0}$ with $f_{0}$; if $f(n+1)=1$, $f_{1}(n+1)=f_{1}(n)+1$ and thus the smallest $\TM$-interpretation satisfying $\NNNEQ{x}{h_{1}(n)+1}\wedge\bigwedge_{i=1}^{h_{1}(n)+1}s(x_{i})=x_{i}$ has $n+1$ elements, meaning $h(n+1)=1$. If, instead, $f(n+1)=0$, $f_{0}(n+1)=f_{0}(n)+1$ and again the smallest $\TM$-interpretation that satisfies 
    \[\NNNEQ{x}{h_{0}(n)+1}\wedge\bigwedge_{i=1}^{h_{0}(n)+1}\neg(s(x_{i})=x_{i})\]
    has $n+1$ elements, leading to $h(n+1)=0$, which again coincides with $f$.
\end{proof}

\subsection{\tp{$\Tgeqn$,
$\Ttwothree$,
$\Tinfty$}{Tgeqn, Ttwothree, Tinfty}}

\begin{lemma}
$\Tgeqn$,
$\Ttwothree$,
$\Tinfty$ have computable minimal model functions w.r.t. their set of sorts.
\end{lemma}

\begin{proof}
    Follows from \Cref{SM+ES=>CMMF}.
\end{proof}

\subsection{\tp{$\Tleqone$,
$\Tleqn$, $\Tmn$}{Tleqone, Tleqn, Tmn}}
\begin{lemma}
The $\Sigma_{1}$-theories $\Tleqone$,
$\Tleqn$ and $\Tmn$ have computable minimal model functions w.r.t. $\{\s\}$.
\end{lemma}

\begin{proof}
    Follows from \Cref{-SI+ES=>CMMF}.
\end{proof}

\subsection{\tp{$\Toddneq$, $\Tnequpinfty$}{Toddneq, Tnequpinfty}}

\begin{lemma}\label{Toddneq has CMMF}
    The $\Sigma_{s}$-theory $\Toddneq$ and the $\Sigma_{s}^{2}$-theory $\Tnequpinfty$ have a computable minimal model function with respect to, respectively, $\{\s\}$ and $\{\s,\s_{2}\}$.
\end{lemma}

\begin{proof}
    We use a construction similar to the one found in \Cref{adding non-convexity}: let $\phi$ be a quantifier-free $\Sigma_{s}$-formula, and $\A$ a $\Sigma_{s}$-interpretation. Let $V=\{z_{1},\ldots,z_{n}\}$ be the set of variables in $\phi$, let $M_{i}$ be the maximum of $j$ such that $s^{j}(z_{i})$ appears in $\phi$, and let $U=\{y_{i,j} : 1\leq i\leq n, 0\leq j\leq M_{i}+1\}$ be fresh variables. Then, replace each atomic subformula $s^{j}(z_{i})=s^{q}(z_{p})$ on $\phi$ by $y_{i,j}=y_{p,q}$, so that we obtain a $\Sigma_{1}$-formula we denote by $\phi^{\star}$. We then define the $\Sigma_{s}$-interpretation $\A^{\star}$ by changing the values assigned by $\A$ to the variables in $U$ so that $y_{i,0}^{\A^{\star}}=z_{i}^{\A}$, and $y_{i,j+1}^{\A^{\star}}=s^{\A}(y_{i,j}^{\A^{\star}})$, so that we can finally define the $\Sigma_{1}$-formula $\phi^{\star}_{\A}=\phi^{\star}\wedge\delta_{U}$, where, $\delta_{U}$ is the arrangement on $U$ induced by $\A^{\star}$.

    Now, let $Eq^{\star}(\phi)$ be the set of equivalences on $U$ such that $\phi^{\star}$ and $\delta_{U}^{E}$ are equivalent, and: either $|U/E|=1$; or $(i)$ $y_{i,j}\overline{E}y_{i,j+1}$ for all $1\leq i\leq n$ and $0\leq j\leq M_{i}$, and $(ii)$ $y_{i,j}Ey_{p,q}$ implies $y_{i,j+1}Ey_{p,q+1}$ for all $1\leq i,p\leq n$, and $0\leq j\leq M_{i}$ and $0\leq q\leq M_{p}$. Finally, we state that
    \[\minmod_{\Toddneq,S}(\phi)=\begin{cases}
        \{m\} & \text{if $m=\min\{|U/E| : E\in Eq^{\star}(\phi)\}$ is odd;}\\
        \{m+1\} & \text{if $m$ is even;}\\
        \emptyset & \text{if $m=\aleph_{0}$,}
    \end{cases}\]
    is a computable minimal model function, where $S=\{\s\}$, being certainly at least a computable function.

    We analyse three cases, starting from the assumption that $Eq^{\star}(\phi)$ is not empty (what will imply that $\phi$ is $\Toddneq$-satisfiable): first, if $m=1$, take the $\Sigma_{s}$-interpretation $\A$ with exactly one element in its domain and $s^{\A}$ the identity, what implies $\A$ is also a $\Toddneq$-interpretation; since $\A$ satisfies $\delta_{U}^{E}$, for $E$ the equivalence on $U$ such that $|U/E|=1$, we have that it also satisfies $\phi^{\star}$. Then, by using that $s^{\A}$ is the identity, $\A$ also satisfies $\phi$; since no interpretation can have fewer elements than $\A$, we have that a minimal model function at $\phi$ is indeed $1$.

    Now, suppose $m$ is an odd number greater than $1$, take an equivalence $E\in Eq^{\star}(\phi)$ such that $|U/E|=m$, and define the function $s^{\A}$ on the domain $\s^{\A}=U/E$ such that $s^{\A}([y_{i,j}])=[y_{i,j+1}]$ if $1\leq j\leq M_{i}$, and $s^{\A}([y_{i,M_{i}+1}])$ as any element of $\s^{\A}$ but $[y_{i,M_{i}+1}]$ itself if that was not previously defined, where $[y_{i,j}]$ is the equivalence class with representative $y_{i,j}$. Now, this is a well-defined function: if $[y_{i,j}]=[y_{p,q}]$, for $1\leq i,p\leq n$, $0\leq j\leq M_{i}$ and $0\leq q\leq M_{p}$, then $y_{i,j}Ey_{p,q}$ and, since $|E/U|>1$, $y_{i,j+1}Ey_{p,q+1}$, implying that $s^{\A}([y_{i,j}])=[y_{i,j+1}]=[y_{p,q+1}]=s^{\A}([y_{p,q}])$. Furthermore, $s^{\A}$ is never the identity, that being clear if we look at an equivalence class $[y_{i,M_{i}}]$ not equal to an $[y_{p,q}]$ with $0\leq q\leq M_{p}$; otherwise, $S^{\A}([y_{i,j}])=[y_{i,j+1}]$, and $y_{i,j+1}\overline{E}y_{i,j}$ by hypothesis, meaning that $[y_{i,j+1}]\neq[y_{i,j}]$. Since the interpretation $\A$ has $m$ elements in its domain, it is a $\Toddneq$-interpretation, and as it satisfies $E$, it also satisfies $\delta_{U}^{E}$ and thus $\phi^{\star}$; we then change the value assigned by $\A$ to the variables in $V$ so to obtain an interpretation $\A^{\prime}$ where $z_{i}^{\A^{\prime}}=y_{i,0}^{\A}$, and then $\A^{\prime}$ is a $\Toddneq$-interpretation that satisfies $\phi$ with $|\s^{\A^{\prime}}|=m$. So, let $\B$ be a $\Toddneq$-interpretation that satisfies $\phi$ with $|\s^{\B}|<m$: we change the value assigned by $\B$ to $U$, and obtain the interpretation $\B^{\prime}$ where $y_{i,0}^{\B^{\prime}}=z_{i}^{\B}$ and $y_{i,j+1}^{\B^{\prime}}=s^{\B}(y_{i,j}^{\B^{\prime}})$, so that $\B^{\prime}$ satisfies $\phi^{\star}$. Letting $F$ be the equivalence on $U$ induced by $\B^{\prime}$, we have that $F\in Eq^{\star}(\phi)$ and so $|U/F|<|\s^{\B^{\prime}}|=|\s^{\B}|<m$, leading to a contradiction, $|U/F|<\min\{|U/E| : E\in Eq^{\star}(\phi)\}$.

    Finally, let us consider the case where $m$ is even, and let again $E$ be an equivalence in $Eq^{\star}(\phi)$ such that $|U/E|=m$: we then define an interpretation $\A$ much like in the previous paragraph, but now $\s^{\A}=U/E\cup\{a\}$, where $a\notin U/E$, and $s^{\A}(a)$ can be any element of $U/E$. After changing the value assigned by such an interpretation to the variables in $U$, we obtain a $\Toddneq$-interpretation $\A^{\prime}$ that satisfies $\phi$ with $|\s^{\A^{\prime}}|=m+1$. Finally, assume that there is a $\Toddneq$-interpretation $\B$ that satisfies $\phi$ with $|\s^{\B}|<m+1$: by changing the values assigned by $\B$ to $U$, we obtain the interpretation $\B^{\prime}$ with $y_{i,0}^{\B^{\prime}}=z_{i}^{\B}$ and $y_{i,j+1}^{\B^{\prime}}=s^{\B}(y_{i,j}^{\B^{\prime}})$, meaning that $\B^{\prime}$ satisfies $\phi^{\star}$. Letting, again, $F$ be the equivalence on $U$ induced by $\B^{\prime}$, one has $F\in Eq^{\star}(\phi)$, and thus $|U/F|<|\s^{\B^{\prime}}|=|\s^{\B}|<m+1$: since $\B$ is a $\Toddneq$-interpretation, $|\s^{\B}|$ is odd, and thus $|\s^{\B}|\leq m-1$, we derive that $|U/F|\leq m-1$, and therefore $|U/F|<m$, leading to a contradiction and proving that $\minmod_{\Toddneq,S}(\phi)$ is indeed a minimal model function for all possible cases.

    The proof for $\Tnequpinfty$ is incredibly long, but more or less copies the one for $\Toddneq$, combining to it some aspects of \Cref{Tupinfty has CMMF}.
    
\end{proof}

\subsection{\tp{$\Ttwothreethree$}{Ttwothreethree}}

\begin{lemma}
    The $\Sigma_{3}$-theory $\Ttwothreethree$ has a computable minimal model function with respect to $\{\s,\s_{2},\s_{3}\}$.
\end{lemma}

\begin{proof}
    Follows easily from the fact that the $\Sigma_{2}$-theory $\Ttwothree$ has a computable minimal model function, so that 
    \[\minmod_{\Ttwothreethree,\{\s,\s_{2},\s_{3}\}}(\phi)=\{(m,n,1) : (m,n)\in\minmod_{\Ttwothree,\{\s,\s_{2}\}}(\phi^{\prime})\}\]
    is a minimal model function, where $\phi^{\prime}$ is obtained by replacing all variables of sort $\s_{3}$ in $\phi$ by the same one fresh variable of sort $\s$.
\end{proof}

\subsection{\tp{$\Tmninfty$}{Tmninfty}}

\begin{lemma}
    The $\Sigma_{2}$-theory $\Tmninfty$ has a computable minimal model function with respect to $\{\s,\s_{2}\}$.
\end{lemma}

\begin{proof}
    Take a quantifier-free $\Sigma_{2}$-formula $\phi$, and let $V_{\s}$ and $V_{\s_{2}}$ be, respectively, the variables of sort $\s$ and $\s_{2}$ of $\phi$; let $\eq{\phi}$ be the set of equivalence relations $E$ on $V=V_{\s}\cup V_{\s_{2}}$ which respect sorts such that $\phi$ and $\delta_{V}^{E}$ are equivalent. Finding $\eq{\phi}$ can be reduced to the satisfiability problem in equality logic, and is therefore decidable. We state that
    \[\minmod_{\Tmninfty,S}(\phi)=\begin{cases}
        \{(n,\aleph_{0}),(m,p)\} & \text{if there is $E\in\eq{\phi}$ with $|V_{\s}/E|\leq n$;}\\
        \{(m,p)\} & \text{if there is no such $E$, but $p$ is finite;}\\
        \emptyset & \text{if $p=\aleph_{0}$}
    \end{cases}\]
is a minimal modal function, where $S=\{\s,\s_{2}\}$, we assume without loss of generality $m>n$, and 
\[p=\min\{|V_{\s_{2}}/E| : E\in\eq{\phi}, |V_{\s}/E|\leq m\}.\]
Of course, this is a computable function that always produces a finite output; notice as well that $p=\aleph_{0}$ iff $\eq{\phi}=\emptyset$, what happens in turn iff $\phi$ is not satisfiable.
\begin{enumerate}
    \item Suppose that $\A$ is a minimal $\Tmninfty$-interpretation that satisfies $\phi$, with $(|\s^{\A}|,|\s_{2}^{\A}|)=(s,t)$. If $s=n$, we state that $t=\aleph_{0}$: indeed, if that were not the case, one would have $t>\aleph_{0}$ and, by the many-sorted L{\"o}wenheim-Skolem theorem, we could obtain a $\Tmninfty$-interpretation $\A^{\prime}$ that satisfies $\phi$ with $(|\s^{\A^{\prime}}|,|\s_{2}^{\A^{\prime}}|)=(n,\aleph_{0})<(n,t)$, contradicting the fact that $(s,t)$ is minimal.

    So, suppose that $s=m$: we first show that there is a $\Tmninfty$-interpretation $\B$ that satisfies $\phi$ with $(|\s^{\B}|,|\s_{2}^{\B}|)=(m,p)$, meaning $t\leq p$. Indeed, take an equivalence $E\in\eq{\phi}$ with $|V_{\s}/E|\leq m$ and $|V_{\s_{2}}/E|=p$, and a $\Sigma_{2}$-interpretation $\B$ with: $|\s^{\B}|=m$ (meaning $\B$ is also a $\Tmninfty$-interpretation); $|\s_{2}^{\B}|=p$; $x^{\B}=y^{\B}$ iff $xEy$, for variables $x,y\in V_{\s}$, and arbitrarily for other variables of sort $\s$; and $u^{\B}=v^{\B}$ iff $uEv$, for variables $u,v\in V_{\s_{2}}$, and arbitrarily for other variables of sort $\s_{2}$. Then $\B$ satisfies $\phi$, as we wanted to show, since it satisfies $\delta_{V}^{E}$. Now, assume that $t<p$: let $F$ be the equivalence induced by $\A$, and we have that $|V_{\s}/F|\leq m$ while $|V_{\s}/F|<p$, contradicting the choice of $p$.

    \item For the reciprocal, we first prove that there always is a minimal $\Tmninfty$-interpretation $\A$ satisfying $\phi$ with $(|\s^{\A}|,|\s_{2}^{\A}|)=(m,p)$, and second that if there is $E\in\eq{\phi}$ with $|V_{\s}/E|\leq n$, then there is a minimal $\Tmninfty$-interpretation $\A$ satisfying $\phi$ with $(|\s^{\A}|,|\s_{2}^{\A}|)=(n,\aleph_{0})$, what will finish the proof, always assuming $\phi$ is satisfiable. So, again take an $E\in\eq{\phi}$ with $|V_{\s}/E|\leq m$ and $|V_{\s_{2}}/E|=p$, and let $\A$ be a $\Sigma_{2}$-interpretation with: $|\s^{\A}|=m$ (making of $\A$ a $\Tmninfty$-interpretation); $|\s_{2}^{\A}|=p$; $x^{\A}=y^{\A}$ if and only if $xEy$, for variables $x$ and $y$ in $\phi$ of sort $\s$, and arbitrarily for other variables of sort $\s$; and $u^{\A}=v^{\A}$ if and only if $uEv$, for variables $u$ and $v$ in $\phi$ or sort $\s_{2}$, and arbitrarily for other variables of sort $\s_{2}$. We henceforth have that $\A$ satisfies $\delta_{V}^{E}$ and thus $\phi$, and that $(|\s^{\A}|,|\s_{2}^{\A}|)=(m,p)$. So, assume $\B$ is a $\Tmninfty$-interpretation that satisfies $\phi$ with $(|\s^{\B}|,|\s_{2}^{\B}|)<(m,p)$: if $|\s^{\B}|<m$, we must have $|\s^{\B}|=n$ and $|\s_{2}^{\B}|\geq\aleph_{0}$, leading to a contradiction; if $|\s^{\B}|=m$ and $|\s_{2}^{\B}|<p$, taking the equivalence $F$ induced by $\B$, we have $|V_{\s}/F|\leq m$ and $|V_{\s_{2}}/F|<p$, contradicting our choice of $p$. 

    So, suppose that there is $E\in\eq{\phi}$ with $|V_{\s}/E|\leq E$, and define the $\Sigma_{2}$-interpretation $\A$ with: $|\s^{\A}|=n$; $|\s_{2}^{\A}|=\aleph_{0}$ (so $\A$ is also a $\Tmninfty$-interpretation); $x^{\A}=y^{\A}$ iff $xEy$, for variables $x,y\in V_{\s}$, and arbitrarily for other variables of sort $\s$; and $u^{\A}=v^{\A}$ iff $uEv$, for variables $u,v\in V_{\s_{2}}$, and arbitrarily for other variables of sort $\s_{2}$. We have that $\A$ satisfies $\delta_{V}^{E}$, and thus $\phi$, and that $(|\s^{\A}|,|\s_{2}^{\A}|)=(n,\aleph_{0})$. So, suppose that $\B$ is a $\Tmninfty$-interpretation that satisfies $\phi$ with $(|\s^{\B}|,|\s_{2}^{\B}|)<(n,\aleph_{0})$: then either $|\s^{\B}|<n$, or $|\s^{\B}|=n$ and $|\s_{2}^{\B}|<\aleph_{0}$, both contradicting the fact that $\B$ is a $\Tmninfty$-interpretation and thus finishing the proof.
    
\end{enumerate}

\end{proof}

\subsection{\tp{$\Tupinfty$}{Tupinfty}}

\begin{lemma}\label{Tupinfty has CMMF}
    The $\Sigma_{2}$-theory $\Tupinfty$ has a computable minimal model function with respect to $\{\s,\s_{2}\}$.
\end{lemma}

\begin{proof}
    Take a quantifier-free $\Sigma_{2}$-formula $\phi$, let $V_{\s}$ be its variables of sort $\s$, and $V_{\s_{2}}$ be its variables of sort $\s_{2}$; let $\eq{\phi}$ be the set of equivalence relations $E$ on $V=V_{\s}\cup V_{\s_{2}}$ which respect sorts (meaning that if $x$ is of sort $\s$, and $u$ of sort $\s_{2}$, then $x\overline{E}u$) and such that a $\phi$ and $\delta_{V}^{E}$ are equivalent. Finding such equivalence boils down to the satisfiability problem of equality logic, and thus producing the set $\eq{\phi}$ is computable. We can then state that the following is a minimal model funtion:
    \begin{multline*}
        \minmod_{\Tupinfty,S}(\phi)=\minimal\big[\{(1,|V_{\s_{2}}/E|) : E\in\eq{\phi}, |V_{\s}/E|=1\}\cup\\\{(\max\{|V_{\s}/E|, |V_{\s_{2}}/E|\},\max\{|V_{\s}/E|, |V_{\s_{2}}/E|\}) : E\in\eq{\phi}, |V_{\s}/E|>1\}\big],
    \end{multline*}
    where $S=\{\s,\s_{2}\}$; notice that, if $\eq{\phi}$ is empty (when $\phi$ is not satisfiable), $\minmod_{\Tupinfty,S}(\phi)$ equals $\{(\aleph_{0},\aleph_{0})\}$, the maximum of $\N\times\N$. Since this is computable and the output is finite, we just need to prove that it is indeed a minimal model function; for simplicity, let the right side of the above equation be $\minimal A(\phi)$, and assume going forward that $\phi$ is satisfiable (or, equivalently, that $\phi$ is $\Tupinfty$-satisfiable, since the theory has models infinite in all cardinalities, and is defined over an empty signature).
    \begin{enumerate}
        \item Suppose first that $\A$ is a minimal $\Tupinfty$-interpretation that satisfies $\phi$, with $(|\s^{\A}|,|\s_{2}^{\A}|)=(m,n)$: if $m=1$, and $E$ is the equivalence induced by $\A$, we get $(m,n)\in A(\phi)$; otherwise, $m=n$, and again letting $E$ be the equivalence induced by $\A$ we get $\max\{|V_{\s}/E|,|V_{\s_{2}}/E|\}\leq m$. In this last inequality we state that one actually has equality: indeed, if that were not true, define the $\Tupinfty$-interpretation $\A^{\prime}$ by making $\s^{\A^{\prime}}=V_{\s}^{\A}$, and $\s_{2}^{\A^{\prime}}$ equal to any subset of $\s_{2}^{\A}$ containing $V_{\s_{2}}^{\A}$ of cardinality $|V_{\s}/E|$ (this under the assumption that $|V_{\s}/E|=\max\{|V_{\s}/E|,|V_{\s_{2}}/E|\}$, otherwise we make $\s_{2}^{\A^{\prime}}=V_{\s_{2}}^{\A}$, and $\s^{\A^{\prime}}$ equal to any subset of $\s^{\A}$ that contains $V_{\s}^{\A}$ with cardinality $|V_{\s_{2}}/E|$); we then make $x^{\A^{\prime}}=x^{\A}$ for any variable in $V$, and arbitrary otherwise, and then $\A^{\prime}$ is a $\Tupinfty$-interpretation that satisfies $\phi$ with $(|\s^{\A^{\prime}}|,|\s_{2}^{\A^{\prime}}|)<(m,m)$, contradicting our choice of $(m,m)$. To summarize, again we obtain $(m,n)\in A(\phi)$.

        So, suppose $(m,n)$ is not minimal in $A(\phi)$, meaning there is a $(p,q)$ in $A(\phi)$ such that $(p,q)<(m,n)$. There are then two cases to consider: if $p=1$, let $F\in\eq{\phi}$ be the equivalence such that $(1,|V_{\s_{2}}/F|)=(1,q)$, and we define $\B$ to be a $\Sigma_{2}$-interpretation with $(|\s^{\B}|,|\s_{2}^{\B}|)=(1,q)$ (what will make of $\B$ a $\Tupinfty$-interpretation), and such that $u^{\B}=v^{\B}$ iff $uFv$ for $u,v\in V_{\s_{2}}$; $\B$ then satisfies $\phi$, and the fact that $(|\s^{\B}|,|\s_{2}^{\B}|)<(m,n)$ contradicts our choice of $(m,n)$. If, otherwise, $p>1$, then $p=q$, and let $F$ be an equivalence on $\eq{\phi}$ such that $(\max\{|V_{\s}/F|, |V_{\s_{2}}/F|\},\max\{|V_{\s}/F|, |V_{\s_{2}}/F|\})=(p,q)$: we define a $\Sigma_{2}$-interpretation $\B$ in this case by making $\s^{\B}=\s_{2}^{\B}$ equal any set with $\max\{|V_{\s}/F|, |V_{\s_{2}}/F|\}$, and $x^{\B}=y^{\B}$ iff $xFy$ (and the same for variables of sort $\s_{2}$), so that $\B$ is a $\Tupinfty$-interpretation that satisfies $\phi$ with $(|\s^{\B}|,|\s_{2}^{\B}|)<(m,n)$, again contradicting our choice of $(m,n)$. So in both possible cases we reach a contradiction, meaning $(m,n)$ is indeed minimal and thus lies in $\minimal A(\phi)$.

        \item Reciprocally, suppose that $(m,n)$ is in $\minimal A(\phi)$: we have that wither $m=1$ or $m>1$, and then $m=n$. In the first case, let $E\in\eq{\phi}$ be an equivalence such that $(|V_{\s}/E|,|V_{\s_{2}}/E|)=(1,n)$, take any $\Sigma_{2}$-interpretation $\A$ with $(|\s^{\A}|,|\s_{2}^{\A}|)=(1,n)$ (necessarily a $\Tupinfty$-interpretation as well), and make $u^{\A}=v^{\A}$ iff $uEv$: as we saw before, $\A$ then satisfies $\phi$. In the second case, let $E\in\eq{\phi}$ be an equivalence such that $(\max\{|V_{\s}/E|, |V_{\s_{2}}/E|\},\max\{|V_{\s}/E|, |V_{\s_{2}}/E|\})=(m,m)$, and then define $\B$ to be any $\Sigma_{2}$-interpretation with $(|\s^{\B}|,|\s_{2}^{\B}|)=(m,m)$ (and $\B$ ends up being a $\Tupinfty$-interpretation), such that $x^{\B}=y^{\B}$ (respectively $u^{\B}=v^{\B}$) if and only if $xEy$ ($uEv$); again we obtain $\B$ satisfies $\phi$.

        Finally, suppose that $\B$ is a $\Tupinfty$-interpretation such that $(|\s^{\B}|,|\s_{2}^{\B}|)<(m,n)$: again we proceed by cases. If $|\s^{\B}|=1$, let $\B^{\prime}$ be the $\Tupinfty$-interpretation with $\s^{\B^{\prime}}=\s^{\B}$, $\s_{2}^{\B^{\prime}}=V_{\s_{2}}^{\B}$, $u^{\B^{\prime}}=u^{\B}$ for all $u\in V_{\s_{2}}$, and arbitrarily otherwise; $(|\s^{\B^{\prime}}|,|\s_{2}^{\B^{\prime}}|)$ then equals $(1,|V_{\s_{2}}/F|)$, for $F$ the equivalence induced by $\B$, and thus lies in $A(\phi)$, contradicting the minimality of $(m,n)$. If, instead, $|\s_{2}^{\B}|>1$, and then $|\s^{\B}|=|\s_{2}^{\B}|$, we take the $\Tupinfty$-interpretation $\B^{\prime}$ such that:
        \begin{enumerate}
            \item if $|V_{\s}^{\B}|\geq |V_{\s_{2}}^{\B}|$, $\s^{\B^{\prime}}=V_{\s}^{\B}$, and $\s_{2}^{\B^{\prime}}$ is any subset of $\s_{2}^{\B}$ containing $V_{\s_{2}}^{\B}$ with $|V_{\s}^{\B}|$ elements;
            \item if $|V_{\s_{2}}^{\B}|\geq |V_{\s}^{\B}|$, $\s_{2}^{\B^{\prime}}=V_{\s_{2}}^{\B}$, and $\s^{\B^{\prime}}$ is any subset of $\s^{\B}$ containing $V_{\s}^{\B}$ with $|V_{\s_{2}}^{\B}|$ elements,
        \end{enumerate}
        and this way $(|\s^{\B^{\prime}}|,|\s_{2}^{\B^{\prime}}|)$ equals $(\max\{|V_{\s}/F|,|V_{\s_{2}}/F|\}, \max\{|V_{\s}/F|,|V_{\s_{2}}/F|\})$, for $F$ the equivalence induced by $\B$, and again we contradict the minimality of $(m,n)$, thus finishing the proof.

    \end{enumerate}
\end{proof}

\subsection{\tp{$\Tbbs$,
$\Tbbeq$,
$\Tbbtwo$,
$\Tbbvee$,
$\Tbbveeeq$,
$\Tbbinfty$,
$\Tonebb$,
$\Tbboneneq$,
$\Tbbinftythree$,
$\Tbbneqinfty$,
$\Tnbb$,
$\Tbbneq$,
$\Tmnbb$
}{Tbbs, Tbbeq, Tbbtwo, Tbbvee, Tbbveeeq, Tbbinfty, Tonebb, Tbboneneq, Tbbinftythree, Tbbneqinfty, Tnbb, Tbbneq, Tmnbb}}

\begin{lemma}
$\Tbbs$,
$\Tbbeq$,
$\Tbbtwo$,
$\Tbbvee$,
$\Tbbveeeq$,
$\Tbbinfty$,
$\Tonebb$,
$\Tbboneneq$,
$\Tbbinftythree$,
$\Tbbneqinfty$,
$\Tnbb$,
$\Tbbneq$,
$\Tmnbb$
do not have a computable minimal model function w.r.t. their set of sorts.
\end{lemma}

\begin{proof}
    Follows from \Cref{CMMF+FM=>FW}.
\end{proof}

\subsection{\tp{$\Toneinfty$, $\Ttwoinfty$}{Toneinfty, Ttwoinfty}}

\begin{lemma}\label{Toneinfty has CMMF}
    The $\Sigma_{2}$-theories $\Toneinfty$ and $\Ttwoinfty$ have a computable minimal model function with respect to $\{\s,\s_{2}\}$.
\end{lemma}

\begin{proof}
    We prove the result for $\Toneinfty$, while for $\Ttwoinfty$ the difference lies in the fact that the minimal model function outputting $\{(2,\aleph_{0})\}$ instead is a computable one. We state that
    \[\minmod_{\Toneinfty,S}(\phi)=\{(1,\aleph_{0})\},\]
    where $S=\{\s,\s_{2}\}$, is a computable minimal model function, being obviously computable and always outputting a finite set; so we have only left to prove that this is indeed a minimal model function.
    \begin{enumerate}
        \item Take a $\Toneinfty$-interpretation $\A$ that satisfies $\phi$: by the many-sorted L{\"o}wenheim-Skolem theorem, we know that there is a $\Toneinfty$-interpretation $\A^{\prime}$ that satisfies $\phi$ whose domains are at most countable, and so $(|\s^{\A^{\prime}}|,|\s_{2}^{\A^{\prime}}|)=(1,\aleph_{0})$. Now, suppose $\B$ is a $\Toneinfty$-interpretation that satisfies $\phi$ with $(|\s^{\B}|,|\s_{2}^{\B}|)<(1,\aleph_{0})$: we must of course have $|\s^{\B}|=1$ and $|\s_{2}^{\B}|<\aleph_{0}$, what contradicts the fact that $\B$ is a $\Toneinfty$-interpretation.

        \item Suppose now that $\A$ is a minimal $\Toneinfty$-interpretation $\A$ satisfying $\phi$, with $(|\s^{\A}|,|\s_{2}^{\A}|)=(m,n)$ that satisfies $\phi$; and since $\A$ is a $\Toneinfty$-interpretation, $m=|\s^{\A}|=1$ and $n=|\s_{2}^{\A}|\geq\aleph_{0}$. We state that one actually has $n=\aleph_{0}$, so suppose this is not actually the case: then, again by the many-sorted L{\"o}wenheim-Skolem theorem, there is a $\Toneinfty$-interpretation $\A^{\prime}$ that satisfies $\phi$ whose domains are at most countable, and so $(|\s^{\A^{\prime}}|,|\s_{2}^{\A^{\prime}}|)<(1,n)$, contradicting the fact that $(m,n)$ is minimal, and thus finisihing the proof.
        
    \end{enumerate}
\end{proof}

\subsection{\tp{$\Toneinftyneq$, $\Ttwoinftyneq$}{Toneinftyneq, Ttwoinftyneq}}

\begin{lemma}\label{Toneinftyneq has CMMF}
    The $\Sigma_{s}$-theories $\Toneinftyneq$ and $\Ttwoinftyneq$ have a computable minimal model function with respect to its only sort.
\end{lemma}

\begin{proof}
    We will prove the result for $\Toneinftyneq$, being the one for $\Ttwoinftyneq$ proved in a very similar way. Take a quantifier-free formula $\phi$ and rewrite it in disjunctive normal form, what can be done algorithmically; we then state that 
    \[\minmod_{\Toneinftyneq,S}(\phi)=\begin{cases}
        \{1\} & \text{if there is a cube in $\phi$ with no negated literals;}\\
        \{\aleph_{0}\} & \text{otherwise,}
    \end{cases}\]
    where $S=\{\s\}$, is a minimal model function. Of course this is computable, so we only need to prove that it is a minimal model function.

    Let us write $\phi$ as $\bigvee_{i=1}^{n}\phi_{i}$, where $\phi_{i}=\bigwedge_{j=1}^{m_{i}}\phi_{i}^{j}$ are cubes, and thus $\phi_{i}^{j}$ are literals; starting from the case where there is a cube in $\phi$ with no negated literals, there is an $1\leq i\leq n$ such that $\phi_{i}^{j}$ equals $s^{p_{j}}(x_{j})=s^{q_{j}}(y_{j})$ for all $1\leq j\leq m_{i}$, where $p_{j},q_{j}\in\mathbb{N}$, and $x_{j}$ and $y_{j}$ are variables. It is then easy to see that the trivial $\Toneinftyneq$-interpretation $\A$ satisfies $\phi$: indeed, it satisfies $s^{p_{j}}(x_{j})=s^{q_{j}}(y_{j})$ for all $1\leq j\leq m_{i}$, and thus $\phi_{i}$ and $\phi$. Of course, there cannot be a $\Toneinftyneq$-interpretation $\B$ that satisfies $\phi$ with $|\s^{\B}|<|\s^{\A}|=1$ simply because no interpretation can have empty domains, so we have proved that a minimal $\Toneinftyneq$-interpretation satisfying $\phi$ has cardinality $1$.

    Now, suppose that all cubes of $\phi$ have a negated literal: without loss of generality, for each $1\leq i\leq n$, let $\phi_{i}^{1}$ be $\neg (s^{p_{i}}(x_{i})=s^{q_{i}}(y_{i}))$, for $p_{i},q_{i}\in\mathbb{N}$, and $x_{i}$ and $y_{i}$ variables. Now, let $\A$ be a $\Toneinftyneq$-interpretation that satisfies $\phi$, and we state that we cannot have $|\s^{\A}|=1$: indeed, if that were true there would be a cube $\phi_{i}$ satisfied by $\A$, and thus each $\phi_{i}^{j}$ would be satisfied by $\A$ for $1\leq j\leq m_{i}$; thus $\neg (s^{p_{i}}(x_{i})=s^{q_{i}}(y_{i}))$, equal to $\phi_{i}^{1}$, would be satisfied by $\A$, what is not possible given the domain of $\A$ is assumed to only have one element. Because of the definition of $\Toneinftyneq$, we then must have $|\s^{\A}|\geq \aleph_{0}$, and by the many-sorted L{\"o}wenheim-Skolem theorem there is a $\Toneinftyneq$-interpretation $\A^{\prime}$ that satisfies $\phi$ with at most countable domain. Of course, by what has been just explained, there cannot be a $\Toneinftyneq$-interpretation $\B$ with $|\s^{\B}|<|\s^{\A^{\prime}}|=\aleph_{0}$ that satisfies $\phi$, and so we have finished the proof.
    
    \end{proof}

\section{\tp{Proof of \Cref{thm:addsresult}}{Proof of thm:addsresult}}

    For a $\Sigma_{2}$-interpretation $\C$, define the $\Sigma_{1}$-interpretation $\subs{\C}$ with $\s^{\subs{\C}}=\s^{\C}$, and $x^{\subs{\C}}=x^{\C}$ for all variables $x$ of sort $\s$: as proven in Lemma 12 of \cite{arxivCADE}, $\C$ is a $\adds{\T}$-interpretation iff $\subs{\C}$ is a $\T$-interpretation, and, more generally, $\subs{\C}$ satisfies a $\Sigma_{1}$-formula $\varphi$ iff $\C$ satisfies $\varphi$. In much the same way, we define, given a $\Sigma_{s}^{2}$-interpretation $\C$, the $\Sigma_{s}$-interpretation $\subs{\C}$ by making: $\s^{\subs{\C}}=\s^{\C}$; $s^{\subs{\C}}=s^{\C}$; and $x^{\subs{\C}}=x^{\C}$, for all variables $x$ of sort $\s$.
    \begin{lemma}\label{tech. lemma adds}
    Let $\C$ be a $\Sigma_{s}^{2}$-interpretation: then, for any $\Sigma_{s}$-formula $\varphi$, $\subs{\C}$ satisfies $\varphi$ iff $\C$ satisfies $\varphi$.
\end{lemma}

\begin{proof}
    The proof is by induction on the complexity of $\varphi$, and follows the step of the proof of Lemma 12 in \cite{arxivCADE}.
    
\end{proof}
    
    Now, take a quantifier-free $\Sigma_{2}$-formula $\phi$, the following definitions working just as fine for $\Sigma_{s}^{2}$: let $V_{\s}$ and $V_{\s_{2}}$ be the sets of variables in $\phi$ of sorts, respectively, $\s$ and $\s_{2}$, and $\eq{\phi}$ the set of equivalence relations $E$ on the variables $V$ of $\phi$ which respect sorts, that is, if $xEy$ then $x$ and $y$ are necessarily of the same sort, and such that a $\phi$ and $\delta^{E}_{V}$ are equivalent; of course $\eq{\phi}$ is computable, given that $\phi$ is quantifier-free. For each $E$ in $\eq{\phi}$ let $E_{\s_{2}}$ be the restriction of $E$ to $V_{\s_{2}}$; and, for each $E\in \eq{\phi}$ and $F=E_{\s_{2}}$, we define the formula $\phi_{E}$ by replacing any equality $u=v$ in $\phi$, for $u$ and $v$ of sort $\s_{2}$, by a $\Sigma_{1}$-tautology if $uFv$ (say $x=x$, for a fixed variable $x$ of sort $\s$ not in $\phi$), and by a $\Sigma_{1}$-contradiction if $u\overline{F}v$ (say $\neg(x=x)$); again, producing the formulas $\phi_{E}$ is computable, the problem being essentially one of satisfiability in equality logic. The following is an interesting technical lemma that allows us to "cut and paste" interpretations. 

    \begin{lemma}\label{gluing lemma}
        If the $\Sigma_{2}$-interpretation ($\Sigma_{s}^{2}$-interpretation) $\C$ satisfies the quantifier-free $\Sigma_{2}$-formula ($\Sigma_{s}^{2}$-formula) $\phi$ and has an induced equivalence relation $E$ and $F=E_{\s_{2}}$, a $\Sigma_{1}$-interpretation ($\Sigma_{s}$-interpretation) $\A$ satisfies $\phi_{E}$ iff the $\Sigma_{2}$-interpretation ($\Sigma_{s}^{2}$-interpretation) $\D$ satisfies $\phi$, where: $\subs{\D}=\A$,$\s_{2}^{\D}=\s_{2}^{\C}$, and $u^{\D}=u^{\C}$ for all variables $u$ of sort $\s_{2}$.
    \end{lemma}

    \begin{proof}
        We prove the result by induction on the complexity of $\phi$, and for only the case of $\Sigma_{1}$ and $\Sigma_{2}$, the case of $\Sigma_{s}$ and $\Sigma_{s}^{2}$ being completely analogous.
        \begin{enumerate}
            \item Suppose $\phi$ is atomic, meaning it equals either $x=y$ or $u=v$, for $x$ and $y$ of sort $\s$, and $u$ and $v$ of sort $\s_{2}$; also assume that $\A$ satisfies $\phi_{E}$. In the former case, $\phi_{E}$ equals $x=y$, and from the fact $\A$ satisfies $\phi_{E}$ we get $x^{\A}=y^{\A}$, meaning $x^{\D}=y^{\D}$ and therefore that $\D$ satisfies $\phi$. If it is the latter, since $\C$ satisfies $\phi$ we get $u^{\C}=v^{\C}$, and thus $u^{\D}=v^{\D}$, meaning $\D$ again satisfies $\phi$.

            Reciprocally, assume $\D$ satisfies $\phi$. If $\phi$ is $x=y$, $x^{\D}=y^{\D}$, meaning that $x^{\A}=y^{\A}$ and thus $\A$ satisfies $\phi_{E}$. If $\phi$ is $u=v$, $u^{\C}=v^{\C}$, thus $uEv$, $uFv$ and therefore $\phi_{E}$ is a tautology, meaning $\A$ satisfies that formula.
            
            \item Suppose $\phi=\neg\phi^{1}$, that the result is true for $\phi^{1}$, and that $\A$ satisfies $\phi_{E}$. Since $\phi_{E}=\neg\phi^{1}_{F}$, we have that $\A$ does not satisfy $\phi^{1}_{F}$, and thus $\D$ does not satisfy $\phi^{1}$, meaning it satisfies $\neg\phi^{1}=\phi$. Reciprocally, if $\D$ satisfies $\phi$, it does not satisfy $\phi^{1}$, and thus $\A$ does not satisfy $\phi^{1}_{E}$, meaning it satisfies $\phi$ as we wanted to prove.

            \item Suppose $\phi=\phi^{1}\vee\phi^{2}$, that the result is true for $\phi^{1}$ and $\phi^{2}$, and that $\A$ satisfies $\phi_{E}$. Since $\phi_{E}=\phi^{1}_{F}\vee\phi^{2}_{F}$, $\A$ satisfies either $\phi^{1}_{F}$ or $\phi^{2}_{F}$, thus $\D$ satisfies either $\phi^{1}$ or $\phi^{2}$, and therefore $\phi^{1}\vee\phi^{2}=\phi$. Reciprocally, if $\D$ satisfies $\phi$, it satisfies either $\phi^{1}$ or $\phi^{2}$, and thus $\A$ satisfies either $\phi^{1}_{F}$ or $\phi^{2}_{F}$, meaning it satisfies $\phi^{1}_{F}\vee\phi^{2}_{F}=\phi_{E}$
        \end{enumerate}
    \end{proof}

\addsresult*

\begin{proof}
We will work the case of $\Sigma_{1}$ and $\Sigma_{2}$, the case of $\Sigma_{s}$ and $\Sigma_{s}^{2}$ being essentially the same. Suppose that $\T$ has a computable minimal model function $\minmod_{\T,S_{1}}$, for $S_{1}=\{\s\}$: we state that
\[\minmod_{\adds{\T},S_{2}}(\phi)=\begin{cases}
\minimal\{(m, |V_{\s_{2}}/E|) : E\in\eq{\phi}, m\in \minmod_{\T,S_{1}}(\phi_{E})\} & \text{if $\eq{\phi}\neq\emptyset$;}\\
\emptyset & \text{otherwise}
\end{cases}\]
is a minimal model function where, for simplicity, $S_{2}=\{\s, \s_{2}\}$. Notice that $\phi$ is unsatisfiable iff $\eq{\phi}=\emptyset$, implying $\minmod_{\adds{\T},S_{2}}(\phi)$ is empty; and if $\phi$ is satisfiable but not $\T$-satisfiable (equivalent to $\phi_{E}$ also being non $\T$-satisfiable), $\minmod_{\T,S_{1}}(\phi)$ is empty, making $\minmod_{\adds{\T},S_{2}}(\phi)$ again empty. Since $\minmod_{\T,S_{1}}(\phi_{E})$ is always finite, so is $\minmod_{\adds{\T},S_{2}}(\phi)$, and since finding $\eq{\phi}$ and $\minmod_{\T,S_{1}}(\phi_{E})$ can be done algorithmically, we have $\minmod_{\adds{\T},S_{2}}$ is computable, so we only have left to prove that this is indeed a minimal model function.

\begin{enumerate}
    \item Suppose that $\C$ is a minimal $\adds{\T}$-interpretation that satisfies $\phi$, with $(|\s^{\C}|, |\s_{2}^{\C}|)=(m,n)$; take the equivalence $E$ induced by $\C$ (that is, where $xEy$ iff $x^{\C}=y^{\C}$, and the same for variables of sort $\s_{2}$), and of course $E\in\eq{\phi}$. 
    
    If $|\s_{2}^{\C}|>|V_{\s_{2}}/E|$ we can define an interpretation $\C^{\prime}$ such that $\s^{\C^{\prime}}=\s^{\C}$, $\s_{2}^{\C^{\prime}}=\s_{2}^{\C}$, $x^{\C^{\prime}}=x^{\prime}$ for all variables $x$ of sort $\s$, $u^{\C^{\prime}}=u^{\C}$ for all variables $u\in V_{\s_{2}}$, and $v^{\C^{\prime}}$ equals any value in $\s_{2}^{\C^{\prime}}$ for all other variables $v$ of sort $\s_{2}$. $\C^{\prime}$ is then a $\adds{\T}$-interpretation given that it agrees with $\C$ on the interpretation of the domain and variables of sort $\s$, which additionally satisfies $\phi$ and $(|\s^{\C^{\prime}}|, |\s_{2}^{\C^{\prime}}|)=(m, |V_{\s_{2}}/E|)<(m,n)$, contradicting our choice of $(m,n)$. So $n=|V_{\s_{2}}/E|$.

    If, on the other hand, $|\s^{\C}|>m_{0}$, for $\minmod_{\T,S_{1}}(\phi_{E})=\{m_{0}\}$ (we know this is a singleton since it is not empty, as $\phi_{E}$ is $\T$-satisfiable), take the $\Sigma_{1}$-interpretation $\subs{\C}$: it satisfies $\phi_{E}$, as it is easy to prove, and we have $|\s^{\subs{\C}}|>m_{0}$. Of course, there is a $\T$-interpretation $\A$ that satisfies $\phi_{E}$ with $|\s^{\A}|=m_{0}$; define then the $\adds{\T}$-interpretation $\D$ according to \Cref{gluing lemma}, meaning that $\s^{\D}=\s^{\A}$, $\s_{2}^{\D}=\s_{2}^{\C}$, $x^{\D}=x^{\A}$ for all variables $x$ of sort $\s$, and $u^{\D}=u^{\C}$ for all variables $u$ of sort $\s_{2}$. $\D$ is a $\adds{\T}$-interpretation since $\D_{1}=\A$ is a $\T$-interpretation, and it satisfies $\phi$ since $\C$ satisfies $\phi$ and $\A$ satisfies $\phi_{E}$, but $(|\s^{\D}|,|\s_{2}^{\D}|)=(|\s^{\A}|, n)<(m,n)$, contradicting our choice of $(m,n)$. So $m=m_{0}$, and thus $(m,n)$ is in $\{(m, |V_{\s_{2}}/E|) : E\in\eq{\phi}, m\in \minmod_{\T,S_{1}}(\phi_{E})\}$, meaning we only have to prove now that it is minimal in this set.

    Take a $(p,q)\in\{(m, |V_{\s_{2}}/E|) : E\in\eq{\phi}, m\in \minmod_{\T,S_{1}}(\phi_{E})\}$, say $q=|V_{\s_{2}}|/F|$ and $p\in\minmod_{\T,S_{1}}(\phi_{F})$, and suppose that $(p,q)<(m,n)$: since $p\in\minmod_{\T,S_{1}}(\phi_{F})$, there is a $\T$-interpretation $\B$ that satisfies $\phi_{F}$ with $|\s^{\B}|=p$. We extend $\B$ into a $\adds{\T}$-interpretation $\E$ by adding $q$ elements of sort $\s_{2}$ so that $\E$ satisfies $\phi$, and we reach a contradiction: $(|\s^{\E}|,|\s_{2}^{\E})=(p,q)<(m,n)$, although $\C$ was chosen to be minimal.
    
    \item Reciprocally, suppose that $(m,n)$ is in $\minimal\{(m,|V_{\s_{2}}/E|) : E\in\eq{\phi}, m\in\minmod_{\T,S_{1}}(\phi_{E})\}$. Take a $\T$-interpretation $\A$ such that $|\s^{\A}|=m$ and define the $\Sigma_{2}$-interpretation $\C$ such that: $\s^{\C}=\s^{\A}$, $\s^{\C}$ is any set of cardinality $|V_{\s_{2}}/E|=n$, $x^{\D}=x^{\A}$ for all variables $x$ of sort $\s$, $u^{\C}=v^{\C}$ iff $uEv$ for variables $u$ and $v$ of sort $\s_{2}$ in $\phi$, and arbitrarily otherwise. It is clear that $\C$ is a $\adds{\T}$-interpretation because $\subs{\C}=\A$; and furthermore, it satisfies $\phi$, given that $\A$ satisfies $\phi_{E}$, $\C$ induces the equivalence $E$ on $V_{\s_{2}}$, and \Cref{gluing lemma}. So $(\s^{\C}, \s_{2}^{\C})=(m,n)$.

    Now, assume that $\D$ is another $\adds{\T}$-interpretation that satisfies $\phi$ with $(m,n)\neq (|\s^{\D}|, |\s_{2}^{\D}|)$, but suppose that $(|\s^{\D}|, |\s_{2}^{\D}|)<(m,n)$. Consider the equivalence $F$ induced by $\D$ on the variables of $\phi$, and since $\subs{\D}$ is a $\T$-interpretation that satisfies $\phi_{F}$, we have $p_{0}\leq |\s^{\D}|$, for $p_{0}\in \minmod_{\T,S_{1}}(\phi_{F})$; furthermore, $|V_{\s_{2}}/F|\leq |\s_{2}^{\D}|$, so $(p_{0}, |V_{\s_{2}}/F|)$ is an element of $\{(m,|V_{\s_{2}}/E|) : E\in\eq{\phi}, m\in \minmod_{\T,S_{1}}(\phi_{E})\}$ strictly smaller than $(m,n)$, contradicting the minimality of this element and finishing this part of the proof.
\end{enumerate}

Now suppose that $\adds{\T}$ has a computable minimal model function $\minmod_{\adds{\T},S_{2}}$, and we state that
\[\minmod_{\T,S_{1}}(\phi)=\{\min\{ m : (m,n)\in\minmod_{\adds{\T},S_{2}}(\phi)\}\}\]
is a minimal model function for $\T$: it is certainly computable, given $\minmod_{\adds{\T},S_{2}}(\phi)$ is computable and finite; and, again because $\minmod_{\adds{\T},S_{2}}(\phi)$ is finite, of course $\minmod_{\T,S_{1}}(\phi)$ must be finite as well. Notice that $\minmod_{\T,S_{1}}(\phi)$ will equal $\{\aleph_{0}\}$ if $\minmod_{\adds{\T},S_{2}}(\phi)$ is empty, which in turn happens only if $\phi$ is not $\adds{\T}$-satisfiable (and thus not $\T$-satisfiable); but, for what follows, we assume $\phi$ is $\T$-satisfiable.
\begin{enumerate}
    \item Suppose that $\A$ is a minimal $\T$-interpretation that satisfies $\phi$, with $|\s^{\A}|=m$, and let $\C$ be a $\adds{\T}$-interpretation with $\subs{\C}=\A$ and $|\s_{2}^{\C}|=1$. We state that $(m,1)\in \minmod_{\adds{\T},S_{2}}(\phi)$: indeed, suppose there is a $\adds{\T}$-interpretation $\D$ with $(|\s^{\D}|, |\s_{2}^{\D}|)<(|\s^{\C}|, |\s_{2}^{\D}|)$, and since $1\leq |\s_{2}^{\D}|\leq |\s_{2}^{\C}|=1$, we must have $|\s^{\D}|<|\s^{\C}|$; but then $\subs{\D}$ is a $\T$-interpretation that satisfies $\phi$ with $|\s^{\subs{D}}|=|\s^{\D}|<|\s^{\C}|=m$, what contradicts our choice of $\A$. We have only left to prove that $m$ is a minimum of the tuples $(m,n)$ in $\minmod_{\adds{\T},S_{2}}(\phi)$, but the argument is the same: if $(p,q)$ lies in $\minmod_{\adds{\T},S_{2}}(\phi)$ and has $p<m$, we know first of all that there is a $\adds{\T}$-interpretation $\E$ with $(|\s^{\E}|, |\s_{2}^{\E}|)=(p,q)$; and then, since $\subs{\E}$ is a $\T$-interpretation that satisfies $\phi$, the fact that $|\s^{\E}|<m$ is absurd, and we are done.
    
    \item Reciprocally, suppose that $m$ is the minimum of $\{ m : (m,n)\in\minmod_{\adds{\T},S_{2}}(\phi)\}$ and let $\C$ be a $\adds{\T}$-interpretation with $(|\s^{\C}|,|\s_{2}^{\C}|)$ in $\minmod_{\adds{\T},S_{2}}(\phi)$ and $m=|\s^{\C}|$: of course, we then have that $|\s^{\subs{C}}|=m$. So, suppose $\A$ is a $\T$-interpretation that satisfies $\phi$ with $|\s^{\A}|<m$: construct then the $\adds{\T}$-interpretation $\D$ with $\subs{\D}=\A$, and $|\s_{2}^{\D}|=1$; $\D$ is then a $\adds{\T}$-interpretation that satisfies $\phi$ with $(|\s^{\D}|,|\s_{2}^{\D}|)<(|\s^{\C}|,|\s_{2}^{\C}|)$ (since $|\s_{2}^{\D}|=1\leq |\s_{2}^{\C}|$), contradicting our choice of $\C$ and finishing the proof.
    
\end{enumerate}

\end{proof}

\section{\tp{Proof of \Cref{thm:addfresult}}{Proof of thm:addfresult}}

Given a $\Sigma^{n}_{s}$-formula $\varphi$, $\subff{\varphi}$ is the $\Sigma_{n}$-formula obtained by replacing any term $s(x)$ from $\varphi$ by $x$ until the result is indeed a $\Sigma_{n}$-formula. Analogously, given a $\Sigma_{s}^{n}$-interpretation $\C$, the $\Sigma_{n}$-interpretation $\subf{\C}$ is the interpretation with $\s_{i}^{\subf{\C}}=\s_{i}^{\C}$ for each $1\leq i\leq n$, and $x^{\subf{\C}}=x^{\C}$ for each variable $x$ of sort $\s_{i}$ (again for any $1\leq i\leq n$). And reciprocally, given a $\Sigma_{n}$-interpretation $\A$, we define a $\Sigma_{s}^{n}$-interpretation $\addfs{\A}$ by making $\s_{i}^{\addfs{\A}}=\s_{i}^{\A}$ and $x^{\addfs{\A}}=x^{\A}$, for each variable $x$ of sort $\s_{i}$ and $1\leq i\leq n$; and $s^{\addfs{\A}}(a)=a$ for each $a\in \s_{1}^{\A}$.

In Lemma 13 of \cite{arxivCADE} we have proved that for any $\Sigma_{s}^{n}$-interpretation $\C$ that satisfies $\Forall{x}s(x)=x$, for $x$ of sort $\s_{1}$, $\C$ satisfies a $\Sigma_{s}^{n}$ formula $\varphi$ iff $\subf{\C}$ satisfies $\varphi$. It is furthermore clear that $\subf{\addfs{\A}}=\A$ for all $\Sigma_{n}$-interpretations $\A$, and $\addfs{\subf{\C}}=\C$ for all $\addf{\T}$-interpretations $\C$.

\addfresult*

\begin{proof}
    Suppose first that $\T$ has a computable minimal model function $\minmod_{\T,S}$: we state that
    \[\minmod_{\addf{\T},S}(\phi)=\minmod_{\T,S}(\subff{\phi}),\]
    for $S=\{\s_{1},\cdots,\s_{n}\}$, is a computable minimal model function for $\addf{\T}$, being obviously computable and always outputting a finite set. For what follows, assume $\phi$ is $\addf{\T}$-satisfiable.

    \begin{enumerate}
        \item Take a minimal $\addf{\T}$-interpretation $\C$ satisfying $\phi$, with $(|\s_{1}^{\C}|,\ldots,|\s_{n}^{\C}|)=(n_{1}, \ldots, n_{n})$. Of course, $\subf{\C}$ is a $\T$-interpretation that satisfies $\subff{\varphi}$ with $(|\s_{1}^{\subf{\C}}|,\ldots,|\s_{n}^{\subf{\C}}|)=(n_{1}, \ldots, n_{n})$.

        Now, suppose that $\A$ is a $\T$-interpretation that satisfies $\subff{\phi}$ with $(|\s_{1}^{\A}|,\ldots,|\s_{n}^{\A}|)\neq (n_{1},\ldots,n_{n})$: we then have that $\addfs{\A}$ is a $\addf{\T}$-interpretation that satisfies $\phi$ with $(|\s_{1}^{\addfs{\A}}|,\ldots,|\s_{n}^{\addfs{\A}}|)=(|\s_{1}^{\A}|,\ldots,|\s_{n}^{\A}|)\neq (n_{1},\ldots,n_{n})$, meaning there must exist a $1\leq i\leq n$ such that $n_{i}<|\s_{i}^{\addfs{\A}}|$, and so $n_{i}<|\s_{i}^{\A}|$, proving $(n_{1},\ldots,n_{n})$ lies in $\minmod_{\T,S}(\subff{\phi})$.

        \item Reciprocally, suppose that $(n_{1},\ldots,n_{n})$ lies in $\minmod_{\T,S}(\subff{\phi})$, and take a $\T$-interpretation $\A$ that satisfies $\subff{\phi}$ with $(|\s_{1}^{\A}|,\ldots,|\s_{n}^{\A}|)=(n_{1},\ldots,n_{n})$. Then it is clear that $\addfs{\A}$ is a $\addf{\T}$-interpretation that satisfies $\phi$ with $(|\s_{1}^{\addfs{\A}}|,\ldots,|\s_{n}^{\addfs{\A}}|)=(n_{1},\ldots,n_{n})$.

        So, suppose that $\C$ is a $\addf{\T}$-interpretation that satisfies $\phi$ with $(|\s_{1}^{\C}|,\ldots,|\s_{n}^{\C}|)\neq (n_{1},\ldots,n_{n})$: in that case $\subf{\C}$ is a $\T$-interpretation that satisfies $\subff{\phi}$ with $(|\s_{1}^{\subf{\C}}|,\ldots,|\s_{n}^{\subf{\C}}|)=(|\s_{1}^{\C}|,\ldots,|\s_{n}^{\C}|)\neq (n_{1},\ldots,n_{n})$, and so there must exist $1\leq i\leq n$ such that $n_{i}<|\s_{i}^{\subf{\C}}|$, meaning that $n_{i}<|\s_{i}^{\C}|$, as we wished to show.
        
    \end{enumerate}

Now, for the other direction of the theorem, suppose that $\addf{\T}$ has a computable minimal model function $\minmod_{\addf{\T},S}$, and we state that
\[\minmod_{\T,S}(\phi)=\minmod_{\addf{\T},S}(\phi)\]
is a minimal model function for $\T$, being obviously computable and always having a finite output. For what follows, assume $\phi$ is $\T$-satisfiable.

\begin{enumerate}
    \item Suppose first that $\A$ is a minimal $\T$-interpretation $\A$ that satisfies $\phi$, with $(|\s_{1}^{\A}|,\ldots,|\s_{n}^{\A}|)=(n_{1},\ldots,n_{n})$, and it is obvious that $\addfs{\A}$ is a $\addf{\T}$-interpretation that satisfies $\phi$ with $(|\s_{1}^{\addfs{\A}}|,\ldots,|\s_{n}^{\addfs{\A}}|)=(|\s_{1}^{\A}|,\ldots,|\s_{n}^{\A}|)=(n_{1},\ldots,n_{n})$.

    Take a $\addf{\T}$-interpretation $\C$ which satisfies $\phi$ and that $(|\s_{1}^{\C}|,\ldots,|\s_{n}^{\C}|)\neq (n_{1},\ldots,n_{n})$: given that $\subf{\C}$ is a $\T$-interpretation that satisfies $\phi$ with $(|\s_{1}^{\subf{\C}}|,\ldots,|\s_{n}^{\subf{\C}}|)=(|\s_{1}^{\C}|,\ldots,|\s_{n}^{\C}|)\neq (n_{1},\ldots,n_{n})$, there must exist $1\leq i\leq n$ such that $n_{i}<|\s_{i}^{\subf{\C}}|$, and thus $n_{i}<|\s_{i}^{\C}|$, proving that $(n_{1},\ldots,n_{n})$ is in $\minmod_{\addf{\T},S}(\phi)$.

    \item Reciprocally, take $(n_{1},\ldots,n_{n})$ in $\minmod_{\addf{\T},S}(\phi)$. There is a $\addf{\T}$-interpretation $\C$ that satisfies $\phi$ and $(|\s_{1}^{\C}|,\ldots,|\s_{n}^{\C}|)=(n_{1},\ldots,n_{n})$, and therefore $\subf{\C}$ is a $\T$-interpretation that satisfies $\subff{\phi}=\phi$ and has 
    \[(|\s_{1}^{\subf{\C}}|,\ldots,|\s_{n}^{\subf{\C}}|)=(|\s_{1}^{\C}|,\ldots,|\s_{n}^{\C}|)=(n_{1},\ldots,n_{n}).\]

    Finally, if $\A$ is a $\T$-interpretation that satisfies $\phi$ with $(|\s_{1}^{\A}|,\ldots,|\s_{n}^{\A}|)\neq (n_{1},\ldots,n_{n})$, $\addfs{\A}$ is a $\addf{\T}$-interpretation that satisfies $\phi$ with $(|\s_{1}^{\addfs{\A}}|,\ldots,|\s_{n}^{\addfs{\A}}|=(|\s_{1}^{\A}|,\ldots,|\s_{n}^{\A}|)\neq (n_{1},\ldots,n_{n})$, meaning there exists $1\leq i\leq n$ where $n_{i}<|\s_{i}^{\addfs{\A}}|$. Of course, then $n_{i}<|\s_{i}^{\A}|$, and thus we are done.
    
\end{enumerate}
    
\end{proof}

\section{\tp{Proof of \Cref{thm:addncresult}}{Proof of thm:addncresult}}\label{adding non-convexity}

Given a quantifier-free $\Sigma_{s}^{n}$-formula $\phi$, and a $\Sigma_{s}^{n}$-interpretation $\C$ that satisfies $\phi$ and $\psi_{\vee}$, we wish to define a formula we will represent by $\phi^{\dag}_{\C}$. First, let $\vars_{\s_{1}}(\phi)=\{z_{1},\ldots,z_{n}\}$, let $M_{i}$ be the maximum of $j$ such that $s^{j}(z_{i})$ appears in $\phi$, and let $U=\{y_{i,j} : 1\leq i\leq n, 0\leq j\leq M_{i}+1\}$ be fresh variables. We define $\phi^{\dag}$ to be the $\Sigma_{n}$-formula $\phi$ with every atomic subformula $s^{j}(z_{i})=s^{q}(z_{p})$ replaced by $y_{i,j}=y_{p,q}$. We also define a $\Sigma_{s}^{n}$-interpretation $\C^{\dag}$, which still validates $\psi_{\vee}$, by changing the values assigned to the variables $y_{i,j}$ so that $y_{i,0}^{\C^{\dag}}=z_{i}^{\C}$, and $y_{i,j+1}^{\C^{\dag}}=s^{\C}(y_{i,j}^{\C^{\dag}})$, and it is easy to prove that $\C^{\dag}$ satisfies both $\phi$ and $\phi^{\dag}$. Finally, let $\delta_{U}$ be the arrangement on $U$ induced by $\C^{\dag}$, and we define $\phi^{\dag}_{\C}$ to be $\phi^{\dag}\wedge\delta_{U}$, also satisfied by $\C^{\dag}$.

\addncresult*

\begin{proof}
Let as usual $S$ denote $\{\s_{1},\ldots,\s_{n}\}$. Assume first that $\addnc{\T}$ has a computable minimal model function $\minmod_{\addnc{\T},S}$: we start by stating that
\[\minmod_{\T,S}(\phi)=\minmod_{\addnc{\T},S}(\phi)\]
is a minimal model function for $\addnc{\T}$; of course, it is computable, and has an always finite output.

\begin{enumerate}
    \item First suppose $\A$ is a minimal $\T$-interpretation that satisfies $\phi$, with $(|\s_{1}^{\A}|,\ldots,|\s_{n}^{\A}|)=(n_{1},\ldots,n_{n})$, and then $\addfs{\A}$ is an interpretation that satisfies $\phi$, $\psi_{\vee}$ (thus making of $\A$ a $\addnc{\T}$-interpretation), and with $(|\s_{1}^{\addfs{\A}}|,\ldots,|\s_{n}^{\addfs{\A}}|)=(|\s_{1}^{\A}|,\ldots,|\s_{n}^{\A}|)=(n_{1},\ldots,n_{n})$.

    Take a $\addnc{\T}$-interpretation $\C$ that satisfies $\phi$ with $(|\s_{1}^{\C}|,\ldots,|\s_{n}^{\C}|)\neq (n_{1},\ldots,n_{n})$: since $\subf{\C}$ is a $\T$-interpretation that satisfies $\phi$ with $(|\s_{1}^{\subf{\C}}|,\ldots,|\s_{n}^{\subf{\C}}|)=(|\s_{1}^{\C}|,\ldots,|\s_{n}^{\C}|)\neq (n_{1},\ldots,n_{n})$, there exists $1\leq i\leq n$ such that $n_{i}<|\s_{i}^{\subf{\C}}|$, and therefore $n_{i}<|\s_{i}^{\C}|$, proving $(n_{1},\ldots,n_{n})\in\minmod_{\addnc{\T},S}(\phi)$.

    \item For the reciprocal, take $(n_{1},\ldots,n_{n})\in\minmod_{\addnc{\T},S}(\phi)$. There is a $\addnc{\T}$-interpretation $\C$ that satisfies $\phi$ with $(|\s_{1}^{\C}|,\ldots,|\s_{n}^{\C}|)=(n_{1},\ldots,n_{n})$, and thus $\subf{\C}$ is a $\T$-interpretation that satisfies $\subff{\phi}=\phi$ with $(|\s_{1}^{\subf{\C}}|,\ldots,|\s_{n}^{\subf{\C}}|)=(|\s_{1}^{\C}|,\ldots,|\s_{n}^{\C}|)=(n_{1},\ldots,n_{n})$.

    Then, if $\A$ is a $\T$-interpretation that satisfies $\phi$ with $(|\s_{1}^{\A}|,\ldots,|\s_{n}^{\A}|)\neq (n_{1},\ldots,n_{n})$, $\addfs{\A}$ is a $\addnc{\T}$-interpretation (as it satisfies $\Forall{x}(s(x)=x)$ and thus $\psi_{\vee}$) that satisfies $\phi$ with $(|\s_{1}^{\addfs{\A}}|,\ldots,|\s_{n}^{\addfs{\A}}|=(|\s_{1}^{\A}|,\ldots,|\s_{n}^{\A}|)\neq (n_{1},\ldots,n_{n})$, meaning there is an $1\leq i\leq n$ with $n_{i}<|\s_{i}^{\addfs{\A}}|$. Then $n_{i}<|\s_{i}^{\A}|$, and thus we are done.
    
\end{enumerate}

For the reciprocal, suppose that $\T$ has a computable minimal model function $\minmod_{\T,S}$. Take a quantifier-free formula $\phi$ with variables of sort $\s_{1}$ equal to $\{z_{1},\ldots,z_{n}\}$, and define $M_{i}$ and $U$ as before. Take $Eq^{\dag}(\phi)$ to be the set of equivalence relations on the set $U$ such that:
\begin{enumerate}
    \item[I.] for all $1\leq i\leq n$ and $0\leq j\leq M_{i}$, $y_{i,j+2}Ey_{i,j}$ or $y_{i,j+2}Ey_{i,j+1}$;
    \item[II.] if $y_{i,j}Ey_{p,q}$, and $j<M_{i}+2$ and $q<M_{p}+2$, then $y_{i,j+1}Ey_{p,q+1}$.
\end{enumerate}
Finding $Eq^{\dag}(\phi)$ can be done algorithmically: there is only a finite number of elements in $U$, and thus the number of equivalence relations in that set is finite; we must then test the former condition on the equivalence, what requires only finitely many steps.

Finally, we state that 
\[\minmod_{\addnc{T},S}(\phi)=\minimal\bigcup_{E\in Eq^{\dag}(\phi)}\minmod_{\T,S}(\phi^{\dag}\wedge\delta_{U}^{E})\]
is a minimal model function for $\addnc{\T}$, which is finite and computable given that $Eq^{\dag}(\phi)$ and $\minmod_{\T,S}(\phi^{\dag}\wedge\delta_{U}^{E})$ are both finite and computable. Notice that if $Eq^{\dag}(\phi)$ is empty (meaning $\phi$ is not satisfiable), or if all $\minmod_{\T,S}(\phi^{\dag}\wedge\delta_{U}^{E})$ are empty for $E\in Eq^{\dag}(\phi)$ (meaning $\phi$ is not $\addnc{\T}$-satisfiable), then $\minmod_{\addnc{\T},S}(\phi)$ will equal $\{(\aleph_{0},\ldots,\aleph_{0})\}$, the maximum element of $(\N)^{n}$.

\begin{enumerate}
    \item Start by assuming that $\C$ is a minimal $\addnc{\T}$-interpretation $\C$ that satisfies $\phi$, with $(|\s_{1}^{\C}|,\ldots,|\s_{n}^{\C}|)=(n_{1},\ldots,n_{n})$. Let then $E$ be the equivalence induced by $\C^{\dag}$, and we have that: $E$ satisfies property $I$, since $y_{i,j}^{\C^{\dag}}=s^{\C}(z_{i}^{\C})$ and $\C$ satisfies $\psi_{\vee}$; and $E$ also satisfies property $II$, since $s^{\C}$ is a function. This means $E\in Eq^{\dag}(\phi)$; additionally, of course $\C^{\dag}$ is a $\addnc{\T}$-interpretation that satisfies the $\Sigma_{n}$-formula $\phi^{\dag}_{\C}=\phi^{\dag}\wedge\delta_{U}^{E}$, meaning $\subf{\C^{\dag}}$ is a $\T$-interpretation that satisfies $\phi^{\dag}\wedge\delta_{U}^{E}$. Since $(|\s_{1}^{\C}|,\ldots,|\s_{n}^{\C}|)=(|\s_{1}^{\subf{\C^{\dag}}}|,\ldots,|\s_{n}^{\subf{\C^{\dag}}}|)$, this means $(n_{1},\ldots, n_{n})$ satisfies the first condition to be in $\minmod_{\T,S}(\phi^{\dag}\wedge\delta_{U}^{E})$.

    Now, suppose that $(n_{1},\ldots,n_{n})$ is not a minimal element of the sets $\minmod_{\T,S}(\phi^{\dag}\wedge\delta^{F}_{U})$ for $F\in Eq^{\dag}(\phi)$, meaning there must exist an equivalence $F$ in $Eq^{\dag}(\phi)$ and a tuple $(m_{1},\ldots,m_{n})\in\minmod_{\T,S}(\phi^{\dag}\wedge\delta_{U}^{F})$ such that $(m_{1},\ldots,m_{n})<(n_{1},\ldots,n_{n})$: let $\A$ be a $\T$-interpretation that satisfies $\phi^{\dag}\wedge\delta^{F}_{U}$ with $(|\s_{1}^{\A}|,\ldots,|\s_{n}^{\A}|)=(m_{1},\ldots,m_{n})$. We now define a $\Sigma_{s}^{n}$-interpretation $\A_{F}$: 
    \begin{enumerate}
        \item for each sort $\s_{i}$, we define $\s_{i}^{\A_{F}}=\s_{i}^{\A}$; 
        \item for all variables $x$, we define $x^{\A_{F}}=x^{\A}$; 
        \item for all elements of $\s_{1}^{\A}$ of the form $y_{i,j}^{\A}$, with $j<M_{i}+2$, we define $s^{\A_{F}}(y_{i,j}^{\A})=y_{i,j+1}^{\A}$; 
        \item and for all other elements $a$ of $\s_{1}^{\A}$, we make $s^{\A_{F}}(a)=a$.
    \end{enumerate}
    Now, it is clear that $\s^{\A_{F}}$ restricted to $\s_{1}^{\A}\setminus U^{\A}$ is a function, indeed the identity, so we must analyse it on $U^{\A}$. To start, it is defined on the whole set since $y_{i,M_{i}+2}^{\A}$ equals either $y_{i,M_{i}}$ or $y_{i,M_{i}+1}$, and for these elements $s^{\A_{F}}$ is defined by the above clauses; $s^{\A_{F}}$ is furthermore single-valued, for if $y_{i,j}^{\A}=y_{p,q}^{\A}$, we have $y_{i,j}Fy_{p,q}$ and thus $y_{i,j+1}Fy_{p,q+1}$, meaning that $y_{i,j+1}^{\A}=y_{p,q+1}^{\A}$ and thus that $s^{\A_{F}}(y_{i,j}^{\A})=s^{\A_{F}}(y_{p,q}^{\A})$. 
    
    Furthermore, we have that either $y_{i,j+2}Fy_{i,j}$ or $y_{i,j+2}Fy_{i,j+1}$, and thus either $y_{i,j+2}^{\A}=y_{i,j}^{\A}$ or $y_{i,j+2}^{\A}=y_{i,j+1}^{\A}$, and therefore $s^{\A_{F}}(s^{\A_{F}}(y_{i,j}^{\A}))$ equals either $y_{i,j}^{\A}$ or $s^{\A_{F}}(y_{i,j}^{\A})$; since $s^{\A_{F}}$ is the identity on the elements of $\s_{1}^{\A}\setminus U^{\A}$, we obtain that $\A_{F}$ satisfies $\psi_{\vee}$, and since $(|\s_{1}^{\A_{F}}|,\ldots,|\s_{n}^{\A_{F}}|)=(|\s_{1}^{\A}|,\ldots,|\s_{n}^{\A}|)$, we get $\A_{F}$ is a $\addnc{\T}$-interpretation. Even more, since $\subf{\A_{F}}=\A$, and $\A$ satisfies $\phi^{\dag}\wedge\delta_{U}^{F}$, we get $\A_{F}$ satisfies $\phi^{\dag}\wedge\delta_{U}^{F}$. Now, change the value assigned to the variables $z_{i}$ so that we obtain a $\addnc{\T}$-interpretation $\A_{F}^{\prime}$ where $z_{i}^{\A_{F}^{\prime}}=y_{i,0}^{\A_{F}}$, and we get that $\A_{F}^{\prime}$ satisfies $\phi$: since we still have $(|\s_{1}^{\A_{F}^{\prime}}|,\ldots,|\s_{n}^{\A_{F}^{\prime}}|)=(m_{1},\ldots,m_{n})$, that contradicts the minimality of $(n_{1},\ldots,n_{n})$, and thus we have proved $(n_{1},\ldots,n_{n})$ is a minimal element of $\bigcup_{E\in Eq^{\dag}(\phi)}\minmod_{\T,S}(\phi^{\dag}\wedge\delta^{E}_{U})$.

    \item Reciprocally, suppose 
    \[(n_{1},\ldots,n_{n})\in\minimal\bigcup_{E\in E^{\dag}(\phi)}\minmod_{\T,S}(\phi^{\dag}\wedge\delta^{E}_{U}),\]
    meaning there is an $E\in Eq^{\dag}(\phi)$ and a $\T$-interpretation $\A$ that satisfies $\phi^{\dag}\wedge\delta^{E}_{U}$ with $(|\s_{1}^{\A}|,\ldots,|\s_{n}^{\A}|)=(n_{1},\ldots,n_{n})$. Using the construction we outlined just above, we can obtain a $\addnc{\T}$-interpretation $\A_{E}$ that satisfies $\phi^{\dag}\wedge\delta^{E}_{U}$ with $(|\s_{1}^{\A_{E}}|,\ldots,|\s_{n}^{\A_{E}}|)=(|\s_{1}^{\A}|,\ldots,|\s_{n}^{\A}|)$; changing the value assigned to $z_{i}$ by $\A_{E}$ so that $z_{i}^{\A_{E}^{\prime}}=y_{i,0}^{\A_{E}}$, $\A_{E}^{\prime}$ satisfies $\phi$, and of course $(|\s_{1}^{\A_{E}^{\prime}}|, \ldots , |\s_{n}^{\A_{E}^{\prime}}|)=(n_{1},\ldots,n_{n})$.

    So, suppose $\C$ is a $\addnc{\T}$-interpretation that satisfies $\phi$ such that $(\s_{1}^{\C},\ldots,\s_{n}^{\C})<(n_{1},\ldots,n_{n})$, call this tuple $(m_{1},\ldots,m_{n})$: we define the $\addnc{\T}$-interpretation $\C^{\dag}$ and the equivalence $F$ it induces on $U$, and this way $\C^{\dag}$ satisfies $\phi^{\dag}\wedge\delta^{F}_{U}$, meaning $\subf{\C^{\dag}}$ is a $\T$-interpretation that satisfies $\phi^{\dag}\wedge\delta^{F}_{U}$, and of course $(|\s_{1}^{\subf{\C^{\dag}}}|, \ldots , |\s_{n}^{\subf{\C^{\dag}}}|)=(\s_{1}^{\C},\ldots,\s_{n}^{\C})$. There are then two cases to consider.
    \begin{enumerate}
        \item If there is $(p_{1},\ldots,p_{n})$ in $\minmod_{\T,S}(\phi^{\dag}\wedge\delta^{F}_{U})$ such that $(p_{1},\ldots,p_{n})\leq (m_{1},\ldots,m_{n})$ we are done, since then $(p_{1},\ldots,p_{n})<(n_{1},\ldots,n_{n})$, what contradicts the fact that $(n_{1},\ldots,n_{n})$ is a minimal element.

        \item Otherwise, by \Cref{alternative definition} we have $(m_{1},\ldots,m_{n})$ is already an element of $\minmod_{\T,S}(\phi^{\dag}\wedge\delta^{F}_{U})$, and again we reach a contradiction.
    \end{enumerate}

\end{enumerate}
\end{proof}

\section{\tp{Adding a Sort to $\Sigma_{s}$}{Adding a Sort to Sigmas}}


In \cite{CADE,FroCoS} it was shown that the 
$\adds{\cdot}$
operator defined in \Cref{def:operators}
preserves all properties (disregarding the computability of a minimal model function, that is addressed in \Cref{thm:addsresult}).
However, this was only proved for $\Sigma_{1}$.
Here extend this result to $\Sigma_{1}^{s}$.


\begin{theorem}
    Let $\T$ be a theory over the signature $\Sigma_{s}$. Then $\adds{\T}$ is stably infinite, smooth, finitely witnessable, strongly finitely witnessable, convex, has the finite model property or is stably finite with respect to $\{\s,\s_{2}\}$ if, and only if, $\T$, respectively, is stably infinite, smooth, finitely witnessable, strongly finitely witnessable, convex, has the finite model property or is stably finite with respect to $\{\s\}$.
\end{theorem}

\begin{proof}

The proofs for stable infiniteness, smoothness, convexity, the finite model property and stable finiteness are almost identical to the corresponding proofs in Lemma 3 of \cite{arxivCADE}. 
We now consider (strong) finite witnessability.

    \begin{enumerate}

        \item 
        \begin{enumerate}

        \item Suppose $\T$ is finitely witnessable, with witness $\wit$. 
        We define a function $\wit_{2}:\qf{\Sigma_{2}}\rightarrow\qf{\Sigma_{2}}$ by
        \[\wit_{2}(\phi)=\phi\wedge(w=w)\wedge\bigvee_{E\in Eq(U)}[\wit(\phi_{E})\wedge\delta_{U}^{E}],\]
        where: $w$ is a fresh variable of sort $\s_{2}$; $U=\vars_{\s_{2}}(\phi)$; $Eq(U)$ is the (computable) set of equivalence relations on $U$; $\delta_{U}^{E}$ is the arrangement on $U$ induced by $E$; and $\phi_{E}$ is obtained from $\phi$ by replacing an equality $u=v$, for $u$ and $v$ of sort $\s_{2}$, by a tautology (whose variables are already in $\phi$, we assume for simplicity, or if there are none we use the fresh variable $z$) in $\Sigma_{s}$ if $uEv$, and a contradiction (again in $\Sigma_{s}$, with variables in $\phi$ or equal to $z$) otherwise. It is easy to see that $\wit_{2}$ maps quantifier-free formulas into themselves, and is computable.

        Suppose now that $\C$ is a $\adds{\T}$-interpretation that satisfies $\phi$: let $E$ be the equivalence induced by $\C$ on $U$ (so $\C$ satisfies $\delta_{U}^{E}$), so that $\subs{\C}$ satisfies $\phi_{E}$ (from \Cref{gluing lemma}), and thus $\Exists{\overarrow{y}}\wit(\phi_{E})$ for $\overarrow{y}=\vars(\wit(\phi_{E}))\setminus\vars(\phi_{E})$, meaning $\C$ satisfies $\Exists{\overarrow{y}}\wit(\phi_{E})$. Since $\overarrow{y}$ is contained in $\overarrow{x}=\vars(\wit(\phi))\setminus\vars(\phi)$, $\C$ satisfies $\Exists{\overarrow{x}}\wit(\phi_{E})$, and thus $\Exists{\overarrow{x}}\wit_{2}(\phi)$ (notice neither $\phi$ nor $\delta_{U}^{E}$ contain the variables on $\overarrow{x}$). The reciprocal is obvious.

        Suppose now that $\C$ satisfies $\wit_{2}(\phi)$, and thus $\wit(\phi_{E})$ for some $E\in Eq(U)$: this means $\subs{\C}$ satisfies $\wit(\phi_{E})$ (due to \Cref{tech. lemma adds}), and so there is a $\T$-interpretation $\A$ that satisfies $\wit(\phi_{E})$ with $\s^{\A}=\vars_{\s}(\wit(\phi_{E}))^{\A}$. Let then $\D$ be a $\adds{\T}$-interpretation with $\subs{\D}=\A$, and $\s_{2}^{\D}=U/E$ (unless $U=\emptyset$, when we make $\s_{2}^{\D}$ a singleton), where $u^{\D}$ for $u\in U$ equals the equivalence class of representant $u$ according to $E$, and $w^{\D}$ is set arbitrarily (where we remind the reader that $w$ is fresh and thus not in $U$). $\D$ satisfies $\wit(\phi_{E})$ and, since $\subs{\D}$ satisfies $\phi_{E}$ and $\D$ induces the equivalence $E$ on $U$, also $\phi$ and $\delta_{U}^{E}$; furthermore, $\s^{\D}$ equals $\vars_{\s}(\wit_{2}(\phi))^{\D}$, as the former set equals $\s^{\subs{\D}}=\s^{\A}$ and the latter contains $\vars_{\s}(\wit(\phi_{E}))^{\A}$, and obviously $\s_{2}^{\D}=\vars_{\s_{2}}(\wit_{2}(\phi))^{\D}$.

         \item Suppose $\adds{\T}$ is finitely witnessable, with witness $\wit$, and given a quantifier-free $\Sigma_{s}$-formula $\phi$, we define a function $\wit_{1}$ by making
        \[\wit_{1}(\phi)=\phi\wedge\bigvee_{E\in Eq(U)}\wit(\phi)_{E},\]
         where: $U=\vars_{\s_{2}}(\wit(\phi))$; $Eq(U)$ is the set of equivalence relations on $U$ (which can be found algorithmically); and $\wit(\phi)_{E}$ is obtained from $\wit(\phi)$ by replacing an equality $u=v$, for $u$ and $v$ of sort $\s_{2}$, by a tautology (whose variables are already in $\phi$, for simplicity) in $\Sigma_{s}$ if $uEv$, and a contradiction (again in $\Sigma_{s}$, with variables in $\phi$) otherwise. Of course $\wit_{1}$ maps quantifier-free formulas into themselves, and is computable.

         First, suppose a $\T$-interpretation $\A$ satisfies $\phi$, and take a $\adds{\T}$-interpretation $\C$ with $\subs{\C}=\A$: $\C$ then satisfies $\phi$ (by \Cref{tech. lemma adds}), and thus $\Exists{\overarrow{y}}\wit(\phi)$ for $\overarrow{y}=\vars(\wit(\phi))\setminus\vars(\phi)$; this way, some interpretation $\C^{\prime}$, differing from $\C$ at most on the value assigned to $\overarrow{y}$, satisfies $\wit(\phi)$ (and thus $\Exists{\overarrow{x}}\wit(\phi)$ and $\phi$). If $E$ is the equivalence induced by $\C^{\prime}$ on $U$, it is clear that $\subs{\C^{\prime}}$ satisfies $\wit(\phi)_{E}$ (again by \Cref{tech. lemma adds}), as well as $\phi$; since $\subs{\C^{\prime}}$ differs from $\A$ at most on the value assigned to the variables in $\overarrow{y}$ of sort $\s$, what equals $\overarrow{x}=\vars(\wit_{1}(\phi))\setminus\vars(\phi)$, $\A$ satisfies $\Exists{\overarrow{x}}\wit(\phi)_{E}$ and $\phi$, and thus $\Exists{\overarrow{x}}\wit_{1}(\phi)$. The reciprocal is obvious, as $\phi$ has none of the variables in $\overarrow{x}$.

         Now, suppose that $\T$-interpretation $\A$ satisfies $\wit_{1}(\phi)$, and thus $\wit(\phi)_{E}$ for some $E\in Eq(U)$: take then a $\adds{\T}$-interpretation $\C$ with $\subs{\C}=\A$ and that induces the equivalence $E$ on $U$, so $\C$ satisfies $\wit(\phi)$. There is then a $\adds{\T}$-interpretation $\D$ that satisfies $\wit(\phi)$ (and thus $\Exists{\overarrow{x}}\wit(\phi)$ and $\phi$) with $\s^{\D}=\vars_{\s}(\wit(\phi))^{\D}$ and $\s_{2}^{\D}=\vars_{\s_{2}}(\wit(\phi))^{\D}$: if $F$ is the equivalence induced by $\D$ on $U$, we have that $\subs{\D}$ satisfies $\phi$ and $\wit(\phi)_{F}$, and therefore $\wit_{1}(\phi)$, and since $\vars_{\s}(\wit_{1}(\phi))\supseteq\vars_{\s}(\wit(\phi))$ we get $\s^{\subs{\D}}=\vars_{\s}(\wit_{1}(\phi))^{\subs{\D}}$, proving $\wit_{1}$ is a witness.

         \end{enumerate}

        \item \begin{enumerate}

        \item Suppose $\T$ is now strongly finitely witnessable, with strong witness $\wit$, and given a quantifier-free $\Sigma_{s}^{2}$-formula $\phi$, and we consider the same function $\wit_{2}$ from item $1(a)$ above. We do not need to prove that $\wit_{2}$ is computable, and that $\phi$ and $\Exists{\overarrow{x}}\wit_{2}(\phi)$ are $\adds{\T}$-equivalent as this is done on item $1(a)$. 
        
        So assume $V$ is a set of variables, $\delta_{V}$ an arrangement on $V$, and $\C$ a $\adds{\T}$-interpretation that satisfies $\wit_{2}(\phi)\wedge\delta_{V}$: let $V_{1}$ be the variables of sort $\s$ in $V$, $V_{2}$ the ones with sort $\s_{2}$, and write $\delta_{V}=\delta_{V_{1}}\wedge\delta_{V_{2}}$ in the obvious way. $\C$ satisfies $\wit(\phi_{E})\wedge\delta_{U}^{E}$ for some $E\in Eq(U)$, and thus $\subs{\C}$ satisfies $\wit(\phi_{E})$ and $\delta_{V_{1}}$, so there is a $\T$-interpretation $\A$ that satisfies $\wit(\phi_{E})\wedge\delta_{V_{1}}$ with $\s^{\A}=\vars_{\s}(\wit(\phi_{E})\wedge\delta_{V_{1}})^{\A}$. We define a $\adds{\T}$-interpretation $\D$ by making: $\subs{\D}=\A$; $\s_{2}^{\D}$ equal to $U\cup V_{2}$ modulo the equivalence relative to $\delta_{V_{2}}\wedge\delta_{U}^{E}$ (this is possible as the formula $\delta_{V_{2}}\wedge\delta_{U}^{E}$ is not contradictory, being satisfied by $\C$), unless $U\cup V_{2}$ is empty when we make $\s_{2}^{\D}$ a singleton; $u^{\D}$ equal to the representant class of $u\in U\cup V_{2}$, and $w^{\D}$ arbitrary. Since $\A$ satisfies $\wit(\phi_{E})\wedge\delta_{V_{1}}$, so does $\D$, and $\D$ satisfies $\delta_{V_{2}}\wedge\delta_{U}^{E}$ by design, meaning it satisfies $\wit_{2}(\phi)$. That $\s^{\D}=\vars_{\s}(\wit_{2}(\phi)\wedge\delta_{V})^{\D}$ and $\s_{2}^{\D}=\vars_{\s_{2}}(\wit_{2}(\phi)\wedge\delta_{V})^{\D}$ follows from the definition of $\D$ and choice of $\A$.

        \item Suppose $\adds{\T}$ is now strongly finitely witnessable, with strong witness $\wit$, and given a quantifier-free $\Sigma_{s}$-formula $\phi$, and we consider the same function $\wit_{1}$ from item $1(b)$ above. We already know that $\wit_{1}$ is computable and that, for $\overarrow{x}=\vars(\wit_{1}(\phi))\setminus\vars(\phi)$, $\phi$ and $\Exists{\overarrow{x}}\wit_{1}(\phi)$ are $\T$-equivalent. 
        
        So, take a set $V$ of variables of sort $\s$, let $\delta_{V}$ be an arrangement on $V$, and $\A$ a $\T$-interpretation that satisfies $\wit_{1}(\phi)\wedge\delta_{V}$: there is an equivalence $E$ on $U$ such that $\A$ satisfies $\wit(\phi)_{E}\wedge\delta_{V}$, so take a $\adds{\T}$-interpretation $\C$ with $\subs{\C}=\A$ and that induces the equivalence $E$ on $U$, so that $\C$ satisfies $\wit(\phi)\wedge\delta_{V}$. There is then a $\adds{\T}$-interpretation $\D$ that satisfies $\wit(\phi)\wedge\delta_{V}$ (and thus $\Exists{\overarrow{x}}\wit(\phi)$ and $\phi$) with $\s^{\D}=\vars_{\s}(\wit(\phi)\wedge\delta_{V})^{\D}$ and $\s_{2}^{\D}=\vars_{\s_{2}}(\wit(\phi)\wedge\delta_{V})^{\D}$: $\subs{\D}$ is then a $\T$-interpretation that satisfies $\phi$, $\wit(\phi)_{F}\wedge\delta_{V}$ (and thus $\wit_{1}(\phi)\wedge\delta_{V}$), for $F$ the equivalence induced by $\D$ on $U$, with $\s^{\subs{\D}}=\vars_{\s}(\wit_{1}(\phi)\wedge\delta_{V})^{\subs{\D}}$, proving $\wit_{1}$ is a strong witness.

        \end{enumerate}

    \end{enumerate}
\end{proof}

\end{document}